\documentclass[a4paper,twocolumn,superscriptaddress,11pt,accepted=2017-11-27]{quantumarticle}
\usepackage[utf8]{inputenc}
\usepackage[english]{babel}
\usepackage[T1]{fontenc}
\usepackage{hyperref}

\usepackage{amsmath,amsthm,latexsym,amssymb,amsfonts,epsfig,psfrag,color,url,dsfont}
\usepackage{bbold}
\usepackage{graphicx}
\usepackage{caption}
\usepackage{float}
\usepackage{subcaption}
\usepackage{sidecap}
\usepackage{wrapfig}
\usepackage{multirow}
\usepackage{cite}
\usepackage{tikz}
\usepackage{lipsum}

\DeclareMathOperator{\Tr}{Tr}

\newtheorem{Theorem}{Theorem}[section]
\newtheorem{Definition}[Theorem]{Definition}
\newtheorem{lem}[Theorem]{Lemma}
\newtheorem{Corollary}[Theorem]{Corollary}
\newtheorem{Assumption}{Assumption}

\newtheorem{Req}{Requirement}
\newtheorem*{R}{Remark}

\newtheorem{pr}{Rule}
\newtheorem*{claim}{Claim}
\def\be{\begin{equation}}
\def\ee{\end{equation}}
\def\ba{\begin{eqnarray}}
\def\ea{\end{eqnarray}}

\newcommand\nn{\nonumber}
\newcommand\q{\quad}
\def\Nl{{\mathchoice
{\setbox0=\hbox{$\displaystyle\rm N$}\hbox{\hbox to0pt
{\kern0.4\wd0\vrule height0.9\ht0\hss}\box0}}
{\setbox0=\hbox{$\textstyle\rm N$}\hbox{\hbox to0pt
{\kern0.4\wd0\vrule height0.9\ht0\hss}\box0}}
{\setbox0=\hbox{$\scriptstyle\rm N$}\hbox{\hbox to0pt
{\kern0.4\wd0\vrule height0.9\ht0\hss}\box0}}
{\setbox0=\hbox{$\scriptscriptstyle\rm N$}\hbox{\hbox to0pt
{\kern0.4\wd0\vrule height0.9\ht0\hss}\box0}}}}
\def\Zl{{\mathchoice
{\setbox0=\hbox{$\displaystyle\rm Z$}\hbox{\hbox to0pt
{\kern0.4\wd0\vrule height0.9\ht0\hss}\box0}}
{\setbox0=\hbox{$\textstyle\rm Z$}\hbox{\hbox to0pt
{\kern0.4\wd0\vrule height0.9\ht0\hss}\box0}}
{\setbox0=\hbox{$\scriptstyle\rm Z$}\hbox{\hbox to0pt
{\kern0.4\wd0\vrule height0.9\ht0\hss}\box0}}
{\setbox0=\hbox{$\scriptscriptstyle\rm Z$}\hbox{\hbox to0pt
{\kern0.4\wd0\vrule height0.9\ht0\hss}\box0}}}}
\def\Ql{{\mathchoice
{\setbox0=\hbox{$\displaystyle\rm Q$}\hbox{\hbox to0pt
{\kern0.4\wd0\vrule height0.9\ht0\hss}\box0}}
{\setbox0=\hbox{$\textstyle\rm Q$}\hbox{\hbox to0pt
{\kern0.4\wd0\vrule height0.9\ht0\hss}\box0}}
{\setbox0=\hbox{$\scriptstyle\rm Q$}\hbox{\hbox to0pt
{\kern0.4\wd0\vrule height0.9\ht0\hss}\box0}}
{\setbox0=\hbox{$\scriptscriptstyle\rm Q$}\hbox{\hbox to0pt
{\kern0.4\wd0\vrule height0.9\ht0\hss}\box0}}}}
\def\Rl{{\mathchoice
{\setbox0=\hbox{$\displaystyle\rm R$}\hbox{\hbox to0pt
{\kern0.4\wd0\vrule height0.9\ht0\hss}\box0}}
{\setbox0=\hbox{$\textstyle\rm R$}\hbox{\hbox to0pt
{\kern0.4\wd0\vrule height0.9\ht0\hss}\box0}}
{\setbox0=\hbox{$\scriptstyle\rm R$}\hbox{\hbox to0pt
{\kern0.4\wd0\vrule height0.9\ht0\hss}\box0}}
{\setbox0=\hbox{$\scriptscriptstyle\rm R$}\hbox{\hbox to0pt
{\kern0.4\wd0\vrule height0.9\ht0\hss}\box0}}}}
\def\Cl{{\mathchoice
{\setbox0=\hbox{$\displaystyle\rm C$}\hbox{\hbox to0pt
{\kern0.4\wd0\vrule height0.9\ht0\hss}\box0}}
{\setbox0=\hbox{$\textstyle\rm C$}\hbox{\hbox to0pt
{\kern0.4\wd0\vrule height0.9\ht0\hss}\box0}}
{\setbox0=\hbox{$\scriptstyle\rm C$}\hbox{\hbox to0pt
{\kern0.4\wd0\vrule height0.9\ht0\hss}\box0}}
{\setbox0=\hbox{$\scriptscriptstyle\rm C$}\hbox{\hbox to0pt
{\kern0.4\wd0\vrule height0.9\ht0\hss}\box0}}}}
\def\Hl{{\mathchoice
{\setbox0=\hbox{$\displaystyle\rm H$}\hbox{\hbox to0pt
{\kern0.4\wd0\vrule height0.9\ht0\hss}\box0}}
{\setbox0=\hbox{$\textstyle\rm H$}\hbox{\hbox to0pt
{\kern0.4\wd0\vrule height0.9\ht0\hss}\box0}}
{\setbox0=\hbox{$\scriptstyle\rm H$}\hbox{\hbox to0pt
{\kern0.4\wd0\vrule height0.9\ht0\hss}\box0}}
{\setbox0=\hbox{$\scriptscriptstyle\rm H$}\hbox{\hbox to0pt
{\kern0.4\wd0\vrule height0.9\ht0\hss}\box0}}}}
\def\Ol{{\mathchoice
{\setbox0=\hbox{$\displaystyle\rm O$}\hbox{\hbox to0pt
{\kern0.4\wd0\vrule height0.9\ht0\hss}\box0}}
{\setbox0=\hbox{$\textstyle\rm O$}\hbox{\hbox to0pt
{\kern0.4\wd0\vrule height0.9\ht0\hss}\box0}}
{\setbox0=\hbox{$\scriptstyle\rm O$}\hbox{\hbox to0pt
{\kern0.4\wd0\vrule height0.9\ht0\hss}\box0}}
{\setbox0=\hbox{$\scriptscriptstyle\rm O$}\hbox{\hbox to0pt
{\kern0.4\wd0\vrule height0.9\ht0\hss}\box0}}}}
\newcommand{\cl}{\mathcal L}

\newcommand{\cq}{\mathcal Q}

\newcommand{\ct}{\mathcal T}

\def\nn{\nonumber}

\newcommand{\eqa}{\begin{eqnarray}}
\newcommand{\neqa}{\end{eqnarray}}

\newcommand{\1}{$\{10j\}$}

\newcommand{\n}{\nabla}

\def\f{\frac}

\usepackage{bbm}

\newcommand{\lalg}[1]{\mathfrak{#1}}  

\newcommand{\SO}{\mathrm{SO}}

\renewcommand{\sl}{\lalg{sl}}

\def\q{{\quad}}

\begin{document}

\title{Toolbox for reconstructing quantum theory from rules on information acquisition}
\date{\today}
\author{Philipp Andres H\"ohn}
\email{hoephil@gmail.com}
\orcid{0000-0001-7108-1184}
\thanks{Present address: Institute for Quantum Optics \& Quantum Information, Austrian Academy of Sciences; and Vienna Center for Quantum Science and Technology, University of Vienna, Boltzmanngasse 3, 1090 Vienna, Austria}
\affiliation{Perimeter Institute for Theoretical Physics, 31 Caroline Street North, Waterloo, Ontario, Canada N2L 2Y5}

\maketitle

\begin{abstract}
We develop an operational approach for reconstructing the quantum theory of qubit systems from elementary rules on information acquisition. The focus lies on an observer $O$ interrogating a system $S$ with binary questions and $S$'s state is taken as $O$'s `catalogue of knowledge' about $S$. The mathematical tools of the framework are simple and we attempt to highlight all underlying assumptions. Four rules are imposed, asserting (1) a limit on the amount of information available to $O$; (2) the mere existence of complementary information; (3) $O$'s total amount of information to be preserved in-between interrogations; and, (4) $O$'s `catalogue of knowledge' to change continuously in time in-between interrogations and every consistent such evolution to be possible. This approach permits a {\it constructive} derivation of quantum theory, elucidating how the ensuing independence, complementarity and compatibility structure of $O$'s questions matches that of projective measurements in quantum theory, how entanglement and monogamy of entanglement, non-locality and, more generally, how the correlation structure of arbitrarily many qubits and rebits arises. The rules yield a reversible time evolution and a quadratic measure, quantifying $O$'s information about $S$. Finally, it is shown that the four rules admit two solutions for the simplest case of a single elementary system: the Bloch ball and disc as state spaces for a qubit and rebit, respectively, together with their symmetries as time evolution groups. The reconstruction for arbitrarily many qubits is completed in a companion paper \cite{hw} where an additional rule eliminates the rebit case. This approach is inspired by (but does not rely on) the relational interpretation and yields a novel formulation of quantum theory in terms of questions. 
\end{abstract}

\section{Introduction}

Tools and concepts from information theory have seen an ever growing number of applications in modern physics, often proving useful for understanding and interpreting specific physical phenomena. Among a vast number of examples, black hole entropy, or more generally space-time horizon entropies, can be understood in terms of entanglement entropy \cite{Sorkin:2014kta,Bombelli:1986rw,Jacobson:1995ab,Jacobson:2012yt,Bianchi:2012ev}, thermodynamics naturally adheres to entropic and thus informational perspectives \cite{jaynes,bennett1982thermodynamics,maruyama2009colloquium,popescu2006entanglement,horodecki2013fundamental}, and, above all, the entire field of quantum information and computation is the natural physical arena for applications of information theoretic tools \cite{nielsen2010quantum}. 

This manuscript, by contrast, is motivated by the question whether information theoretic concepts, apart from their useful applications to concrete physical situations, can also tell us something deeper about physics, namely about the physical content and architecture of theories. The overriding idea is that elementary rules or restrictions of certain informational activities, e.g.\ information acquisition or communication, should be deeply intertwined with the structure of the appropriate theory. In this article, we shall address this question by means of the concrete example of quantum theory. Our ambition is to develop a novel informational framework for deriving the formalism and structure of quantum theory for systems of arbitrarily many qubits from elementary operational postulates -- a task which is completed in the companion paper \cite{hw}. While neither this question nor the fact that one can reconstruct quantum theory from elementary axioms is new and has been extensively explored before in various contexts \cite{Hardy:2001jk,Dakic:2009bh,masanes2011derivation,Mueller:2012ai,Masanes:2012uq,chiribella2011informational,de2012deriving,Mueller:2012pc,Hardy:2013fk,Barnum:2014fk,kochen2013reconstruction,Oeckl:2014uq}, we shall approach both from a novel constructive perspective and with a stronger emphasis on the conceptual content of the theory. The ultimate goal of this work is therefore very rudimentary: to redo a well established theory -- albeit in a novel way which is especially engineered for exposing its informational and logical structure, physical content and distinctive phenomena more clearly. In other words, we shall attempt to rebuild quantum theory for qubit systems from scratch.

In such an information based context it is natural to follow an operational approach, describing physics from the perspective of an observer. Accordingly, we shall work under the premise that we may only speak about the information an observer has access to in an experiment. Our approach will thus be purely operational and epistemic (i.e.\ knowledge based) by construction and shall survive without ontic statements (i.e.\ references to `reality'). We shall thus say nothing about whether or not `hidden variables' could give rise to the experiences of the observer. Under these circumstances we adopt an `inside view of physics', holding properties of systems as being {\it relationally}, rather than absolutely defined. 

Indeed, more generally the replacement of absolute by relational concepts goes in hand with the establishment of {\it universal (i.e.\ observer independent) limits}. For instance, the crucial step from Galilean to special relativity is the realization that the speed of light $c$ constitutes a {\it universal limit} for information communication among observers. The fact that all observers agree on this limit is the origin of the relativity of space and time. Similarly, the crucial step from classical to quantum mechanics is the recognition that the Planck constant $\hbar$ establishes a {\it universal limit} on how much simultaneous information is accessible to an observer. While less explicit than in the case of special relativity, this simple observation suggests a relational character of a system's quantum properties. More precisely, the {\it process of information acquisition through measurement establishes an informational relation between the observer and system}. Only if there was no limit on the acquisition of information would it make sense to speak about an absolute state of a system within a purely operational approach (unless one accepts the existence of an omniscient and absolute observer as an external standard). But thanks to the existence of complementarity, implied by $\hbar$, an observer may not access all conceivable properties of the system at once. Furthermore, the observer can choose the experimental setting and thereby which property of the system she would like to reveal (although, clearly, she cannot choose the experimental outcome). Under our purely operational premise, we shall treat the situation as if the system does not have any other properties than those accessible to the observer at any moment of time. 
In particular, the system's state is naturally interpreted as representing the observer's state of information about the system. These ideas are in agreement with earlier proposals in the literature \cite{hartle1968quantum,zheng1996quantum} and, most specifically, with the {\it relational interpretation of quantum mechanics} \cite{Rovelli:1995fv,Smerlak:2006gi}.

Of course, in order for different observers who may communicate (by physical interaction) to have a basis for agreeing on the description of a system, some of its attributes must be observer independent such as its state space, the set of possible measurements on it and possibly a limit on its information content. But without adhering to an external standard against which measurement outcomes and states could be defined, it is as meaningless to assert a system's physical state to be independent of its relations to other systems as it is to relate a system's dynamics to an absolute Newtonian background time. 

It is worthwhile to investigate what we can learn about physics from such an operational and informational approach. For this endeavour we shall adopt the general conviction, which has been voiced in many different (even conflicting) ways before in the literature \cite{hartle1968quantum,Rovelli:1995fv,peres1995quantum,zeilinger1999foundational,Brukner:1999qf,Brukner:vn,Brukner:ys,Brukner:2002kx,Fuchs:fk,Caves:1996nq,Caves:2002uq,caves2002unknown,spekkens2007evidence,Spekkens:2014fk}, that quantum theory is best understood as an operational framework governing an observer's acquisition of information about a system. While most earlier works take quantum theory as given and attempt to characterize and interpret its physical content with an emphasis on information inference, here and in \cite{hw} we take a step back and show that one can actually derive quantum theory from this perspective. This will require a focus on the informational relation between an observer and a system and the rules governing the observer's acquisition of information. More precisely, our approach will be formulated in terms of the observer interrogating a system with elementary questions.

While the present work has been inspired by relational ideas, we emphasize that the sequel does {\it not} actually rely on them so this should not discourage a reader unsympathetic with relational interpretations of quantum theory. Our approach will reconstruct and produce a novel formulation of quantum theory, but will clearly not single out the relational interpretation as `the right one'. However, it will support a partial interpretation, namely that quantum theory is a law book governing an observer's acquisition of information about physical systems.

This is clearly not the only physical situation to which such a relational approach applies; it likewise opens up a novel perspective on elementary space-time structure which is encoded in the informational relations among different observers and can be exposed by a communication game.  For instance, without presupposing a particular space-time structure -- and thus without assuming an externally given transformation group between different reference frames -- one can also derive the Lorentz group as the minimal group translating between different observer's descriptions of physics from their informational relations, established by communication with quantum systems \cite{Hoehn:2014vua}.

More fundamentally, relational ideas are actually required and commonly employed in the context of background independent quantum gravity approaches where the notion of coordinates disappears together with a classical notion of space-time within which a dynamics could be defined. Instead, one has to resort to dynamical degrees of freedom to define physical reference frames (i.e., dynamical `rods' and `clocks') relative to which a meaningful dynamics can be formulated in the first place. This constitutes the relational paradigm of dynamics \cite{Rovelli:2004tv,Rovelli:1989jn,Rovelli:1990ph,Rovelli:1990pi,Dittrich:2005kc,Tambornino:2011vg,Bojowald:2010qw,Bojowald:2010xp,Hohn:2011us} but goes beyond a purely informational and operational approach and thus clearly beyond the scope of this work. 

The remainder of this manuscript is organized as follows.

\onecolumngrid

\tableofcontents

 \vspace*{.4cm}
\twocolumngrid

%
 
The first part of the article up to and including section \ref{sec_postulates} includes a substantial amount of conceptual elaborations. The second part, by contrast, will become more technical upon putting the novel postulates to use in sections \ref{sec_qstructure}-\ref{sec_n1}. The reconstruction for arbitrarily many qubits is performed in the companion article \cite{hw} where an additional postulate eliminates rebits (two-level systems over real Hilbert spaces) in favour of qubit quantum theory. The rebit case is considered separately in \cite{hw2}.

\section{Why a(nother) reconstruction of quantum theory?}\label{sec_motivation}

Given that we have a beautifully working theory, one may wonder why one should bother to reconstruct quantum theory from operational statements. There are various motivations for this endeavour:
\begin{enumerate}
\item To equip the standard, physically obscure textbook axioms for quantum theory with an operational sense. In particular, in addition to its empirical success, a derivation from operational statements can conceptually justify the formulation of the theory in terms of Hilbert spaces, complex numbers, tensor product rule for composite systems, etc.
\item To better understand quantum theory within a larger context. By singling out quantum theory with operational statements one can answer the question ``what makes quantum theory special?'', thereby establishing a bird's-eye perspective on the formalism and conceivable alternatives.
\item It may help to understand why or why not quantum theory in its present form should be fundamental and thus why it should or should not be modified in view of attempting to construct fundamental theories. By dropping or modifying some of its defining physical principles, one obtains a handle for systematic generalizations of quantum theory. This may also be interesting in view of quantum gravity phenomenology (away from the deep quantum regime). More fundamentally, the question arises whether an informational perspective could be beneficial for quantum gravity in general.
\item The hope has been voiced that a clear interpretation of the theory may finally emerge from a successful reconstruction, in analogy to how the interpretation of special and general relativity follows naturally from its underlying principles \cite{Rovelli:1995fv,zeilinger1999foundational,Fuchs:fk}.
\end{enumerate}

It is fair to say that the hope alluded to under point 4 has not been realized thus far because the existing successful reconstructions \cite{Hardy:2001jk,Dakic:2009bh,masanes2011derivation,Mueller:2012ai,Masanes:2012uq,chiribella2011informational,de2012deriving,Mueller:2012pc,Hardy:2013fk,Barnum:2014fk,kochen2013reconstruction,Oeckl:2014uq} are fairly neutral as far as an interpretation is concerned. While most of them emphasize the operational character of the theory, a particular interpretation of quantum theory is not strongly suggested.

The language and concepts of the present reconstruction are different. It will emphasize and concretize the view (or partial interpretation) that quantum theory is a framework governing an observer's acquisition of information about the observed system \cite{hartle1968quantum,Rovelli:1995fv,peres1995quantum,zeilinger1999foundational,Brukner:1999qf,Brukner:vn,Brukner:ys,Brukner:2002kx,Fuchs:fk,Caves:1996nq,Caves:2002uq,caves2002unknown,spekkens2007evidence,Spekkens:2014fk}. The new postulates are simple and conceptually comprehensible, concerning only the relation between an observer and the system. Due to the simplicity, the ensuing derivation is mathematically quite elementary and, in contrast to previous derivations, yields the formalism, state spaces and time evolution groups explicitly and in a more constructive manner. The disadvantage, compared to other reconstructions, is that a large number of detailed steps is required. The advantage, on the other hand, is the simplicity of the principles and mathematical tools, and the fact that the reconstruction affords natural explanations for many quantum phenomena, including entanglement and monogamy, and elucidates the origin of the unitary group.

\section{Landscape of inference theories}\label{sec_landscape}

The ambition of a (re-)construction of quantum theory is to derive its formalism, state spaces, time evolution groups and permissible operations from physical principles -- in some rough analogy to the construction of relativity theory from the principle of relativity and the equivalence principle. But in order to formulate physical postulates, we clearly have to presuppose some mathematical structure within which a precise meaning can be given to them. (For example, also the construction of special and general relativity certainly presupposed a substantial amount of mechanical structure.) 

The procedure is thus to firstly define some {\it landscape of theories}, which hopefully contains quantum theory and classical information theory, but within which theories are generally {\it not} formulated in terms of the usual complex Hilbert spaces, tensor product rules, etc. The mathematical formulation of the landscape must therefore be more elementary and, in particular, operational. That is, for the time being, we have to forget about the usual -- mathematically crisp but physically rather obscure -- textbook axioms of quantum theory. While different theories will have the mathematical and physical structure of the landscape in common, they may have otherwise very different physical and informational properties; e.g., they may admit much stronger correlations than quantum theory \cite{popescu1994quantum}, or weaker correlations as classical probability theory, they may allow exotic communication and information processing tasks to be accomplished \cite{pawlowski2009information,Paterek:2010fk}, and so on. Secondly, given the language of this landscape, one can attempt to formulate comprehensible physical statements which single out quantum theory from {\it within} it. Going to a larger theory landscape and beyond the language of Hilbert spaces is precisely what allows us to ask the question ``what makes quantum theory special?'' and, ultimately, to find an operational and physical justification for the usual textbook axioms and the standard Hilbert space formulations.

The goal of this section is precisely to build such an appropriate landscape of inference theories both conceptually and mathematically from scratch within which we shall subsequently formulate those elementary rules, governing an observer's acquisition of information about a system, that single out qubit quantum theory. This novel landscape of inference theories employs a different language and is conceptually distinct from the by now standard landscape of generalized probabilistic theories (GPTs) which are commonly employed for characterizations (or generalizations) of quantum theory. For contrast and to put the new tools into a larger perspective, we begin with a brief synopsis of the GPT language before we establish the landscape of inference theories underlying this manuscript.

\subsection{The standard landscape of generalized probabilistic theories}

It has become a standard in the literature to employ the formalism of generalized probabilistic theories (GPTs) for operational characterizations or derivations of quantum theory \cite{Hardy:2001jk,barrett2007information,masanes2011derivation,Mueller:2012ai,Dakic:2009bh,hardy2011foliable,hardy2013formalism,chiribella2010probabilistic,Masanes:2012uq,chiribella2011informational,Masanes:2011kx,de2012deriving,Mueller:2012pc,pfister2013information,Hardy:2013fk,Barnum:2014fk,Barnum:2010uq}. The setup of GPTs is exclusively operational and one considers three kinds of operation devices (see figure \ref{fig_gpt}): (1) a preparation device which can spit out systems in some set of states defined by a vector of probabilities for the outcome of fiducial measurements. The state spaces of the systems are necessarily required to be convex to permit convex mixtures of states. (2) A transformation device can perform physical operations on the prepared systems (e.g., a rotation) which may change the state of the system (e.g., by some group action on the state vector), but must allow it to continue its journey to (3), a measurement device, which detects certain experimental outcomes. The measurement devices are mathematically described by so-called `effects' which are assumed to be dual to, i.e.\ linear functionals on, the states. (For the interested reader we note, however, that Holevo has shown how the mathematical incarnation of measurements as linear functionals on states follows from other simple operational assumptions \cite{holevo}.)
\begin{figure}[hbt!]
\begin{center}
\psfrag{p}{\small preparation}
\psfrag{t}{\small transformation}
\psfrag{m}{\small measurement}
\psfrag{s}{systems}
\psfrag{c}{\hspace*{-.2cm}convex state spaces}
\psfrag{cb}{\footnotesize cbit}
\psfrag{gb}{\hspace*{-.35cm}\footnotesize square bit}
\psfrag{rb}{\footnotesize rebit}
\psfrag{qb}{\footnotesize qubit}
\psfrag{e}{\small `effects'}
\psfrag{d}{\footnotesize dual to states}
{\includegraphics[scale=.33]{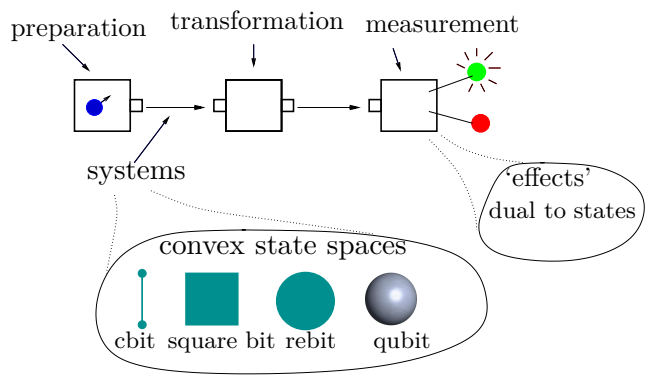}}
\caption{\small The standard operational setup of generalized probabilistic theories with examples of allowed convex state spaces of elementary two-level systems.}\label{fig_gpt}
\end{center}
\end{figure}

{\it Measurements} and {\it states} are at the heart of GPTs;  `effects' directly determine the outcome probabilities of measurements and thereby, when a complete set is engaged, reveal the 
probabilistic state of a system. 
An observer assumes a supporting role, 
giving intuitive meaning to the notion of preparation, transformation and measurements of systems. The reconstructions of quantum theory within the GPT formalism depart from rather global operational axioms, restricting the sets of possible preparations (i.e.\ state spaces), transformations and measurements. The concrete acquisition of information of the observer about the systems is otherwise not accentuated. In particular, since the outcome probabilities are the primary concept of GPTs, these axioms say little about what an observer experiences in individual experimental runs, but instead focus on the totality of a large number of experimental runs. 
This has led to a whole wave of successful quantum theory reconstructions, employing GPT concepts in one way or another \cite{Hardy:2001jk,Dakic:2009bh,masanes2011derivation,Mueller:2012ai,Masanes:2012uq,chiribella2011informational,de2012deriving,Mueller:2012pc,Hardy:2013fk,Barnum:2014fk}, most of which, however, are fairly abstract on account of the rather global axioms.

It is perhaps one of the great strengths of GPTs that they constitute a functional and purely operational framework which is interpretationally fairly neutral. It is not the ambition of this framework to elucidate the measurement problem, to clarify what happens to a state during a measurement, what probabilities are or, ultimately, how to interpret quantum mechanics (except that it highlights its operational character). As such, this framework is compatible with most interpretations of quantum theory.

\subsection{A novel landscape of inference theories }\label{sec_infland}

The success of GPTs notwithstanding, we shall now change semantics and perspective to define a new landscape of theories and a novel framework for (re)constructing and understanding quantum theory. Henceforth, we shall fully engage the observer and give primacy to his acquisition of fundamentally {\it limited} information from observed systems. This will include being more explicit about what general sort of information will be available to the observer even in individual experimental runs. Probabilities, on the other hand, can be viewed as secondary and as a consequence of the limited information available to the observer -- although, clearly, probabilities will assume a pivotal role too (after all we want to reconstruct quantum theory). 

\subsubsection{Questions and answers}\label{sec_QandA}

As schematically depicted in figure \ref{fig_inter}, we shall consider an observer $O$ who can only interact with a system $S$ through interrogation 
 via questions $Q_i$ from some set of questions $\cq$ which we shall further constrain below. (At this stage we make no assumption about whether $\cq$ is continuous or discrete.) The only information which we allow $O$ to acquire about $S$ is by asking questions from this set $\cq$.
 
 In principle, of course, $O$ could conceive of and ask $S$ all kinds of questions, but $S$ could not always give a meaningful answer; $S$ may simply not have the desired properties or carry the information $O$ is inquiring about (e.g., $S$ may not be complex enough, or $S$ does not interact with other systems carrying the information in question), or possibly the question is not a senseful one in the first place. We shall call a question $Q$ physically {\it implementable on $S$} if $O$ can acquire a `meaningful' answer from $S$ to $Q$. An elementary restriction on $\cq$ is that any question $Q_i$ from this set be {\it implementable} on $S$ and that, whenever $O$ asks $Q_i$ to $S$, $S$ will give an answer to $O$. Clearly, $\cq$ depends on $S$.
 
 But how does $O$ know whether $S$ will give an answer and how can he judge whether the latter is `meaningful'? This requires $O$, like any experimenter, to have developed, from previous experiences, a theoretical model by means of which he describes and interprets his interactions with $S$. An answer can only be `meaningful' in the context of this model such that our notion of implementability is actually dependent on $O$'s model. We shall come back to this model frequently.


\begin{figure*}[htb!]
\begin{center}
\psfrag{p}{Preparation}
\psfrag{i}{Interrogation}
\psfrag{S}{$S$}
\psfrag{O}{$O$}
\psfrag{q}{$Q_i$?}
{\includegraphics[scale=.5]{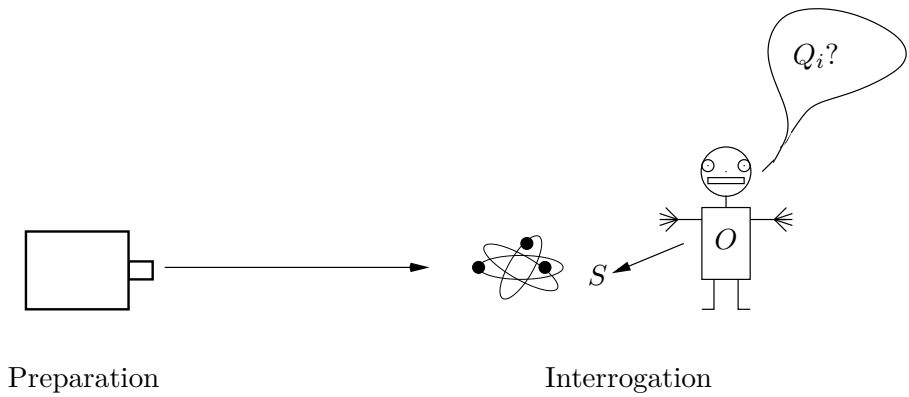}}
  \caption{\small Schematic representation of an observer $O$ interrogating a system $S$. }\label{fig_inter}
\end{center}
\end{figure*}

The central ingredients of this framework will thus be {\it questions} and {\it answers} -- and $O$'s information about answer outcomes and their relations. This will lead to a novel question calculus from which many crucial quantum properties will be derived. That is, rather than focusing mostly on probabilistic properties as in GPTs, this novel framework will connect more directly to what can be measured in experiments through its emphasis on questions and their relations. As a consequence, this framework will also produce operationally more compelling explanations of typical quantum phenomena.

In the sequel, we shall solely speak about the {\it information $O$ has about $S$} and, correspondingly, about the state that $O$ {\it assigns} to $S$ based on this information. Such a state of $S$ is then defined {\it relative} to $O$ (whether or not some hidden variables give rise to this state is a question we shall not address). The act of information acquisition establishes a {\it relation between $O$ and $S$} and this will be the center of our attention. (A priori, a different observer $O'$ may establish a {different} relation with $S$.) Although this framework will also not give rise to a unique interpretation, it will support the {\it partial} interpretation that quantum theory is a law book governing an observer's acquisition of information about physical systems and might therefore be considered interpretationally less neutral than GPTs. While this connects with general ideas underlying, e.g., the {relational} \cite{Rovelli:1995fv,Smerlak:2006gi}, Brukner-Zeilinger informational \cite{zeilinger1999foundational,Brukner:1999qf,Brukner:vn,Brukner:ys,Brukner:2002kx}, or QBist \cite{Fuchs:fk,Caves:1996nq,Caves:2002uq,caves2002unknown} interpretations of quantum mechanics, we emphasize that this reconstruction does {\it not} rely on them and should thus not discourage readers uncomfortable with these specific interpretations. 

We shall not explicitly deal with transformation and measurement (`effect') devices as in GPTs; instead, these will be replaced more generally by time evolution and questions, respectively. The set of all possible (information preserving) operations that $O$ could perform on $S$ can later be identified with the set of all possible time evolutions of $S$. However, in analogy to GPTs, we will assume that $O$ has access to some method of preparing $S$ in different ways such that $S$'s answers to $O$'s questions depend on the precise way of preparation. ($O$ could either control himself a preparation device or have a distinct observer $O'$ prepare systems for him.)

When constructing the new landscape $\cl$ of theories describing $O$'s acquisition of information about $S$ within which we shall later, in section \ref{sec_postulates}, formulate our postulates, we will make a number of restrictions and assumptions. 
In order to facilitate future generalizations and improvements of the present construction of quantum theory -- which constitutes a proof of principle -- we shall attempt to be as clear as possible about the assumptions made throughout this work.

As a starter, we would like to keep $\cq$ as simple as possible, while still having non-trivial questions. In particular, we do not wish to consider trivial propositions which are always true or always false. We shall therefore assume the following.
\begin{Assumption}\label{assump1}
The set of questions $\cq$ which we shall permit $O$ to ask $S$ only contains \emph{binary} questions $Q_i$. Any $Q_i\in\cq$ is a non-trivial question such that $S$'s answer (`yes' or `no') is \emph{not} independent of its preparation. Furthermore, any $Q_i\in\cq$ is \emph{repeatable} such that $O$, by asking the same $S$ the same $Q_i$ $m$ times in succession will receive $m$ times the same answer.
\end{Assumption}
The restriction to elementary `yes-no'-questions greatly simplifies the discussion and, ultimately, will give rise to the quantum theory of qubit systems. For instance, in quantum theory, a binary question could be `is the spin of the qubit up in $x$-direction?' However, it will not be too difficult to generalize $\cq$ to also consist of ternary, quaternary, quinary, etc. questions, but we shall not attempt to do so here. Since trivial questions are not considered, we already see that not even all implementable questions will be taken into account. We shall impose further restrictions on $\cq$ such that, ultimately, it will be a strict subset of all possible binary questions which $O$ could, in principle, ask $S$.

Since this is an operational approach, it is fair to assume that $O$ can record the answers to his questions asked to any system (e.g., by writing them on a piece of paper) and that he can do statistics over the outcomes (e.g., by counting the frequency of outcomes). We shall require that every possible way of preparing $S$ will give rise to a particular statistics over the answers to all $Q_i\in\cq$; $O$ could test these statistics by interrogating a large number $n$ of identically prepared systems\footnote{Given the present structure, two systems $S_1$ and $S_2$ could only be considered as distinct in nature if either the (maximal) set of questions $\cq$ which $O$ can ask the systems or the totality of answer statistics for all possible preparations were distinct for $S_1$ and $S_2$ (if always the same two $S_1,S_2$ are interrogated and thereafter freshly prepared again). If $O$ can {\it not} distinguish $S_1$ and $S_2$ in this way for sufficiently many trials (ideally infinitely many), we shall call them identical. Let $S_1$ and $S_2$ be identical and $O$ prepare both systems with the same procedure (for instance, the setting of the preparation device is the same for both systems). If the answer statistics (frequencies) for $S_1,S_2$ for all $Q_i\in\cq$ become indistinguishable after sufficiently many trials (of preparing and interrogating the same systems with the same procedure), we consider these systems as `identically prepared'. (After completion of this work, the author was made aware that this notion is similar to the definition of `operational equivalence' put forward in \cite{spekkens2005contextuality}.)}  $S_a$, $a=1,\ldots,n$, sufficiently often with (at least ideally) all $Q_i\in\cq$. In fact, this is precisely how $O$ will operationally distinguish different preparations of systems. 

By having interrogated, in this manner, the $n$ always identically prepared $S_a$ for all possible ways of preparation, we shall assume $O$ to have gained a `stable' knowledge of the {\it set $\Sigma$ of all possible answer statistics for $S$ over $\cq$}, i.e., the answer statistics for all $Q\in\cq$ and all possible ways of preparing a system $S$.\footnote{We assume the preparation method to be ideal in the sense that it accounts for any possible answer statistics which $S$ can admit in $O$'s world for questions in $\cq$. That is, there do not exist other methods which can prepare $S$ in ways that $O$'s method does not encompass.} While ideally $n\rightarrow\infty$ is necessary, `practically' $n$ should be large enough for $O$ to develop a theoretical model of both $\cq$ and $\Sigma$ up to some accuracy which agrees with his observations. It is not our ambition to clarify further what $n$ is nor how precisely $O$ has developed his theoretical model, instead, we shall henceforth just assume that $O$ has puzzled out the pair $(\cq,\Sigma)$. As any experimenter in an actual laboratory, $O$ shall interpret the outcomes of his interrogations by means of his model for $(\cq,\Sigma)$ and he can decide whether a given question is contained in the set $\cq$ or not. Henceforth, we assume that $O$ only asks questions from $\cq$. Our task will be to establish what this model is, based on the ensuing assumptions and postulates.

\subsubsection{From information to probabilities: the state of $S$ relative to $O$}

The previous subsection did not refer to the notion of probabilities. But it defined two preparations as being identical if they give identical statistics. This permits us to regard $\Sigma$ equivalently as the set of all possible answer statistics or as the set of all (distinct) preparations.  

Based on this notion, we shall now identify probabilities as degrees of belief. The `knowledge' of what $\Sigma$ is for a given $S$ will permit $O$ to assign probabilities to the outcomes of his questions. It is very natural for $O$ to assign probabilities to questions because he deals with statistical fluctuations and furthermore, as we shall see later, with systems about which he always has incomplete information in the sense that the corresponding $\Sigma$ is such that he can never know the answers to all $Q_i\in\cq$ at the same time.

More precisely, for a specific $S$ and any $Q_i\in\cq$ that he may ask the system next, $O$ can assign a probability $y_i$ that the answer will be `yes' (or a probability $n_i$ that the answer will be `no'), according to
\begin{itemize}
\item[(i)] $O$'s knowledge of $\Sigma$, and
\item[(ii)] any prior information that $O$ may have about the specific $S$.
\end{itemize}
Given our setup, the only prior information that $O$ may have about the particular $S$ (apart from what the associated $\Sigma$ and $\cq$ are) must result either from having interrogated some ensemble of systems identically prepared to $S$ with some subset of $\cq$ beforehand\footnote{If $O$ asks more than one question to any $S$, the ordering of the questions may matter. But then $O$ could ask the questions for any $S$ always in the same order.} and from the corresponding accumulated statistics of the asked $Q_j\in\cq$ (e.g., that $O$ may have recorded on a piece of paper), or from any other prior information about the method of preparation. In particular, this previous ensemble could have been empty.

For instance, if every time that $O$ asked the specific $Q_i$ to any of the identically prepared systems gave a `yes' answer before, he will assign the prior probability $y_i=1$ to $Q_i$ and to the next identically prepared $S$ that he will interrogate. If, on the other hand, the number of `yes' and `no' answers to some other $Q_j$ was equal for the previously identically prepared systems, $O$ will assign, as a best guess, the prior probability $y_j=\f{1}{2}$ to this $Q_j$ and to the next $S$. Similarly, for any other answer statistics, $O$ would assign $y_i$ to the next $S$ depending on the recorded frequencies of `yes' answers. But, thanks to his knowledge of $\Sigma$ and therefore of any possible relations in the answer statistics, $O$ can also assign prior probabilities $y_k$ to questions $Q_k$ that he did {\it not} ask the previous set of identically prepared systems. For example, $\Sigma$ may be such that whenever $S$ gives a `yes' answer to $Q_i$, it will give a `no' answer to an immediately following $Q_k$. Accordingly, if $O$ assigns a prior probability $y_i=1$ to $Q_i$ as above, he will also assign a prior $y_k=0$ to $Q_k$ without previously having asked $Q_k$. But other relations between questions will be permitted too. In particular, it may be that the information gained from the questions he previously asked the identically prepared systems and the structure of $\Sigma$ make it equally likely that the answer to $Q_k$ asked to the next $S$ will be `yes' or `no'. In this case, $O$ will assign $y_k=\f{1}{2}$ to $Q_k$ that he may ask the next $S$. This will become more precise along the way.

We therefore take a broadly\footnote{We add the qualifier `broadly' here since we also allow for the typical laboratory situation of ensembles of systems (which may or may not contain more than one element).} {\it Bayesian perspective} on probabilities: $O$ assigns probabilities to questions according to his `degree of belief' about $S$. These probabilities $y_i$ are thereby {\it relative to the observer $O$}. A different observer $O'$ may have different information about $S$ and thereby assign different probabilities to the various outcomes of questions posed to $S$ (for a discussion, within quantum theory, of the consistency of different observers having different information about a system, see \cite{Rovelli:1995fv,Smerlak:2006gi,Cavescompat,mermin2002compatibility}). E.g., $O'$ could be the one preparing $S$. She could `know' the statistics for the specific preparation setting (from previous tests) and then send $O$ the specifically prepared $S$ without informing him about her knowledge. 

For consistency, we tacitly assume the set $\Sigma$ of all possible answer statistics to coincide with the set of possible `beliefs'. That is, to every equivalence class of preparations of $S$ (with `identically prepared' defining the equivalence relation) there shall correspond a unique `belief' $\{y_i\}_{Q_i\in\cq}$ and vice versa. Since the only way for $O$ to acquire information about $S$ is by interrogation with questions in $\cq$, the prior probabilities $y_i$ that $O$ assigns to every $Q_i\in\cq$ encode the entire information that $O$ has about $S$. Hence, we shall make the following identification.
\begin{Definition}{\bf(State of $S$ relative to $O$)}
The collection of all probabilities $y_i$ $\forall\, Q_i\in\cq$ \emph{is} the state of $S$ relative to $O$. Accordingly, the set $\Sigma$ of all possible answer statistics on $\cq$ which $S$ admits is the \emph{state space of $S$}.
\end{Definition}
Of course, ultimately not all $y_i$ will be independent such that the full collection of probabilities will yield a redundant parametrization of the state. However, this is not important for the moment and we shall come back to this shortly. 

This definition of the state of a system $S$ explicitly identifies it with the `state of information' that $O$ has acquired about $S$; $O$ assigns this state to $S$ according to his information about the $Q_i\in\cq$. As such, {\it the state of system $S$ is epistemic (i.e.\ a `state of knowledge') and a priori only meaningful relative to the observer $O$}. The interpretation of the quantum state as a `state of information' is certainly not new and has been proposed in various ways before (see also, e.g., \cite{Caves:1996nq,Caves:2002uq,Fuchs:fk,spekkens2007evidence}). However, the above definition is closest in spirit to the 
 ideas underlying Relational Quantum Mechanics \cite{Rovelli:1995fv,Smerlak:2006gi} and the Zeilinger-Brukner interpretation \cite{zeilinger1999foundational,Brukner:ys,Brukner:1999qf,Brukner:vn,Dakic:2009bh} and thereby generalizes them to the landscape of theories describing $O$'s acquisition of information which we are in the process to establish. 

While the state of $S$ is thus a priori only meaningful relative to $O$, we emphasize that both the set of questions $\cq$ which $O$ may ask and the state space $\Sigma$ are to be intrinsic to the system $S$. Otherwise, it would be difficult for two observers to agree on the description of a given $S$.

\subsubsection{`Belief' updating and `collapse' of the state}\label{sec_bayes}

At this stage it is important to distinguish {\it single} from {\it multiple shot interrogations}. In a 
\begin{description}
\item[single shot interrogation] $O$ interrogates a single system $S$, in some prior state, with a number of questions from $\cq$ {\it without} intermediate re-preparations of $S$. The definite answers to these questions give $O$ {\it definite} information about this specific $S$ after the interrogation. Furthermore, his knowledge of $\Sigma$ and any prior knowledge of $S$ (acquired through previous interrogations of identically prepared systems) give him statistical information about any questions he did {\it not} ask $S$. In conjunction, the new answers and his prior knowledge thus determine the state of $S$ {\it after} the interrogation. This will constitute a {\it posterior state update rule}, and thereby a `belief' update form a prior to a posterior state for a {\it single} system which we shall turn to shortly (and which clearly depends on the specific way $O$ interrogates $S$). This {\it posterior} state of $S$ will reflect $O$'s definite information about every asked question $Q_i$ by featuring either $y_i=0$ or $y_i=1$ due to repeatability, depending on whether the answer was `no' or `yes', respectively (assuming for now, of course, that $\Sigma$ is such that the answers to the selection of questions that $O$ asked can be known simultaneously). 

If this {\it posterior} state does not coincide with the prior state that $O$ assigned to $S$ before the interrogation, based on his prior information about $S$, then $S$'s state  has `collapsed' relative to $O$ during the interrogation. Hence, a state `collapse' only occurs if $O$'s {\it posterior} information about $S$ does not coincide with his prior information about $S$, i.e.\ if $O$ experienced an {\it information gain} about $S$ via the interrogation. We shall therefore view a state `collapse' as $O$'s information gain about this specific $S$ rather than a `disturbance'\footnote{A `disturbance' of the system $S$ is only meaningful if there was an underlying {\it ontic} state (i.e.\ `state of reality') to which, however, $O$ would have no access. Here we shall merely speak about the information that $O$ has access to and therefore not make any ontic statements, regarding them as excess baggage for our purposes.} of $S$ (we refer the reader also to \cite{hartle1968quantum,Fuchs:1995xa,Fuchs:1996hz,pfister2013information} for a related discussion).

\item[multiple shot interrogation] $O$ interrogates an {\it ensemble} of identically prepared systems $S_a$, $a=1,\ldots,n$, where the interrogation of every $S_a$ is a single shot interrogation. $O$ will carry out such a multiple shot interrogation to do {\it state tomography}, i.e.\ to estimate the state of the ensemble $\{S_a\}$ for the specific setting of preparation or, in other words, the state of any of the systems {\it prior} to being interrogated by $O$. 

This will also be a `belief' updating, however, a {\it prior (or ensemble) state updating}. After every interrogation of a system in the ensemble, $O$ will assign probabilities $y_i$ to the $Q_i$ in the manner described above. This will then define the prior state of the next system in the ensemble to be interrogated. By interrogating more and more systems, $O$ will gain more and more information about the ensemble state such that his assignments of the $y_i$ will fluctuate less the larger the number of interrogated systems. This process gives rise to an ensemble state updating. Independent of this updating, the prior state that $O$ assigns to any individual system may experience a `collapse' during the interrogation of that specific system because his information about the specific system may have changed. Accordingly, $O$ will have to distinguish the ensemble state from the {\it posterior} state of any system in the ensemble. But the collection of {\it posterior} states determines the ensemble state.

\end{description}


\subsubsection{Elementary structure on $\Sigma$ and $\cq$}\label{sec_elstruc}

The present structure is still too rudimentary for a quantum (re)construction. We therefore need more. 

Firstly, as in GPTs we will need $O$ to be able to assign a single prior state to any pair of identical systems whenever he flips a biased coin in order to decide which of the two systems he will interrogate. (Equivalently, the preparation method could involve a biased coin toss, the outcome of which determines the preparation setting.) That is, $O$ will be permitted to build convex combinations of states.
\begin{Assumption}\label{assump2}
The state space $\Sigma$ of $S$ is a closed convex set.
\end{Assumption}
Closure of $\Sigma$ is assumed because points lying in the boundary of $\Sigma$ can be arbitrarily well approximated by states in the interior. {\it Perfect preparation} and {\it arbitrarily good preparation} are operationally indistinguishable and so adding the boundary does not change the physical predictions \cite{masanes2011derivation}.

Next, we need to define some elementary structure on $\cq$ in order to meaningfully speak about relations among questions (and answers). To this end, we shall establish additional structure on the pair $(\cq,\Sigma)$. We declared in assumption \ref{assump1} that there are to be no trivial questions in $\cq$ the answers to which would be independent of $S$'s preparation. We also insisted before that $O$ will distinguish the different ways of preparing a system $S$ -- and thus the different states he can assign -- by the particular answer statistics. We shall now strengthen these requirements, by asserting that 
there exists a distinguished state of `no information', corresponding to the situation that $O$'s prior knowledge about $S$ makes it equally likely for him that the answer to {\it any} $Q\in\cq$ is `yes' or `no'.\footnote{We emphasize that GPTs are more general by, in principle, permitting state spaces which do not contain such a distinguished state. However, most operationally interesting GPTs do possess such a state.} We note that this assumption is a restriction on $O$'s model for the \underline{pair} $(\cq,\Sigma)$.\footnote{For instance, we note that on account of the existence of binary POVMs with an inherent bias, such as $(E_0=2/3\cdot\mathds{1},E_1=1/3\cdot\mathds{1})$, the pair given by $\cq=\{\text{binary POVMs}\}$ and $\Sigma=\{\text{unit trace density matrices}\}$ cannot satisfy this condition because no unit trace density matrix exists which yields probability $1/2$ for $(E_0,E_1)$. This pair will therefore ultimately {\it not} be the solution of this reconstruction. However, a subset of $\{\text{binary POVMs}\}$ together with the full quantum state space will be the solution.}
\begin{Assumption}\label{assump3}
There exists a special state in $\Sigma$, called the \emph{state of no information}, which is given by $y_i=\f{1}{2}$, $\forall\,Q_i\in\cq$.
\end{Assumption}
This state shall be the prior state that $O$ assigns to $S$ (or an ensemble thereof) in a belief updating whenever he has `no prior information'. For example, a distinct observer $O'$ may prepare a system $S$ and send it to $O$ in such a way that the latter knows only the associated $\Sigma$ but not the preparation setting. In this case, as a prior, $O$ will assign the state of no information to $S$, i.e.\ $y_i=\f{1}{2}$, $\forall\,Q_i\in\cq$. Similarly, there will exist a special preparation setting which is such that a multiple shot interrogation on an ensemble prepared in this setting will give totally random answers to $O$ such that he will assign the state of no information to the ensemble. But note that the state of any individual system {\it after} the interrogation of that system will {\it not} be the state of no information because, through the interrogation, $O$ will have acquired information about that specific system (see the discussion of the state `collapse' above).

Given the state of no information, we shall preliminarily quantify the amount of information $\alpha_i$ that the definite answer to any $Q_i\in\cq$ defines as one \texttt{bit}. Similarly, whenever $y_i=\f{1}{2}$, as in the state of no information, we shall say that $O$ has $\alpha_i=0$ \texttt{bits} of information about $Q_i$. In general, under the premise that information can neither be negative nor complex, $O$'s information about $Q_i$ should satisfy
\ba
0\,\texttt{bit}\leq\alpha_i\leq1\,\texttt{bit}.\label{infoineq}
\ea 
We shall not propose an explicit information measure $\alpha_i$ (as a function on $\Sigma$) here because this must follow from the rules on information acquisition postulated below and we shall indeed derive it therefrom later in section \ref{sec_infomeasure}. Until then it will be sufficient to work with this implicit notion of quantifying $O$'s information about any $Q_i$. 

But the questions $O$ can ask, and the information that the corresponding answers define, may not be independent. However, a priori the notion of independence of questions, in the sense of stochastic independence, is state dependent. For example, for a pair of qubits in quantum theory the questions $Q_1$, `is the spin of qubit 1 up in $x$-direction?', and $Q_2$, `is the spin of qubit 2 up in $x$-direction?', are stochastically independent relative to the completely mixed state, but fully dependent relative to an entangled state (with correlation in $x$-direction). As a result of the state dependence, the independence of questions may also be viewed as observer dependent. For instance, in quantum theory an observer $O'$ could send $O$ an entangled pure state (with correlation in $x$-direction) and refuse to tell $O$ which state it is. Relative to $O'$, $Q_1$ and $Q_2$ will be dependent, but they will be independent relative to $O$ because the latter will assign the completely mixed state to the pair prior to measurement.

Since the notion of independence of questions is state dependent, we need a {\it distinguished} state in order to unambiguously define it -- this is the second purpose of assumption \ref{assump3} and brings us in contact with a state update rule. Indeed, suppose $O$ acquires $S$ in the state of no information and poses the question $Q_i\in\cq$. By assumption \ref{assump1}, any $Q_i\in\cq$ is repeatable such that, upon receiving either the answer `yes' or `no', $O$ will assign $y_i=1$ or $y_i=0$, respectively, as the probability for $S$ giving the answer `yes' if posing $Q_i$ again. The posterior state update rule, which enables $O$ to update his information about a specific $S$ in compliance with the given answers, must respect this repeatability. It depends on the details of this update rule what the probabilities $y_j$ for all other $Q_j\in\cq$ is after having asked only $Q_i$. Rather than fully specifying at this stage what this rule is, we shall simply assume that $O$ employs one which is consistent. Whatever this posterior state update rule, we shall refer to $Q_i,Q_j\in\cq$ as
\begin{description}
\item[independent] if, after having asked $Q_i$ to $S$ in the state of no information, the probability $y_j=\f{1}{2}$. That is, if the answer to $Q_i$ relative to the state of no information tells $O$ `nothing' about the answer to $Q_j$. We shall require this relation to be symmetric, i.e., $Q_i$ is independent of $Q_j$ if and only if $Q_j$ is independent of $Q_i$.\footnote{It should be noted that, while this is true for projective measurements in quantum theory, it does {\it not} hold for generalized measurements. I thank Tobias Fritz for pointing this out.} This is equivalent to saying that $Q_i,Q_j$ are stochastically independent with respect to the state of no information, i.e.\ the joint probabilities factorize relative to the latter, $p(Q_i,Q_j)=y^*_{i}\cdot y^*_j=\f{1}{2}\cdot\f{1}{2}=\f{1}{4}$. Here $p(Q_i,Q_j)=p(Q_j,Q_i)$ denotes the probability that $Q_i$ and $Q_j$ give `yes' answers if asked in {sequence} to the {\it same} $S$ arriving in the state of no information and $y^*_{i},y^*_j$ are the individual `yes'-probabilities in this distinguished state.\footnote{We emphasize that this is a definition of {\it maximal} independence. For example, for a qubit two linearly independent directions $\vec{n}_1,\vec{n}_2$ in the Bloch sphere with $\vec{n_1}\cdot\vec{n}_2\neq0$ would define two spin observables $\vec{n}_1\cdot\vec{\sigma}$ and $\vec{n}_2\cdot\vec{\sigma}$ whose corresponding projectors would {\it not} be maximally independent according to this definition. The corresponding questions would be partially dependent (see below) because whenever the observer knows the answer to one question, the probability for a `yes' answer to the other would be distinct from $\f{1}{2}$.} 

\item[dependent] if, after having asked $Q_i$ to $S$ in the state of no information, the probability $y_j=0,1$. That is, if the answer to $Q_i$ relative to the state of no information implies also the answer to $Q_j$. Again, we require this relation to be symmetric. This is equivalent to saying that, relative to the state of no information, $Q_i,Q_j$ are stochastically fully dependent as either $p(Q_i,Q_j)=y^*_i=y^*_j=\f{1}{2}=p(\neg Q_i,\neg Q_j)$ or $p(Q_i,\neg Q_j)=y^*_i=y^*_j=\f{1}{2}=p(\neg Q_i,Q_j)$, where $\neg Q$ is the negation of $Q$.\footnote{The most trivial example of a pair of dependent questions are clearly $Q$ and $\neg Q$. But there exist less trivial ones. E.g., see theorem \ref{thm_qubit} below.} 
\item[partially dependent] if, after having asked $Q_i$ to $S$ in the state of no information, the probability $y_j\neq0,\f{1}{2},1$. That is, if the answer to $Q_i$ relative to the state of no information gives $O$ partial information about the answer to $Q_j$. Again, this relation is required to be symmetric. This is equivalent to saying that the joint probabilities $p(Q_i,Q_j)$ relative to the state of no information do {\it not} factorize and the answer to one question does not fully imply the answer to the other.
\end{description}
We shall henceforth assume that $\Sigma$ is such that two questions $Q_i,Q_j$ which are fully or partially dependent relative to the state of no information are also fully or partially dependent relative to any other state in $\Sigma$. These definitions of (in-)dependence depend a priori on the update rule through the joint probabilities. For our purposes it will turn out not to be necessary to fully specify this update rule. 

Next, we need a notion of compatibility and complementarity. $Q_i,Q_j$ are
\begin{description}
\item[maximally compatible] if $O$ may know the answers to both $Q_i,Q_j$ simultaneously, i.e.\ if there exists a state in $\Sigma$ such that $y_i,y_j$ can be simultaneously $0$ or $1$.

\item[maximally complementary] if maximal information about the answer to $Q_i$ forbids $O$ to have {\it any} information about the answer to $Q_j$ at the same time (and vice versa). That is, every state in $\Sigma$ which features $y_i=0,1$ will necessarily have $y_j=\f{1}{2}$ (and vice versa). 

\end{description}
Consequently, complementary questions are independent, but independent questions are not necessarily complementary. 
Finally, $Q_i,Q_j$ are partially compatible (or complementary) if maximal information about one precludes maximal, but permits non-maximal information about the other.

This permits us to further constrain the update rule. Firstly, maximal complementarity has an obvious consequence for an update rule. Secondly, we shall assume the following.
\begin{Assumption}\label{assump5b}
If $Q_i,Q_j$ are maximally compatible and independent then asking either shall not change $O$'s information about the other, regardless of $S$'s state. That is, asking $Q_i$ must leave $y_j$ invariant -- and vice versa for $i,j$ interchanged. 
\end{Assumption}
This is to prevent $O$ from losing or gaining information about some question by asking another question which is compatible with but independent of the first. These constraints on the update rule turn out to be sufficient for our reconstruction.

Assumption \ref{assump5b} leads to a first result which we shall use at times: it implies what sometimes is referred to as `Specker's principle' \cite{specker1960logik} (see also \cite{liang2011specker,cabello2013simple,cabello2012specker,chiribella2014measurement}). We shall call $Q_1,\ldots,Q_n\in\cq$ {\it mutually maximally compatible} if there exists a state of $S$ where the answers to all of $Q_1,\ldots,Q_n$ are known simultaneously to $O$. 

\begin{Theorem}\label{assump5} {\bf (`Specker's principle')}
If $Q_1,\ldots,Q_n\in\cq$ are pairwise maximally compatible and pairwise independent then they are also mutually maximally compatible.
\end{Theorem}

\begin{proof}
Let $Q_1,\ldots,Q_n\in\cq$ be pairwise independent and pairwise maximally compatible. On account of repeatability, after asking $S$ the question $Q_1$, $O$ will assign the probability $y_1$ to be either $0$ or $1$, depending on the answer. $O$ may subsequently ask $S$ the question $Q_2$ upon which the probability $y_2$ will likewise be either $0$ or $1$. According to assumption \ref{assump5b}, this will not change $y_1$ because, by assumption, both are maximally compatible and independent. By the same argument, upon next asking and receiving an answer to $Q_3$, $O$ will assign either $y_3=1$ or $y_3=0$, while both $y_1,y_2$ are unchanged. Repeating this argument recursively, it is clear that $O$ can thereby generate a state of $S$ in which all $y_i$, $i=1,\ldots,n$, are either $0$ or $1$ which represents a state of $S$ where $O$ knows the answers to all of $Q_1,\ldots,Q_n$.
\end{proof}

Notice that this property is satisfied for classical bit theory and in quantum theory for projective, however, not for generalized measurements.

\subsubsection{Logical compositions of questions and rules of inference}\label{sec_ri}

Given multiple questions, nothing, in principle, stops $O$ from composing them via logical connectives to all kinds of propositions which themselves constitute questions. The issue is, however, whether he will be able to get an answer from $S$ to such a question, i.e.\ whether it is implementable on $S$ and thus whether it could be contained in $\cq$ (see section \ref{sec_QandA}). For example, given any two questions $Q_i,Q_j\in\cq$, $O$ could consider the correlation question $Q_{ij}$,  `are the answers to $Q_i,Q_j$ the same?'. Do we allow $Q_{ij}$ to {\it also} be implementable on $S$? 
Clearly, if $Q_i,Q_j$ are maximally compatible, then $Q_{ij}$ is implementable because $O$ can always find the answer to the latter by asking $Q_i,Q_j$. In this case, since $Q_i,Q_j$ are simultaneously defined relative to $O$, we can also write 
 \ba
 Q_{ij}:=Q_i\leftrightarrow Q_j,\label{correlation}
 \ea
 where $\leftrightarrow$ is the logical biconditional or (XNOR).\footnote{That is $Q_{ij}=$ `yes' if $Q_i=Q_j=$ `yes' or `no' and $Q_{ij}=$ `no' otherwise.} $Q_{ij}$ will then automatically be maximally compatible with both  $Q_i,Q_j$. 

Let us now contrast this with the situation in which $Q_i,Q_j$ are (partially or maximally) complementary. The structure introduced thus far does not preclude the `correlation question' $Q_{ij}$ to also be implementable and, ultimately, contained in $\cq$ even if $Q_i,Q_j$ are complementary. In fact, it is {\it not} possible to decide {\it without a theoretical model for describing and interpreting} $(\cq,\Sigma)$ whether $Q_{ij}$ is implementable on $S$ 
in the case that $Q_i,Q_j$ are at least partially complementary. What, however, is unambiguously clear is that {\it for all practical purposes} $Q_{ij}$ will not be implementable in this case. 
Namely, $Q_{ij}$ would be a statement about the correlation of two complementary questions $Q_i,Q_j$ and $O$ could say `the answers are the same', but he can never directly test them individually and see that they are actually `the same' at the same time. Indeed, $Q_i,Q_j,Q_{ij}$ would need to form a {\it mutually} (partially or maximally) 
complementary set. Thus, from a purely operational perspective alone, $O$ could never tell, even in principle, that $Q_{ij}$ was implementable; 
he simply cannot get an answer to $Q_{ij}$ which he could interpret, based on operationally accessible information alone, as a correlation of the complementary $Q_i,Q_j$.

It thus depends on the model for $(\cq,\Sigma)$ whether $Q_{ij}$ is implementable when $Q_i,Q_j$ are complementary. Since it is impossible for $O$ to settle this question operationally, it would require a model which employs hidden and operationally inaccessible ontic information and which is devoid of complementarity at an ontic level to conclude that $Q_{ij}\in\cq$ in this case. For example, Spekkens' elegant toy model \cite{spekkens2007evidence} and the `black boxes' of \cite{Paterek:2010fk} employ ontic states, satisfy the structure established thus far (at least at the epistemic level and modulo restrictions on the notion of convexity), and explicitly feature such a triple of questions.\footnote{The reader familiar with Spekkens' toy model \cite{spekkens2007evidence} will recall the simplest (epistemic) 1-bit system which has four ontic states `1', `2', `3' and `4'. An epistemic restriction forbids an observer to know the ontic state. Instead, the epistemic states of maximal knowledge correspond to either of the following three questions (and their negations) 
\begin{eqnarray}
&&Q_1: ``1\vee 2" \q\includegraphics[width=15mm]{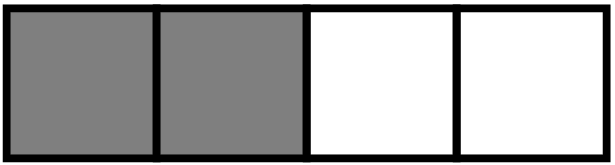}\,,\q\q
Q_2: ``2\vee3"\q\includegraphics[width=15mm]{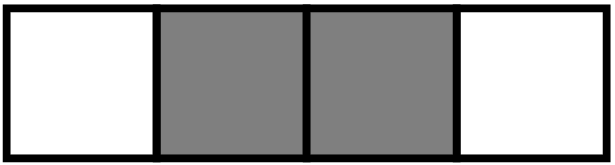}\,,\q\q\notag\\
&&Q_3: ``2\vee4"\q\includegraphics[width=15mm]{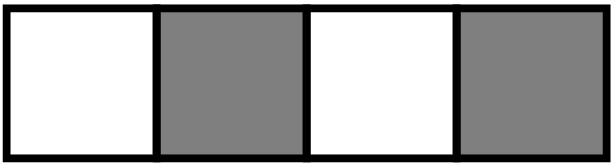}\,,   \notag
\end{eqnarray}
where $\vee$ is to be read as `or'. $Q_1,Q_2,Q_3$ are mutually complementary and it can be easily checked that $Q_3$ coincides with the `correlation' $Q_{12}$ of $Q_1$ and $Q_2$: it gives `yes' when the (ontic) answers to $Q_1,Q_2$ are equal and `no' otherwise. This relation is cyclic: $Q_1$ is also the `correlation' $Q_{23}$ of $Q_2,Q_3$ and $Q_2$ is the `correlation' $Q_{13}$ of $Q_1,Q_3$. Accordingly, there are three complementary questions for this 1-bit system, in principle the correct number for a qubit. 
Later in section \ref{sec_3D}, we will develop a different approach, {\it without} ontic states, in order to reason for the three-dimensionality of the Bloch-sphere. } However, such a model is in conflict with our premise of following a purely operational approach which only speaks about information that $O$ has access to via direct interrogation. 
The model by means of which $O$ interprets the answers that he gets from $S$ -- and, hence, the information he can acquire about $S$ -- shall be based entirely on operational statements that $O$ can, in principle, check through interrogation, and not on propositions that require hidden and inaccessible ontic information. Accordingly, $O$'s theoretical model for $\cq$ should {\it not} contain (a question which is logically equivalent to) $Q_{ij}$ whenever $Q_i,Q_j$ are at least partially complementary. 

Of course, the `correlation' (represented by the XNOR connective $\leftrightarrow$) is only one of many possible logical connectives. More generally, we shall require that $O$'s model for $\cq$ does {not} contain {\it any} logical connectives of complementary questions. He can only logically connect questions which are maximally compatible -- and thus are simultaneously defined with respect to him -- such that he could meaningfully write down a truth table for the questions to be connected and the connective question. If he cannot connect questions, he can also not ask for the connective.

 \begin{Assumption}\label{assump4}
 Let $Q_i,Q_j\in\cq$ and $*$ be a logical connective. According to $O$'s theoretical model for $\cq$, $Q_{i}*Q_j$ is implementable if and only if $Q_i,Q_j$ are maximally compatible.  
 \end{Assumption}
We recall from section \ref{sec_QandA} that $O$'s theoretical model for $\cq$ only contains implementable questions (but not necessarily all implementable questions). 

Since $*$ is only applied to connect maximally compatible questions which have also operationally simultaneous truth values relative to $O$, we shall allow $*$ to be any of the 16 binary connectives (or binary Boolean functions) $\neg,\vee,\wedge,\leftrightarrow,\ldots$ of classical Boolean or propositional logic.
   
We stress that assumption \ref{assump4} is a restriction on $O$'s model for $\cq$; this assumption clearly does not rule out that there may exist other (namely, ontic) models which also yield a consistent description of $O$'s experiences with the systems he is interrogating, yet which ascribe $Q_i*Q_j$ to be (logically equivalent to a question) contained in $\cq$ regardless of whether $Q_i,Q_j$ are compatible or not. However, we shall not worry about such models here. 

We note that assumption \ref{assump4} does not only apply to logical connectives of two questions, but arbitrarily many. For example, if $Q_i*Q_j$ is implementable (and contained in $\cq$), 
according to $O$'s model, then $(Q_i*Q_j)*Q_k$ is implementable too 
iff $Q_i*Q_j$ and $Q_k\in\cq$ are compatible, and so on. It is also important to note that assumption \ref{assump4} is a statement about which questions can be {\it directly} connected logically. It does {\it not} entail that $O$, according to this model, can only form meaningful logical expressions containing exclusively questions which are mutually maximally compatible. Namely, questions which are compatible might be logically further decomposable into other questions for which different compatibility relations hold. For instance, it might be that $Q_k$ in the expression $(Q_i*Q_j)*Q_k$ is only compatible with $Q_i*Q_j$, yet not with $Q_i$ or $Q_j$ individually. In this case, there is no harm in nevertheless forming the expression $(Q_i*Q_j)*Q_k$ even though, say, $Q_j$ and $Q_k$ might be complementary such that $O$ cannot write down a truth table with all three $Q_i,Q_j,Q_k$ separately. However, it is clear that in this case there is a hierarchy in which the logical connectives are to be executed. More specifically, the connective $*$ in $(Q_i*Q_j)*Q_k$ cannot be associative since $Q_i*(Q_j*Q_k)$ is not implementable 
according to $O$'s model because $Q_j,Q_k$ are, by assumption, complementary and can therefore not be directly connected. That is to say, the left $*$ in $(Q_i*Q_j)*Q_k$ must be executed first and $Q_k$ can only be logically connected with the result. Thus, assumption \ref{assump4} does admit logical expressions containing also mutually complementary questions as long as the latter are not {\it directly} connected through a logical connective. What assumption \ref{assump4} entails is that the ordering of the execution of the various logical connectives in a composition has to respect a hierarchy which ensures that at every step two compatible subexpressions are connected. This will become important later in the reconstruction where we shall encounter concrete such examples.

This brings us to the rules of inference by means of which $O$ may 
transform or evaluate logical expressions of questions and derive logical identities, e.g.\ to establish possible logical relations among various question outcomes. 
The state update rules need, in particular, respect these rules of inference. Again, the rules of inference which $O$ may employ depend on his theoretical model for $(\cq,\Sigma)$. For instance, in ontic models all questions will typically have simultaneous values, in contrast to purely operational ones, which will require different rules of inference.

In line with our operational premise, we shall require that the rules of inference of $O$'s model must not only be consistent with $O$'s experiences with the systems he is interrogating, but also testable. In particular, any logical identities derived from these rules must be testable through interrogations. 
It is thus appropriate to call them {\it operational rules of inference}.

Classical rules of inference require that the questions or propositions in a logical expression have truth values simultaneously. In $O$'s model any truth value must be operational such that only mutually maximally compatible questions have simultaneous meaning. Nothing stops $O$ from applying classical rules to such kind of questions. We shall thus require that it is appropriate for $O$ to employ classical rules of inference, according to Boolean logic, -- and only those --
for any logical expression or subexpression which only contains questions that are mutually maximally compatible. In all other cases (which still must abide by assumption \ref{assump4}), any possible rule of inference transforming the composition of questions will directly involve mutually complementary questions whose truth values, however, have no simultaneous operational meaning to $O$. 
Accordingly, classical rules will be operationally inappropriate in such cases and we shall demand more precisely the following.

\begin{Assumption}\label{assump4b}
$O$'s model for $\cq$ allows him to apply exclusively classical rules of inference, 
according to the rules of Boolean logic,  
to (transform or evaluate) \emph{any} logical expression (or subexpression of a larger expression) which is composed purely of mutually maximally compatible questions. This holds regardless of whether the mutually compatible questions can be logically decomposed into other questions which feature different compatibility relations. In all other cases, no classical rules of inference (and thus no classical logical identities) are valid.
\end{Assumption}
In consequence, $O$ can take any set of mutually compatible questions and treat them, with classical logic, as a Boolean algebra. The assumption states what is possible for compositions of compatible questions and what is {\it not} possible for compositions involving also complementary questions. This will become useful later in section \ref{sec_bell} where we shall also see what {\it is} possible for compositions of complementary questions.

We emphasize that assumption \ref{assump4b} by itself does {\it not} severely constrain the nature of any conceivable `hidden variable model' which could also consistently describe $O$'s world. However, in conjunction with two of the subsequent quantum principles, it will rule out {\it local} `hidden variables' in section \ref{sec_bell}.

\subsubsection{Parametrization of $S$'s state and tomography}

Now that we have a notion of independence on $\cq$ we can say more about the parametrization and thus representation of $S$'s state relative to $O$. Not all $Q_i\in\cq$ will be necessary to describe the state; the pairwise independent questions shall be the fundamental building blocks of the landscape $\cl$ of inference theories. 

Suppose there is a maximal set $\cq_M=\{Q_1,\ldots,Q_D\in\cq\}$ of $D$ pairwise independent (but not necessarily compatible) questions, such that no further question $Q\in\cq\setminus\cq_M$ exists which is pairwise independent from all members of $\cq_M$ too. Then, every other $Q\in\cq\setminus\cq_M$ is either (i) dependent on exactly one $Q_j\in\cq_M$ and independent of all other $\cq_M\ni Q_l\neq Q_j$ (if $Q$ was dependent on $Q_j\in\cq_M$ and partially dependent on $\cq_M\ni Q_l\neq Q_j$, then $Q_j,Q_l$ could not be independent), (ii) partially dependent on some and independent of the other questions in $\cq_M$, or (iii) partially dependent on all $Q_j\in\cq_M$. While $O$ will not be able to infer maximal information about the answers to questions of cases (ii) and (iii) from his information about individual (or even subsets of) members of $\cq_M$ alone, the question is whether his information about the full set $\cq_M$ will be sufficient to do so. 

\begin{Definition}{\bf(Informational Completeness)}
A maximal set $\cq_M=\{Q_1,\ldots,Q_D\in\cq\}$ of pairwise independent questions is said to be \emph{informationally complete} if $O$'s information about the questions in $\cq_M$ determines his information about all other $Q\in\cq\setminus\cq_M$ in such a way that the probabilities $y_i$ which $O$ assigns to every $Q_i\in\cq_M$ are sufficient in order for him to compute the probabilities $y_j\,\forall\,Q_j\in\cq$ for all preparations of $S$. In this case, the set of probabilities $\{y_i\}_{i=1}^D$ of the $Q_i\in\cq_M$ parametrizes the state that $O$ assigns to $S$ and thereby yields a complete description of the state space $\Sigma$. We shall call $D$ the dimension of $\Sigma$.
\end{Definition}

We emphasize that the dimension $D$ of $\Sigma$ will ultimately {\it not} be the Hilbert space dimension in quantum theory, but the dimension of the set of density matrices.

If $\cq_M$ was not informationally complete, $O$ would require further questions that are partially dependent on at least some of the elements in $\cq_M$ in order to fully describe the system $S$ and its state. This situation cannot be precluded, given the structure we have devised so far. However, we deem it undesirable, given that we would like to employ pairwise independent questions as building blocks for system descriptions. We shall therefore require that no more independent information about $S$ can be learned from any question in addition to a maximal set $\cq_M$.


\begin{Assumption}\label{assump6}
Every maximal set $\cq_M$ of pairwise independent questions is informationally complete.
\end{Assumption}

There may exist (even continuously) many such informationally complete sets of questions on $\cq$ which, at this stage, may still be either discrete or continuous. However, their dimensions are equal.

\begin{lem}
The dimensions of all maximal sets $\cq_M$ on $\cq$ are equal (if finite).\end{lem}
\begin{proof}
Consider any $\cq_M$. On account of pairwise independence, for each $Q_i\in\cq_M$ there must exist a state such that $O$ only knows the answer to this $Q_i$ with certainty, but nothing about any other $Q_{j\neq i}\in\cq_M$, i.e.\ $y_i=0$ or $1$ and all other $y_{j\neq i}=1/2$. (Namely, $O$ could have asked $Q_i$ to $S$ prepared in the state of no information.) But since also in the state of no information $y_{j\neq i}=1/2$, $y_i$ cannot in general be computed from the $y_{j\neq i}$. Hence, given an informationally complete set $\cq_M$, all associated probabilities $\{y_i\}_{i=1}^D$ are necessary to parametrize the full state space $\Sigma$. 

Now $\Sigma$ is a convex set (see assumption \ref{assump2}). Given the lack of redundancy in $\{y_i\}$, a given informationally complete $\cq_M$ with $D$ elements therefore encodes $\Sigma$ as a $D$-dimensional convex subregion of $\mathbb{R}^D$. Suppose there is a second $\cq_M'$ with $D'$ elements. In terms of $\cq_M'$, $\Sigma$ would be described as a $D'$-dimensional convex subregion of $\mathbb{R}^{D'}$. But clearly, these two descriptions of $\Sigma$ must be isomorphic which, for $D$ finite, is only possible if $D\equiv D'$.
\end{proof}

Finiteness of $D$ will later be established in section \ref{sec_qstructure} with the help of the principles.


Any such $\cq_M$ establishes a {\it question reference frame} on $\cq$ (see also \cite{brukner2009information}) and thereby also a `coordinate system' on $\Sigma$. In particular, in order to do {\it state tomography} with a multiple shot interrogation, as outlined in section \ref{sec_bayes}, it will be sufficient for $O$ to interrogate an ensemble of identically prepared $S$ with the questions within a given $\cq_M$ only. Given a specific $\cq_M$, there are now three equivalent ways for $O$ to describe $S$'s state: he could represent it by either the $D$-dimensional {\it yes-} or {\it no-vector}
\ba
\!\!\!\!\vec{y}_{O\rightarrow S}:=\left(\begin{array}{c}y_1 \\y_2 \\\vdots \\y_D\end{array}\right),\q\q \vec{n}_{O\rightarrow S}:=\left(\begin{array}{c}n_1 \\n_2 \\\vdots \\n_D\end{array}\right),
\ea
of probabilities $y_i$ and $n_i$, $i=1,\ldots,D$, that the answers to question $Q_i\in\cq_M$ are `yes' and `no', respectively. That is,
\ba
\vec{y}_{O\rightarrow S}+\vec{n}_{O\rightarrow S}=p\,\vec{1},\label{norm}
\ea
where $\vec{1}$ is a $D$-dimensional vector with a $1$ in each of its entries. Here we have introduced a new parameter: $p$ is the probability that $S$ is present at all. (For example, the method of preparation could be such that $O$ cannot always tell with certainty when a system is spit out by the preparation device.) This probability will rescale all $y_i$ and $n_i$ simultaneously. However, it only becomes relevant in subsection \ref{subsec_time} and otherwise is usually taken to be $p=1$. Evidently, the assignment of which answer to $Q_i$ is `yes' and which is `no' is arbitrary, but any consistent such assignment is fine for us. 

But he could also represent the state redundantly as a $2D$-dimensional vector
\ba
\vec{P}_{O\rightarrow S}=\left(\begin{array}{c}\vec{y} \\\vec{n}\end{array}\right)\label{2dvec}
\ea
which will turn out to be convenient especially when $p<1$. We shall write a state with a subscript $O\rightarrow S$ to emphasize that it is the state of $O$'s information about $S$.

Lastly, this structure also puts us into the position to specify $O$'s total amount of information about $S$. Clearly, the total amount of information must be a function of the state. Let $\cq_M=\{Q_1,\ldots,Q_D\}$ be an informationally complete set of questions in $\cq$. Given that these questions carry the entire information $O$ may know about $S$, we define the total information $I_{O\rightarrow S}$ as the sum of $O$'s information about the $Q_i\in\cq_M$, as measured by the $\alpha_i$ (\ref{infoineq}):
\ba
I_{O\rightarrow S}\left(\vec{y}_{O\rightarrow S}\right):=\sum_{i=1}^D \,\alpha_i.\label{totalinfo}
\ea
(We imagine $O$ as an agent who can write down results on a piece of paper and add these up.) The specific relation between $\alpha_i$ and $\vec{y}_{O\rightarrow S}$ will be derived later in section \ref{sec_infomeasure}; it will {\it not} be the Shannon entropy which, as discussed in section \ref{sec_shannon}, describes average information gains in repeated experiments rather than the information content in the state $\vec{y}_{O\rightarrow S}$.

\subsubsection{Composite systems}

For later purpose we need to clarify the notion of a composite system. Since we are pursuing a purely operational approach, the notion of a composition of systems must be defined in terms of the information accessible to $O$ through interrogation. 
$O$ should, in principle, be able to tell a composite system apart into its constituents. 
Accordingly, we require that the set of questions which can be posed to the composite system contains all questions about the individual subsystems and that the remaining questions are literally composed of these individual questions or iteratively composed of compositions of them.\footnote{For the GPT specialists, we emphasize that this definition has nothing to do with the usual requirement of {\it local tomography} \cite{masanes2011derivation,Mueller:2012ai,barrett2007information,Dakic:2009bh,Masanes:2011kx,de2012deriving,Masanes:2012uq} in GPTs. To give a concrete example, we shall see later that two-level systems over real Hilbert spaces (rebits) satisfy this definition while violating local tomography.}

\begin{Definition}\label{def_comp}{\bf(Composite System)}
Let $\cq_A,\cq_B$ denote the full question sets associated to $S_A,S_B$, respectively. Two systems $S_A,S_B$ are said to form a \emph{composite system} $S_{AB}$ if any $Q_a\in\cq_A$ is maximally compatible with and independent of any $Q_b\in\cq_B$ and if
\ba
\cq_{AB}=\cq_A\cup\cq_B\cup\tilde{\cq}_{AB},
\label{composite}
\ea
where $\tilde{\cq}_{AB}$ only contains composite questions which are iterative compositions, $Q_a\,*_{\tiny1}\,Q_b, Q_a\,*_2(Q_{a'}*_3Q_b), (Q_a*_4Q_b)*_5Q_{b'}, (Q_a*_6Q_b)*_7(Q_{a'}*_8Q_{b'}),\ldots$, via some logical connectives $*_1,*_2,*_3,\cdots$, of individual questions $Q_a,Q_{a'},\ldots\in\cq_A$ about $S_A$ and $Q_b,Q_{b'},\ldots\in\cq_B$ about $S_B$. For an $N$-partite system, we use this definition recursively.
 \end{Definition}
 
Of course, thanks to assumption \ref{assump4}, $O$ can only directly compose questions with a logical connective $*$ if they are compatible. 
 
There are further repercussions: let $\cq_{M_A},\cq_{M_B}$ be informationally complete sets for $S_A,S_B$, respectively. Then an informationally complete set $\cq_{M_{AB}}$ for a composite $S_{AB}$ can be formed iteratively by joining (in a set union sense) $\cq_{M_A},\cq_{M_B}$ and adding to it a maximal pairwise independent set of questions which are the logical connectives of members of $\cq_{M_A}$ with elements of $\cq_{M_B}$ and/or iterative compositions of such compositions and possibly questions from $\cq_{M_A},\cq_{M_B}$. In section \ref{sec_connective} below, we shall determine which logical connectives $*$ we may employ to build up $\cq_{M_{AB}}$ from $\cq_{M_A},\cq_{M_B}$. 

\subsubsection{Time evolution of $S$'s state}\label{subsec_time}

We shall assume that $O$ has access to a clock and begin with a definition:
\begin{description}
\item[\bf static states:] A state $\vec{P}_{O\rightarrow S}$ which is {\it constant} in time -- corresponding to the situation that $O$ will always assign the same probability to each of his question outcomes -- is called {\it static}.
\end{description}
If all the states $O$ would assign to any system $S$ were {\it static} -- according to his information -- $O$'s world would be a rather boring place. In order not to let $O$ die of boredom, we allow the probability vector $\vec{P}_{O\rightarrow S}$ to change in time. (But we require that $\cq$ and $\Sigma$ are {\it time independent}.) We shall now make use of operational reasoning, to briefly consider how $\vec{P}_{O\rightarrow S}$ evolves under time evolution. The following argument bears some analogy to an argument typically employed in the GPT framework \cite{Hardy:2001jk,barrett2007information,masanes2011derivation,Mueller:2012ai,Dakic:2009bh,Masanes:2012uq} concerning convex mixtures and measurements, however, is distinct in nature as we instead consider states and their mixtures under time evolution. 

Let $O$ have access to two {\it identical}\footnote{That is, both $S_1$ and $S_2$ carry the same $\cq$ and $\Sigma$.} non-interacting systems $S_1$ and $S_2$. $O$'s information about $S_1$ and $S_2$ can be different such that they are allowed to be in distinct states $\vec{y}_{O\rightarrow S_1}$ and $\vec{y}_{O\rightarrow S_2}$ relative to $O$. 
Now let $O$ perform a (biased) coin flip which yields `heads' with probability $\lambda$ and `tails' with probability $(1-\lambda)$.\footnote{Equivalently, $O$ may use another system to which he has assigned a stable state vector by repeated interrogations on identically prepared systems. Given an arbitrary elementary question $Q$, he can use the probability that the answer is `yes' as $\lambda$ and the probability that the outcome is `no' as $(1-\lambda)$ and thus use $Q$ as the `coin'.} 
Given that $S_1$ and $S_2$ are identical, $O$ can ask the same questions to both systems. If the coin flip yields `heads', $O$ will interrogate $S_1$, if it yields `tails' $O$ will interrogate $S_2$. 
\begin{eqnarray}
\psfrag{O}{$O$}
\psfrag{s1}{$S_1$}
\psfrag{s2}{$S_2$}
\psfrag{l}{$\lambda$}
\psfrag{n}{$1-\lambda$}
\includegraphics[scale=.3]{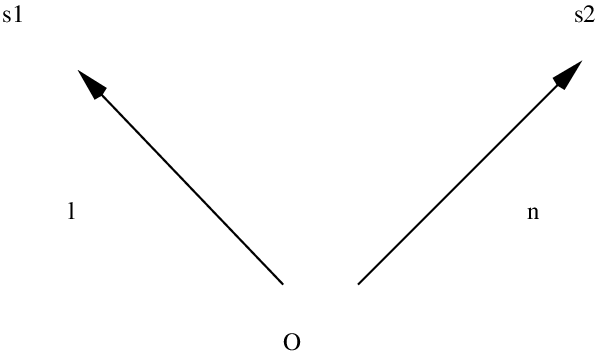}   \notag
\end{eqnarray}

In particular, before tossing the coin, say at time $t_1$, the probability he assigns to receiving a `yes' answer to question $Q_i$ (asked to either $S_1$ or $S_2$ depending on the outcome of the coin flip) is simply the convex sum $y^{12}_i(t_1)=\lambda \,y^1_i(t_1)+(1-\lambda)\,y^2_i(t_1)$. ($O$ is allowed to build this convex combination thanks to assumption \ref{assump2}.) But this holds at {\it any} time $t$ before $O$ tosses the coin (which $O$ can also choose never to do).\footnote{We assume $\lambda$ to be constant in time.} As will become clear shortly, it is convenient to consider the states $\vec{P}_{O\rightarrow S}$ with $0\leq p\leq1$ for now. We shall return to the yes-vector later in section \ref{sec_time}. 
We then have for all $t$
\ba
\!\!\!\!\!\!\!\!\vec{P}_{O\rightarrow S_{12}}(t)=\lambda\,\vec{P}_{O\rightarrow S_1}(t)+(1-\lambda)\,\vec{P}_{O\rightarrow S_2}(t).\label{convex}
\ea
This convex combination is the state of an `effective' system $S_{12}$ (identical to $S_1,S_2$) describing $O$'s information about the coin flip scenario. Note that $S_{12}$ is {\it not} a composite system according to definition \ref{def_comp} because $\cq_{12}=\cq_1=\cq_2$.

Denote by $T_k$ the time evolution map of the state of system $S_k$, $k=1,2,12$, from $t=t_1$ to $t=t_2$. Under the assumption that $\vec{P}_{O\rightarrow S_{1,2}}(t)$ evolve independently of each other, given that $S_1$ and $S_2$ do not interact, equation (\ref{convex}) implies
\ba
\! T_{12}[\vec{P}_{O\rightarrow S_{12}}(t_1)]&=&T_{12}[\lambda\,\vec{P}_{O\rightarrow S_1}(t_1)\nn\\
&&\q\,\,\q+(1-\lambda)\,\vec{P}_{O\rightarrow S_2}(t_1)]\nn\\
&=&\lambda\,T_1[\vec{P}_{O\rightarrow S_1}(t_1)]\label{T}\\&&\,\,\q\q+(1-\lambda)\,T_2[\vec{P}_{O\rightarrow S_2}(t_1)].\nn
\ea
All three states $\vec{P}_{O\rightarrow S_{k}}$, $k=1,2,12$, are elements of the same $\Sigma$. As such, we assume the following.

\begin{Assumption}\label{assump7}
Every time evolution map can act on {\it any} state in $\Sigma$.
\end{Assumption}

Next, we restrict to the situation where $O$ exposes both $S_1,S_2$ to evolve under the same time evolution $T:=T_1=T_2$. Thanks to assumption \ref{assump7}, (\ref{T}) becomes in this case
\ba
T[\vec{P}_{O\rightarrow S_{12}}(t_1)]&=&\lambda\,T[\vec{P}_{O\rightarrow S_1}(t_1)]\nn\\&&\q\q\q+(1-\lambda)\,T[\vec{P}_{O\rightarrow S_2}(t_1)]\nn
\ea
and $T$ is \emph{convex linear}. This is clear since $O$ may, in particular, prepare $S_1,S_2$ in identical states $\vec{P}_{O\rightarrow S_{1}}(t_1) = \vec{P}_{O\rightarrow S_{2}}(t_1)=\vec{P}_{O\rightarrow S_{12}}(t_1)$. This special case, when inserted into (\ref{T}), yields $T_{12}=T$. 

This scenario can be easily generalized to arbitrarily many identical systems. Remarkably, it can be shown that convex linearity of operations on probability vectors, in fact, implies full linearity of the operation \cite{Hardy:2001jk,barrett2007information}. Hence, the time evolution of the state must be {\it linear}:
\ba
\vec{P}_{O\rightarrow S}(t_2)&=&T[\vec{P}_{O\rightarrow S}(t_1)]\nn\\&=&A(t_1,t_2)\,\vec{P}_{O\rightarrow S}(t_1)+\vec{V}(t_1,t_2),\nn
\ea
where $A(t_1,t_2)$ is a $2D\times2D$ matrix and $\vec{V}$ some $2D$-dimensional vector. Given our assumption above that $O$'s world is not {\it static}, $A(t_1,t_2)$ will generally be a non-trivial matrix. Requiring that the last relation also holds at $t_1=t_2$ for all initial states, it follows that $\vec{V}(t,t)=0$ and $A(t,t)=\mathbb{1}$. Demanding further that the special state $\vec{P}_{O\rightarrow S}=\vec{0}$, corresponding to $p=0$ and the absence of a system, is invariant under time evolution (such that no system or any information is created out of `nothing'), finally yields $\vec{V}\equiv\vec{0}$, $\forall\,t_1,t_2$.\footnote{A similar argument would {\it not} work for the evolution of $\vec{y}_{O\rightarrow S}$ because, in principle, $\vec{n}_{O\rightarrow S}\neq\vec{0}$ is possible even if $\vec{y}_{O\rightarrow S}=\vec{0}$ (where $\vec{0}$ is a $D$-dimensional zero vector). This is the reason why here it is more convenient to work with $\vec{P}_{O\rightarrow S}$, which contains the information about both $\vec{y}_{O\rightarrow S},\vec{n}_{O\rightarrow S}$, rather than $\vec{y}_{O\rightarrow S}$ alone. In fact, we shall see later in section \ref{sec_time} that the evolution of the latter {\it does} involve an affine $D$-dimensional vector $\vec{V}'\neq\vec{0}$.} Given that $\vec{P}_{O\rightarrow S}(t)$ is a probability vector for all $t$ and (\ref{norm}) must always hold, $A(t_1,t_2)$ must be a {\it nonnegative} (real) matrix which is {\it stochastic} in any pair of components $i$ and $i+D$ of $\vec{P}_{O\rightarrow S}$.

Consequently, the elementary and natural assumption \ref{assump7} rules out non-linear and state dependent time evolution, such as in Weinberg's non-linear extension of quantum mechanics \cite{Weinberg:1989cm}, and thereby also eliminates the many operational problems that arise with it (e.g., superluminal signaling) \cite{gisin1990weinberg,Polchinski:1990py,Bennett:2009rt}.\footnote{It should be emphasized that assumption \ref{assump7} is also tacitly made in the GPT framework \cite{Hardy:2001jk,barrett2007information,masanes2011derivation,Mueller:2012ai,Dakic:2009bh,Masanes:2012uq} for any kind of transformations on states, such that the present setup  is in this regard not less potent.}

For later purpose, we impose a further condition on the evolution of $S$'s state. 
\begin{Assumption}\label{assump8}{\bf(Temporal Translation Invariance)}
There exist no distinguished instants of time in $O$'s world such that $O$ is free to set any instant he desires as the instant $t=0$. Time evolution, as perceived by $O$, is therefore (temporally) translation invariant $A(t_1,t_2)=A(t_2-t_1)$.
\end{Assumption}
Since the time evolution matrix $A$ can only depend on the duration elapsed, but not on the particular instant of time, we can collect the above results in the simple form:
\ba
\!\!\!\!\!\!\!\!\!\!\!\vec{P}_{O\rightarrow S}(t_2)=A(\Delta t)\,\vec{P}_{O\rightarrow S}(t_1),\q\!\Delta t=t_2\!-\!t_1.\label{linear}
\ea

We do not yet have sufficient evidence to conclude that a given time evolution will be described by a continuous one-parameter matrix group because we neither know (i) that $A$ is invertible for all $\Delta t$, (ii) that the evolution is actually continuous, nor (iii) that every composition of evolution matrices is again an evolution matrix. We shall return to this question in section \ref{sec_time} with the help of the set of principles for quantum theory and also defer the discussion of the time evolution of $\vec{y}_{O\rightarrow S}$ until then.

For now we note, however, firstly, that any $A(\Delta t)$ acting on $\vec{P}_{O\rightarrow S}$ implies a unique map $T_{\Delta t}(\vec{y}_{O\rightarrow S})$ acting on the yes-vector thanks to (\ref{norm}, \ref{2dvec}). Secondly, a multiplicity of time evolutions of $S$ is possible, depending on the physical circumstances (interactions) to which $O$ may subject $S$. The set of all possible time evolutions   which $O$ shall henceforth be able to implement will be denoted by $\ct$. This set need not contain all physically possible time evolutions. Indeed, upon imposing the quantum principles, $\ct$ will become the set of unitaries (rather than of arbitrary completely positive maps). Clearly, $\ct$ is part of $O$'s model for describing $S$; his model is thus a triple $(\cq,\Sigma,\ct)$.

\subsubsection{The theory landscape $\cl$}\label{sec_cl}

The landscape $\cl$ of theories describing $O$'s acquisition of information about a physical system $S$ is now the set of all theories which comply with the structure and assumptions established in this section.

In the sequel, we shall restrict $O$'s attention solely to composite systems of $N\in\mathbb{N}$ generalized bits (gbits), where a single gbit is characterized by the fact that $O$ can maximally know the answer to a single question at once such that it can carry at most one \texttt{bit} of information. 
Every inference theory specifies for every system of $N$ gbits a triple $(\cq_N,\Sigma_N,\ct_N)$. We shall be concerned with the landscape of gbit theories $\cl_{\rm gbit}$ which contains all gbit theories, satisfying assumptions \ref{assump1}--\ref{assump8}. To name concrete examples at this stage, $\cl_{\rm gbit}$ contains, among a continuum of other theories, classical bit, rebit and qubit theory for all $N\in\mathbb{N}$. We shall see more of this later, but for now we summarize their characteristics for $N=1$, 
\begin{description}
\item[classical bit theory] gives $\cq_1=\{Q,\neg Q\}$, $\Sigma_1\simeq[0,1]$ (for normalized states) with extremal points corresponding to $Q=$ `yes' and `no', respectively, and $\ct_1\simeq\mathbb{Z}_2$.\footnote{There are precisely two states of maximal information that $O$ can assign to a classical bit
\ba
\vec{P}_{O\rightarrow S}=\left(\begin{array}{c}1 \\0\end{array}\right),\q\q\q\vec{P}_{O\rightarrow S}= \left(\begin{array}{c}0 \\1\end{array}\right),\nn
\ea
corresponding to `yes' and `no' answers to $Q$. Time evolution is described by the abelian group $\mathbb{Z}_2$, given by
\ba
\mathds{1}=\left(\begin{array}{cc}1 & 0 \\0 & 1\end{array}\right),\q\q\q\q P=\left(\begin{array}{cc}0 & 1 \\1 & 0\end{array}\right),\nn
\ea
where $P$ swaps the two states. For classical bit theory, the permitted time evolution is therefore {\it discrete}.}

\item[rebit theory] (two-level systems on {\it real} Hilbert spaces) yields $\cq_1\simeq S^1$ and every two maximally complementary questions are informationally complete. $\Sigma_1\simeq D^2$ (for normalized states) and $\ct_1\simeq\rm{SO}(2)$.

\item[qubit theory] has $\cq_1\simeq S^2$ and every triple of mutually maximally complementary questions are informationally complete. $\Sigma_1\simeq B^3$ (for normalized states) and $\ct_1\simeq\rm{SO}(3)$.
\end{description}

With all the assumptions made along the way, the theory landscape devised here is perhaps not quite as general as the commonly employed landscape of generalized probability theories \cite{Hardy:2001jk,barrett2007information,masanes2011derivation,Mueller:2012ai,Dakic:2009bh,Masanes:2012uq,pfister2013information,Barnum:2014fk,Barnum:2010uq}. In particular, GPTs easily handle arbitrary finite dimensional systems, while this remains to be done within the present framework. Nevertheless, the new landscape $\cl_{\rm gbit}$ 
 is large enough and provides new tools for a non-trivial and instructive (re)construction of qubit quantum theory that approaches the latter from a conceptually and technically new angle. To facilitate future generalizations of this work, we have also attempted to spotlight the assumptions underlying $\cl_{\rm gbit}$ as clearly as possible in this section.

\section{Principles for the quantum theory of qubits as rules on information acquisition}\label{sec_postulates}

We shall now use the landscape $\cl_{\rm gbit}$ of gbit theories and formulate the principles for the quantum theory of qubit systems on it as rules on an observer's information acquisition. These principles constitute a set of `coordinates' of qubit quantum theory on $\cl_{\rm gbit}$.

\subsection{The rules}\label{sec_principles}

As we are reconstructing a technically and empirically well-established theory, we are in the fortunate position to avail ourselves of empirical evidence and earlier ideas on characterizing quantum theory in order to motivate a set of basic postulates. However, the ultimate justification for these postulates will be their success in singling out qubit quantum theory within the inference theory landscape $\cl_{\rm gbit}$ which will be completed in \cite{hw}. As `coordinates' on a theory space these principles will not be unique and one could find other equivalent sets. (As usual, at least many roads lead to Rome.) But we shall take the below ones as a first working set which restricts the informational relation between $O$ and $S$, where $S$ will be a composite system (c.f.\ definition \ref{def_comp}) of $N\in\mathbb{N}$ gbits. Each principle will be first expressed as an intuitive and colloquial statement, followed by its mathematical meaning within $\cl_{\rm gbit}$.

We take it as an empirical fact that there exist physical systems about which only a limited amount of information can be known at any one moment of time. The standard quantum example is a spin-$\f{1}{2}$ particle about which an experimenter may only know its polarization in a given spatial direction, but nothing independent of that; the polarization `up' or `down' corresponds to one \texttt{bit} of information. But there is also a typical classical example, namely a ball which may be located in either of two identical boxes, the definite position `left' or `right' corresponding to one \texttt{bit} of information. Informationally, these two examples incarnate the most elementary of systems, a gbit, supporting maximally just one proposition at a time. But, clearly, there are more complicated systems supporting other limited amounts of information. We shall take this simple observation and raise the existence of an information limit to the level of a principle.

More precisely, in analogy to von Weizs\"acker's `ur-theory' \cite{weizsaeckerbook,gornitz2003introduction,lyre1995quantum}, we shall restrict $O$'s world to be a world of elementary alternatives which thereby consists only of systems which can be decomposed into elementary gbits.\footnote{However, in contrast to \cite{weizsaeckerbook,gornitz2003introduction,lyre1995quantum}, we shall be much less ambitious here and will not attempt to deduce the dimension of space or space-time symmetry groups from systems of elementary alternatives. Recent developments \cite{Mueller:2012pc,Masanes:2011kx,Hoehn:2014vua,Dakic:2013fk}, on the other hand, unravel a deep relation between (a) simple conditions on operations with systems carrying finite information (e.g., communication with physical systems) and (b) the dimension and symmetry group of the ambient space or space-time.} All physical quantities in $O$'s world are to be finite such that he can record his information about them on a finite register. We wish to characterize the composite systems in $O$'s world according to the finite limit of $N$\,\texttt{bits} of information he can maximally inquire about them.   
\begin{pr}\label{lim}{\bf(Limited Information)}
\emph{``The observer $O$ can acquire maximally $N\in\mathbb{N}$ {\it independent} \texttt{bits} of information about the system $S$ at any moment of time.''} \\
There exists a maximal set $Q_i$, $i=1,\ldots,N$, of $N$ mutually independent and maximally compatible questions in $\cq_N$ and no subset in $\cq_N$ can contain more than $N$ questions with that property.
\end{pr}
In other words, $O$ can ask maximally $N$ independent questions\footnote{Obviously, combinations (e.g., `correlations') of the compatible $Q_i$ will define other \texttt{bits} of information which, however, will be dependent once the $Q_1,\ldots,Q_N$ are asked (we shall return to this in detail in section \ref{sec_qstructure}).} at a time to $S$. Accordingly, this rule immediately implies that $O$ can distinguish maximally $2^N$ states of $S$ in a single shot interrogation because there will be $2^N$ possible answers to the $Q_1,\ldots,Q_N$.  Since $\cq_N,\Sigma_N$ are intrinsic to $S$, also the maximal amount of $N$ \texttt{bits} that each $S$ can carry must be intrinsic to it and thus be observer independent. In fact, this rule can be regarded as a defining property of {\it all} inference theories in the gbit landscape $\cl_{\rm gbit}$ and can thus clearly not distinguish between classical bits, rebits, qubits, etc.

We take another empirical fact and elevate it to a fundamental principle: despite the limited information accessible to an experimenter at any moment of time, there always exists {\it additional} independent information that she may learn about the observed system at other times. This is Bohr's complementarity principle \cite{Bohr}. Consider, for instance, the prototypical quantum physics experiment: Young's double slit experiment. The experimenter can choose whether to obtain which-way-information or an interference pattern, but not both; the complete knowledge of whether the particle went through the left or right slit is at the expense of total ignorance about the information pattern and vice versa. (For an informational discussion of the particle-wave duality in Young's double slit experiment and a Mach-Zehnder interferometer, see also \cite{Brukner:2002kx,Brukner:1999qf}.) Similarly, in a Stern-Gerlach experiment an experimenter may determine the polarization of a spin-$\f{1}{2}$ particle in $x$-direction, but will be entirely oblivious about the polarization in $y$- and $z$-direction. A subsequent measurement of the spin of the same particle in $y$-direction will render her previous information about the polarization in $x$-direction obsolete and keep her ignorant about the spin in $z$-direction and so on. That is, systems empirically admit many more independent questions than they are able to answer at a time -- thanks to the information limit. We shall now return to the relation between $O$ and $S$ and accordingly stipulate that complementarity exists in $O$'s world, however, we shall say nothing more about {\it how much} complementary information may exist.

\begin{pr}\label{unlim}{\bf(Complementarity)}
\emph{``The observer $O$ can always get up to $N$ \emph{new} independent \texttt{bits} of information about the system $S$. But whenever $O$ asks $S$ a new question, he experiences no net loss in his total information about $S$.''}\\
There exists another maximal set $Q_i'$, $i=1,\ldots,N$, of $N$ mutually independent and maximally compatible questions in $\cq_N$ such that $Q'_i,Q_i$ are maximally complementary and $Q'_i,Q_{j\neq i}$ are maximally compatible and independent.\footnote{Alternatively, one could formulate the technical part of this rule as follows: \emph{In the $N=1$ case of a single gbit there exists, for every $Q_1\in\cq_1$, another $Q_1'\in\cq_1$ such that $Q_1,Q_1'$ are maximally complementary.} Definition \ref{def_comp} would then immediately imply that a composite system of $N$ gbits features two sets of questions, $Q_i$, $i=1,\ldots,N$, and $Q_j'$, $j=1,\ldots,N$, with the following properties: $Q_i,Q_i'$ are questions about the elementary gbit labeled by $i$, all $Q_i$ are maximally independent and maximally compatible, all $Q_j'$ are maximally independent and maximally compatible and any pair $Q_i,Q_{j\neq i}'$ is maximally independent and compatible and every pair $Q_i,Q_i'$ corresponding to the same gbit is maximally complementary. This is a priori not equivalent to the current formulation of the technical part of rule \ref{unlim} in the main text as the latter does not necessarily refer to questions about individual subsystems of a composite system. However, the alternative formulation would also be enough to get the correct structure of the informationally complete sets for $N$ qubits and $N$ rebits below. While the alternative formulation is simpler, we keep the other one as it is also published in this form in the companion article \cite{hw}.}
\end{pr}
That is, after asking $S$ a set of $N$ independent elementary propositions $Q_i$, $i=1,\ldots,N$, in a single shot interrogation, $O$ can pose a new $(N+1)$th elementary question $Q'_1$ to $S$. Since $O$ can only know $N$ independent \texttt{bits} of information about $S$ at a time and asking a new question does not lead to a {\it net} loss of information, the single \texttt{bit} of his previous information about $Q_1$ must have become obsolete upon learning the answer to $Q'_1$, while $O$'s total information about $S$ is still $N$ \texttt{bits}. The rule allows $O$ to perform the same procedure until he replaces his $N$ old by $N$ {\it new} \texttt{bits} of information about $S$. 
As $O$'s information about $S$ has changed, the state of $S$ relative to $O$ will necessarily experience a `collapse' whenever he asks a new complementary question. As a result, it is the observer who decides {\it which} information (e.g., spin in $x$- or $y$-direction) he will obtain about the system by asking specific questions. But clearly $O$ will have no influence on what the answer to these questions will be and any answer will come at the price of total ignorance about complementary questions.

A few further explanations concerning this complementarity rule are in place. First of all, this rule clearly rules out classical bit theory. Secondly, the postulate asserts the existence of maximally complementary questions in $O$'s world, however, makes no statement about whether partially independent or partially complementary questions may exist too. Thirdly, the peculiar requirement that every question of rule \ref{unlim} is complementary to exactly one from rule \ref{lim} is chosen such that each subsystem of a composite system features complementarity. For example, consider the $N=1$ case, corresponding to a single gbit, for which rules \ref{lim} and \ref{unlim} entail a complementary pair $Q_1,Q_1'$. This complementarity {\it per} gbit generalizes to arbitrary $N$ since every $Q_i$ in rule \ref{lim} and every $Q'_j$ in rule \ref{unlim} may correspond to one of $N$ gbits. More complicated complementarity relations arising from rules \ref{lim} and \ref{unlim} will be discussed in section \ref{sec_qstructure}.

Notice that the complementarity rule implies the existence of a notion of {\it `superposition'} -- {\it even in states of maximal information}. For example, take $N=1$ and let $O$ know the answer to $Q_1$ with certainty, $y_1=1$, such that his information about $S$ saturates the limit of one \texttt{bit}. This will leave him oblivious about the complementary $Q'_1$ such that he would have to assign $y'_1=\f{1}{2}$. The state of information he has about $S$ can be interpreted as being in a `superposition' of the $Q'_1$ alternatives `yes' and `no', but in this case with `equal weight' because both alternatives are equally likely according to $O$'s knowledge.\footnote{We shall later see that, even in a state of maximal information of $N$ \texttt{bits}, $O$ may possibly have only partial or incomplete knowledge about the outcomes of {\it any} question in an informationally complete set $\cq_{M_N}$ of pairwise independent questions (c.f.\ assumption \ref{assump6}). In this case, $O$'s information about all questions in $\cq_{M_N}$ will be in a general state of superposition because their corresponding probabilities cannot all be $\f{1}{2}$, otherwise it would be the state of no information.  }
 We emphasize that the necessary presence of superpositions of elementary alternatives, as perceived by $O$, is a consequence of the information limit and complementarity.

We would like to stress that rules \ref{lim} and \ref{unlim} are conceptually motivated by related proposals which have been put forward first by Rovelli within the context of relational quantum mechanics \cite{Rovelli:1995fv} and later, independently, by Brukner and Zeilinger within attempts to understand the structure of quantum theory via limited information \cite{zeilinger1999foundational,Brukner:ys,Brukner:1999qf,Brukner:2002kx,brukner2009information}. However, in order to complete these ideas to a full reconstruction of qubit quantum theory, we have to impose further rules. 

Specifically, rules \ref{lim} and \ref{unlim} say nothing about what happens in-between interrogations.
We require that $O$ shall not gain or lose information about an otherwise non-interacting $S$ {\it without} asking questions.
\begin{pr}\label{pres}{\bf(Information Preservation)}
\emph{``The total amount of information $O$ has about (an otherwise non-interacting) $S$ is preserved in-between interrogations.''}\\
$I_{O\rightarrow S}$ is \emph{constant} in time in-between interrogations for (an otherwise non-interacting) $S$.
\end{pr}
Correspondingly, $I_{O\rightarrow S}$ is a `{\it conserved charge}' of time evolution; this is a simple observation which will become extremely useful later. In fact, notice that rule \ref{pres} could also be viewed as defining the notion of non-interacting systems (as perceived by $O$).

Finally, we come to the last rule of this manuscript which is about time evolution. Empirical evidence suggests, at least to a good approximation, that the time evolution of an experimenter's `catalogue of knowledge' about an observed system is {\it continuous} -- in-between measurements. More precisely, the specific probabilistic statements an experimenter can make about the outcomes of measurements, i.e.\ his {\it actual} information about the systems, change continuously in time. We promote this to a further postulate for $O$'s world, however, with a further requirement. 

We note that the state space and the time evolutions are interdependent as every legal time evolution must map a legal state to a legal state, i.e.\ $T_{\Delta t}(\vec{y}_{O\rightarrow S})\in\Sigma_N$, $\forall\,T_{\Delta t}\in\ct_N$ and $\forall\,\vec{y}_{O\rightarrow S}\in\Sigma_N$. Accordingly, what is the set of legal states depends on what is the set of legal time evolutions -- and vice versa. We would like the pair $(\Sigma_N,\ct_N)$ to be as `big' as compatibility with the other rules allows in order to equip the other rules with as general a validity as possible. But there are multiple ways of `maximizing' the pair. Namely, the interdependence implies that the larger the number of states, the tighter the constraints on the set of time evolutions -- and vice versa. We required before that $O$'s world should not be `boring' and therefore feature a non-trivial time evolution of states. We shall now sharpen this requirement: it is more interesting for $O$ to live in a world which `maximizes' the number of possible ways in which any given state of $S$ can change in time, rather than the number of states which it can be in relative to $O$. We shall thus require that {\it any} consistent, continuous  time evolution of $O$'s `catalogue of knowledge' about $S$ in-between interrogations is physically realizable. Since different ways of evolving the state of a system correspond to different interactions (e.g., among the members of a composite system) we thereby maximize the set of possible interactions among systems. $O$'s world is a maximally interactive place!

\begin{pr}\label{time}{\bf(Time Evolution)}
\emph{``$O$'s `catalogue of knowledge' about $S$ evolves \emph{continuously} in time in-between interrogations and every consistent such evolution is physically realizable.''}\\
$\ct_N$ is the maximal set of transformations $T_{\Delta t}$ on states 
which is \emph{continuous} in $\Delta t$ and compatible with rules \ref{lim}-\ref{pres} and the structure of $\cl_{\rm gbit}$.
\end{pr}

Of course, during an interrogation, $\vec{y}_{O\rightarrow S}$ may change discontinuously, i.e.\ `collapse', on account of complementarity. As innocent as the requirement of continuity of time evolution in rule \ref{time} appears, it turns out to be absolutely crucial in order to single out a unique information measure $\alpha_i(\vec{y}_{O\rightarrow S})$. Notice also that classical bit theory violates this postulate due to its discrete time evolution group (see section \ref{sec_cl}).

It is now our task to verify what the triples of $(\cq_N,\Sigma_N,\ct_N)$ for each $N\in\mathbb{N}$ are which obey rules \ref{lim}--\ref{time}. Remarkably, it turns out that these four rules cannot distinguish between real and complex quantum theory, i.e.\ between real and complex Hilbert spaces. But the state spaces and the orthogonal and unitary groups of time evolutions of rebit and qubit quantum theory are, in fact, the only `solutions' within $\cl_{\rm gbit}$ to these rules. In the sequel of this manuscript we shall employ rules \ref{lim}--\ref{time} in order to develop the necessary tools, within $\cl_{\rm gbit}$, for eventually proving the following more precise claim in \cite{hw,hw2}. In fact, as a simple example of the newly developed tools, we shall already prove the claim for $N=1$ at the end of this article.
\begin{claim}
$\cl_{\rm gbit}$ contains only two solutions for the pair $(\Sigma_N,\ct_N)$ which are compatible with rules \ref{lim}--\ref{time}:
\begin{description}
\item[1.\ rebit quantum theory \cite{hw2},] where $\Sigma_N$ coincides with the space of $2^N\times2^N$ density matrices\footnote{A real density matrix is a symmetric, positive semidefinite matrix on $(\mathbb{R}^{2})^{\otimes N}$.} over $(\mathbb{R}^2)^{\otimes N}$ and $\ct_N$ is $\rm{PSO}(2^N)$.

\item[2.\ qubit quantum theory \cite{hw},] where $\Sigma_N$ coincides with the space of $2^N\times2^N$ density matrices over $(\mathbb{C}^2)^{\otimes N}$ and $\ct_N$ is $\rm{PSU}(2^N)$. Furthermore, states evolve unitarily according to the von Neumann evolution equation.
\end{description}
\end{claim}
We remind the reader that the time evolution groups $\ct_N$ for density matrices in rebit and qubit quantum theory are {\it projective} because they correspond to $\rho\mapsto U\,\rho\,U^\dag$, where for (1) rebits $\rho$ is a $2^N\times2^N$ real symmetric matrix, $U\in\rm{SO}(2^N)$ and $\dag$ denotes matrix transpose; and (2) for qubits $\rho$ is a $2^N\times2^N$ hermitian matrix, $U\in\rm{SU}(2^N)$ and $\dag$ denotes hermitian conjugation.

Furthermore, using one additional rule on $\cq_N$ which allows $O$ to ask $S$ any question that ``makes (probabilistic) sense'', we show in \cite{hw,hw2} that
\begin{description}
\item[1. in the rebit case \cite{hw2},] \emph{$\cq_N$ is (isomorphic to) the set of projective measurements onto the $+1$ eigenspaces of Pauli operators\footnote{The set of Pauli operators over $\mathbb{R}^{2^N}$ is the set of traceless real symmetric $2^N\times2^N$ matrices with eigenvalues $\pm1$.}} over $\mathbb{R}^{2^N}$. This is the set of rank-$(2^{N-1})$-projectors.
\item[2. in the qubit case \cite{hw},]\emph{ $\cq_N$ is (isomorphic to) the set of projective measurements onto the $+1$ eigenspaces of Pauli operators\footnote{The set of Pauli operators over $\mathbb{C}^{2^N}$ is the set of traceless hermitian $2^N\times2^N$ matrices with eigenvalues $\pm1$.} over $\mathbb{C}^{2^N}$ (i.e., the set of rank-$(2^{N-1})$-projectors) and the outcome probability for any $Q\in\cq_N$ to be answered with `yes' by $S$ in some state is given by the Born rule for projective measurements.}
\end{description}

Since both real and complex quantum theory come out of rules \ref{lim}--\ref{time}, we shall impose a further rule on $O$'s information acquisition in \cite{hw} in order to also eliminate rebit theory. However, we note that rebit theory is mathematically contained in qubit quantum theory and that rebits can actually be produced in laboratories. Hence, one might also hold rules \ref{lim}--\ref{time} and the rule that $O$ may ask $S$ any question that ``makes sense'' as physically sufficient.

Notice that it is necessary to directly reconstruct the space of density matrices over Hilbert spaces from the rules rather than the underlying Hilbert spaces themselves. The reason is that the latter contain physically redundant information (norm and global phase), while the rules refer only to information which is directly accessible to $O$.

\subsection{Strategy for building the necessary tools and proving the claim}

Our strategy and procedure for developing tools and intermediate results before proving the main claim is best summarized as a diagram (see figure \ref{fig_strategy}).

\begin{widetext}
\begin{center}
\begin{figure}[hbt!]
\begin{center}
\psfrag{L}{{limited information}}
\psfrag{C}{{complementarity}}
\psfrag{p}{{information preservation}}
\psfrag{t}{{time evolution}}
\psfrag{o}{time evolution \hspace*{.5cm}}
\psfrag{q}{independence, compatibility and cor-}
\psfrag{q1}{relation structure on $\cq_N$ (in sec.\ \ref{sec_qstructure})}
\psfrag{i}{reversible time evolution (in sec.\ \ref{sec_time})}
\psfrag{i1}{and information measure (in sec.\ \ref{sec_infomeasure})}
\psfrag{1}{$\Sigma_1$ is a ball with $D=2,3$, and}
\psfrag{4}{$\ct_1$ is either $\rm{SO}(2)$ or $\rm{SO}(3)$ (in sec.\ \ref{sec_n1})}
\psfrag{2}{$\ct_2$ is either $\rm{PSO}(4)$ or $\rm{PSU}(4)$ $\Rightarrow$}
\psfrag{3}{ $\Sigma_2$ convex hull of $\mathbb{R}P^3$ or $\mathbb{C}P^3$ (in \cite{hw,hw2})}
\psfrag{n}{$(\cq_N,\Sigma_N,\ct_N)$ for $\!N>2$ {\footnotesize (in \cite{hw,hw2})}}
{\includegraphics[scale=.55]{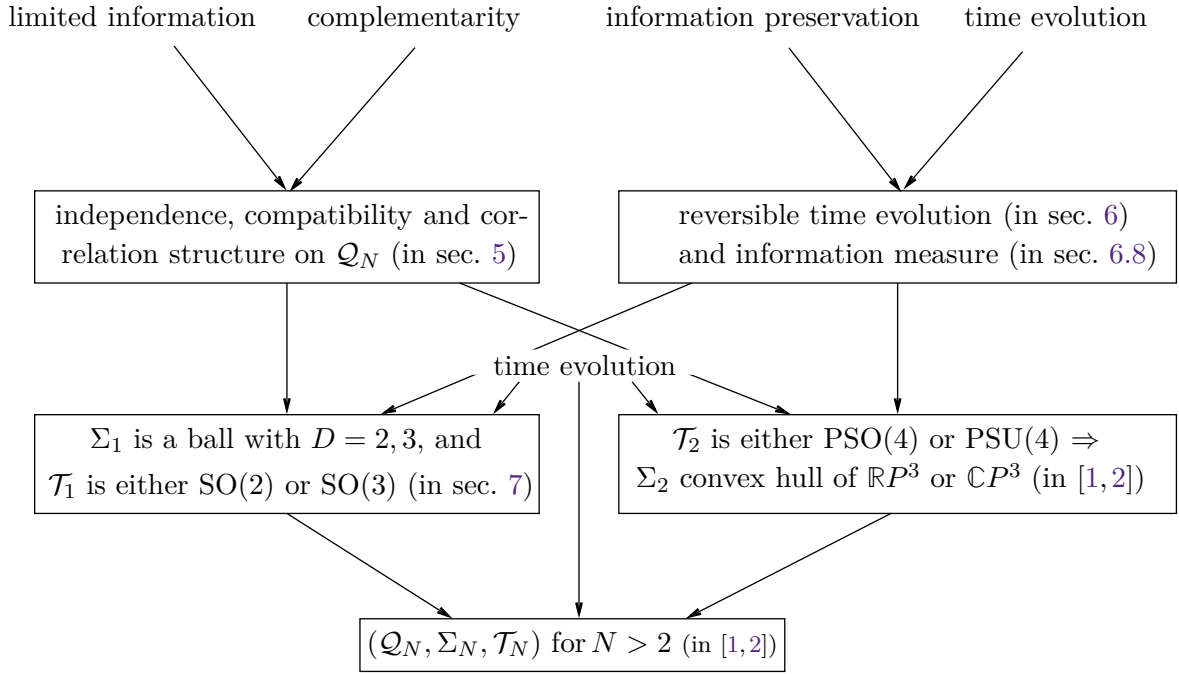}}
\caption{\small Strategy and steps for the reconstruction.}\label{fig_strategy}
\end{center}
\end{figure}
\end{center}
\end{widetext}

That is, we shall firstly ascertain, in section \ref{sec_qstructure}, the independence, compatibility and correlation structure on $\cq_N$ which is induced by rules \ref{lim} and \ref{unlim}. This section will be particularly instructive and deliver many elementary technical results which, besides becoming important later, yield simple explanations for typical quantum phenomena such as, e.g., entanglement, non-locality and monogamy relations. This is also where we shall determine the dimensionality of $\Sigma_1$ by a simple argument and, using this result, derive the dimensionalities of all other state spaces too. Next, in section \ref{sec_time}, we shall use rules \ref{pres} and \ref{time} in order to conclude that a specific time evolution, is reversible and, in fact, described by a continuous one-parameter matrix group. These results will then be used in section \ref{sec_infomeasure}, together with rules \ref{pres} and \ref{time}, to derive an explicit information measure $\alpha_i(\vec{y}_{O\rightarrow S})$ which is unique under elementary consistency conditions. In section \ref{sec_n1}, the conjunction of these outcomes will then be employed 
to show that for $N=1$ one either obtains a two- or three-dimensional Bloch ball as a state space $\Sigma_1$ and that the group of all possible time evolutions $\ct_1$ is accordingly either $\rm{SO}(2)$ or $\rm{SO}(3)$. The claim for $N>1$ will require substantial additional work and will be proven in the companion paper \cite{hw} for qubits and in \cite{hw2} for rebits.


\section{Question structure and correlations}\label{sec_qstructure}

In this section, we shall solely employ rules \ref{lim} and \ref{unlim} in order to deduce the independence, compatibility and correlation structure for the questions in an informationally complete set $\cq_{M_N}$ (c.f.\ assumption \ref{assump6}). To this end it is instructive to look at the individual cases $N=1,2,3$ in some detail before considering general $N\in\mathbb{N}$. We shall denote by $D_N$ the dimension of the state space $\Sigma_N$.

\subsection{A single gbit}

Rules \ref{lim} and \ref{unlim} taken together imply that for $N=1$ there exists at least one maximally complementary pair $Q_1,Q_1'$. Since $N=1$ is fixed, it is now more convenient to count the independent questions by an index, rather than a prime.\footnote{Rules \ref{lim} and \ref{unlim} count compatible questions, here we need to count maximally complementary questions.} Let us therefore slightly change the notation and write $Q_1'$ henceforth as $Q_2$, such that the rules imply the existence of a maximally complementary pair $Q_1,Q_2$. But, applied to a single gbit only, rules \ref{lim} and \ref{unlim} tell us nothing more about how many questions maximally complementary to $Q_1$ exist. There may arise another $Q_3$ maximally complementary to both $Q_1,Q_2$ etc. Notice that, for a single gbit, an informationally complete set of pairwise independent questions must be given by a maximal set of mutually maximally complementary questions (i.e., no further $Q\in\cq_1$ can be added to this set without destroying mutual maximal complementarity) 
\ba
\cq_{M_1}=\{Q_1,Q_2,\ldots,Q_{D_1}\}\label{qm1}
\ea
because rule \ref{lim} prohibits further pairwise independent (partially or maximally) compatible questions. (Recall that maximally complementary questions are automatically independent.) Rules \ref{lim} and \ref{unlim} imply $D_1\geq2$. It turns out that we need to consider two gbits in order to upper bound $D_1$.

We shall call such questions for a single gbit {\it individual questions}.


\subsection{Two gbits}\label{sec_n2}

The $N=2$ case requires substantially more work. First of all, since this is a composite system (c.f.\ definition \ref{def_comp}), the individual questions about the two single gbits must be contained in an informationally complete set $\cq_{M_2}$. We shall again slightly change the notation, as compared to section \ref{sec_principles}: the maximal complementary set $\cq_{M_1}$ (\ref{qm1}) for gbit 1 will be denoted by $Q_1,\ldots,Q_{D_1}$, while the elements of $\cq_{M_1}$ for gbit 2 will be denoted with a prime $Q'_1,\ldots,Q'_{D_1}$. It will be convenient to depict the question structure graphically. Representing individual questions henceforth as {\it vertices}, $\cq_{M_2}$ certainly contains the following
\vspace*{.2cm}
\begin{eqnarray}
\psfrag{q1}{\hspace*{-.3cm}gbit 1}
\psfrag{q2}{gbit 2}
\psfrag{Q1}{$Q_1$}
\psfrag{Q2}{$Q_2$}
\psfrag{Q3}{$Q_3$}
\psfrag{QD}{$Q_{D_1}$}
\psfrag{P1}{$Q'_1$}
\psfrag{P2}{$Q'_2$}
\psfrag{P3}{$Q'_3$}
\psfrag{PD}{$Q'_{D_1}$}
\psfrag{d}{$\vdots$}
{\includegraphics[scale=.2]{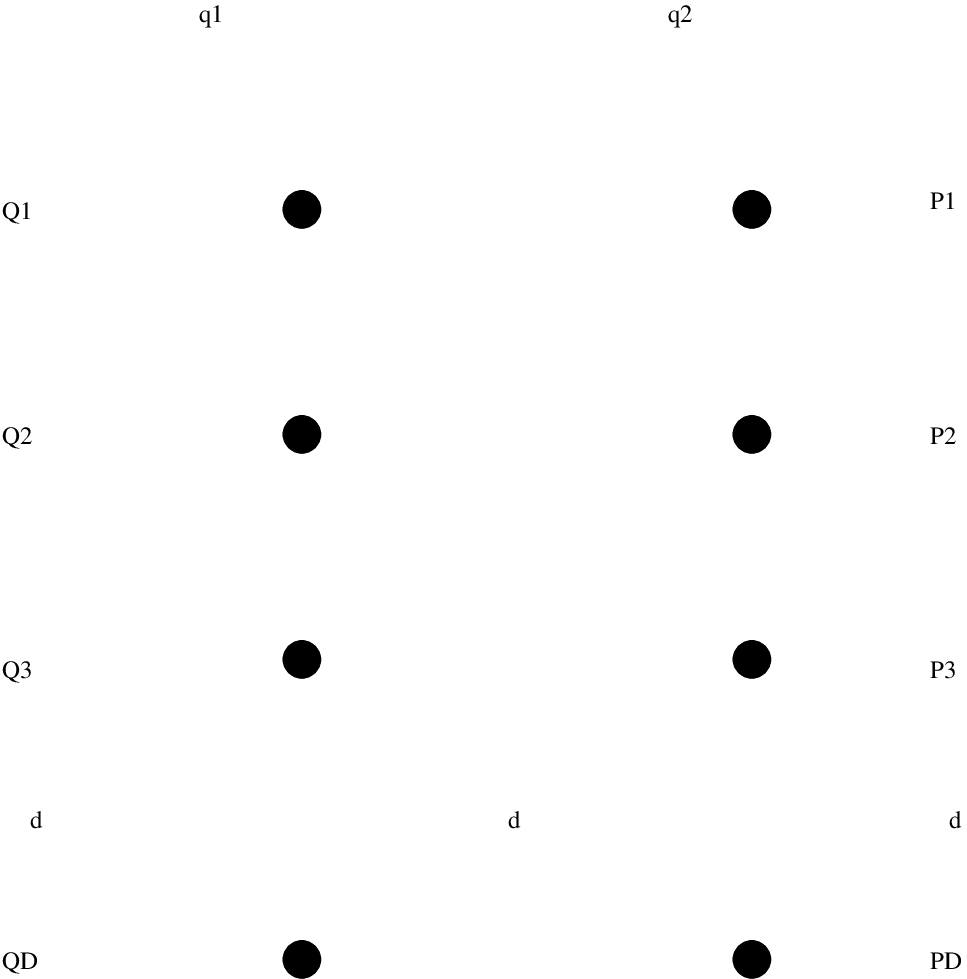}}.\label{indiv}
\end{eqnarray}
This structure abides by rules \ref{lim} and \ref{unlim} and any $Q_i$ will be independent of and maximally compatible with any $Q'_j$ (cf.\ definition \ref{def_comp}).

\subsubsection{Logical connectives of single gbit questions}\label{sec_connective}

But a composite system may also admit composite questions, built with logical connectives of the subsystem questions. We must now determine which composite questions are allowed and which must be added to the individual questions in order to form an informationally complete $\cq_{M_2}$ of {\it pairwise independent} questions. To this end, we must firstly unravel which logical connective $*$ $O$ can employ at all in order to construct composite questions which are pairwise independent of the individual ones. Recall assumption \ref{assump4} that $O$, according to his theoretical model, may only compose compatible questions with logical connectives $*$ and consider $Q_i,Q'_j$ as an exemplary compatible pair. In truth tables, we shall henceforth symbolize `yes' by `$1$' and `no' by `$0$'. 

\begin{R} {\bf(`$=$' in logical expressions.)}
To avoid confusion, we emphasize that whenever we write an equality sign `$=$' in a logical expression in the sequel, we do \emph{not} interpret it as another logical connective $*$, but as an actual equality of the values of the propositions on the left and right hand side. For example, the expression $Q_i=Q_j'$ is \emph{not} itself a proposition which takes truth values $0$ or $1$, but is intended to mean that the truth value which $Q_i$ takes is identical to that which $Q_j'$ takes.
\end{R}

While permitted binary connectives, it is clear that, e.g., the AND and OR operations $\wedge$ and $\vee$, respectively, cannot be employed alone to build a new independent question from $Q_i,Q'_j$ because $Q_i\wedge Q'_j=1$ implies $Q_i=Q'_j=1$ and $Q_i\vee Q'_j=0$ implies $Q_i=Q'_j=0$ such that certain answers to these two connectives imply answers to the individuals. But if $Q_i*Q'_j$ is to be pairwise independent from $Q_i,Q'_j$, an answer to it must not imply the answer to either of $Q_i,Q'_j$ and, conversely, the answer to only $Q_i$ or only $Q'_j$ cannot imply $Q_i*Q'_j$. As can be easily checked, it must satisfy the following truth table:
\vspace*{.2cm}
\begin{eqnarray}
\!\!\!\!\!\!\!\! \!\!\! \begin{tabular}{ ||c | c || c|| }
    \hline
    $Q_{i}$ & $Q'_{j}$ & $Q_i*Q'_j$ \\ \hline\hline
    0 & 1 & a \\ \hline
    1 & 0 & a \\ \hline
    1&1&b\\\hline
    0&0&b\\\hline
  \end{tabular} \q a\neq b\q a,b\in\{0,1\}.\label{truth}
\end{eqnarray}
\vspace*{.2cm}

The only logical connectives $*$ satisfying this truth table are the XNOR $\leftrightarrow$ for $a=0$ and $b=1$ and its negation, the XOR $\oplus$ for $a=1$ and $b=0$. As $Q_i\oplus Q'_j = \neg(Q_i\leftrightarrow Q'_j)$, we may only choose one of the two binary connectives to define new pairwise independent questions. We henceforth choose to use the XNOR, already introduced in (\ref{correlation}), and operationally interpret it as a `correlation question'; $Q_{ij}:=Q_i\leftrightarrow Q'_j$ is to be read as `are the answers to $Q_i$ and $Q'_j$ the same?'. But the XOR could equivalently be employed.\footnote{As an aside, the XNOR can be expressed in terms of the basic Boolean operations as $Q_i\leftrightarrow Q'_j = (\neg Q_i\vee Q'_j)\wedge(Q_i\vee \neg Q'_j)$.}

Since, according to assumption \ref{assump6}, $O$ must be able to build an informationally complete set for a composite system which, in particular, involves pairwise independent composite questions, we would like to conclude that $Q_i\leftrightarrow Q'_j$ are not only implementable, but also contained in a $\cq_{M_2}$ together with the individuals. However, this is subject to a few consistency checks on pairwise independence.

The XNOR $\leftrightarrow$ is a symmetric logical connective, $Q_{ij}=Q_i\leftrightarrow Q'_j=Q'_j\leftrightarrow Q_i$ (but note that $Q_{ij}\neq Q_{ji}$), and, thanks to its associativity,
 \ba\label{correlation2}
  Q_i\leftrightarrow Q_{ij} &=& Q_i \leftrightarrow (Q_i\leftrightarrow Q'_j)\nn\\ &=& \underset{1}{\underbrace{(Q_i \leftrightarrow Q_i)}}\leftrightarrow Q'_j = Q'_j\,,\nn\\
  Q_{ij}\leftrightarrow Q'_j&=&Q_i\, \ea
such that $\leftrightarrow$ defines a closed relation on $Q_i,Q'_j,Q_{ij}$. This has the following ramifications: (1) For any compatible pair $Q_i,Q'_j$, the `correlation' operation $\leftrightarrow$ gives rise to precisely one additional question $Q_{ij}$, and (2) the set $Q_i,Q'_j,Q_{ij}$ will indeed be {\it pairwise} independent, according to the definition in section \ref{sec_elstruc}, because neither $Q_i$ nor $Q'_j$ alone can determine $Q_{ij}$ relative to the state of no information (not even partially), otherwise, together with the determined $Q_{ij}$, they would (at least partially) determine each other via (\ref{correlation2}) -- in contradiction with the independence of $Q_i,Q'_j$. That is, $Q_{ij}$ is independent of both $Q_i$ and $Q'_j$ and since we assume independence to be a symmetric relation, it must also be true the other way around.

Consequently, the $D_1^2$ correlation questions $Q_{ij}$, $i,j=1,\ldots,D_1$, are candidates for additional questions in $\cq_{M_2}$ besides the $2D_1$ individual ones of (\ref{indiv}). Graphically, we shall represent $Q_{ij}$ as the {\it edge} connecting the vertices corresponding to $Q_i$ and $Q'_j$. For example, the following question graphs
\begin{widetext}
\begin{eqnarray}
\psfrag{q1}{\hspace*{-.3cm}gbit 1}
\psfrag{q2}{gbit 2}
\psfrag{Q1}{$Q_1$}
\psfrag{Q2}{$Q_2$}
\psfrag{Q3}{$Q_3$}
\psfrag{QD}{$Q_{D_1}$}
\psfrag{P1}{$Q'_1$}
\psfrag{P2}{$Q'_2$}
\psfrag{P3}{$Q'_3$}
\psfrag{PD}{$Q'_{D_1}$}
\psfrag{q11}{$Q_{11}$}
\psfrag{q31}{$Q_{31}$}
\psfrag{q22}{$Q_{22}$}
\psfrag{q23}{$Q_{23}$}
\psfrag{qdd}{$Q_{D_1D_1}$}
\psfrag{d}{$\vdots$}
{\includegraphics[scale=.2]{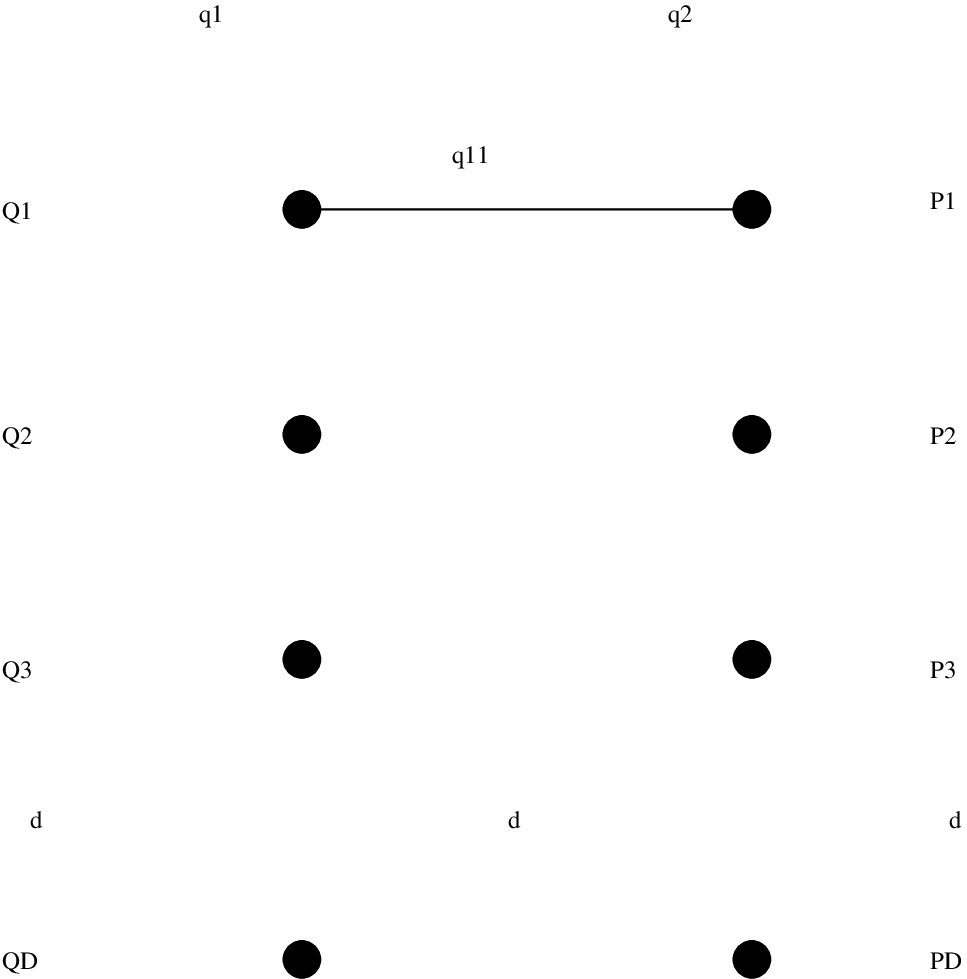}}\q\q\q\q\q\q\q{\includegraphics[scale=.2]{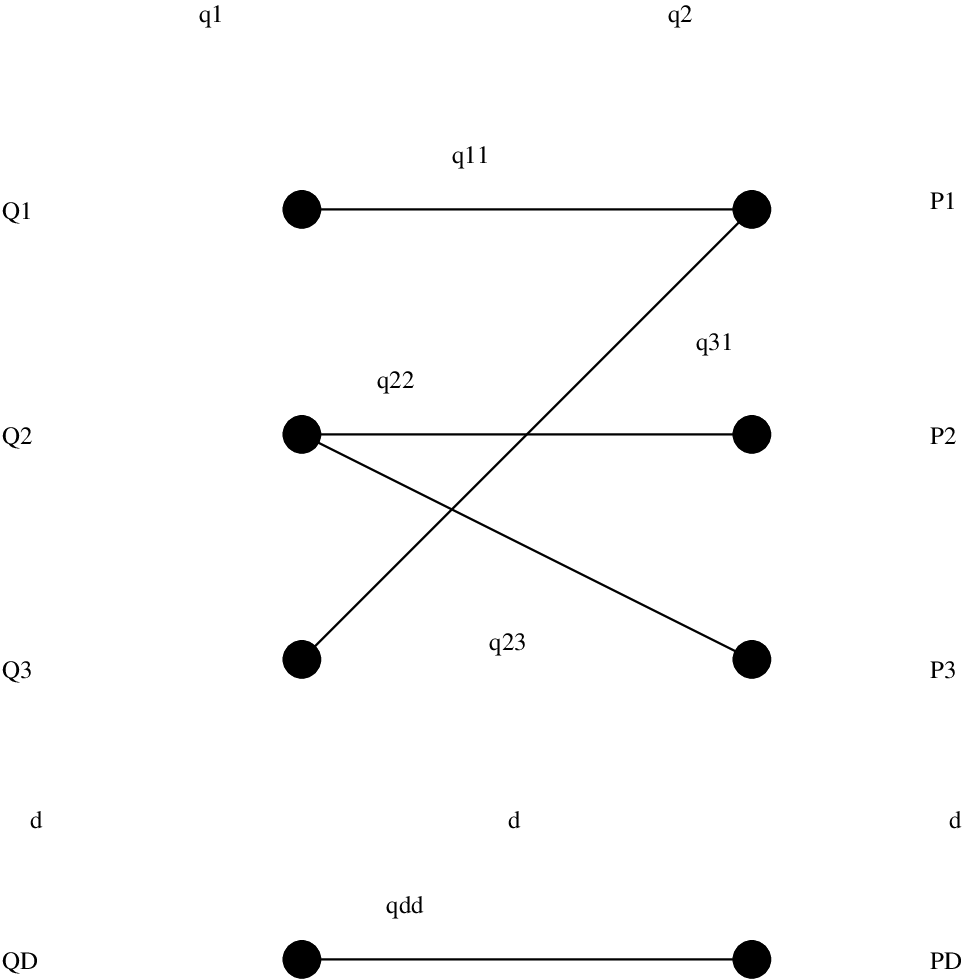}}\q\q\q\q\q\q\q{\includegraphics[scale=.2]{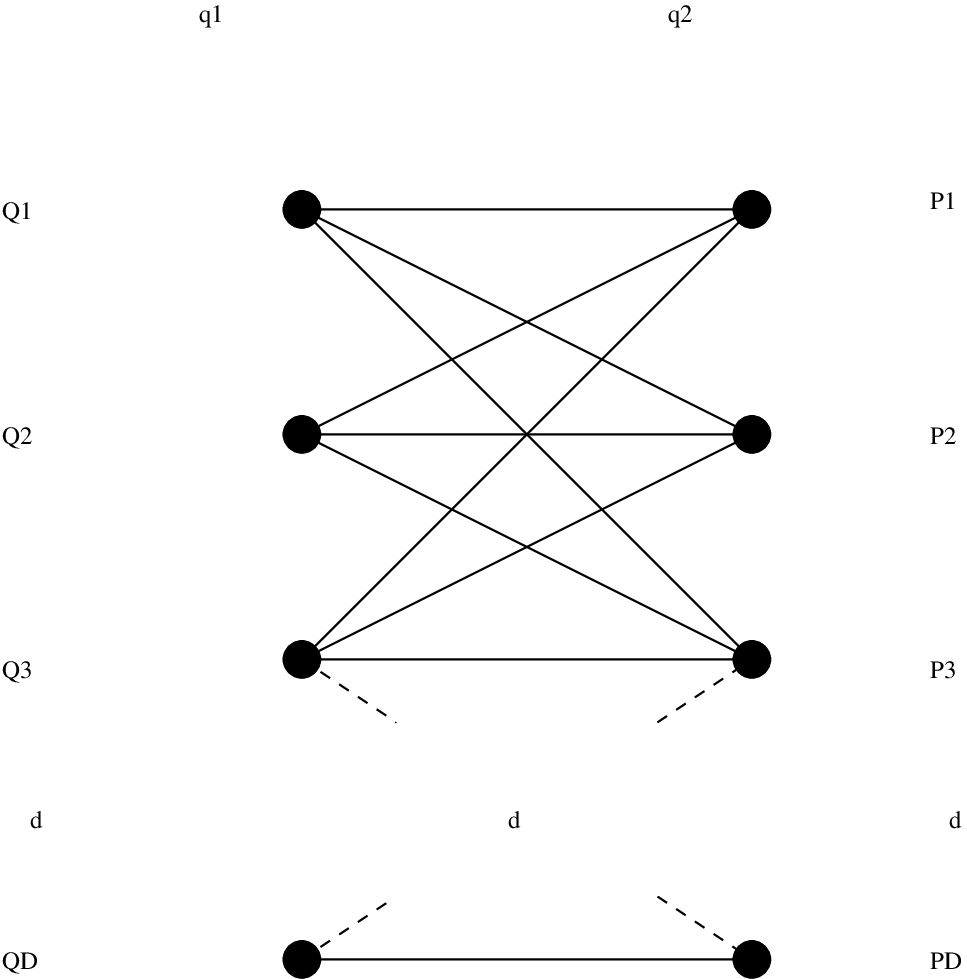}}\notag
\end{eqnarray} 
\end{widetext}
represent legal sets of questions. Note that due to assumption \ref{assump4}, there can only be edges between gbit 1 and gbit 2 but not between the vertices of a single gbit, e.g., $Q_1$ and $Q_2$, because they are complementary.
 
 Our next task is to analyse the independence and compatibility structure of the $Q_{ij}$.

\subsubsection{Independence, complementarity and entanglement}

We begin with a simple, but important observation which we will frequently make use of.

\begin{lem}\label{lem0}
Let $Q_1,Q_2,Q_3\in\cq_N$ be such that $Q_1$ is maximally compatible with and pairwise independent of both $Q_2,Q_3$ and that $Q_3$ is maximally complementary to $Q_1\leftrightarrow Q_2$. Then $Q_2,Q_3$ are maximally complementary.
\end{lem}

\begin{proof}
Suppose $Q_2,Q_3$ were at least partially compatible. In that case there must exist a state of $S$ in which $O$ has maximal information $\alpha_{2}=1$ \texttt{bit} about $Q_{2}$ and at least partial information $\alpha_3>0$ \texttt{bit} about $Q_3$ (or vice versa). Let $O$ now ask $Q_1$ which is maximally compatible with and pairwise independent of both $Q_2,Q_{3}$ to $S$ in this state. Hence, according to assumption \ref{assump5b}, by asking $Q_1$, $O$ cannot change his information about either $Q_2,Q_{3}$. Thus, after asking $Q_1$, $O$ will have maximal information about $Q_1,Q_{2}$ and at least partial information about $Q_3$. But maximal information about $Q_1,Q_{2}$ implies maximal information about $Q_1\leftrightarrow Q_2$ which is maximally complementary to $Q_3$. Consequently, $Q_2$ and $Q_{3}$ must be maximally complementary too.
\end{proof}

This has immediate useful implications.
 
\begin{lem}\label{lem1}
$Q_i$ is maximally compatible with $Q_{ij}$, $\forall\,j=1,\ldots,D_1$ and maximally complementary to $Q_{kj}$, $\forall\,k\neq i$ and $\forall\,j=1,\ldots,D_1$. That is, graphically, an individual question $Q_i$ is maximally compatible with a correlation question $Q_{ij}$ if and only if its corresponding vertex is a vertex of the edge corresponding to $Q_{ij}$. By symmetry, the analogous result holds for $Q'_j$.
\end{lem}
 
 \begin{proof}
$Q_i$ and $Q_{ij}$ are maximally compatible by construction. Consider therefore $Q_i$ and $Q_{kj}$ for $k\neq i$ and $j=1,\ldots,D_1$. Clearly, $Q_i$ and $Q'_j$ are maximally compatible and pairwise independent and so are $Q'_j$ and $Q_{kj}$. Maximal complementarity of $Q_i$ and $Q_{kj}$ for $k\neq i$ now follows from $Q_k=Q_{kj}\leftrightarrow Q'_j$ (see (\ref{correlation2})), which is maximally complementary to $Q_i$, and lemma \ref{lem0}.
 \end{proof}
 
For example, $Q_1$ and $Q_{12}$ are maximally compatible, while $Q_1$ and $Q_{22}$ are maximally complementary. The intuitive explanation for the incompatibility of $Q_1$ and $Q_{22}$ is as follows: if $O$ knew the answers to both simultaneously, he would know more than one \texttt{bit} of information about gbit 1 because $Q_1$ defines a full \texttt{bit} of information about it while $Q_{22}$ could be regarded as defining half a \texttt{bit} about each of gbit 1 and 2. But in view of rule \ref{lim}, $O$ should never know more than one \texttt{bit} about a single gbit, even in a composite system. Considering the information contained in correlation questions to equally correspond to gbit 1 and 2, lemma \ref{lem1} suggests that $O$'s information will always be equally distributed over the two gbits for states of maximal knowledge, i.e.\ $O$ will know equally much about gbit 1 and 2.\footnote{Clearly, this is not true for non-maximal knowledge. E.g., let $O$ always ask only $Q_1$.} We shall make this more precise in \cite{hw}.

We saw before that $Q_i,Q'_j,Q_{ij}$ are pairwise independent. Lemma \ref{lem1} implies that also $Q_i$ and $Q_{kj}$ for $i\neq k$ are independent. We can make use of this result to show the following.

\begin{lem}\label{lem2}
The $Q_{ij}$, $i,j=1,\ldots,D_1$ are pairwise independent.
\end{lem}

\begin{proof}
Consider $Q_{ij}$ and $Q_{kl}$ 
and suppose $i\neq k$. Then $Q_i$ and $Q_{ij}$ are maximally compatible, while $Q_i$ and $Q_{kl}$ are maximally complementary by lemma \ref{lem1}. Hence, when knowing $Q_i$ and $Q_{ij}$ $O$ cannot know $Q_{kl}$ such that $Q_{ij}$ and $Q_{kl}$ must be independent. (Recall from section \ref{sec_elstruc} that dependence requires the answer to $Q_{ij}$ to always imply at least partial knowledge about $Q_{kl}$, i.e., $y_{ij}=0,1$ implies $y_{kl}\neq\f{1}{2}$). The analogous argument holds for when $j\neq l$. 
\end{proof}
 
 Recalling from subsection \ref{sec_connective} that the set of composite questions in $\cq_{M_2}$ must be non-empty, lemmas \ref{lem1} and \ref{lem2} have an important corollary.
 
 \begin{Corollary}
 $Q_i,Q'_j,Q_{ij}$ are pairwise independent for all $i,j=1,\dots,D_1$ and, thanks to assumption \ref{assump6}, will thus be part of an informationally complete set $\cq_{M_2}$.
 \end{Corollary}

Next, we consider the compatibility and complementarity structure of the correlation questions $Q_{ij}$.

\begin{lem}\label{lem3}
$Q_{ij}$ and $Q_{kl}$ are maximally compatible if and only if $i\neq k$ and $j\neq l$. That is, graphically, $Q_{ij}$ and $Q_{kl}$ are maximally compatible if their corresponding edges do \emph{not} intersect in a vertex and maximally complementary if they intersect in one vertex.
\end{lem}

\begin{proof}
Suppose $Q_{ij}$ and $Q_{kj}$ with $i\neq k$ were at least partially compatible. Then there must exist a state of $S$ in which the answer to $Q_{ij}$ is fully known to $O$ and in which also $\alpha_{kj}>0$ \texttt{bit} (or vice versa). Let $O$ ask $Q_j'$ to $S$ in this state. Since, by lemma \ref{lem1}, this question is maximally compatible with both $Q_{ij},Q_{kj}$ (and pairwise independent of both), according to assumption \ref{assump5b}, by asking $Q_j'$, $O$ cannot change his information about both $Q_{ij},Q_{kj}$. That is, after having also asked $Q_j'$, $O$ must have maximal information about $Q_{ij},Q_j'$ and partial information about $Q_{kj}$. But maximal information about $Q_{ij},Q_j'$ implies also maximal information about $Q_i$ which, however, is maximally complementary to $Q_{kj}$ by lemma \ref{lem1}. Hence, $Q_{ij},Q_{kj}$ with $i\neq k$ must be maximally complementary.

%

Consider now $Q_{ij}$ and $Q_{kl}$ with $i\neq k$ and $j\neq l$. Let $O$ ask $S$ both $Q_i$ and $Q'_j$ the answers to which imply the answer to $Q_{ij}$ according to (\ref{correlation}). Thanks to rule \ref{lim}, this defines a state of maximal knowledge of $2$ independent \texttt{bits} and $O$ may not know any further independent information. Next, let $O$ ask the same $S$ the question $Q_{kl}$. From lemma \ref{lem1} it follows that $Q_{kl}$ is maximally complementary to both $Q_i$ and $Q'_j$ such that the answer to $Q_{kl}$ will give one new independent \texttt{bit} of information but renders $O$'s information about $Q_i,Q'_j$ obsolete. But by rule \ref{unlim} $O$ cannot experience a net loss of information by asking a new question and after asking $Q_{kl}$ he must still know $2$ \texttt{bits} about $S$. Hence, after acquiring the answer to $Q_{kl}$ he must still know the answer to $Q_{ij}$ such that both are maximally compatible.
\end{proof} 
 
 To give graphical examples, $Q_{11}$ and $Q_{21}$ are maximally complementary due to the intersection in $Q_1'$, while $Q_{13}$ and $Q_{21}$ are maximally compatible because their edges do not intersect in vertices
 \vspace*{.1cm}
 \begin{eqnarray}
\psfrag{q1}{\hspace*{-.3cm}gbit 1}
\psfrag{q2}{gbit 2}
\psfrag{Q1}{$Q_1$}
\psfrag{Q2}{$Q_2$}
\psfrag{Q3}{$Q_3$}
\psfrag{QD}{$Q_{D_1}$}
\psfrag{P1}{$Q'_1$}
\psfrag{P2}{$Q'_2$}
\psfrag{P3}{$Q'_3$}
\psfrag{PD}{$Q'_{D_1}$}
\psfrag{q11}{$Q_{11}$}
\psfrag{q21}{$Q_{21}$}
\psfrag{q22}{$Q_{22}$}
\psfrag{q13}{$Q_{13}$}
\psfrag{qdd}{$Q_{D_1D_1}$}
\psfrag{d}{$\vdots$}
{\includegraphics[scale=.2]{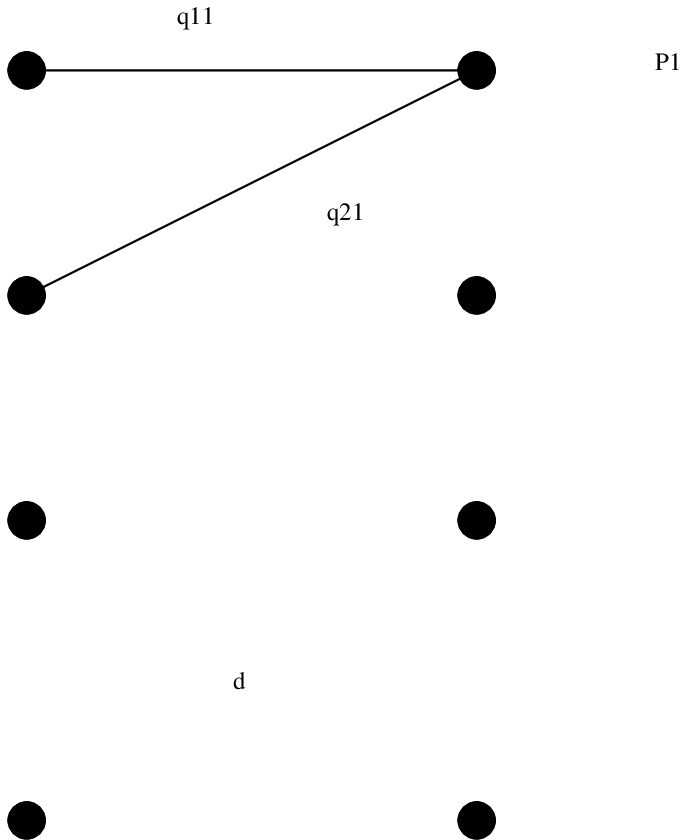}}\q\q\q\q\q\q\q{\includegraphics[scale=.2]{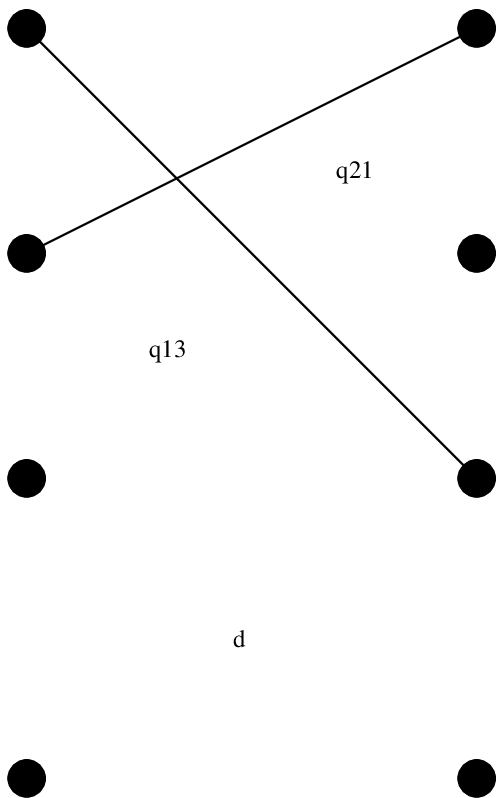}}\notag
\end{eqnarray}

This question structure has significant consequences: since the correlation questions $Q_{ij}$ and $Q_{kl}$ are both independent and maximally compatible for $i\neq k$ and $j\neq l$, $O$ can ask both of them simultaneously, thereby spending the maximal amount of $N=2$ independent \texttt{bits} of information he may acquire according to rule \ref{lim} over composite questions. For example, he may ask $S$ $Q_{11}$ and $Q_{22}$ simultaneously
\vspace*{.2cm}
 \begin{eqnarray}
\psfrag{q1}{\hspace*{-.3cm}gbit 1}
\psfrag{q2}{gbit 2}
\psfrag{Q1}{$Q_1$}
\psfrag{Q2}{$Q_2$}
\psfrag{Q3}{$Q_3$}
\psfrag{QD}{$Q_{D_1}$}
\psfrag{P1}{$Q'_1$}
\psfrag{P2}{$Q'_2$}
\psfrag{P3}{$Q'_3$}
\psfrag{PD}{$Q'_{D_1}$}
\psfrag{q11}{$Q_{11}$}
\psfrag{q21}{$Q_{21}$}
\psfrag{q22}{$Q_{22}$}
\psfrag{q13}{$Q_{13}$}
\psfrag{qdd}{$Q_{D_1D_1}$}
\psfrag{d}{$\vdots$}
{\includegraphics[scale=.2]{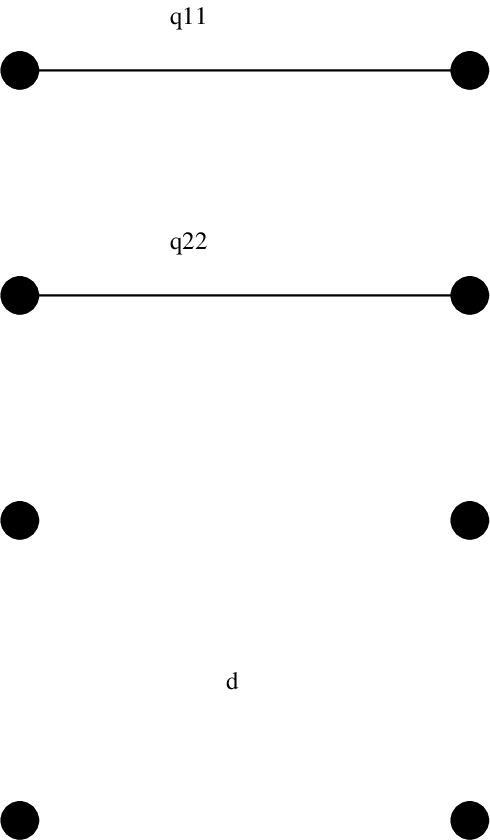}}\label{bell}
\end{eqnarray}   
upon which he must be entirely oblivious about the individual gbit properties represented by $Q_i,Q'_j$. (Lemma \ref{lem1} implies that no individual question is maximally compatible with two non-intersecting bipartite questions at once.) That is, $O$ has only composite, but no individual information about the two gbits; they are maximally correlated. But this is precisely {\it entanglement}. Indeed, the question graph (\ref{bell}) will ultimately correspond to four Bell states in either of the two question bases $\{Q_1,Q'_1\}$ and $\{Q_2,Q'_2\}$, representing the four possible answer configurations `yes-yes', `yes-no', `no-yes' and `no-no' to the correlation questions $Q_{11},Q_{22}$ (see \cite{zeilinger1999foundational} for a related perspective within quantum theory). In fact, if $O$ now `marginalized' over gbit 2, corresponding to discarding any information about questions involving gbit 2, he would find gbit 1 in the state of no information relative to him because the discarded information also contained everything he knew about gbit 1. Other compatible edges will correspond to Bell states in different question bases. 

Inspired by the compelling proposal within quantum theory in \cite{Brukner:vn,Brukner:ys}, we shall define states $\vec{y}_{O\rightarrow S}$ of a bipartite gbit system in $\cl_{\rm gbit}$ as 
\begin{description}
\item[entangled] if the composite information satisfies
$
\sum_{i,j=1}^{D_1}\,\alpha_{ij}>1\,\texttt{bit}$,
\end{description}
where $\alpha_{ij}(\vec{y}_{O\rightarrow S})$ is $O$'s information about the correlation question $Q_{ij}$. While there will be states with $\sum_{i,j=1}^{D_1}\,\alpha_{ij}\leq1\,\texttt{bit}$ which can indeed be considered as classically composed,\footnote{E.g., consider a state with $y_{11}=1$ and $y_i=\f{1}{2}$ for all other questions in $\cq_{M_2}$, corresponding to $O$ knowing with certainty that $Q_{11}=1$ and nothing else. The discussion surrounding (\ref{infoineq}) implies already that $\sum_{i,j=1}^{D_1}\,\alpha_{ij}=1\,\texttt{bit}$, while the individual information is zero. In quantum theory, this state would correspond to the separable state $\rho=\f{1}{4}(\mathds{1}+\sigma_x\otimes\sigma_x)$.} in contrast to \cite{Brukner:vn,Brukner:ys}, we emphasize that the above will only be a {\it sufficient}, but not a necessary condition for entanglement, as in the quantum case there will exist entangled states violating it.\footnote{For the quantum case, states with $\sum_{i,j=1}^{D_1=3}\,\alpha_{ij}\leq1\,\texttt{bit}$ where defined as classically composed in \cite{Brukner:vn,Brukner:ys}. This is, however, in conflict with the usual definition of entangled states as those which are non-separable. Namely, consider the family of quantum states (on the r.h.s.\ expressed in z-basis) for $-\f{1}{3}\leq a\leq\f{1}{3}$
\ba
\rho_a&:=&\f{1}{4}\left(\mathds{1}+\f{1}{3}(\sigma_z\otimes\mathds{1}+\mathds{1}\otimes\sigma_z-\sigma_z\otimes\sigma_z)\right.\nn\\&&\left.+2a(\sigma_x\otimes\sigma_x+\sigma_y\otimes\sigma_y)\right)=\left(\begin{array}{cccc}\f{1}{3} & 0 & 0 & 0 \\0 & \f{1}{3} & a & 0 \\0 & a & \f{1}{3} & 0 \\0 & 0 & 0 & 0\end{array}\right).\nn
\ea
This clearly is a family of non-separable states for $a\neq0$. In particular, for $a=\f{1}{3}$ this is the reduced state which one obtains from a W-state upon marginalizing over a qubit. These states yield $y_{z_1}=y_{z_2}=\f{2}{3}$, $y_{zz}=\f{1}{3}$ and $y_{xx}=y_{yy}=a+\f{1}{2}$ and all other $y_i=\f{1}{2}$, where, e.g., $y_{z_1}$ and $y_{zz}$ are the probabilities that the spin of qubit 1 is up in z-direction and that the spins of both qubits are the same in z-direction, respectively. Using the quadratic information measure $\alpha_{ij}=(2\,y_{ij}-1)^2$ which we derive in subsection \ref{sec_infomeasure} and which was also proposed for quantum theory by the authors of \cite{Brukner:vn,Brukner:ys}, these states give
\ba
I_{\rm comp}:=\sum_{i,j=1}^{3}\!\alpha_{ij}&=&(2\,y_{xx}-1)^2\!+\!(2\,y_{yy}-1)^2\!+\!(2\,y_{zz}-1)^2\nn\\&&+0\cdots+0=\left(\f{1}{9}+8a^2\right)\,\,\texttt{bit}\leq 1\,\,\texttt{bit},\nn
\ea
and for the individual information $I_{\rm indiv}:=(2y_{z_1}-1)^2+(2y_{z_2}-1)^2+0\cdots+0=2/9$ \texttt{bit}. $I_{\rm comp}$ can be arbitrarily close to $1/9$ \texttt{bit} for a sufficiently small $a\neq0$ so that even $I_{\rm comp}<I_{\rm indiv}$, and yet the corresponding state is entangled in violation of the proposed condition of classical composition in \cite{Brukner:vn,Brukner:ys}.}
It would be an interesting question to investigate what the precise sufficient and necessary conditions for entanglement are in the quantum case in this informational formulation. However, we shall leave this open for now.

Since there are only $N=2$ independent \texttt{bits} of information to be gained, according to rule \ref{lim}, the above sufficient condition for entanglement thus means, in particular, that $O$ has more composite than individual information. In general, two gbits will be referred to as {\it maximally} entangled relative to $O$ if he has {\it only} composite information 
about them which is incompatible with any individual information. By contrast, two gbits would be in a `product state' if $O$ has maximal individual knowledge of $2$ \texttt{bits} about them. For instance, if he knows that $Q_1=Q'_1=1$, $S$ would be in such a product state relative to him. But notice that even in this case, $O$ would have one dependent \texttt{bit} of information about the correlations because clearly $Q_{11}=1$ too. 

Note also that $O$ can `collapse' two gbits into an entangled state: as the most extreme example, consider the case that $O$ receives a `product ensemble state' of two gbits in a multiple shot interrogation. $O$ may then ask the next pair of gbits the questions $Q_{11},Q_{22}$. Upon receiving the answers, the two gbits will have `collapsed' into a maximally entangled {\it posterior} state relative to $O$, despite having been in a `product state' prior to interrogation. Entanglement is thus a property of $O$'s information about $S$.

We emphasize that, for systems with limited information content, {\it entanglement is a direct consequence of complementarity} -- at least for states of maximal information. To illustrate this observation, consider two classical bits (cbits). Since there is no complementarity in this case, there are only three pairwise independent questions: the individuals $Q_1,Q'_1$ about cbit 1 and cbit 2, respectively, and the correlation $Q_{11}$ would form an informationally complete $\cq_{M_N}$, graphically represented as
\vspace*{.2cm}
\begin{eqnarray}
\psfrag{q1}{\hspace*{-.3cm}cbit 1}
\psfrag{q2}{cbit 2}
\psfrag{Q1}{$Q_1$}
\psfrag{Q2}{$Q_2$}
\psfrag{Q3}{$Q_3$}
\psfrag{QD}{$Q_{D_1}$}
\psfrag{P1}{$Q'_1$}
\psfrag{P2}{$Q'_2$}
\psfrag{P3}{$Q'_3$}
\psfrag{PD}{$Q'_{D_1}$}
\psfrag{q11}{$Q_{11}$}
\psfrag{q21}{$Q_{21}$}
\psfrag{q22}{$Q_{22}$}
\psfrag{q13}{$Q_{13}$}
\psfrag{qdd}{$Q_{D_1D_1}$}
\psfrag{d}{$\vdots$}
{\includegraphics[scale=.2]{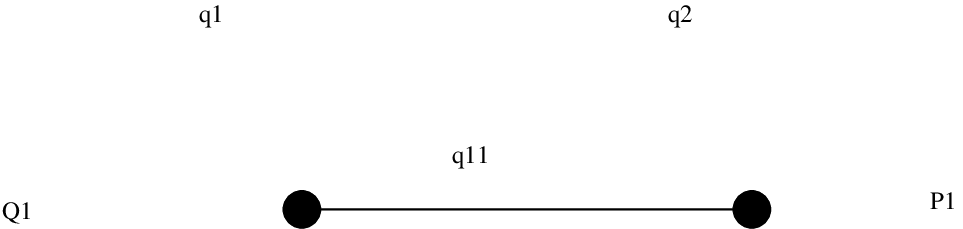}}\notag
\end{eqnarray}   
\vspace*{.1cm}

\noindent $Q_1,Q'_1,Q_{11}$ are mutually maximally compatible and this composite system satisfies rule \ref{lim} too such that $O$ may only acquire maximally $N=2$ independent \texttt{bits} of information about this pair of cbits. Clearly, also in this classical case it is possible for $O$ to acquire {\it only} composite information about the pair of cbits, namely by {\it only} asking $Q_{11}$ and not bothering about the individuals. $O$ would be able to find out whether identically prepared pairs of cbits in an ensemble are correlated by always asking $Q_{11}$ in a multiple shot interrogation. There will indeed exist states such that $y_{11}=1$ and $y_1=y'_1=\f{1}{2}$, however, 
the pair of cbits will not be in a state of maximal information relative to $O$ because he has only spent one of his two available \texttt{bits}. A classical state of maximal information corresponds to $O$ knowing the answers to two questions of the three $Q_1,Q'_1,Q_{11}$, but then by (\ref{correlation}, \ref{correlation2}) $O$ will also know the answer to the third. That is, in a state of maximal information about two cbits, $O$ will always have maximal individual information about the pair. For cbits, $O$ cannot spend the second \texttt{bit} also in composite information. By contrast, it is a consequence of the complementarity rule \ref{unlim} that in our case here, $O$ can actually exhaust the entire information available to him by composite questions, thereby giving rise to entanglement.

Thus far, we have only considered individual and correlation questions. All are part of $\cq_{M_2}$. But can there be more pairwise independent questions in $\cq_{M_2}$? These would have to be obtained from the individuals and correlations via another composition with an XNOR. Composing the individuals $Q_i,Q'_j$ with the correlation questions via XNOR, e.g.\ $Q_i\leftrightarrow Q_{il}$, will yield nothing new because questions need to be maximally compatible in order to be logically connected by $O$ and for compatible pairs it already follows from (\ref{correlation2}) that the individual and correlation questions are logically closed under $\leftrightarrow$. However, what about correlations of correlation questions, e.g., what about $Q_{11}\leftrightarrow Q_{22}$? 

\subsubsection{A logical argument for the dimensionality of the Bloch-sphere}\label{sec_3D}
 
The answer to the last question, in fact, is intertwined with the dimensionality $D_1$ of the state space $\Sigma_1$ of a single gbit and thus, ultimately, with the dimensionality of the Bloch-sphere. Deriving the dimension of the Bloch sphere has also been a crucial step in the successful reconstructions of quantum theory via the GPT framework. While Hardy's pioneering reconstruction \cite{Hardy:2001jk} did not fully settle this issue (it required an operationally somewhat obscure `simplicity axiom' to obtain $D_1=3$), it was later explicitly solved by impressive group-theoretic arguments in \cite{Masanes:2011kx} which show that, under certain information theoretic constraints, qubit quantum theory is the only theory with non-trivial entangling dynamics. This result was then exploited to relate the dimension of the Bloch-sphere to the dimension of space via a communication thought experiment \cite{Mueller:2012pc}. However, unfortunately, the ingredients of these arguments neither fit mathematically nor conceptually fully into our framework which relies on distinct structures than GPTs. 

On the other hand, although not yielding full quantum theory reconstructions, approaches asserting an epistemic restriction (i.e.\ a restriction of knowledge) over ontic states admit very simple and elegant arguments for a 1-\texttt{bit} system to have a state space spanned by three independent epistemic states \cite{spekkens2007evidence,Spekkens:2014fk,Paterek:2010fk}. These arguments can be carried out by considering a single system over a $2$-\texttt{bit} ontic state which at the epistemic level, however, is effectively a $1$-\texttt{bit} system due to epistemic restrictions. The arguments involve binary connectives (at the ontic level) of complementary questions which, while non-problematic when dealing with `hidden variables', are, however, deemed illegal in $O$'s model of the world according to our assumption \ref{assump4}. We do not wish to make any ontological commitments in our purely operational approach. Consequently, we must reason differently and without ontic states to deduce the dimensionality of $\Sigma_1$. In our case, this requires to consider {\it two} gbits rather than just one and involves entanglement, in analogy to the group-theoretic GPT derivation in \cite{Masanes:2011kx,Mueller:2012pc} which likewise requires two gbits and entanglement.

\begin{Theorem}\label{thm_3d}
$D_1=2$ or $3$.
\end{Theorem}

\begin{proof}
By lemma \ref{lem3}, $Q_{ij}$ and $Q_{kl}$ are independent and maximally compatible if $i\neq k$ and $j\neq l$. But any maximal set of pairwise maximally compatible correlation questions then contains precisely $D_1$ questions. Graphically, this is easy to see: one can always find $D_1$ non-intersecting edges between the $D_1$ vertices of gbit 1 and the $D_1$ vertices of gbit 2. For example, the $D_1$ `diagonal correlations' $Q_{ii}$, $i=1,\ldots,D_1$
 \begin{eqnarray}
\psfrag{d}{$\vdots$}
\psfrag{Q11}{$Q_{11}$}
\psfrag{Q22}{$Q_{22}$}
\psfrag{QDD}{$Q_{D_1D_1}$}
\psfrag{Q33}{$Q_{33}$}
{\includegraphics[scale=.2]{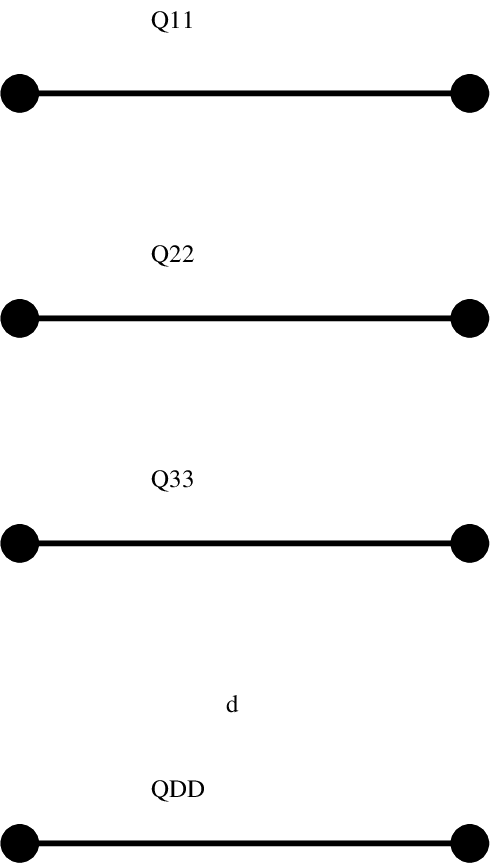}}\notag
\end{eqnarray}   
are pairwise maximally compatible and pairwise independent (lemma \ref{lem2}) such that, by theorem \ref{assump5} (Specker's principle), they must also be mutually maximally compatible. Hence, $O$ may acquire the answers to all $D_1$ $Q_{ii}$ at the same time. However, rule \ref{lim} forbids $O$'s information about $S$ to exceed the limit of $N=2$ independent \texttt{bits}. Accordingly, the $D_1$ correlation questions $Q_{ii}$ cannot be mutually independent if $D_1>2$ and 
the answers to any pair of them must imply the answers to all others. For instance, suppose $O$ has asked $Q_{11}$ and $Q_{22}$. The pair must determine the answers to $Q_{33},\ldots,Q_{D_1D_1}$ such that the latter must be Boolean functions of $Q_{11},Q_{22}$. Hence, $Q_{jj}=Q_{11}*Q_{22}$, $j\neq1,2$, for some logical connective $*$ which preserves that $Q_{11},Q_{22},Q_{jj}$ are pairwise independent. Table (\ref{truth}) implies that $*$ must either be the XNOR $\leftrightarrow$ or the XOR $\oplus$ such that either
\ba
Q_{jj}&=&\,\,\,\,\,\,Q_{11}\leftrightarrow Q_{22},\q\q\q\text{or}\nn\\
Q_{jj}&=&\neg(Q_{11}\leftrightarrow Q_{22}),\q\q\q\forall\, j=3,\ldots,D_1.\nn
\ea
But then clearly $Q_{jj}$, $j=3,\ldots,D_1$ could {\it not} be pairwise independent if $D_1>3$. The same argument can be carried out for {\it any} set of $D_1$ pairwise independent and maximally compatible correlations $Q_{ij}$, i.e.\ for any question graph with $D_1$ non-intersecting edges, and any choice of a pair of maximally compatible correlations which $O$ may ask first. We must therefore conclude that $D_1\leq3$, while rule \ref{unlim} implies that $D_1\geq2$.
\end{proof}

We shall refer to $D_1=2$ as the `rebit case' and to $D_1=3$ as the `qubit case'; for the moment this is just suggestive terminology,
however, we shall see later and in \cite{hw2} that the $D_1=2$ case will, indeed, result in rebit quantum theory (two-level systems over {\it real} Hilbert spaces), while the $D_1=3$ case will give rise to standard quantum theory of qubit systems \cite{hw}.

\subsubsection{An informationally complete set for $N=2$ qubits}\label{sec_n2complete}

The two cases have distinct properties as regards questions which ask for the correlations of bipartite questions and it is necessary to consider them separately. We begin with the simpler qubit case $D_1=3$ for which correlations of correlation questions do not define new information for $O$ about $S$ such that six individual and nine correlation questions constitute an informationally complete set $\cq_{M_2}$. These will ultimately correspond to the propositions `the spin is up in $x-,y-,z-$direction' for the individual qubits 1 and 2 and `the spins of qubit 1 in $i-$ and qubit 2 in $j-$ direction are the same' where $i,j=x,y,z$. Note that the density matrix for two qubits has 15 parameters.

\begin{widetext}
\begin{Theorem}\label{thm_qubit}{\bf(Qubits)}
If $D_1=3$ then the correlation questions $Q_{ij}$ are logically closed under $\leftrightarrow$ and $\cq_{M_2}=\{Q_i,Q'_j,Q_{ij}\}_{i,j=1,2,3}$ constitutes an informationally complete set for $N=2$ with $D_2=15$. Furthermore, for any two permutations $\sigma,\sigma'$ of $\{1,2,3\}$ either
\ba
Q_{\sigma(3)\sigma'(3)}=Q_{\sigma(1)\sigma'(1)}\leftrightarrow Q_{\sigma(2)\sigma'(2)},\q\q\q\text{or}\q\q\q Q_{\sigma(3)\sigma'(3)}=\neg(Q_{\sigma(1)\sigma'(1)}\leftrightarrow Q_{\sigma(2)\sigma'(2)})\label{qbit1}
\ea
and either
\ba
Q_{\sigma(3)\sigma'(3)}=Q_{\sigma(1)\sigma'(2)}\leftrightarrow Q_{\sigma(2)\sigma'(1)},\q\q\q\text{or}\q\q\q Q_{\sigma(3)\sigma'(3)}=\neg(Q_{\sigma(1)\sigma'(2)}\leftrightarrow Q_{\sigma(2)\sigma'(1)})\label{qbit2}
\ea
such that either
\ba
Q_{\sigma(1)\sigma'(1)}\leftrightarrow Q_{\sigma(2)\sigma'(2)}&=&\q Q_{\sigma(1)\sigma'(2)}\leftrightarrow Q_{\sigma(2)\sigma'(1)},\q\q\q\text{or}\nn\\ Q_{\sigma(1)\sigma'(1)}\leftrightarrow Q_{\sigma(2)\sigma'(2)}&=&\neg(Q_{\sigma(1)\sigma'(2)}\leftrightarrow Q_{\sigma(2)\sigma'(1)}).\label{qbit3}
\ea
\end{Theorem}
\end{widetext}

\begin{proof}
Statements (\ref{qbit1}, \ref{qbit2}) are an immediate consequence of the argument in the proof of theorem \ref{thm_3d} which can be applied to any correlation question graph with $D_1=3$ non-intersecting edges and thus to any two permutations $\sigma,\sigma'$ of $\{1,2,3\}$. From this it directly follows that the correlation questions $Q_{ij}$ are logically closed under $\leftrightarrow$ such that in the case $D_1=3$ correlations of correlations such as $Q_{ij}\leftrightarrow Q_{kl}$ do {\it not} define any new independent questions. Accordingly, 
\ba
\cq_{M_2}=\{Q_i,Q'_j,Q_{ij}\}_{i,j=1,2,3}\nn
\ea
is an informationally complete set of questions for $N=2$ and $D_1=3$ such that $D_2=15$.
\end{proof}

For example, if $\sigma,\sigma'$ are both trivial then the theorem applied to the following two question graphs
\vspace*{.2cm}
 \begin{eqnarray}
\psfrag{d}{$\vdots$}
\psfrag{Q11}{$Q_{11}$}
\psfrag{Q22}{$Q_{22}$}
\psfrag{Q12}{$Q_{12}$}
\psfrag{Q21}{$Q_{21}$}
\psfrag{Q33}{$Q_{33}$}
{\includegraphics[scale=.2]{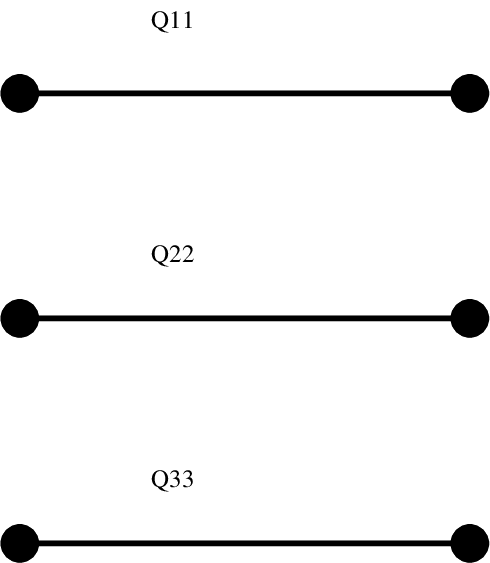}}\q\q\q\q\q\q\q\q\q{\includegraphics[scale=.2]{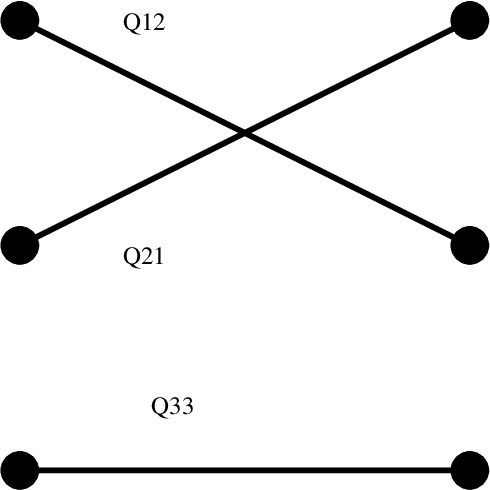}}\notag
\end{eqnarray} 
implies for $D_1=3$ either
\begin{widetext}
\ba
Q_{33}=Q_{11}\leftrightarrow Q_{22},\q\q\q\text{or}\q\q\q Q_{33}=\neg(Q_{11}\leftrightarrow Q_{22})\label{qbit4}
\ea
and either
\ba
Q_{33}=Q_{12}\leftrightarrow Q_{21},\q\q\q\text{or}\q\q\q Q_{33}=\neg(Q_{12}\leftrightarrow Q_{21})\label{qbit5}
\ea
such that either
\ba
Q_{11}\leftrightarrow Q_{22}=Q_{12}\leftrightarrow Q_{21},\q\q\q\text{or}\q\q\q Q_{11}\leftrightarrow Q_{22}=\neg(Q_{12}\leftrightarrow Q_{21})\label{qbit6}
\ea
\end{widetext}
In sections \ref{sec_bell} and \ref{sec_corr} we shall determine whether negations $\neg$ occur in the expressions (\ref{qbit1}--\ref{qbit3}). For $Q,Q',Q''$ maximally compatible and related by an XNOR, we shall henceforth distinguish between
\begin{description}
\item[even correlation:] if $Q=Q'\leftrightarrow Q''$, and
\item[odd correlation:] if $Q=\neg(Q'\leftrightarrow Q'')$.\footnote{We note that this also implies $Q'=\neg(Q\leftrightarrow Q'')$ and $Q''=\neg(Q\leftrightarrow Q')$.}
\end{description}

\subsubsection{An informationally complete set for $N=2$ rebits}

Next, let us consider the rebit case $D_1=2$. $O$ can ask the four individual questions $Q_1,Q_2,Q'_1,Q'_2$ and the four correlations $Q_{11},Q_{12},Q_{21},Q_{22}$. But $O$ can also {\it define} the two new correlation of correlations questions
\ba
\!\!\!\!\!\!\!\!\!\!\!\!\!\!\!Q_{33}:=Q_{12}\leftrightarrow Q_{21},\q\q\tilde{Q}_{33}:=Q_{11}\leftrightarrow Q_{22}\label{q33re}
\ea
corresponding to the two correlation questions graphs
\vspace*{.4cm}
 \begin{eqnarray}
\psfrag{Q1}{$Q_1$}
\psfrag{Q2}{$Q_2$}
\psfrag{Q3}{$Q_3$}
\psfrag{QD}{$Q_{D_1}$}
\psfrag{P1}{\!\!$Q'_1$}
\psfrag{P2}{\!\!$Q'_2$}
\psfrag{P3}{$Q'_3$}
\psfrag{PD}{$Q'_{D_1}$}
\psfrag{q11}{$Q_{11}$}
\psfrag{q21}{$Q_{21}$}
\psfrag{q22}{$Q_{22}$}
\psfrag{q12}{$Q_{12}$}
\psfrag{q33}{$Q_{33}$}
\psfrag{qt33}{$\tilde{Q}_{33}$}
\psfrag{d}{$\vdots$}
\!\!\!\!{\includegraphics[scale=.2]{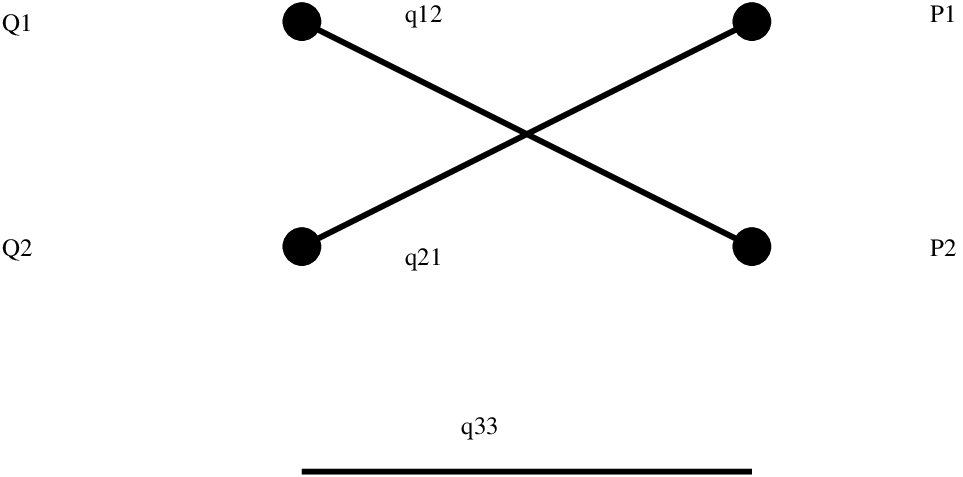}}\q\q\q\q{\includegraphics[scale=.2]{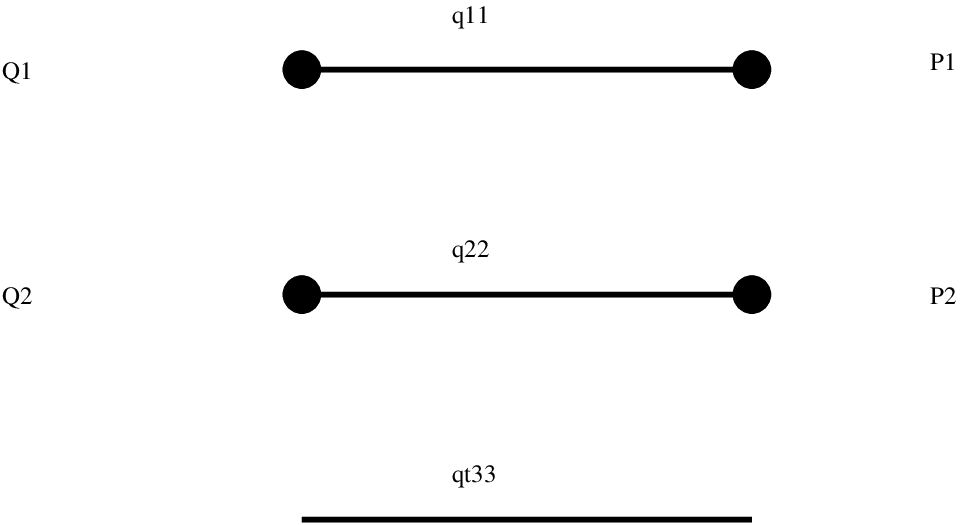}}\notag
\end{eqnarray}  
\vspace*{.1cm}

\noindent(we add the correlation of correlations as a new edge without vertices to the graph). The difference to the qubit case is that $Q_{33},\tilde{Q}_{33}$ can now {\it not} be written as correlations of individual questions $Q_3,Q'_3$ since the latter do not exist. Clearly, $Q_{12},Q_{21},Q_{33}$ and $Q_{11},Q_{22},\tilde{Q}_{33}$ are two pairwise independent sets. The question is whether $Q_{33},\tilde{Q}_{33}$ are independent of each other and of the individuals. The following lemma even asserts complementarity to the latter. (Note that this lemma holds trivially in the case $D_1=3$.)

\begin{lem}\label{lem4}
The questions asking for the correlation of the bipartite correlations, $Q_{12}\leftrightarrow Q_{21}$ and $Q_{11}\leftrightarrow Q_{22}$, are maximally complementary to and independent of any $Q_1,Q_2,Q'_1,Q'_2$.
\end{lem}

\begin{proof}
We firstly demonstrate independence of $Q_1$ and $Q_{11}\leftrightarrow Q_{22}$. This follows from noting that $Q_{22}$ is maximally complementary to $Q_1$ (lemma \ref{lem1}) and maximally compatible with $Q_{11}\leftrightarrow Q_{22}$ and arguments completely analogous to those in the proof of lemma \ref{lem2}. Next, we establish complementarity of $Q_{11}\leftrightarrow Q_{22}$ and $Q_1$. Suppose the contrary, namely, that $Q_{11}\leftrightarrow Q_{22}$ and $Q_1$ were at least partially compatible such that there must exist a state with $\alpha_{Q_{11}\leftrightarrow Q_{22}}=1$ \texttt{bit} and $\alpha_1>0$ \texttt{bit}. Clearly, both are maximally compatible with and independent of $Q_{11}$. According to assumption \ref{assump5b}, asking $Q_{11}$ to this state does not change $O$'s information about $Q_{11}\leftrightarrow Q_{22}$ and $Q_1$. However, maximal information about $Q_{11}\leftrightarrow Q_{22}$ and $Q_{11}$ also implies maximal information about $Q_{22}$ which is maximally complementary to $Q_1$. Hence, $Q_1$ and $Q_{11}\leftrightarrow Q_{22}$ must be maximally complementary. The arguments works analogously for the other cases.
\end{proof}

Finally, we show that (\ref{qbit4}--\ref{qbit6}) hold analogously for the rebit case $D_1=2$. From this it follows that the four individual questions $Q_i,Q'_j$, the four correlations $Q_{ij}$, $i,j=1,2$, and the correlation of correlations $Q_{33}$ form an informationally complete set $\cq_{M_2}$ for rebits. That is, in contrast to the qubit case, there is now one non-trivial correlation of correlations question which is pairwise independent from the individuals and correlations.

\begin{Theorem}\label{thm_rebit}{\bf (Rebits)}
If $D_1=2$ then $\cq_{M_2}=\{Q_i,Q_j',Q_{ij},Q_{33}\}_{i,j=1,2}$ constitutes an informationally complete set for $N=2$ with $D_2=9$ and either
\ba
Q_{11}\leftrightarrow Q_{22}&=&\q Q_{12}\leftrightarrow Q_{21},\q\q\q\text{or}\nn\\ Q_{11}\leftrightarrow Q_{22}&=&\neg(Q_{12}\leftrightarrow Q_{21}).\nn
\ea
\end{Theorem}

\begin{proof}
$O$ can begin by asking $S$ the maximally compatible $Q_{11}$ and $Q_{22}$ upon which he also knows the answer to $\tilde{Q}_{33}$. $O$ then possesses the maximal amount of $N=2$ independent \texttt{bits} of information about $S$. Next, $O$ can ask $Q_{12}$ (or $Q_{21}$) which, according to lemma \ref{lem3}, is maximally complementary to $Q_{11},Q_{22}$. But, by rule \ref{unlim}, $O$ is not allowed to experience a net loss of information. Hence, $Q_{12}$ (or $Q_{21})$ must be maximally compatible with $\tilde{Q}_{33}$ (they are also independent). That is, $Q_{12},Q_{21},\tilde{Q}_{33}$ are mutually maximally compatible according to theorem \ref{assump5} (Specker's principle) and $O$ may also ask all three at the same time. But then the same argument as in the proof of theorem \ref{thm_3d} applies such that either $\tilde{Q}_{33}=Q_{12}\leftrightarrow Q_{21}$ or $\tilde{Q}_{33}=\neg(Q_{12}\leftrightarrow Q_{21})$. This implies either $Q_{33}=\tilde{Q}_{33}$ or $Q_{33}=\neg\tilde{Q}_{33}$. Accordingly, there is only one independent correlation of correlations $Q_{33}$. Since $Q_{11},Q_{12},Q_{21},Q_{22}$ and $Q_{33}$ are then, by construction, logically closed under the XNOR $\leftrightarrow$ and $Q_{33}$ is maximally complementary to all individuals $Q_1,Q_2,Q'_1,Q'_2$, no further pairwise independent question can be built from this set such that it is informationally complete. Hence, $D_2=9$.
\end{proof}

\subsubsection{Entanglement and non-local tomography for rebits}\label{sec_rebtang}

The rebit case $D_1=2$ thus has a very special question and correlation structure: $Q_{33}=Q_{12}\leftrightarrow Q_{21}$ is the only composite question which is maximally complementary to all individuals, but maximally compatible with all correlations $Q_{ij}$. By contrast, e.g., $Q_{11}$ is maximally complementary to $Q_2,Q'_2,Q_{12},Q_{21}$ and maximally compatible with $Q_1,Q'_1,Q_{22}$ and maximally compatible with the correlation of correlations $Q_{33}$. Consequently, $Q_{33}$ assumes a special role in the entanglement structure: once $O$ knows the answer to this question he may no longer have {\it any} further information about the outcomes of individual questions, but may have information about correlation questions. That is, two rebits can be in a state of non-maximal information of $1$ \texttt{bit} relative to $O$, corresponding to the latter only having maximal knowledge about the answer to $Q_{33}$ and no information otherwise, and still be {\it maximally} entangled because any individual information is forbidden in that situation. Notice that this is {\it not} true for pairs of qubits because even if everything $O$ knew about the pair was the answer to $Q_{33}$, he could still acquire information about the individuals $Q_3,Q'_3$ such that one could not consider such a state as maximally entangled. 

Even stronger, $O$ will {\it always} know the answer to $Q_{33}$ if the two rebits are maximally entangled and he has maximal information about $S$. This follows from (\ref{totalinfo}): for a maximally entangled state of maximal information (no information about individuals), the following must hold
\ba
I_{O\rightarrow S}&=&\underset{=0}{\underbrace{\sum_{i=1}^2(\alpha_i+\alpha'_i)}}+\sum_{i,j=1}^2\alpha_{ij} +\alpha_{33}\nn\\
&=&\alpha_{11}+\alpha_{12}+\alpha_{21}+\alpha_{22}+\alpha_{33}=3\,\texttt{bits}.\nn
\ea
The last equality follows from the fact that once $O$ knows the answers to two maximally compatible questions, he will also know the answer to their correlation and since this correlation is in the pairwise independent question set the total information defined in (\ref{totalinfo}) will always yield 3 \texttt{bits} for $N=2$ independent \texttt{bit} systems in states of maximal knowledge.\footnote{We are implicitly using here also rules \ref{pres} and \ref{time}.} If we now require that $O$ cannot know more than a single bit about maximally complementary questions then, in a state of maximal information of $N=2$ independent \texttt{bits}, we must have
\ba
\alpha_{11}+\alpha_{12}=\alpha_{21}+\alpha_{22}=1\,\texttt{bit}\q\q\q\q\q\q\nn\\\Rightarrow\q\q\q \alpha_{33}=1\,\texttt{bit}\nn
\ea
such that $O$ must have maximal information about $Q_{33}$. Therefore, the correlation of correlations $Q_{33}$ can be viewed as the litmus test for entanglement of two rebits.

We note that the individuals $Q_1,Q_2,Q_1',Q_2'$ will ultimately correspond to projections onto the $+1$ eigenspaces of $\sigma_x,\sigma_z$, while the $Q_{ij}$ correspond to projections onto the $+1$ eigenspaces of $\sigma_i\otimes\sigma_j$, $i,j=x,z$, and $Q_{33}$ corresponds to the projection onto the $+1$ eigenspace of $\sigma_y\otimes\sigma_y$, where $\sigma_x,\sigma_y,\sigma_z$ are the Pauli matrices \cite{hw,hw2}. Observables and density matrices on a real Hilbert space correspond to real symmetric matrices. This is the reason why $\sigma_y$ is {\it not} an observable on $\mathbb{R}^2$ (it corresponds to the `missing' $Q_3$), but $\sigma_y\otimes\sigma_y$ {\it is} a real symmetric matrix and thus an observable on $\mathbb{R}^2\otimes\mathbb{R}^2$. 

This gives a novel and simple explanation for the discovery that $\sigma_y\otimes\sigma_y$ determines the entanglement of rebits \cite{caves2001entanglement}: a two-rebit density matrix $\rho$ is separable if and only if $\Tr(\rho\,\sigma_y\otimes\sigma_y)=0$, i.e.\ if the state has no $\sigma_y\otimes\sigma_y$ component. This statement means in our language that a state is separable if and only if $\alpha_{33}=0$ and is consistent with our observation above because {\it any} information about the individuals is incompatible with information about $Q_{33}$ due to complementarity.

This question structure also has severe repercussions for rebit state tomography: it must ultimately be {\it non-local}. For rebits, $D_1=2$, the probability that $Q_{33}=1$ could be written as 
\ba
y_{33}=p(Q_{12}=1,Q_{21}=1)+p(Q_{12}=0,Q_{21}=0)\nn
\ea
where $p(Q_i,Q_j)$ denotes here the joint probability distribution over $Q_i,Q_j$. That is, in a multiple shot interrogation, $O$ could ask both $Q_{12},Q_{21}$ to the identically prepared rebit couples and from the statistics over the answers, he could also determine $y_{33}$. But this probability cannot be decomposed into joint probabilities over the individual questions according to $Q_{33}=``(Q_1\leftrightarrow Q_2')\leftrightarrow(Q_2\leftrightarrow Q'_1)"$ because $Q_1,Q_2$ and $Q'_1,Q'_2$ are maximally complementary. Therefore, $O$ would not be able to determine $y_{33}$ by only asking individual questions to the two rebits and the statistics over these answers.\footnote{For example, in a multiple shot interrogation $O$ could first ask $Q_1,Q'_2$ on a set of identically prepared rebit couples to find out whether the answers are correlated. On a second identically prepared set, $O$ could then ask $Q_2,Q'_1$ which are maximally complementary to the first questions he asked. From the statistics over the answers, $O$ would be able to determine also the probabilities for $Q_ {12}$ and $Q_{21}$. But since he needed two separate interrogation runs to determine the statistics for $Q_{12}$ and $Q_{21}$, he would not be able to infer from this {\it any} information whatsoever about the statistics of answers to $Q_{33}$.} That is, for rebits, state tomography would always require correlation questions and in this sense be non-local. Note that this stands in stark contrast to qubit pairs where $Q_{33}=Q_3\leftrightarrow Q'_3$ can be written in terms of individual questions such that also $y_{33}=p(Q_3=1,Q'_3=1)+p(Q_3=0,Q'_3=0)$ can be determined by the statistics over the answers to $Q_3,Q'_3$ only. Qubit systems (and quantum theory in general) are thus tomographically local.

The requirement of {\it tomographic locality}, according to which the state of a composite system can be determined by doing statistics over measurements on its subsystems, is a standard condition in the GPT framework \cite{masanes2011derivation,Mueller:2012ai,barrett2007information,Dakic:2009bh,Masanes:2011kx,de2012deriving,Masanes:2012uq} and thus directly rules out rebit theory. However, in contrast to derivations within the GPT landscape, we shall not implement local tomography here because it will be interesting to see the differences between real and complex quantum theory from the perspective of information inference and we shall thus carry out the reconstruction of both here and in \cite{hw,hw2}.

Local tomography is usually taken as the origin of the tensor product structure for composite systems in quantum theory \cite{Masanes:2011kx,masanes2011derivation,Mueller:2012ai,barrett2007information}. However, one has to be careful with this statement because there exist {\it two distinct} tensor products: there is (a) the tensor product of Hilbert spaces, e.g., $(\mathbb{C}^2)^{\otimes N}$ for qubits and $(\mathbb{R}^2)^{\otimes N}$ for rebits, and (b) the tensor product of {\it unnormalized} probability vectors or density matrices. A tensor product of type (b) defines a sufficient support only for composite qubit systems, but not for composite rebit systems; the space of hermitian matrices over $(\mathbb{C}^2)^{\otimes N}$ is the $N$-fold tensor product of hermitian matrices over $\mathbb{C}^2$, but the space of symmetric matrices over $(\mathbb{R}^2)^{\otimes N}$ is {\it not} the $N$-fold tensor product of symmetric matrices over $\mathbb{R}^2$ (see also \cite{caves2001entanglement}). What local tomography implies is the tensor product (b), but as the example of rebits shows, the tensor product of type (a) also exists without it. 


\subsubsection{A Bell scenario with questions: ruling out {\it local} hidden variables}\label{sec_bell}

We shall now settle the issue of the relative negation $\neg$, i.e.\ whether
\ba
Q_{11}\leftrightarrow Q_{22}&=&\q Q_{12}\leftrightarrow Q_{21},\q\q\q\text{or}\nn\\ Q_{11}\leftrightarrow Q_{22}&=&\neg(Q_{12}\leftrightarrow Q_{21}),\nn
\ea
for both rebits and qubits; one of these equations must be true (see theorems \ref{thm_qubit} and \ref{thm_rebit}). 
It is instructive to illustrate these equations by means of a question configuration analogous to a Bell scenario.

{\bf(A)} Suppose $Q_{11}\leftrightarrow Q_{22}=Q_{12}\leftrightarrow Q_{21}$ was true. Before we determine whether it is consistent with our background assumptions and postulates, we note that this is the case of classical (or realist) logic: $O$ can consistently interpret any such configuration by means of a {\it local} `hidden variable' model. For example, consider the case $Q_{11}=Q_{22}=1$ (which is compatible with the rules) such that $Q_{11}\leftrightarrow Q_{22}=1$ and, consequently, also $Q_{12}\leftrightarrow Q_{21}=1$. We represent the last two equations by the following two graphs
\vspace*{-.05cm}
\begin{eqnarray}
\psfrag{Q12}{$\!\textcolor{red}{Q_{12}}$}
\psfrag{Q21}{$\,\textcolor{red}{Q_{21}}$}
 \psfrag{Q11}{$\textcolor{blue}{Q_{11}}$}
\psfrag{Q22}{$\textcolor{blue}{Q_{22}}$}
{\includegraphics[scale=.2]{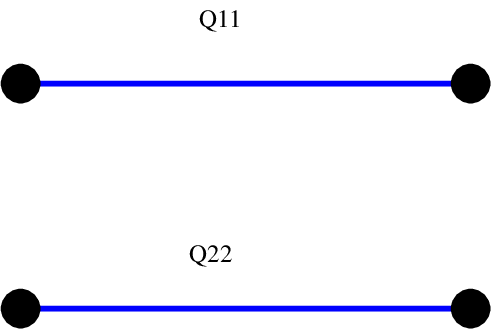}}\q\q\q\q\q\q\q\q{\includegraphics[scale=.2]{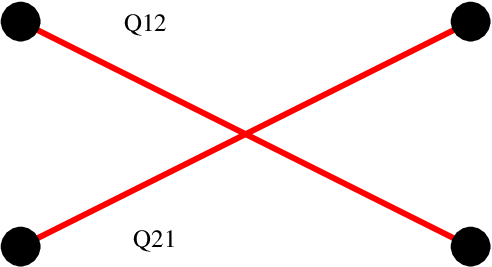}}\nn 
\end{eqnarray}
where both edges solid and of equal colour means that the corresponding questions are correlated. $O$ could consistently read this configuration as follows: ``Since $``Q_1=Q'_1"$ and $``Q_2=Q'_2"$, I would precisely have to conclude that $``Q_{12}=Q_{21}"$, i.e.\ $Q_{12}\leftrightarrow Q_{21}=1$, if the individuals $Q_1,Q_2,Q'_1,Q'_2$ had definite values which, however, I do not know.'' 
For instance, $Q_{11}= Q_{22}=Q_{12}\leftrightarrow Q_{21}=1$ would be consistent with the four ontic states:
\begin{widetext}
 \begin{eqnarray}
 \psfrag{1}{$1$}
 \psfrag{0}{$0$}
{\includegraphics[scale=.2]{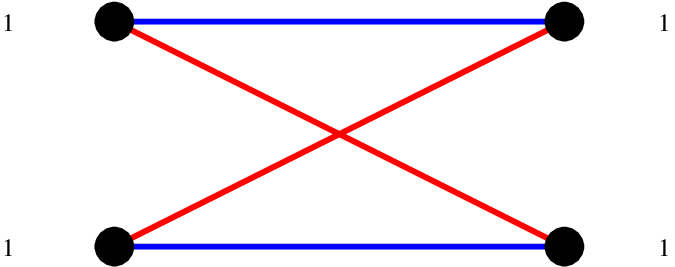}}\q\q\q\q{\includegraphics[scale=.2]{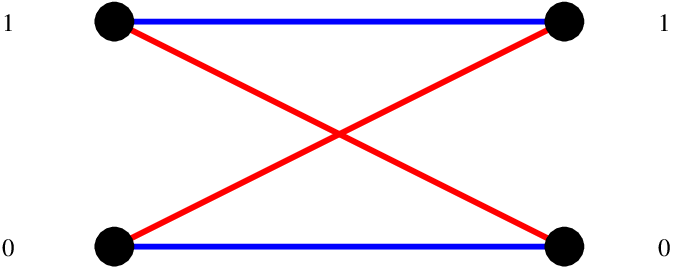}}\q\q\q\q{\includegraphics[scale=.2]{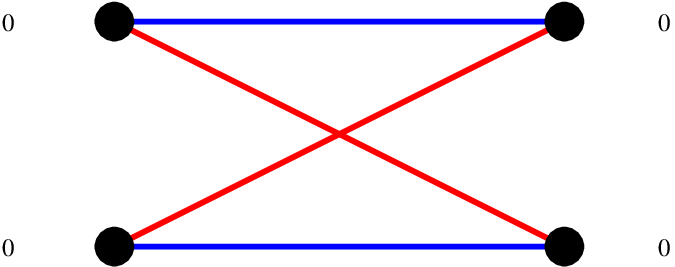}}\q\q\q\q{\includegraphics[scale=.2]{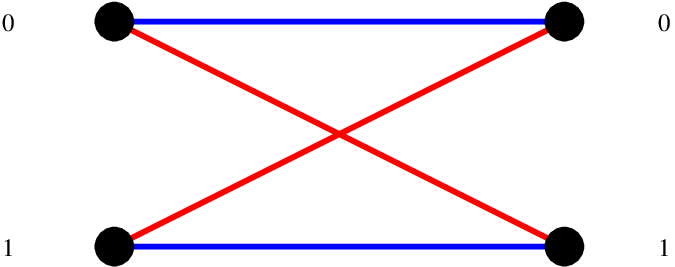}}\,\,\,\,.\nn 
\end{eqnarray} 
\end{widetext}
Therefore, $O$ could join the two graphs consistently (without knowing the truth values of the individuals)
\begin{eqnarray}
{\includegraphics[scale=.2]{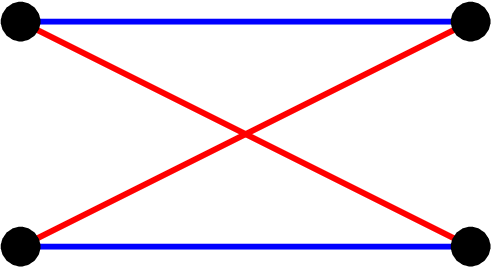}}\label{join}
\end{eqnarray}
i.e., draw all four edges simultaneously which corresponds to all four edges having definite (albeit unknown) values at the same time. 
Similarly, $O$ can interpret any other configuration with $Q_{11}\leftrightarrow Q_{22}=Q_{12}\leftrightarrow Q_{21}$ in terms of a local `hidden variable' model which assigns definite values to the individuals and, conversely, as can be easily checked, any assignment of definite (or ontic) values to the individuals $Q_1,Q_2,Q'_1,Q'_2$ leads necessarily to $Q_{11}\leftrightarrow Q_{22}=Q_{12}\leftrightarrow Q_{21}$. 

We now preclude this possibility. We note that $Q_{11}\leftrightarrow Q_{22}=Q_{12}\leftrightarrow Q_{21}$ can be rewritten in terms of the individuals as
\begin{widetext}
\ba
(Q_1\leftrightarrow Q_1')\leftrightarrow(Q_2\leftrightarrow Q_2')=(Q_1\leftrightarrow Q_2')\leftrightarrow(Q_2\leftrightarrow Q_1')\label{noclassfo}
\ea
\end{widetext}
which is a {\it classical logical identity} that is {\it always} true for classical Boolean logic. It states that the brackets can be equivalently regrouped, i.e.\ that the order in which the XNOR $\leftrightarrow$ is executed in the logical expression can be switched. In classical Boolean logic, this identity follows from the associativity and symmetry of the XNOR. Namely, suppose for a moment that $Q_1,Q_2,Q_1',Q_2'$ had simultaneous truth values. Then Boolean logic would tell us that
\begin{widetext}
\ba
(Q_1\leftrightarrow Q_1')\leftrightarrow(Q_2\leftrightarrow Q_2')&=&\left((Q_1\leftrightarrow Q_1')\leftrightarrow Q_2\right)\leftrightarrow Q_2'=\left(Q_1\leftrightarrow(Q_2\leftrightarrow Q_1')\right)\leftrightarrow Q_2'\nn\\
&=&Q_2'\leftrightarrow\left(Q_1\leftrightarrow(Q_2\leftrightarrow Q_1')\right)=(Q_1\leftrightarrow Q_2')\leftrightarrow(Q_2\leftrightarrow Q_1').\label{noclassfo2}
\ea
\end{widetext}
Of course, relative to $O$, $Q_1,Q_2,Q_1',Q_2'$ do not have simultaneous truth values. Instead, assumption \ref{assump4b} implies that, 
for a classical logical identity to hold in $O$'s model, either (a) all questions in the involved logical expressions are mutually maximally compatible or (b) the identity follows from applying classical rules of inference to (sub-)expressions which are entirely written in terms of mutually maximally compatible questions and subsequently decomposing the result logically into other questions (with possibly distinct compatibility relations).
However, since $Q_1,Q_2$ and $Q_1',Q_2'$ form maximally complementary pairs, neither is the case here. In particular, the right hand side in (\ref{noclassfo}) does {\it not} follow from applying any rules of Boolean logic to the total expression $Q_{11}\leftrightarrow Q_{22}$ or those subexpressions $Q_1\leftrightarrow Q_1'$ and $Q_2\leftrightarrow Q_2'$ on the left hand side which are fully composed of maximally compatible questions.\footnote{The only relevant rules in this case, as in (\ref{noclassfo2}), follow from the properties of the XNOR, namely its symmetry and associativity in Boolean logic. However, it is clear that, thanks to the occurring complementarity, the arguments of (\ref{noclassfo2}) no longer apply. For instance, the first step in (\ref{noclassfo2}) is illegal, according to assumption \ref{assump4}, because $Q_2$ and $Q_1\leftrightarrow Q_1'$ are complementary such that they may not be directly connected. } Hence, (\ref{noclassfo}) violates assumption \ref{assump4b} 
such that we must rule out the possibility $Q_{11}\leftrightarrow Q_{22}=Q_{12}\leftrightarrow Q_{21}$.

{\bf(B)} We are thus already forced to the conclusion that 
\ba
Q_{11}\leftrightarrow Q_{22}=\neg(Q_{12}\leftrightarrow Q_{21}).\label{qbit7}
\ea
must hold. Indeed, its equivalent representation
\ba
(Q_1\leftrightarrow Q_1')\leftrightarrow(Q_2\leftrightarrow Q_2')\q\q\q\q\q\q\q\nn\\=\neg(Q_1\leftrightarrow Q_2')\leftrightarrow(Q_2\leftrightarrow Q_1')\nn
\ea
does {\it not} constitute a classical logical identity and is thus consistent with assumption \ref{assump4b}. It is also consistent with assumption \ref{assump4} since only maximally compatible questions are {\it directly} connected.

It is impossible for $O$ to interpret this situation in terms of {\it local} `hidden variables' which assign definite values to the individuals simultaneously. Namely, consider, again, the case $Q_{11}=Q_{22}=1$ such that $Q_{11}\leftrightarrow Q_{22}=1$ and, consequently, now $Q_{12}\leftrightarrow Q_{21}=0$. This configuration may be graphically represented as
\vspace*{.2cm}
\begin{eqnarray}
\psfrag{Q12}{$\!\textcolor{red}{Q_{12}}$}
\psfrag{Q21}{$\,\textcolor{red}{Q_{21}}$}
 \psfrag{Q11}{$\textcolor{blue}{Q_{11}}$}
\psfrag{Q22}{$\textcolor{blue}{Q_{22}}$}
{\includegraphics[scale=.2]{Figures/even2.eps}}\q\q\q\q\q\q\q\q{\includegraphics[scale=.2]{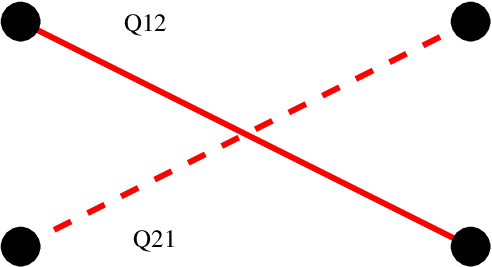}}\nn 
\end{eqnarray}
\vspace*{.2cm}

\noindent where one edge solid the other dashed, but same colour, means that the corresponding answers are anti-correlated. One can easily check that, in contrast to (\ref{join}), it is impossible to consistently join the two diagrams by drawing all four edges simultaneously (which would correspond to all individuals and thus all edges having definite truth values), as one would obtain a frustrated graph 
\vspace*{.2cm}
 \begin{eqnarray}
\psfrag{Q12}{}
\psfrag{Q21}{}
 \psfrag{Q11}{}
\psfrag{Q22}{}
{\includegraphics[scale=.2]{Figures/even2.eps}}\q\q\q\nLeftrightarrow\q\q\q{\includegraphics[scale=.2]{Figures/odd3.eps}}\,\,\,.\label{frustrated}
\end{eqnarray} 
\vspace*{.2cm}

\noindent$O$ must conclude that neither the four individuals nor the four edges can have definite (but unknown) values all at the same time. 
The same verdict holds for any other configuration with $Q_{11}\leftrightarrow Q_{22}=\neg(Q_{12}\leftrightarrow Q_{21})$.  


Since either (A) or (B) must be true by theorems \ref{thm_qubit} and \ref{thm_rebit}, and (A) violates assumption \ref{assump4b}, while (B) is consistent with both assumptions \ref{assump4} and \ref{assump4b}, we conclude that (\ref{qbit7}) is the correct relation.

Note that this argument holds for both rebits and qubits and, for qubits, also for any other pairs $Q_i,Q_{j\neq i}$ and $Q_k',Q'_{l\neq k}$ of pairs of maximally complementary individuals (see theorem \ref{thm_qubit}). In fact, even more generally, the same argument can be made for {\it any} four questions (i.e., not necessarily individuals) $Q,Q',Q'',Q'''$ which are such that $Q,Q'$ and $Q'',Q'''$ are maximally complementary pairs, while $Q$ and $Q'$ are each maximally compatible with both $Q'',Q'''$; also in this case one would have to conclude that
\begin{widetext}
\ba
(Q\leftrightarrow Q'')\leftrightarrow (Q'\leftrightarrow Q''')=\neg\left((Q\leftrightarrow Q''')\leftrightarrow(Q'\leftrightarrow Q'')\right).\label{qbit8}
\ea 
\end{widetext}
This will become relevant later on. For now we observe that exchanging the positions of the complementary $Q'',Q'''$ from the left to the right hand side introduces a negation $\neg$. 

This relative negation $\neg$ in (\ref{qbit8}) precludes a classical reasoning for the distribution of truth values over $O$'s questions. In fact, as just seen, it rules out {\it local} `hidden variable models', analogously to the Bell arguments. We also recall that assumption \ref{assump4b} was a statement about which rules of inference and logical identities are {\it not} applicable to logical compositions involving mutually complementary questions. By contrast, (\ref{qbit8}) is now a first non-classical logical identity showing one possible non-classical way of reasoning in the presence of complementarity. 
 
The correlation structure for rebits is now clear from these results; (\ref{qbit7}) implies $Q_{33}=\neg\tilde{Q}_{33}$ for the correlations of correlations defined in (\ref{q33re}). This settles the fate of all possible relative negations $\neg$ for rebits. However, these results do {\it not} fully determine the odd and even correlation structure for qubits: we still have to clarify whether there is an overall negation $\neg$ relative to $Q_{33}$ in (\ref{qbit4}, \ref{qbit5}) and more generally in theorem \ref{thm_qubit} for other permutations of non-intersecting edges. This difference between rebits and qubits results again from the fact that $Q_{33}$ is defined as the correlation of the individuals $Q_3,Q'_3$ for qubits, while it is a correlation of correlations for rebits. This has a remarkable consequence: rebit theory is its own `logical mirror image', while qubit theory's `logical mirror image' is distinct from qubit theory. This topic will be deferred to section \ref{sec_corr} because we firstly need to understand the question structure for the $N=3$ case in order to discuss the odd and even correlation structure of two qubits further.

 \vspace*{.5cm}
\subsection{Three gbits}
\vspace*{.2cm}

 It will be both useful and instructive to explicitly consider the $N=3$ case for rebits and qubits. As a composite system, we can view three gbits, labeled by $A,B,C$, either as three individual systems, as three combinations of one individual and a bipartite composite system or as a tripartite system:
 \begin{widetext}
  \begin{eqnarray}
 \psfrag{a}{$A$}
 \psfrag{b}{$B$}
 \psfrag{c}{$C$}
{\includegraphics[scale=.2]{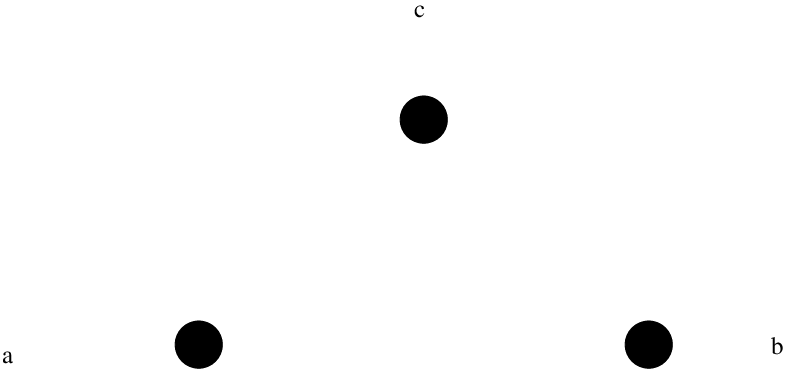}}\q\q\q\q\q{\includegraphics[scale=.2]{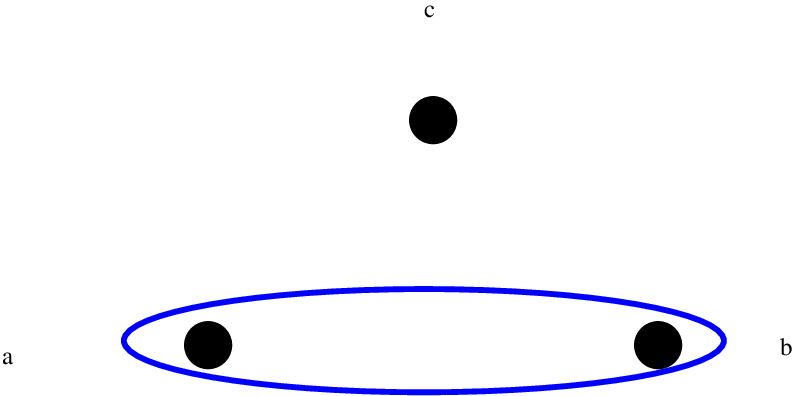}}\q\q\q\q\q{\includegraphics[scale=.2]{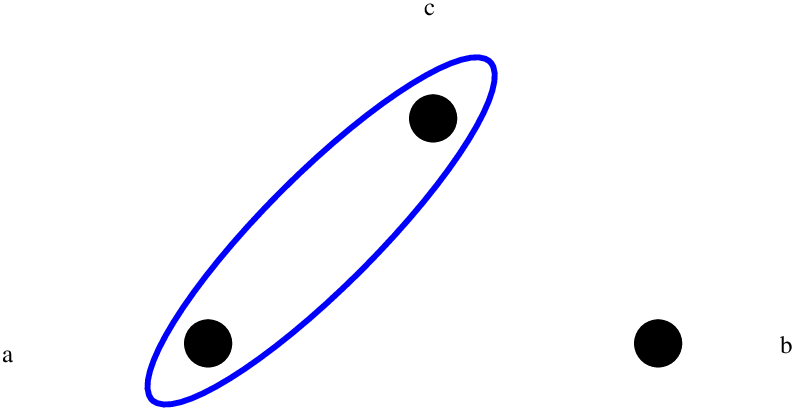}}\nn
\end{eqnarray}
\begin{eqnarray}
 \psfrag{a}{$A$}
 \psfrag{b}{$B$}
 \psfrag{c}{$C$}
 {\includegraphics[scale=.2]{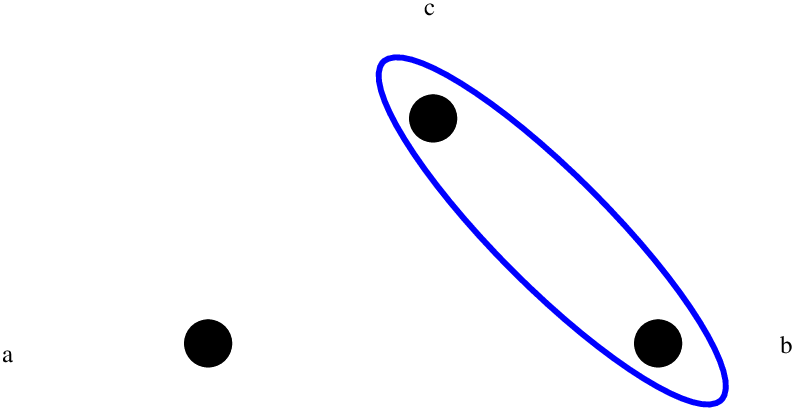}}\q\q\q\q\q{\includegraphics[scale=.2]{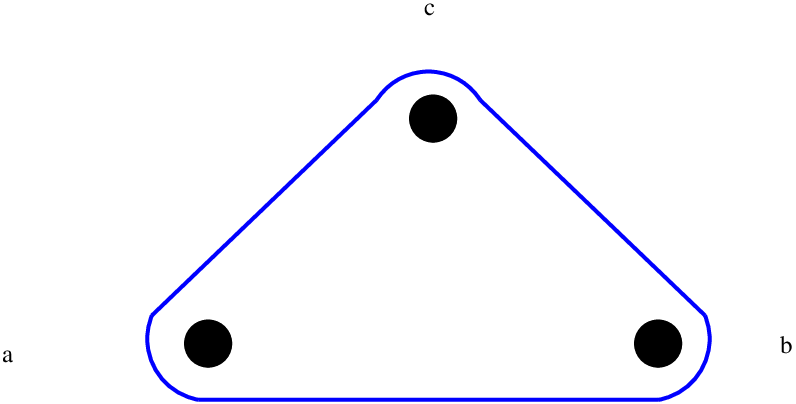}}\label{3gbits}
\end{eqnarray}  
\end{widetext}
According to definition \ref{def_comp} of a composite system, $\cq_3$ must then contain the individual questions of all three gbits, any bipartite correlation questions and any permissible logical connectives thereof. This results in different structures for rebits and qubits.

 \subsubsection{Three qubits}
 
We shall begin with the qubit case $D_1=3$. Clearly, according to definition \ref{def_comp}, all $3\times3=9$ individual questions, henceforth denoted as $Q_{i_A},Q_{j_B},Q_{k_C}$, and all $3\times9=27$ bipartite correlation questions, from now on written as $Q_{i_Aj_B},Q_{i_A,k_C},Q_{j_Bk_C}$, $i,j,k=1,2,3$, will be part of an informationally complete set $\cq_{M_3}$. As before, we represent individuals and correlations graphically as vertices and edges, respectively, e.g.
 \begin{eqnarray}
 \psfrag{a}{$A$}
 \psfrag{b}{$B$}
 \psfrag{c}{$C$}
 \psfrag{a1}{$Q_{1_A}$}
 \psfrag{a2}{$Q_{2_A}$}
 \psfrag{a3}{$Q_{3_A}$}
  \psfrag{b1}{$Q_{1_B}$}
 \psfrag{b2}{$Q_{2_B}$}
 \psfrag{b3}{$Q_{3_B}$} 
  \psfrag{c1}{$Q_{1_C}$}
 \psfrag{c2}{$Q_{2_C}$}
 \psfrag{c3}{$Q_{3_C}$} 
\psfrag{11}{$Q_{1_A1_C}$}
\psfrag{133}{$Q_{1_A3_B}$}
\psfrag{13}{$Q_{1_B3_C}$}
\psfrag{22}{$Q_{2_A2_B}$}
\psfrag{222}{$Q_{2_B2_C}$}
\psfrag{111}{$Q_{111}$}
\psfrag{333}{$Q_{333}$}
\psfrag{322}{$Q_{322}$}
{\includegraphics[scale=.2]{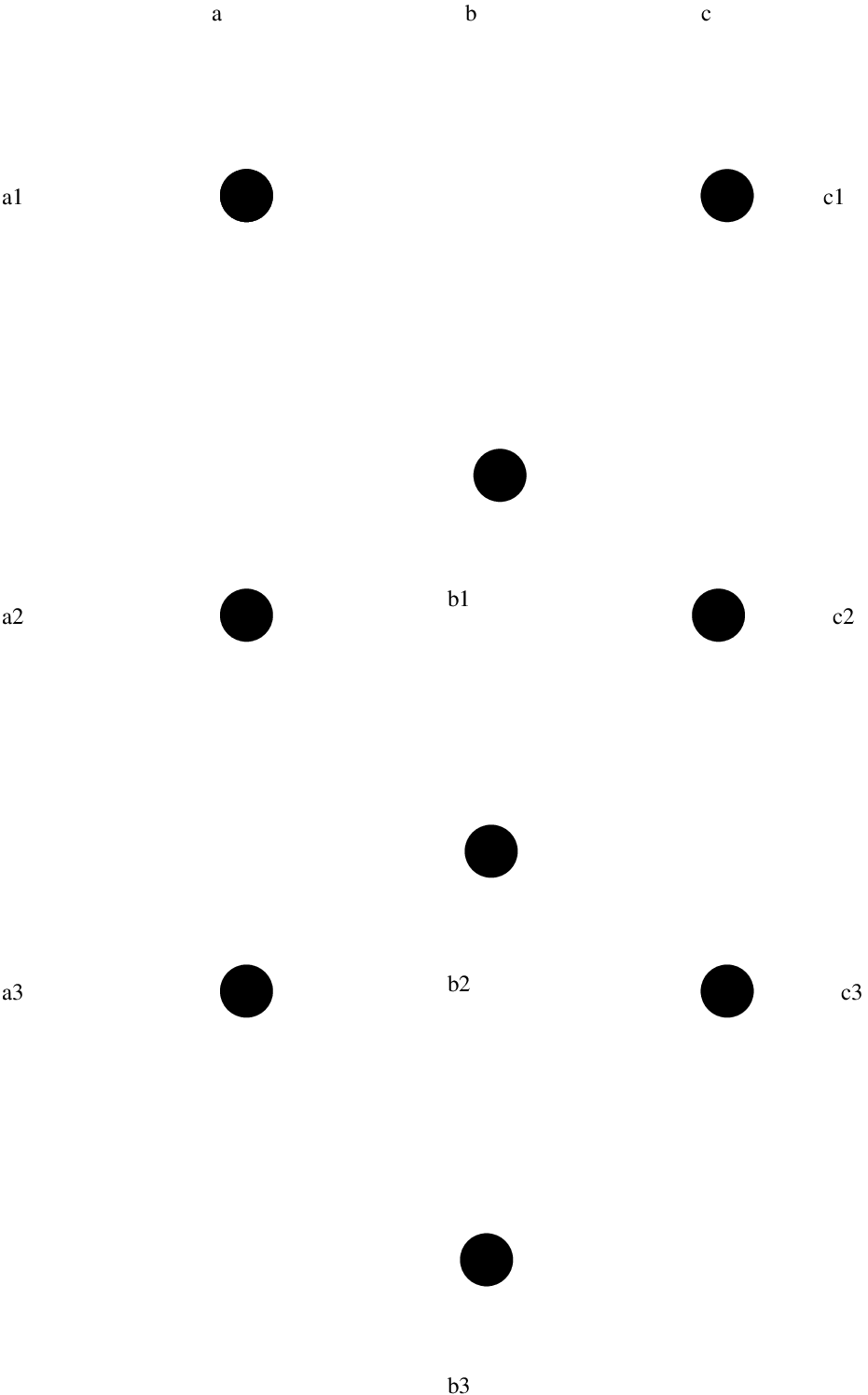}}\nn\\
\nn\\
\nn\\
\nn\\
\psfrag{a}{$A$}
 \psfrag{b}{$B$}
 \psfrag{c}{$C$}
 \psfrag{a1}{$Q_{1_A}$}
 \psfrag{a2}{$Q_{2_A}$}
 \psfrag{a3}{$Q_{3_A}$}
  \psfrag{b1}{$Q_{1_B}$}
 \psfrag{b2}{$Q_{2_B}$}
 \psfrag{b3}{$Q_{3_B}$} 
  \psfrag{c1}{$Q_{1_C}$}
 \psfrag{c2}{$Q_{2_C}$}
 \psfrag{c3}{$Q_{3_C}$} 
\psfrag{11}{$Q_{1_A1_C}$}
\psfrag{133}{$Q_{1_A3_B}$}
\psfrag{13}{$Q_{1_B3_C}$}
\psfrag{22}{$Q_{2_A2_B}$}
\psfrag{222}{$Q_{2_B2_C}$}
\psfrag{111}{$Q_{111}$}
\psfrag{333}{$Q_{333}$}
\psfrag{322}{$Q_{322}$}
{\includegraphics[scale=.2]{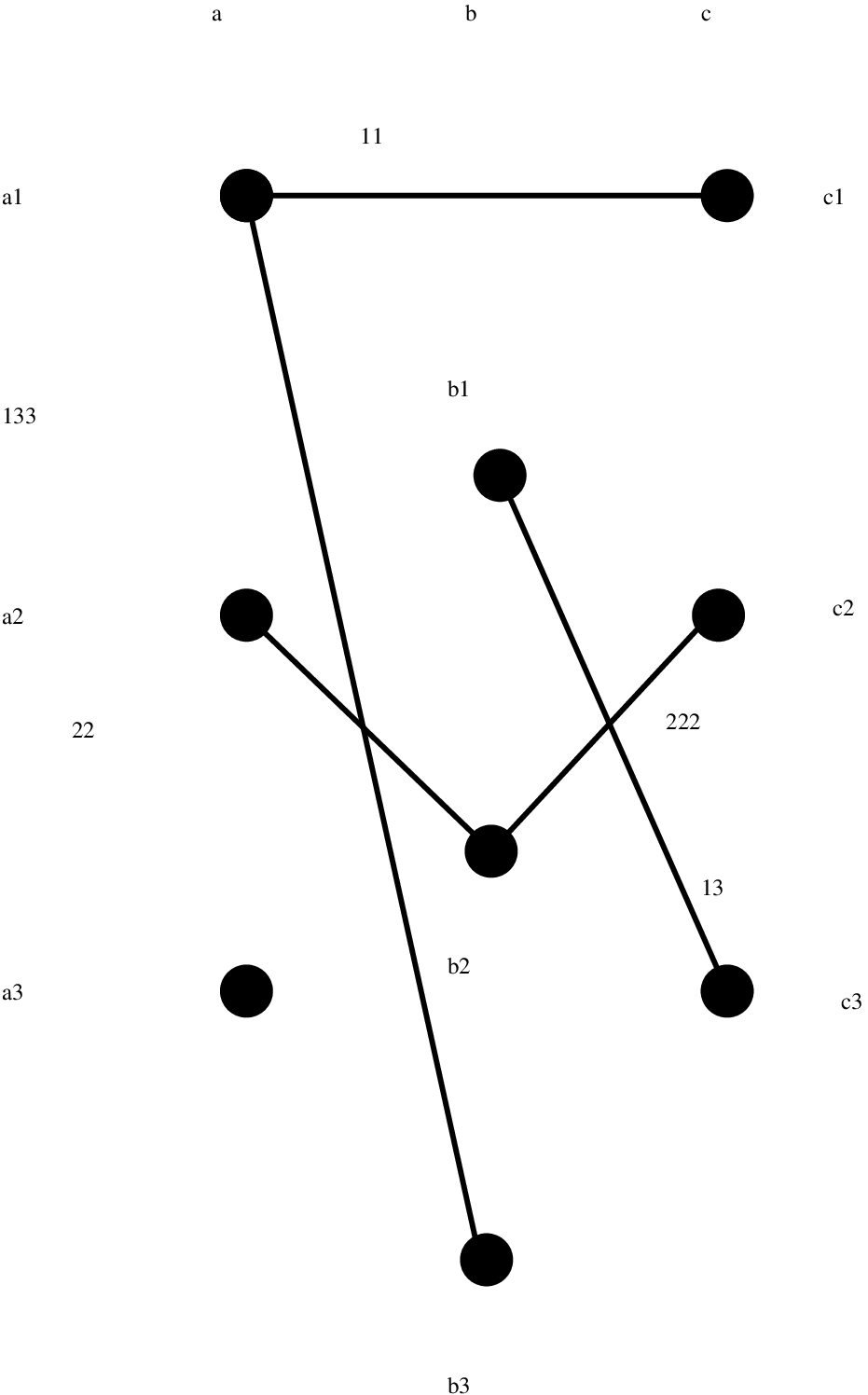}}\nn 
\end{eqnarray} 
depict valid question graphs. 

But in order to complete the individuals and bipartite correlations to $\cq_{M_3}$ we now have to consider logical connectives of these questions which are pairwise independent. This will necessarily involve tripartite questions because the bipartite structure is already exhausted with individuals and bipartite correlations. Clearly, we cannot add a question $\tilde{Q}_{111}$, representing the proposition ``the answers to $Q_{1_A},Q_{1_B},Q_{1_C}$ are the same" to the individuals and bipartite correlations because, e.g., $\tilde{Q}_{111}=1$ would always imply $Q_{1_A1_B}=Q_{1_A1_C}=Q_{1_B1_C}=1$. Also, $Q_{1_A1_B}=0$ would imply $\tilde{Q}_{111}=0$ such that the bipartite correlations and $\tilde{Q}_{111}$ would not be pairwise independent. 

From (\ref{truth}) we already know that the logical connective yielding new pairwise independent questions must either be the XNOR or the XOR. For consistency with the bipartite structure, we continue to employ the XNOR. There is an obvious candidate for an independent tripartite question, namely
\ba
Q_{ijk}=Q_{i_A}\leftrightarrow Q_{j_B}\leftrightarrow Q_{k_C}\label{tri1}
\ea
which thanks to the associativity and symmetry of $\leftrightarrow$ can also equivalently be written as
\ba
Q_{ijk}&=&Q_{i_A}\leftrightarrow Q_{j_Bk_C}=Q_{i_Aj_B}\leftrightarrow Q_{k_C}\nn\\&=&Q_{i_Ak_C}\leftrightarrow Q_{j_B}.\label{tripartite}
\ea
(Since the notation for this tripartite question is unambiguous from the ordering of $i,j,k$ we drop the subscripts $A,B,C$ in $Q_{ijk}$.) This structure is also natural from the different compositions in (\ref{3gbits}). $Q_{ijk}$ thus defined is by construction maximally compatible with $Q_{i_A},Q_{j_B},Q_{k_C},Q_{i_Aj_B},Q_{i_Ak_C}$, $Q_{j_Bk_C}$ and, for similar reasons to the independence of $Q_{i_Aj_B}$ from $Q_{i_A},Q_{j_B}$ (see the discussion below (\ref{correlation2})), also pairwise independence of the latter. Note that this question does not stand in one-to-one correspondence with the proposition ``the answers to $Q_{i_A},Q_{j_B},Q_{k_C}$ are the same"; e.g., $Q_{i_A}=1$ and $Q_{j_B}=Q_{k_C}=0$ also gives $Q_{ijk}=1$. It is easier to interpret this question via (\ref{tripartite}) as either of the three questions ``are the answers to $Q_{i_A},Q_{j_Bk_C}$/$Q_{i_Aj_B},Q_{k_C}$/$Q_{i_Ak_C},Q_{j_B}$ the same?".

There are $3\times3\times3=27$ such tripartite questions $Q_{ijk}$, $i,j,k=1,2,3$. We shall represent them graphically as triangles. For example, $Q_{111},Q_{322},Q_{333}$ are depicted as follows:
  \begin{eqnarray}
 \psfrag{a}{$A$}
 \psfrag{b}{$B$}
 \psfrag{c}{$C$}
 \psfrag{a1}{$Q_{1_A}$}
 \psfrag{a2}{$Q_{2_A}$}
 \psfrag{a3}{$Q_{3_A}$}
  \psfrag{b1}{$Q_{1_B}$}
 \psfrag{b2}{$Q_{2_B}$}
 \psfrag{b3}{$Q_{3_B}$} 
  \psfrag{c1}{$Q_{1_C}$}
 \psfrag{c2}{$Q_{2_C}$}
 \psfrag{c3}{$Q_{3_C}$} 
\psfrag{11}{$Q_{1_A1_C}$}
\psfrag{133}{$Q_{1_A3_B}$}
\psfrag{13}{$Q_{1_B3_C}$}
\psfrag{22}{$Q_{2_A2_B}$}
\psfrag{222}{$Q_{2_B2_C}$}
\psfrag{111}{$Q_{111}$}
\psfrag{333}{$Q_{333}$}
\psfrag{322}{$Q_{322}$}
{\includegraphics[scale=.2]{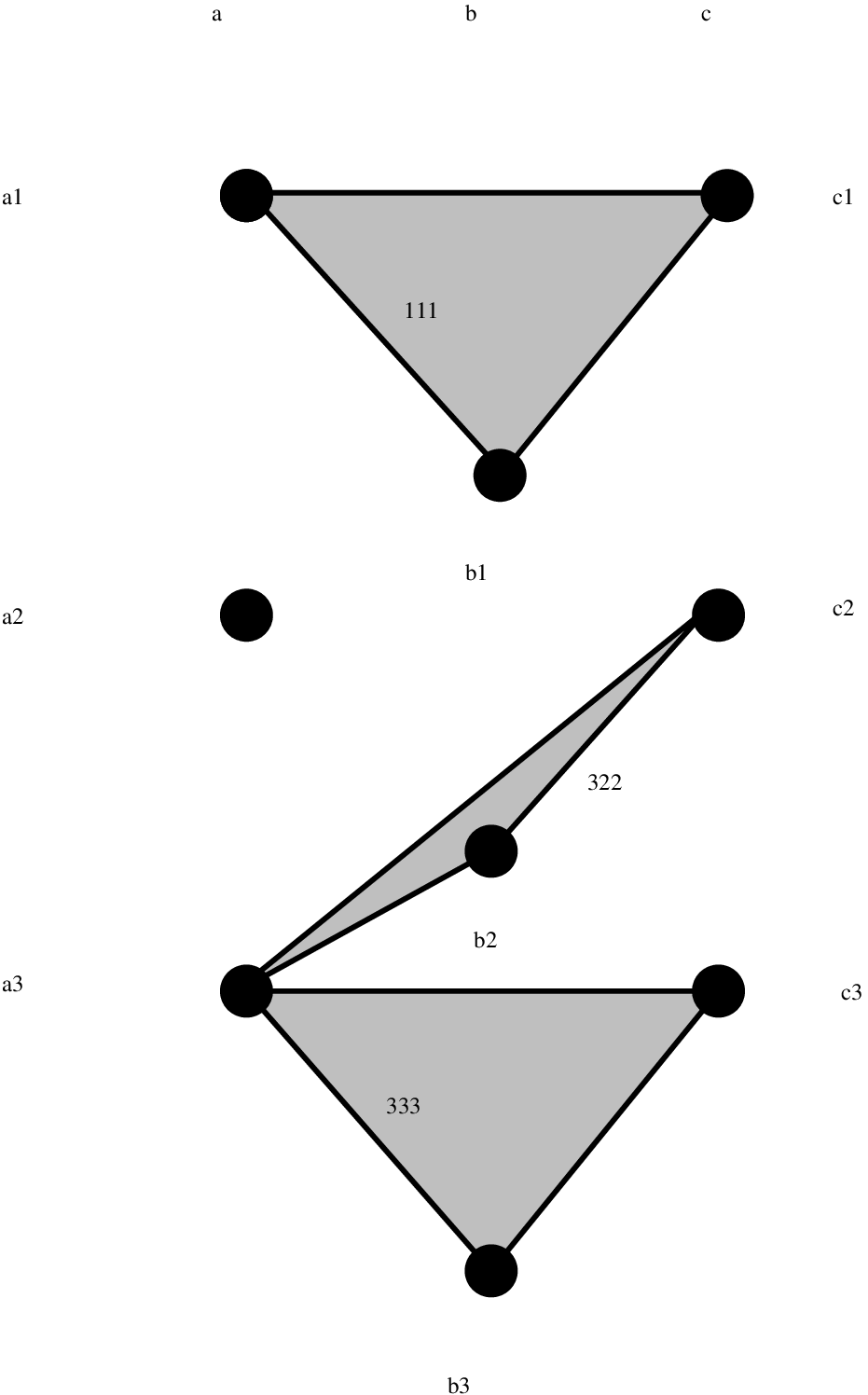}}\nn 
\end{eqnarray}  
 
 \subsubsection{Independence and compatibility for three qubits}

We have to delve into some technical details on the independence and compatibility structure to explain and understand monogamy and the entanglement structure of three qubits. These results will also be used to prove the analogous results by induction in section \ref{sec_nqbits} for $N$ qubits. 

Individual questions from qubit $A$ are maximally compatible with the individual questions from qubits $B$ and $C$, etc. But what about the compatibility of bipartite and tripartite correlations? The compatibility structure of bipartite correlations of a fixed qubit pair is clear from lemma \ref{lem3}, but we have to investigate compatibility of bipartite correlation questions involving all three qubits.
 
 \begin{lem}\label{lem5}
 $Q_{i_Aj_B}$ and $Q_{l_Bk_C}$ are maximally complementary if $j\neq l$. On the other hand, $Q_{i_Aj_B}$ and $Q_{j_Bk_C}$ are maximally compatible and it holds
 \ba
 Q_{i_Aj_B}\leftrightarrow Q_{j_Bk_C}=Q_{i_Ak_C}.\nn
 \ea
 The analogous statements hold for any permutation of $A,B,C$. That is, graphically, two bipartite correlations involving three qubits are maximally compatible if the corresponding edges intersect in a vertex and maximally complementary otherwise.
 \end{lem}
 
 \begin{proof}
Maximal complementarity of $Q_{i_Aj_B}$ and $Q_{l_Bk_C}$ for $j\neq l$ is proven by noting that both are maximally compatible with and independent of $Q_{i_A}$, the relation $Q_{j_B}=Q_{i_A}\leftrightarrow Q_{i_Aj_B}$ and lemma \ref{lem0}. 

$Q_{i_Aj_B}$ and $Q_{j_Bk_C}$ are evidently maximally compatible since $Q_{i_A},Q_{j_B},Q_{k_C}$ are maximally compatible. Moreover, 
\ba
 Q_{i_Aj_B}\leftrightarrow Q_{j_Bk_C}&=&Q_{i_A}\leftrightarrow\underset{=1}{\underbrace{(Q_{j_B}\leftrightarrow Q_{j_B})}}\leftrightarrow Q_{k_C}\nn\\
 &=&Q_{i_Ak_C}\nn
 \ea
 thanks to the associativity of $\leftrightarrow$.
 \end{proof}
 
 For example, $Q_{2_A2_B}$ and $Q_{2_B2_C}$ intersect in $Q_{2_B}$ and are thus maximally compatible, while $Q_{2_A2_B}$ and $Q_{1_B1_C}$ do not share a vertex and are therefore maximally complementary:
 \vspace*{.2cm} \begin{eqnarray}
 \psfrag{a}{$A$}
 \psfrag{b}{$B$}
 \psfrag{c}{$C$}
 \psfrag{a1}{$Q_{1_A}$}
 \psfrag{a2}{$Q_{2_A}$}
 \psfrag{a3}{$Q_{3_A}$}
  \psfrag{b1}{$Q_{1_B}$}
 \psfrag{b2}{$Q_{2_B}$}
 \psfrag{b3}{$Q_{3_B}$} 
  \psfrag{c1}{$Q_{1_C}$}
 \psfrag{c2}{$Q_{2_C}$}
 \psfrag{c3}{$Q_{3_C}$} 
\psfrag{11}{$Q_{1_B1_C}$}
\psfrag{133}{$Q_{1_A3_B}$}
\psfrag{13}{$Q_{1_B3_C}$}
\psfrag{22}{$Q_{2_A2_B}$}
\psfrag{222}{$Q_{2_B2_C}$}
\psfrag{111}{$Q_{111}$}
\psfrag{333}{$Q_{333}$}
\psfrag{322}{$Q_{322}$}
\!\!\!\!\!\!\!\!\!\!\!\!{\includegraphics[scale=.2]{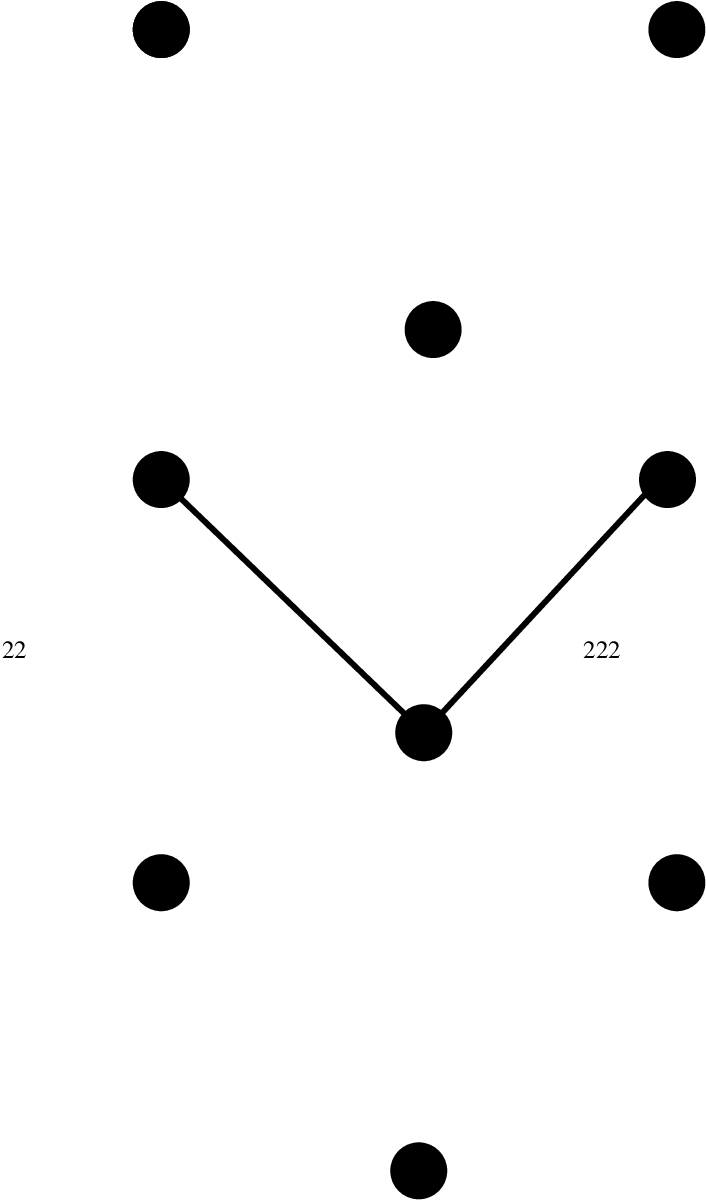}}\q\q\q\q\q\q\q
{\includegraphics[scale=.2]{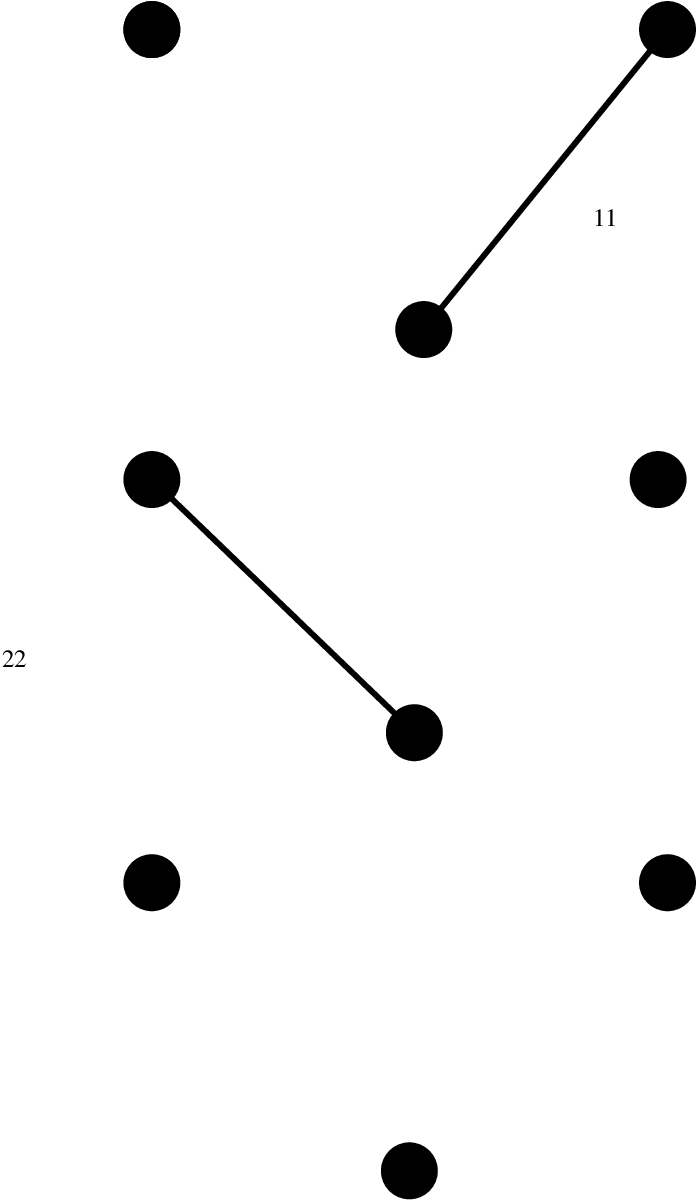}}\nn 
\end{eqnarray} 
 \vspace*{.2cm}
  
 We continue with the tripartite questions.

 \begin{lem}\label{lem_tri1}
$Q_{ijk}$ is maximally compatible with $Q_{i_A},Q_{j_B},Q_{k_C}$ and maximally complementary to $Q_{l_A\neq i_A}$, $Q_{m_B\neq j_B},Q_{n_C\neq k_C}$.   That is, graphically, $Q_{ijk}$ is maximally compatible with an individual $Q_{i_{A,B,C}}$ if the corresponding vertex is one of the vertices of the triangle representing $Q_{ijk}$ and maximally complementary otherwise.
 \end{lem}

\begin{proof}
$Q_{ijk}$ is by construction maximally compatible with $Q_{i_A},Q_{j_B},Q_{k_C}$. On the other hand, complementarity of $Q_{ijk}$ and $Q_{l_A\neq i_A}$ is shown by noting that both are maximally compatible with and independent of $Q_{j_Bk_C}$ and lemma \ref{lem0}. 
One argues analogously for the individuals of qubits $B$ and $C$.
\end{proof}

For instance, $Q_{111}$ is maximally compatible with $Q_{1_C}$ and maximally complementary to $Q_{2_C}$:
 \vspace*{.2cm}  \begin{eqnarray}
 \psfrag{a}{$A$}
 \psfrag{b}{$B$}
 \psfrag{c}{$C$}
 \psfrag{a1}{$Q_{1_A}$}
 \psfrag{a2}{$Q_{2_A}$}
 \psfrag{a3}{$Q_{3_A}$}
  \psfrag{b1}{$Q_{1_B}$}
 \psfrag{b2}{$Q_{2_B}$}
 \psfrag{b3}{$Q_{3_B}$} 
  \psfrag{c1}{$Q_{1_C}$}
 \psfrag{c2}{$Q_{2_C}$}
 \psfrag{c3}{$Q_{3_C}$} 
\psfrag{11}{$Q_{1_B1_C}$}
\psfrag{133}{$Q_{1_A3_B}$}
\psfrag{13}{$Q_{1_B3_C}$}
\psfrag{22}{$Q_{2_A2_B}$}
\psfrag{222}{$Q_{2_B2_C}$}
\psfrag{111}{$Q_{111}$}
\psfrag{333}{$Q_{333}$}
\psfrag{322}{$Q_{322}$}
{\includegraphics[scale=.2]{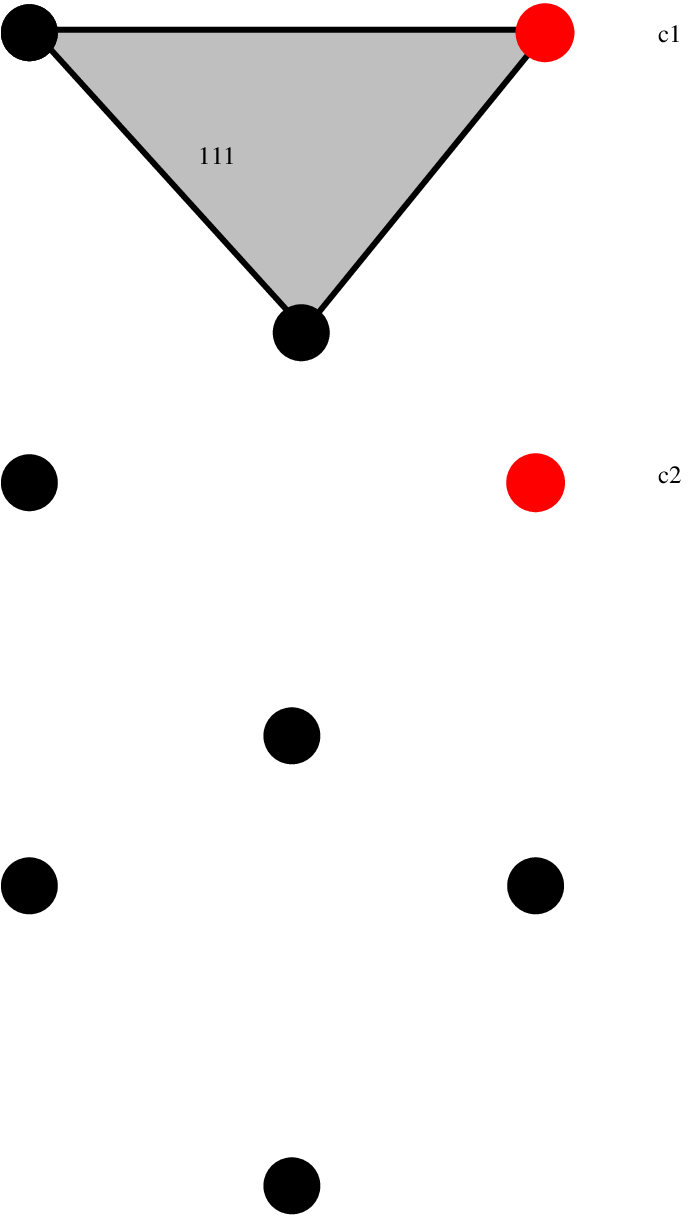}}\nn
\end{eqnarray} 
 \vspace*{.2cm}
 
 \noindent This lemma also directly implies that any individual and any tripartite correlation question are pairwise independent because (1) maximally complementary questions are in particular independent, and (2) for maximally compatible question pairs such as $Q_{i_A},Q_{ijk}$ independence arguments analogous to those surrounding (\ref{correlation2}) apply.

Next we consider bipartite and tripartite correlation questions.
 
\begin{lem}\label{lem_tri2}
$Q_{ijk}$ is maximally compatible with $Q_{i_Aj_B},Q_{i_Ak_C},Q_{j_Bk_C}$ and, furthermore, with $Q_{m_Bn_C}$, $Q_{l_An_C}$ and $Q_{l_Am_B}$ for $l\neq i$, $m\neq j$ and $k\neq n$. On the other hand, $Q_{ijk}$ is maximally complementary to $Q_{i_Am_B},Q_{i_An_C},Q_{l_Aj_B}$, $Q_{j_Bn_C},Q_{l_Ak_C},Q_{m_Bk_C}$ for $l\neq i$, $m\neq j$ and $k\neq n$. That is, graphically, $Q_{ijk}$ is maximally compatible with a bipartite correlation if the edge of the latter is either an edge of the triangle corresponding to $Q_{ijk}$ or if the edge and triangle do not intersect. $Q_{ijk}$ is maximally complementary to a bipartite correlation question if the edge of the latter and the triangle corresponding to $Q_{ijk}$ share one common vertex.
\end{lem}

\begin{proof}
$Q_{ijk}$ is by construction maximally compatible with $Q_{i_Aj_B},Q_{i_Ak_C},Q_{j_Bk_C}$. $Q_{ijk}=Q_{i_A}\leftrightarrow Q_{j_Bk_C}$ and $Q_{m_Bn_C}$ are also maximally compatible for $j\neq m$ and $k\neq n$ because $Q_{m_Bn_C}$ is maximally compatible with $Q_{i_A}$ and thanks to lemma \ref{lem3} also with $Q_{j_Bk_C}$. Complementarity of $Q_{ijk}$ and $Q_{i_Am_B}$ for $j\neq m$ follows from noting that both are maximally compatible with and independent of $Q_{k_C}$ and lemma \ref{lem0}. 
The reasoning for all other cases is analogous.
\end{proof}

To give a graphical example, $Q_{111}$ is maximally compatible with $Q_{1_B1_C}$ and $Q_{3_A3_C}$ and maximally complementary to $Q_{1_A2_B}$:
 \vspace*{.2cm}\begin{eqnarray}
 \psfrag{a}{$A$}
 \psfrag{b}{$B$}
 \psfrag{c}{$C$}
 \psfrag{a1}{$Q_{1_A}$}
 \psfrag{a2}{$Q_{2_A}$}
 \psfrag{a3}{$Q_{3_A}$}
  \psfrag{b1}{$Q_{1_B}$}
 \psfrag{b2}{$Q_{2_B}$}
 \psfrag{b3}{$Q_{3_B}$} 
  \psfrag{c1}{$Q_{1_C}$}
 \psfrag{c2}{$Q_{2_C}$}
 \psfrag{c3}{$Q_{3_C}$} 
\psfrag{11}{$Q_{1_B1_C}$}
\psfrag{133}{$Q_{1_A3_B}$}
\psfrag{12}{$Q_{1_A2_B}$}
\psfrag{33}{$Q_{3_A3_C}$}
\psfrag{222}{$Q_{2_B2_C}$}
\psfrag{111}{$Q_{111}$}
\psfrag{333}{$Q_{333}$}
\psfrag{322}{$Q_{322}$}
\!\!\!\!{\includegraphics[scale=.2]{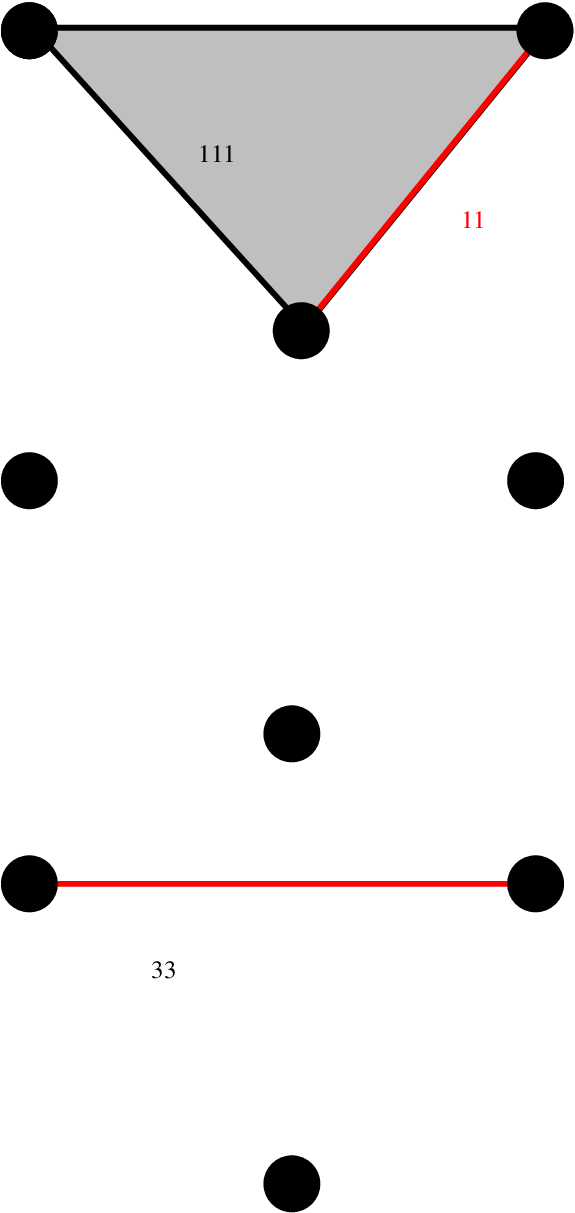}}\q\q\q\q\q\q\q\q\q\q
{\includegraphics[scale=.2]{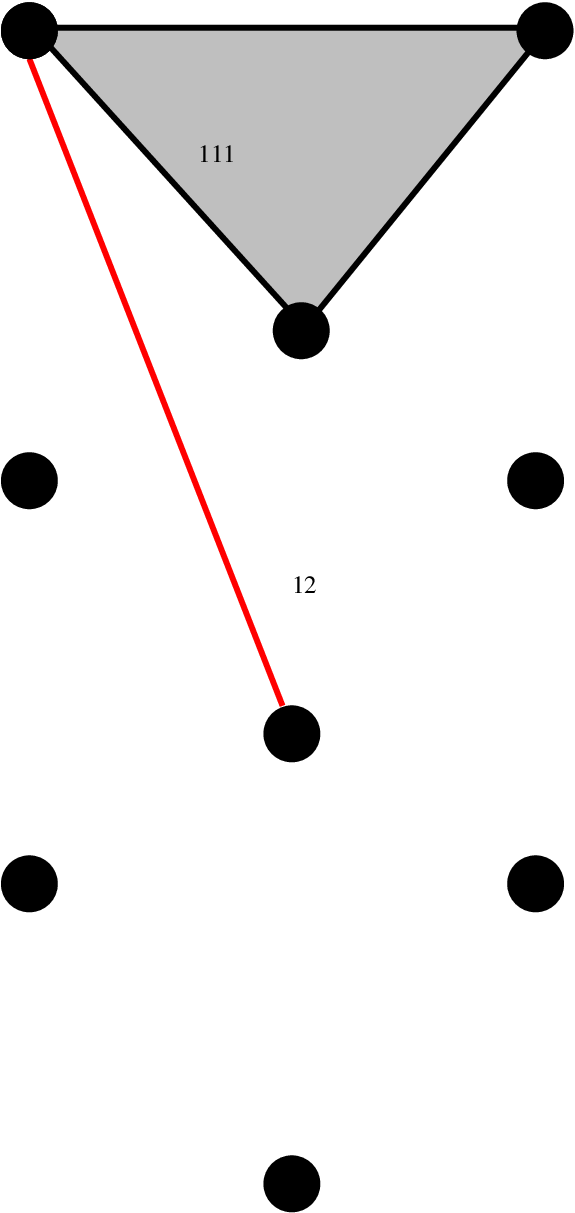}}\nn 
\end{eqnarray} 
 \vspace*{.2cm}

\noindent We still have to check pairwise independence of the bipartite and tripartite correlation questions.
\begin{lem}\label{lem8}
Any bipartite $Q_{m_Bn_C},Q_{l_An_C},Q_{l_Am_B}$ and any tripartite correlation question $Q_{ijk}$ are independent from one another.
\end{lem}
\begin{proof}
Lemma \ref{lem_tri2} implies that we only have to check pairwise independence of $Q_{m_Bn_C},Q_{l_An_C},Q_{l_Am_B}$ from $Q_{ijk}$ for $l\neq i$, $m\neq j$ and $k\neq n$ because maximally complementary questions are by definition independent and $Q_{ijk}$ and $Q_{i_Aj_B},Q_{i_Ak_C},Q_{j_Bk_C}$ are pairwise independent. Consider therefore $Q_{ijk}$ and $Q_{m_Bn_C}$ for $j\neq m$ and $k\neq n$. By lemma \ref{lem_tri1}, $Q_{k_C}$ is maximally compatible with $Q_{ijk}$ and by lemma \ref{lem1} maximally complementary to $Q_{m_Bn_C}$. This implies, using the arguments from the proof of lemma \ref{lem2}, independence of $Q_{ijk},Q_{m_Bn_C}$. The other cases follow similarly.
\end{proof}

\begin{lem}\label{lem9}
The tripartite correlation questions $Q_{ijk}$, $i,j,k=1,2,3$ are pairwise independent.
\end{lem}

\begin{proof}
Consider $Q_{ijk}$ and $Q_{lmn}$ for $i\neq l$. By lemma \ref{lem_tri1}, $Q_{i_A}$ is maximally compatible with $Q_{ijk}$ and maximally complementary to $Q_{lmn}$. Using the analogous arguments from the proof of lemma \ref{lem2}, this implies that $Q_{ijk},Q_{lmn}$ are independent. The same reasoning holds when $j\neq m$ and $k\neq n$.
\end{proof}
 
This has an immediate consequence:

\begin{Corollary}
The individuals $Q_{i_A},Q_{j_B},Q_{k_C}$, the bipartite $Q_{i_Aj_B},Q_{i_Ak_C},Q_{j_Bk_C}$ and the tripartite $Q_{ijk}$, $i,j,k=1,2,3$ are pairwise independent and thus, thanks to assumption \ref{assump6}, contained in an informationally complete set $\cq_{M_3}$.
\end{Corollary} 

Lastly, we consider the complementarity and compatibility structure of the tripartite correlations.

\begin{lem}\label{lem_tri3}
$Q_{ijk}$ and $Q_{lmn}$ are maximally compatible if $\{i,j,k\}$ and $\{l,m,n\}$ overlap in one or three indices and maximally complementary if $\{i,j,k\}$ and $\{l,m,n\}$ overlap in zero or two indices. That is, graphically, $Q_{ijk}$ and $Q_{lmn}$ are maximally compatible if their corresponding triangles intersect in one vertex (or coincide) and maximally complementary if the triangles share an edge or do not intersect.
\end{lem}

\begin{proof}
Compatibility for an overlap in all three indices is trivial. But also $Q_{ijk}=Q_{i_A}\leftrightarrow Q_{j_Bk_C}$ and $Q_{imn}=Q_{i_A}\leftrightarrow Q_{m_Bn_C}$ are clearly maximally compatible for $j\neq m$ and $k\neq n$ because by lemma \ref{lem3} $Q_{j_Bk_C},Q_{m_Bn_C}$ are maximally compatible in this case. Compatibility for the other cases of an overlap of $Q_{ijk}$ and $Q_{lmn}$ in one index follows by permutation. 

The proof of the complementarity of $Q_{ijk}$ and $Q_{lmn}$ for $i\neq l$, $j\neq m$ and $k\neq n$ follows from lemma \ref{lem0}. One may use the fact that, by lemma \ref{lem_tri2}, both questions are maximally compatible with and independent of $Q_{j_Bk_C}$ and that, by lemma \ref{lem_tri1}, $Q_{i_A\neq l_A}$ is maximally complementary to $Q_{lmn}$.  

Similarly, one proves complementarity of $Q_{ijk}=Q_{i_A}\leftrightarrow Q_{j_Bk_C}$ and $Q_{ljk}=Q_{l_A}\leftrightarrow Q_{j_Bk_C}$  for $i\neq l$ by using that both are maximally compatible with and independent of $Q_{j_Bk_C}$. 
Complementarity of tripartite correlations for other overlaps in precisely two indices follows by permutation.
\end{proof}
 
For example, $Q_{111}$ and $Q_{212}$ intersect in the vertex $Q_{1_B}$ and are thus maximally compatible. By contrast, $Q_{111}$ shares the edge $Q_{1_A1_B}$ with $Q_{113}$ and does not intersect at all with $Q_{333}$ such that $Q_{111}$ is maximally complementary to both. $Q_{113}$ and $Q_{333}$ intersect in the vertex $Q_{3_C}$ and are therefore maximally compatible:
  \vspace*{.2cm} \begin{eqnarray}
\psfrag{111}{$Q_{111}$}
\psfrag{333}{$Q_{333}$}
\psfrag{113}{$Q_{113}$}
\psfrag{212}{$Q_{212}$}
{\includegraphics[scale=.2]{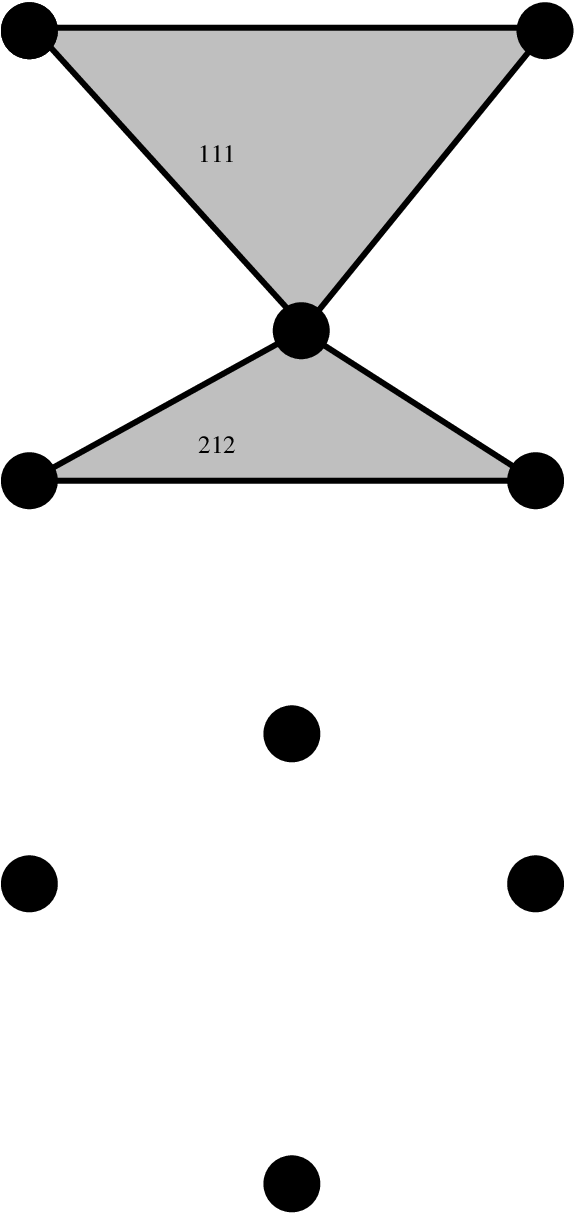}}\q\q\q\q\q\q\q\q\q\q
{\includegraphics[scale=.2]{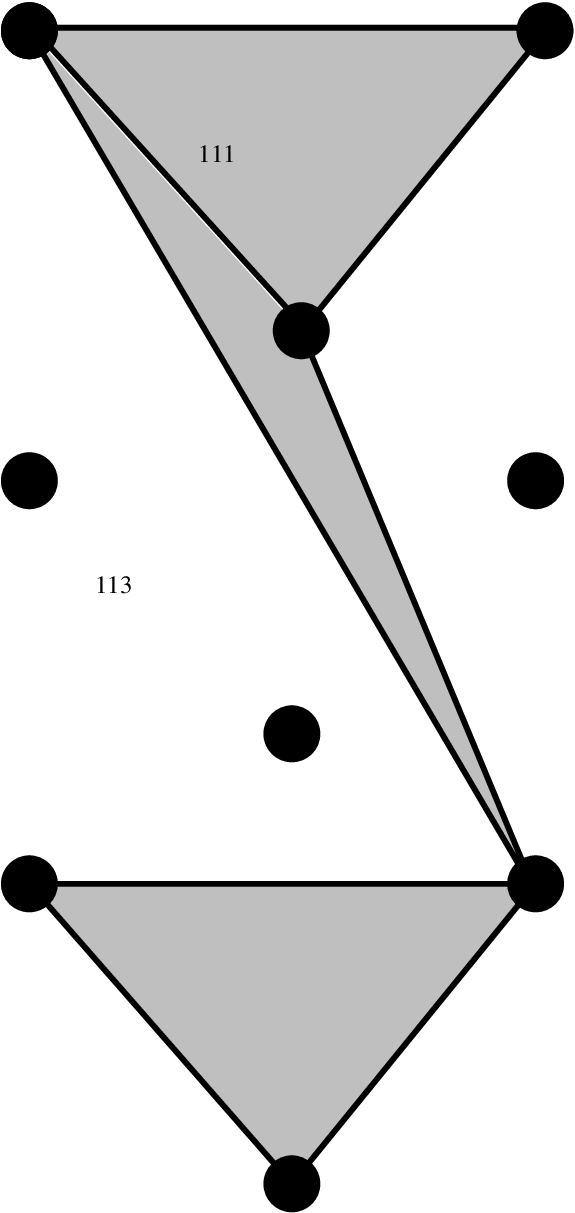}}\nn 
\end{eqnarray}

\subsubsection{An informationally complete set for three qubits}

We now show that the $9$ individuals, the $27$ bipartite and the $27$ tripartite correlation questions form an informationally complete set $\cq_{M_3}$ for $N=3$ qubits.

\begin{Theorem}\label{thm_qubit2}{\bf(Qubits)}
The individuals $Q_{i_A},Q_{j_B},Q_{k_C}$, the bipartite $Q_{i_Aj_B},Q_{i_Ak_C},Q_{j_Bk_C}$ and the tripartite $Q_{ijk}$, $i,j,k=1,2,3$ are logically closed under $\leftrightarrow$ such that they form an informationally complete set $\cq_{M_3}$ with $D_3=63$ for $D_1=3$.
\end{Theorem}

\begin{proof}
For individuals and bipartite correlations of any pair of qubits the logical closure under the XNOR was already shown in section \ref{sec_n2}. Correlation of individuals with a bipartite correlation question from another pair of qubits yields the tripartite correlation questions. Lemma \ref{lem5} shows that also the bipartite correlation questions involving distinct pairs of qubits are logically closed under $\leftrightarrow$. We thus only have to check logical closure of XNOR combinations involving tripartite correlation questions.

Lemma \ref{lem_tri1} shows that tripartite correlations and individuals are only maximally compatible if the individual is a vertex of the tripartite triangle. But then a combination such as 
\ba
Q_{ijk}\leftrightarrow Q_{i_A}=(Q_{i_A}\leftrightarrow Q_{i_A})\leftrightarrow Q_{j_Bk_C}=Q_{j_Bk_C}\nn
\ea
produces another bipartite correlation, thanks to the associativity of $\leftrightarrow$. Similarly, lemma \ref{lem_tri2} asserts that tripartite and bipartite correlations are only maximally compatible if the edge of the bipartite correlation is either contained in the tripartite triangle or if the edge and triangle do not intersect. However, combinations of such maximally compatible pairs also do not yield any new questions because, e.g., 
\ba
Q_{ijk}\leftrightarrow Q_{i_Aj_B}=Q_{k_C}\nn
\ea
and, using again the associativity of XNOR and theorem \ref{thm_qubit},
\begin{widetext}
\ba
Q_{\sigma_A(1)\sigma_B(1)k}\leftrightarrow Q_{\sigma_A(2)\sigma_B(2)}=\underset{=\,Q_{\sigma_A(3)\sigma_B(3)}\q\text{or}\q\neg Q_{\sigma_A(3)\sigma_B(3)}}{\underbrace{(Q_{\sigma_A(1)\sigma_B(1)}\leftrightarrow Q_{\sigma_A(2)\sigma_B(2)})}}\leftrightarrow Q_{k_C}=\,Q_{\sigma_A(3)\sigma_B(3)k}\q\text{or}\q\neg Q_{\sigma_A(3)\sigma_B(3)k},\nn
\ea
\end{widetext}
where $\sigma_A,\sigma_B$ are permutations of $\{1,2,3\}$. Finally, lemma \ref{lem_tri3} entails that tripartite correlations are only maximally compatible if they intersect in one vertex. But, using theorem \ref{thm_qubit} once more,
\begin{widetext}
\ba
Q_{\sigma_A(1)\sigma_B(1)k}\leftrightarrow Q_{\sigma_A(2)\sigma_B(2)k}=Q_{\sigma_A(1)\sigma_B(1)}\leftrightarrow Q_{\sigma_A(2)\sigma_B(2)}=\,Q_{\sigma_A(3)\sigma_B(3)}\q\text{or}\q\neg Q_{\sigma_A(3)\sigma_B(3)}.\nn
\ea
\end{widetext}
Permuting these examples implies that all individuals, bipartite and tripartite correlations are logically closed under $\leftrightarrow$ which by (\ref{truth}) is the only independent logical connective possibly yielding new independent questions. These questions therefore form an informationally complete set $\cq_{M_3}$ with $D_3=63$.
\end{proof}

\subsubsection{Entanglement of three qubits and monogamy}\label{sec_qbitmono}

Let us now put all these detailed results to good use and reward ourselves with the observation that they naturally explain monogamy of entanglement for the qubit case. This is best illustrated with an example. Let $O$ have asked the questions $Q_{1_A1_B}$ and $Q_{2_A2_B}$ to qubits $A,B$ such that the two are in a state of maximal information of $N=2$ independent \texttt{bits} and maximal entanglement relative to $O$. 

 \vspace*{.2cm} \begin{eqnarray}
\psfrag{11}{$Q_{1_A1_B}$}
\psfrag{22}{$Q_{2_A2_B}$}
{\includegraphics[scale=.2]{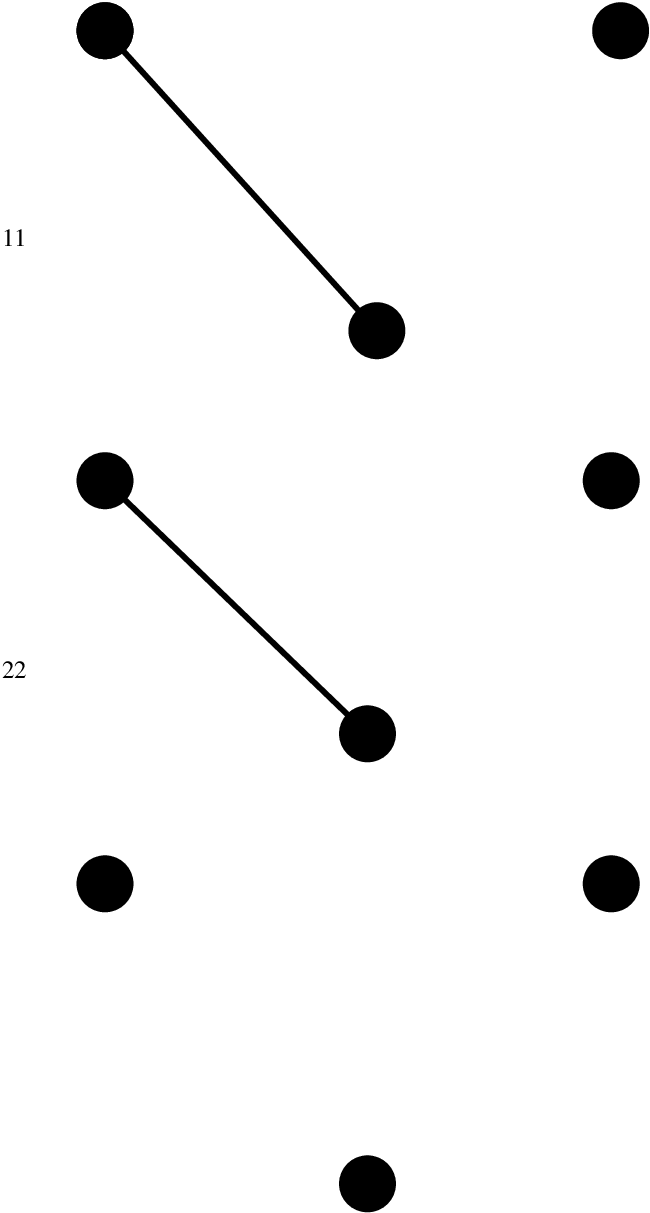}}\nn 
\end{eqnarray}  
 \vspace*{.2cm}

\noindent Intuitively, this already suggests monogamy of entanglement because $O$ has spent the maximally attainable amount of information over the pair $A,B$ such that even if now a third qubit $C$ enters the game he should no longer be able to ask {\it any} question to the triple which gives him any further simultaneous independent information about the pair $A,B$. But any correlation of $A$ or $B$ with $C$ would constitute additional information about $A$ or $B$ which would violate the information bound for the subsystem of the composite system of three qubits. 

We shall now make this more precise. The question is, which bipartite or tripartite correlations involving qubit $C$ are maximally compatible with $Q_{1_A1_B},Q_{2_A2_B}$. Firstly, lemma \ref{lem5} states that bipartite correlations of two distinct qubit pairs are only maximally compatible if their corresponding edges intersect in a vertex. Clearly, there can then not exist {\it any} edge which is maximally compatible with {\it both} $Q_{1_A1_B},Q_{2_A2_B}$. Secondly, lemma \ref{lem_tri2} asserts that a bipartite correlation is only maximally compatible with a tripartite correlation if the corresponding edge is either contained in the tripartite triangle or does not intersect the triangle. This means that the only tripartite correlations maximally compatible with $Q_{1_A1_B},Q_{2_A2_B}$ are (a) the correlations of the latter with the individuals $Q_{1_C},Q_{2_C},Q_{3_C}$ of qubit $C$ and (b) $Q_{33k}$, $k=1,2,3$. 

For example, in the first of the following two question graphs $Q_{2_B2_C}$ is maximally compatible with $Q_{2_A2_B}$ but maximally complementary to $Q_{1_A1_B}$, while $Q_{3_B3_C}$ is maximally complementary to both. But $O$ could ask $Q_{222}$ together with $Q_{1_A1_B},Q_{2_A2_B}$, as depicted in the second graph:

 \vspace*{.2cm}
  \begin{eqnarray}
\psfrag{11}{$Q_{1_A1_B}$}
\psfrag{22}{$Q_{2_A2_B}$}
\psfrag{222}{$Q_{222}$}
\psfrag{2}{$Q_{2_C}$}
\psfrag{33}{$Q_{3_B3_C}$}
\psfrag{223}{$Q_{2_B2_C}$}
\!\!\!\!\!\!\!\!{\includegraphics[scale=.2]{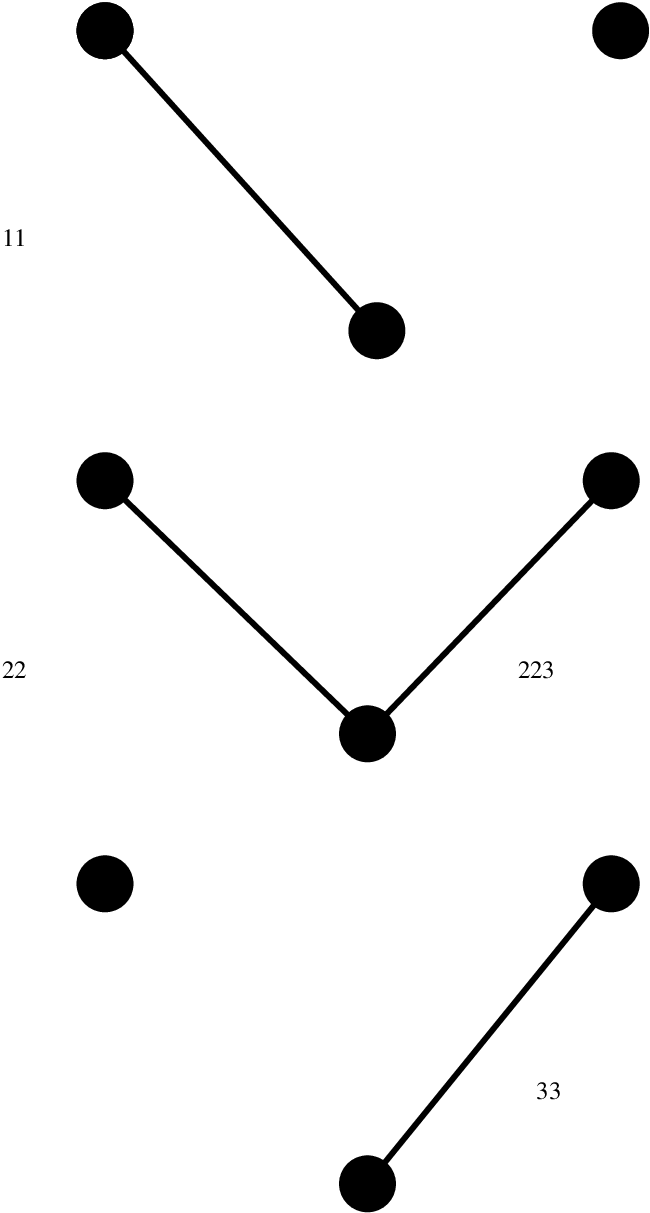}}\q\q\q\q\q\q\q\,\,
{\includegraphics[scale=.2]{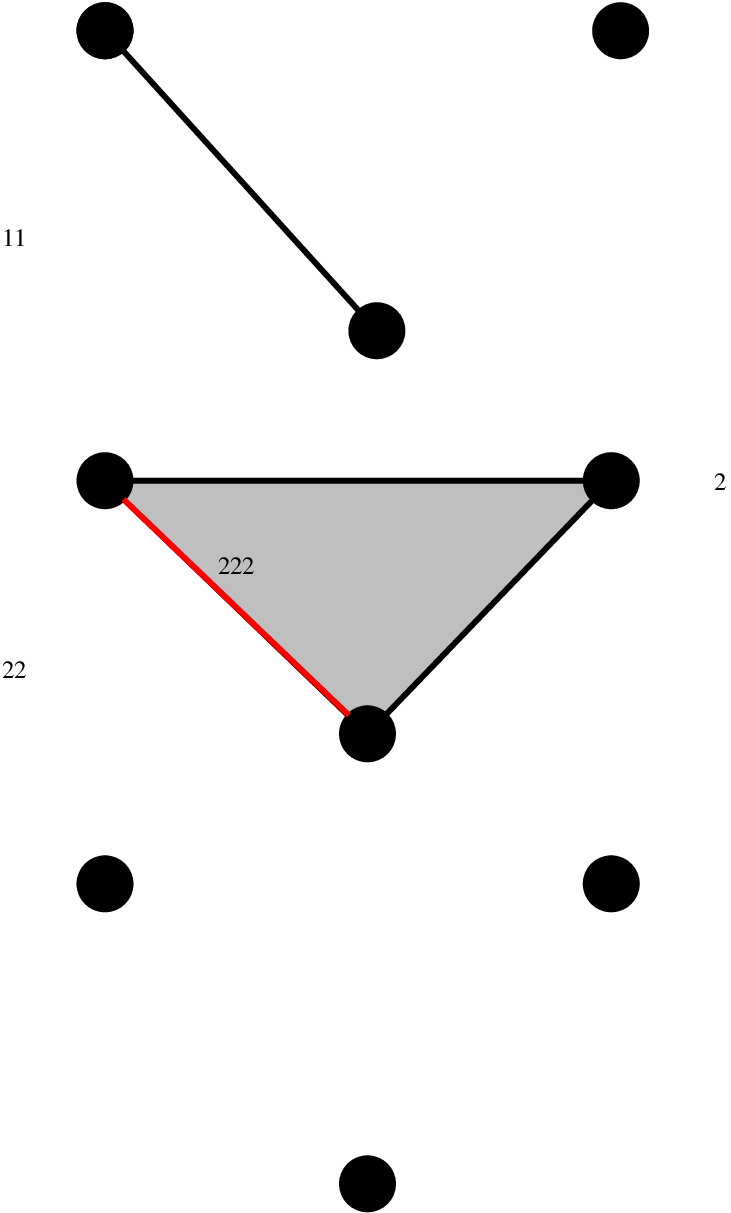}}\nn 
\end{eqnarray}  
 \vspace*{.2cm}

\noindent However, asking $Q_{1_A1_B},Q_{2_A2_B}$ and $Q_{222}$ simultaneously is equivalent to asking $Q_{1_A1_B},Q_{2_A2_B}$ and $Q_{2_C}$ because $Q_{2_A2_B}\leftrightarrow Q_{222}=Q_{2_C}$. The analogous conclusion holds for any other tripartite question maximally compatible with $Q_{1_A1_B},Q_{2_A2_B}$. Therefore, once $O$ has asked the last to bipartite correlations about qubits $A,B$, he can {\it only} acquire individual information corresponding to $Q_{1_C},Q_{2_C},Q_{3_C}$ about qubit $C$. Clearly, the same state of affairs is true for any permutation of the three qubits. This is {\it monogamy of entanglement} in its most extreme form: two qubits which are maximally entangled can {\it not} be correlated whatsoever with any other system.

The non-extremal form of monogamy, namely the case that a qubit pair is not maximally entangled and can thus share a bit of entanglement with a third qubit, can also be explained. To this end, we recall from quantum information theory that monogamy is generally described with so-called {\it monogamy inequalities}, e.g., the Coffman-Kundu-Wootters inequality \cite{Coffman:2000fk}
\ba
\tau_{A|BC}\geq \tau_{AB}+\tau_{AC},\label{CKW}
\ea
where $0\leq\tau_{A|BC}=2(1-\Tr\rho_A^2)\leq1$ measures the entanglement between qubit $A$ and the qubit pair $B,C$ and $\rho_A$ is the marginal state of qubit $A$ obtained after tracing qubit $B$ and $C$ out of the (not necessarily pure) tripartite state. $\tau_{AB},\tau_{AC}$, on the other hand, measure the bipartite entanglement of the pairs $A,B$ and $A,C$ and are similarly obtained by tracing either $C$ or $B$ out of the tripartite qubit state. The inequality (\ref{CKW}) formalizes the intuition that the correlation between $A$ and the pair $B,C$ is at least as strong as the correlation of $A$ with $B$ and $C$ individually. This is the general form of monogamy. One usually also defines the so-called {\it three-tangle} \cite{Coffman:2000fk} as the difference between left and right hand side
\ba
\tau_{ABC}=\tau_{A|BC}-\tau_{AB}-\tau_{AC}\in[0,1]\nn
\ea
to measure the genuine tripartite entanglement shared among all three qubits. This three-tangle turns out to be permutation invariant $\tau_{ABC}=\tau_{BCA}=\tau_{CAB}$.

In our current case, a measure of entanglement sharing must be informational. At this stage, in line with our previous definition of entanglement, we simply {\it define} the analogous entanglement measures to be the sum of $O$'s information about the various independent bipartite or tripartite correlation questions
\ba
&&\tilde{\tau}_{A|BC}:=\sum_{i,j,k=1}^3\alpha_{ijk}+\sum_{i,j=1}^3\alpha_{i_Aj_B}+\sum_{i,k=1}^3\alpha_{i_Ak_C}\nn\q\\
&&\tilde{\tau}_{AB}:=\sum_{i,j=1}^3\alpha_{i_Aj_B},\q\q\,\,
\tilde{\tau}_{AC}:=\sum_{i,k=1}^3\alpha_{i_Ak_C}.\label{monomeasure}
\ea
With this definition, an informational inequality analogous to the inequality (\ref{CKW}) is trivial
\ba
\tilde{\tau}_{A|BC}\geq\tilde{\tau}_{AB}+\tilde{\tau}_{AC}\nn
\ea
as is the fact that the informational three-tangle
\ba
\tilde{\tau}_{ABC}:=\tilde{\tau}_{A|BC}-\tilde{\tau}_{AB}-\tilde{\tau}_{AC}=\sum_{i,j,k=1}^3\alpha_{ijk}\nn
\ea
is permutation invariant and measures only the information contained in the tripartite questions and thus genuine tripartite entanglement. This is how one can describe the general form of monogamy of entanglement in our language. Of course, at this stage the information measure $\alpha_i$ about the various questions $Q_{i}$ is only implicit, but we shall derive its explicit form in section \ref{sec_infomeasure}, upon which the above statements become truly quantitative. These informational versions of monogamy inequalities and tangles naturally suggest themselves for simple generalizations to arbitrarily many qubits, thereby complementing current efforts in the quantum information literature (e.g., see \cite{Regula:2014uq}).

Before we move on, we briefly emphasize that monogamy is a consequence of complementarity -- as is entanglement. For example, three classical bits also satisfy the limited information rule \ref{lim}, but due to the absence of complementarity, $O$ could ask all correlations $Q_{AB},Q_{AC},Q_{BC}$ and $Q_{ABC}$ at once

 \vspace*{.2cm} \begin{eqnarray}
\psfrag{a}{$A$}
\psfrag{b}{$B$}
\psfrag{c}{$C$}
{\includegraphics[scale=.2]{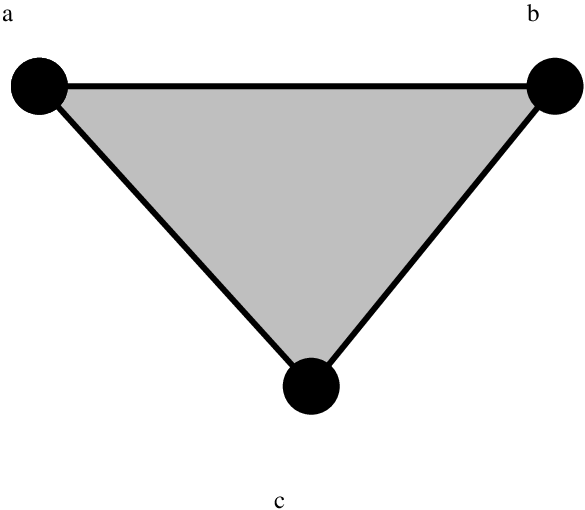}}\nn 
\end{eqnarray}  
which, however, is equivalent to asking the three individuals $Q_A,Q_B,Q_C$.

\subsubsection{Maximal entanglement for three qubits}\label{sec_ghz}

Let us elucidate maximal tripartite entanglement for three qubits, corresponding to $O$ asking three mutually maximally compatible and independent tripartite correlation questions such that he exhausts the information bound of $N=3$ independent \texttt{bits} with tripartite information. Lemma \ref{lem_tri3} asserts that tripartite questions are maximally compatible if and only if they intersect in exactly one vertex. For example, $O$ could ask $Q_{211},Q_{121},Q_{112}$ simultaneously. These are mutually independent because by (\ref{qbit4}, \ref{qbit5}, \ref{qbit7})
\ba
Q_{211}\leftrightarrow Q_{121}&=&Q_{3_A3_B}\q\text{or}\q\neg Q_{3_A3_B}\nn\\
Q_{211}\leftrightarrow Q_{112}&=&Q_{3_A3_C}\q\text{or}\q\neg Q_{3_A3_C}\nn\\
 Q_{121}\leftrightarrow Q_{112}&=&Q_{3_B3_C}\q\text{or}\q\neg Q_{3_B3_C},\label{ghz1}
\ea
i.e., their binary connectives with an XNOR do not imply each other and the bipartite correlations are pairwise independent from the tripartite ones. Accordingly, asking $Q_{211},Q_{121},Q_{112}$ will provide $O$ with three independent \texttt{bits} of information about the qubit triple. We note that lemma \ref{lem_tri1} implies that once $O$ has posed the questions $Q_{211},Q_{121},Q_{112}$, every individual question $Q_{i_A},Q_{j_B},Q_{k_C}$ is maximally complementary to at least one of the three tripartite correlations. That is, $O$ cannot acquire {\it any} individual information about the three qubits and will only have composite information. This is a necessary condition for maximal entanglement.

 In addition to the three implied binary correlations (\ref{ghz1}), $O$ would clearly also know the tripartite correlation of all three $Q_{211},Q_{121},Q_{112}$, which amounts to
\ba
Q_{211}\leftrightarrow Q_{121}\leftrightarrow Q_{112}=Q_{222}\q\text{or}\q\neg Q_{222}.\nn
\ea
This exhausts the list of questions about which $O$ would have information by asking $Q_{211},Q_{121},Q_{112}$. The list can be represented by the following question graph: 

 \vspace*{.2cm}
  \begin{eqnarray}
{\includegraphics[scale=.2]{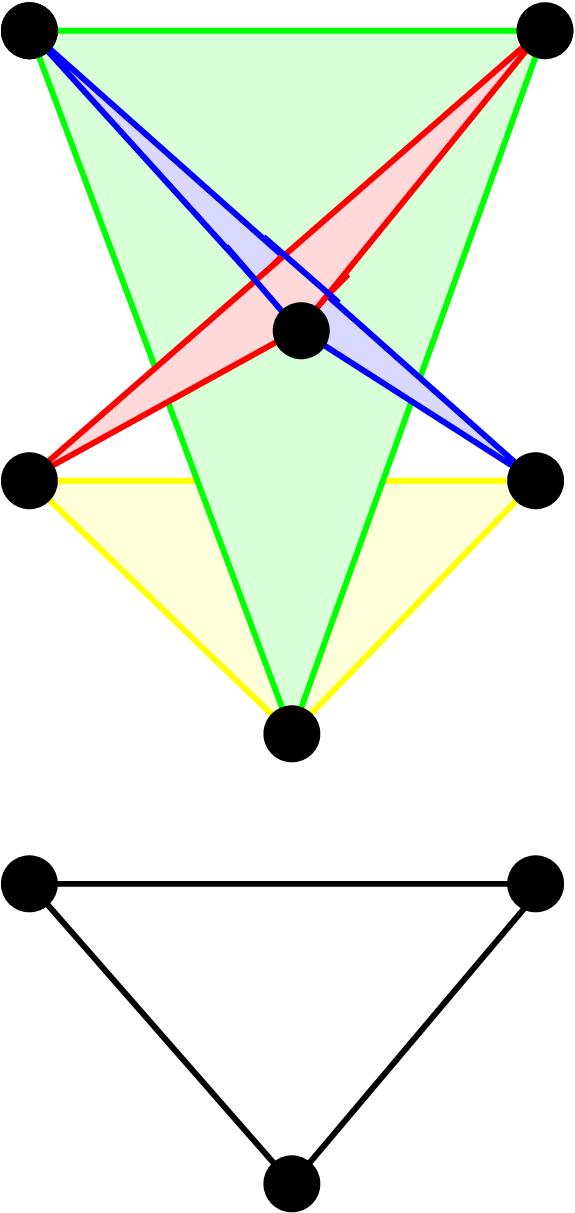}}\label{ghz2}
\end{eqnarray}  
 \vspace*{.2cm}

\noindent This graph will ultimately correspond to the eight possible Greenberger-Horne-Zeilinger (GHZ) states in either of the three question bases $\{Q_{2_A},Q_{1_B},Q_{1_C}\}$, $\{Q_{1_A},Q_{2_B},Q_{1_C}\}$ and $\{Q_{1_A},Q_{1_A},Q_{2_C}\}$, corresponding to the eight answer configurations `yes-yes-yes', `yes-yes-no', `yes-no-no', `yes-no-yes', `no-yes-yes', `no-yes-no', `no-no-yes' and `no-no-no' to the three tripartite correlations $Q_{211},Q_{121},Q_{112}$. Similarly, any other three maximally compatible tripartite correlations will correspond to GHZ states in other question bases.\footnote{As an aside, we note that these propositions solve the little riddle in Zeilinger's festschrift for D.\ Greenberger \cite{zeilinger1999foundational}.}

The correlation information in the graph (\ref{ghz2}) is clearly democratically distributed over the three qubits $A,B,C$. (Intuitively, one could even view the \texttt{bit} stemming from, say, the answer to $Q_{2_A1_B1_C}$ as being carried to one-third by each of $A,B,C$.) Furthermore, if $O$ now wanted to `marginalize' over qubit $C$, i.e.\ discard any information involving qubit $C$, all that would be left of his knowledge about the qubit triple would be the answer to $Q_{3_A3_B}$. But, as argued in section \ref{sec_rebtang}, this resulting `marginal state' of qubits $A,B$ relative to $O$ cannot be considered entangled. The analogous results hold for `marginalization'  over either $A$ or $B$. As a result, the question graph (\ref{ghz2}) corresponds to genuine maximal tripartite entanglement.

\subsubsection{Three rebits}

We shall now repeat the same exercise for three rebits but can benefit from the results of the $N=3$ qubit case. We shall thus be briefer here. The reader exclusively interested in qubits may jump to section \ref{sec_corr}.

According to definition \ref{def_comp}, $O$'s questions to the three rebit system must contain the $6$ individuals $Q_{i_A},Q_{j_B},Q_{k_C}$, $12$ bipartite correlations $Q_{i_Aj_B},Q_{i_Ak_C},Q_{j_Bk_C}$, $i,j,k=1,2$, and $3$ bipartite correlations of correlations $Q_{3_A3_B},Q_{3_A3_C},Q_{3_B3_C}$. To render this set informationally complete, we have to top it up with tripartite questions. There are now $8$ tripartite correlations $Q_{ijk}$, $i,j,k=1,2$, of the kind (\ref{tri1}, \ref{tripartite}) and, moreover, $6$ tripartite correlations of a rebit individual question with the question $Q_{33}$ asking for the correlation of bipartite correlations of the other rebit pair
\ba
Q_{i33}&:=&Q_{i_A}\leftrightarrow Q_{3_B3_C},\nn\\ Q_{3j3}&:=&Q_{j_B}\leftrightarrow Q_{3_A3_C},\nn\\ Q_{33k}&:=&Q_{3_A3_B}\leftrightarrow Q_{k_C},\q\q\q i,j,k=1,2.\nn
\ea
Given that the individuals $Q_{3_A},Q_{3_B},Q_{3_C}$ do not exist for rebits there is now also no tripartite correlation $Q_{333}$.

\subsubsection{Independence and compatibility for three rebits} 
 
The independence, compatibility and complementarity structure for the questions {\it not} involving an index $i,j,k=3$ directly follows from the qubit discussion as lemmas \ref{lem5}--\ref{lem_tri3} also hold in the present case for $i,j,k=1,2$. But we now have to clarify the question structure once two indices are equal to $3$ (an odd number of indices cannot be $3$). The status of any such purely bipartite relations was clarified in section \ref{sec_n2} such that here we have to consider the case that all three rebits are involved.

\begin{lem}\label{lem_re1}
$Q_{3_B3_C}$ is maximally complementary to $Q_{i_Aj_B}$, $i,j=1,2$, and maximally compatible with $Q_{3_A3_B}$. Furthermore,
\ba
Q_{3_A3_B}\leftrightarrow Q_{3_B3_C}= Q_{3_A3_C}.\label{re-close}
\ea
The same holds for any permutations of $A,B,C$.
\end{lem}
\begin{proof}
One proves complementarity of $Q_{i_Aj_B}$ and $Q_{3_B3_C}=Q_{1_B2_C}\leftrightarrow Q_{2_B1_C}$ by employing that both are maximally compatible with and independent of $Q_{i_A}$ and lemma \ref{lem0}. This also requires invoking that $Q_{j_B}$ and $Q_{3_B3_C}$ are maximally complementary by lemma \ref{lem4}. 

Compatibility of $Q_{3_A3_B}$ and $Q_{3_B3_C}$ follows indirectly. Namely, $Q_{1_A2_B},Q_{2_B1_C}$ and $Q_{2_A1_B},Q_{1_B2_C}$ are two maximally compatible pairs by lemma \ref{lem5}. We can then apply the reasoning around (\ref{qbit8}) to find
\begin{widetext}
\ba
Q_{3_A3_B}\leftrightarrow Q_{3_B3_C}&=&\q(Q_{1_A2_B}\leftrightarrow Q_{2_A1_B})\leftrightarrow (Q_{1_B2_C}\leftrightarrow Q_{2_B1_C})\nn\\
&=&\neg(\underset{=\,Q_{1_A1_C}}{\underbrace{(Q_{1_A2_B}\leftrightarrow Q_{2_B1_C})}}\leftrightarrow\underset{=\,Q_{2_A2_C}}{\underbrace{(Q_{1_B2_C}\leftrightarrow Q_{2_A1_B})}})\nn\\
&=&\neg(Q_{1_A1_C}\leftrightarrow Q_{2_A2_C})=Q_{3_A3_C}.\nn
\ea
\end{widetext}
The last equality holds thanks to theorem \ref{thm_rebit}. 

The same reasoning applies to any permutation of $A,B,C$.
\end{proof}

Next, we discuss the tripartite correlations of individuals with the questions $Q_{33}$ asking for the correlation of bipartite correlations of rebit pairs. 

\begin{lem}
$Q_{i33}$ is maximally compatible with $Q_{i_A}$ and maximally complementary to $Q_{j_B},Q_{k_C}$. The same holds for any permutation of $A,B,C$.
\end{lem}
 
 \begin{proof}
 Compatibility of $Q_{i33}$ with $Q_{i_A}$ is true by construction. Complementarity of $Q_{j_B}$ and $Q_{i33}=Q_{i_A}\leftrightarrow (Q_{1_B2_C}\leftrightarrow Q_{2_B1_C})$ follows from the observation that $Q_{i33}$ and $Q_{j_B}$ are both maximally compatible with and independent of $Q_{i_A}$ and a similar reasoning to the previous proof.
 \end{proof}
 
 \begin{lem}\label{lem_re3}
 $Q_{i33}$ is maximally compatible with $Q_{3_B3_C}$, $Q_{j_Bk_C}$, $j,k=1,2$, and $Q_{l_Ak_C},Q_{l_Aj_B}$ for $i\neq l$. On the other hand, $Q_{i33}$ is maximally complementary to $Q_{i_Aj_B},Q_{i_Ak_C}$ and $Q_{3_A3_B},Q_{3_A3_C}$. Furthermore, $Q_{ijk}$ is maximally compatible with $Q_{3_A3_B}$. The same holds for any permutation of $A,B,C$.
 \end{lem}

\begin{proof}
Compatibility of $Q_{i33}$ with $Q_{j_Bk_C}$ and $Q_{3_B3_C}$, as well as compatibility of $Q_{ijk}$ with $Q_{3_A3_B}$ is obvious. Compatibility of $Q_{i33}$ with $Q_{l_Ak_C}$ for $i\neq l$ follows indirectly by noting that $Q_{i_A},Q_{l_A}$ and $Q_{3_B3_C},Q_{k_C}$ are two pairs of maximally complementary questions, but that the questions in one pair are maximally compatible with both questions of the other. In this case the reasoning of (\ref{qbit8}) applies and entails the correlation of $Q_{i_A}$ with $Q_{3_B3_C}$ must be maximally compatible with the correlation of $Q_{l_A}$ with $Q_{k_C}$. Compatibility of $Q_{i33}$ with $Q_{l_Aj_B}$ follows similarly.
 
Complementarity of $Q_{i33}$ and $Q_{i_Aj_B}$ follows from the fact that both are maximally compatible with and independent of $Q_{i_A}$ and using arguments as in previous proofs. Likewise, $Q_{i33}$ and $Q_{3_A3_B}$ follows similarly by noting that both are maximally compatible with $Q_{3_B3_C}$. \end{proof}

\begin{lem}
Any of $Q_{i33},Q_{3j3},Q_{33k}$ is independent from any bipartite correlation question $Q_{i_Aj_B},Q_{i_Ak_C}$, $Q_{j_Bk_C}$ and any bipartite correlation of correlations question$Q_{3_A3_B},Q_{3_A3_C},Q_{3_B3_C}$. $Q_{ijk}$ is also pairwise independent from the latter. Furthermore, the $Q_{ijk}$ and $Q_{i33},Q_{3j3},Q_{33k}$ are pairwise independent.
\end{lem}

\begin{proof}
The proof is completely analogous to the proofs of lemmas \ref{lem2}, \ref{lem8} and \ref{lem9}.
\end{proof}

As before this has an important consequence.
\begin{Corollary}
The individuals $Q_{i_A},Q_{j_B},Q_{k_C}$, the bipartite correlations $Q_{i_Aj_B},Q_{i_Ak_C},Q_{j_Bk_C}$, the bipartite correlations of correlations $Q_{3_A3_B},Q_{3_A3_C},Q_{3_B3_C}$, the tripartite correlations $Q_{ijk}$ and the tripartite $Q_{i33},Q_{3j3},Q_{33k}$, $i,j,k=1,2$, are pairwise independent and thus, thanks to assumption \ref{assump6}, part of an informationally complete set $\cq_{M_3}$.
\end{Corollary}
 
The compatibility and complementarity structure of questions involving the `correlation of correlations' $Q_{33}$ is analogous to lemma \ref{lem_tri3}.
\begin{lem}\label{lem_re5}
$Q_{ijk}$ is maximally compatible with $Q_{i33},Q_{3j3},Q_{33k}$ and maximally complementary to $Q_{l33},Q_{3m3},Q_{33n}$ for $i\neq l$, $j\neq m$ and $k\neq n$. Furthermore, $Q_{i33}$ is maximally compatible with $Q_{3j3},Q_{33k}$, but $Q_{133}$ and $Q_{233}$ are maximally complementary. The analogous result holds for all permutations of $A,B,C$.
\end{lem}

\begin{proof}
Compatibility of $Q_{ijk}$ with $Q_{i33},Q_{3j3},Q_{33k}$ follows from the fact that the constituents of the latter $Q_{i_A},Q_{3_B3_C},\ldots$ are maximally compatible with $Q_{ijk}$. Complementarity of, e.g., $Q_{ijk}$ and $Q_{l33}$ for $i\neq l$ can be shown by noting that both are maximally compatible with and independent of $Q_{3_B3_C}$ and lemma \ref{lem0}. Compatibility of, say, $Q_{i33}$ and $Q_{3j3}$ can be demonstrated by using (\ref{qbit8}) and noting that $Q_{i_A},Q_{3_A3_C}$ and $Q_{j_B},Q_{3_B3_C}$ are two pairs of maximally complementary questions which are such that each question in one pair is maximally compatible with both questions of the other pair. Finally, $Q_{133}$ and $Q_{233}$ are maximally complementary because both are maximally compatible with $Q_{3_B3_C}$ and $Q_{1_A},Q_{2_A}$ are maximally complementary.
\end{proof}
 
 This finishes our considerations of the independence and complementarity structure of three rebits.

\subsubsection{An informationally complete set for three rebits} 
 
The rebit questions considered thus far comprise an informationally complete set of $35$ elements.
 \begin{Theorem}\label{thm_rebit2}{\bf(Rebits)}
The individuals $Q_{i_A},Q_{j_B},Q_{k_C}$, the bipartite correlations $Q_{i_Aj_B},Q_{i_Ak_C},Q_{j_Bk_C}$, the bipartite correlations of correlations $Q_{3_A3_B},Q_{3_A3_C},Q_{3_B3_C}$, the tripartite correlations $Q_{ijk}$ and the tripartite $Q_{i33},Q_{3j3},Q_{33k}$, $i,j,k=1,2$, are logically closed under $\leftrightarrow$ and thus form an informationally complete set $\cq_{M_3}$ with $D_3=35$ for $D_1=2$. 
 \end{Theorem}

\begin{proof}
Logical closure under the XNOR for any pair of rebits follows from section \ref{sec_n2}. Combining individuals with bipartite questions of another rebit pair produces the tripartite questions. Lemma \ref{lem5} and \ref{lem_re1} assert logical closure of bipartite correlation questions and the bipartite questions $Q_{33}$ asking for the correlation of bipartite correlations among all three rebits. We must therefore only check logical closure of combinations involving tripartite questions which involve indices taking the value $3$ (the other cases are covered by theorem \ref{thm_qubit2} for $i,j,k=1,2$).

Using theorem \ref{thm_rebit} and lemmas \ref{lem_re1}--\ref{lem_re5}, the proof is mostly analogous to the proof of theorem \ref{thm_qubit2} such that here we will only show two non-trivial cases which are treated differently. For example, by lemma \ref{lem_re3} $Q_{133}$ and $Q_{2_A2_B}$ are maximally compatible. The conjunction with $\leftrightarrow$ yields
\begin{widetext}
\ba
Q_{133}\leftrightarrow Q_{2_A2_B}&\underset{\tiny(\ref{q33re}, \ref{qbit7})}{=}&(Q_{1_A}\leftrightarrow Q_{3_B3_C})\leftrightarrow \neg (Q_{1_A1_B}\leftrightarrow Q_{3_A3_B})\nn\\
&\underset{\tiny(\ref{qbit8}, \ref{re-close})}{=}&Q_{1_B}\leftrightarrow Q_{3_A3_C}=Q_{313}.\nn
\ea
\end{widetext}
Similarly, by lemma \ref{lem_re5}, $Q_{i33},Q_{3j3}$ are maximally compatible. Their XNOR conjunction gives
\begin{widetext}
\ba
Q_{i33}\leftrightarrow Q_{3j3}=(Q_{i_A}\leftrightarrow Q_{3_B3_C})\leftrightarrow(Q_{3_A3_C}\leftrightarrow Q_{j_B})
\underset{\tiny(\ref{qbit8})}{=}\neg(Q_{3_A3_B}\leftrightarrow Q_{i_Aj_B})\nn
\ea
\end{widetext}
which again coincides with some $Q_{l_Am_B}$ with $i\neq l$ and $j\neq m$ (or the negation thereof). In complete analogy to these explicit examples and to the proof of theorem \ref{thm_qubit2}, one shows the logical closure under the XNOR of all other cases. This gives the desired result.
\end{proof}

\subsubsection{Monogamy and maximal entanglement for three rebits}
 
Rebits are considered as non-monogamous in the literature \cite{wootters2012entanglement,woottersrebit}. However, this conclusion depends somewhat on one's notion of monogamy and can be clarified within our language. Firstly, rebits, just as qubits, are clearly monogamous in the following sense: if two rebits $A,B$ are maximally entangled in a state of {\it maximal} information relative to $O$, then they cannot share any entanglement whatsoever with a third rebit $C$. This can be seen by repeating the argument of section \ref{sec_qbitmono} which holds analogously for three rebits. 

The situation changes slightly for entangled states of non-maximal information involving either of $Q_{3_A3_B},Q_{3_A3_C},Q_{3_B3_C}$ and for tripartite maximally entangled states. As an example for the former case, $O$ could ask $Q_{3_A3_B},Q_{3_B3_C}$ simultaneously which gives him two independent \texttt{bits} of information about the rebit triple and implies $Q_{3_A3_C}$ by (\ref{re-close}) as well such that his information could be summarized as

 \vspace*{.2cm} \begin{eqnarray}
 \psfrag{33}{$Q_{3_A3_B}$}
 \psfrag{331}{$Q_{3_A3_C}$}
 \psfrag{332}{$Q_{3_B3_C}$}
{\includegraphics[scale=.2]{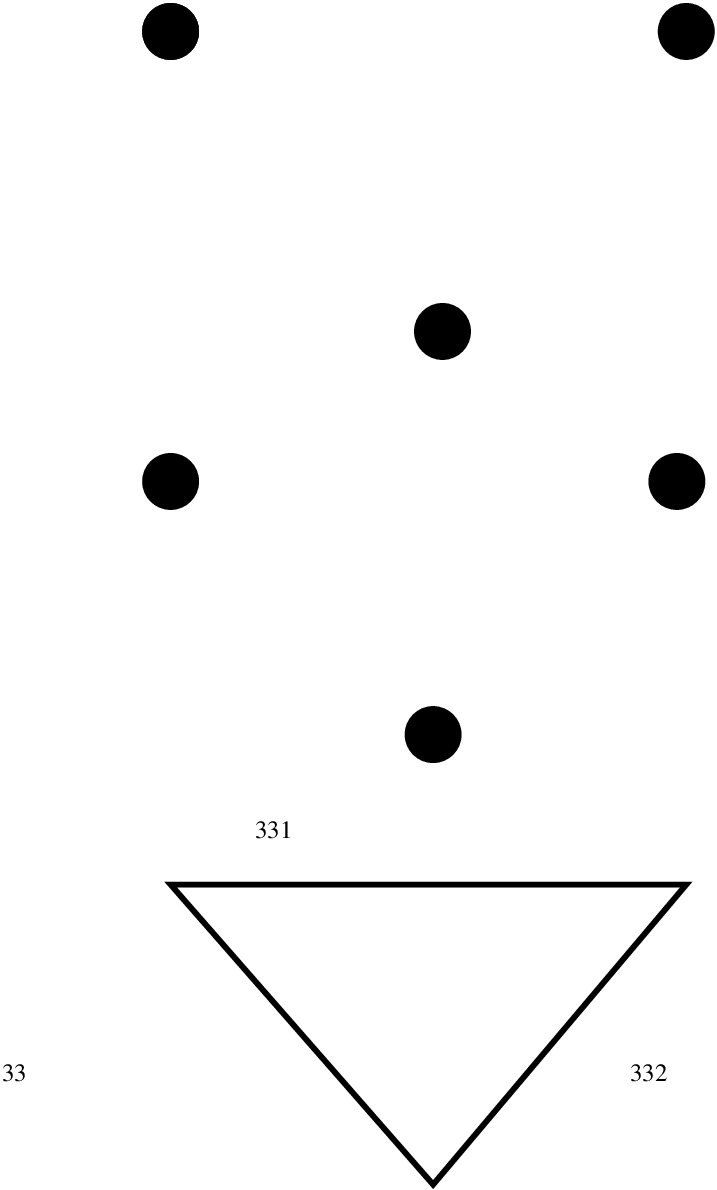}}\nn
\end{eqnarray}
 \vspace*{.1cm}
 
\noindent (we depict the $Q_{3_A3_B},Q_{3_A3_C},Q_{3_B3_C}$ without vertices to emphasize the absence of the individuals $Q_{3_A},Q_{3_B}$, $Q_{3_C}$ for rebits). This graph corresponds to non-monogamously entangled rebit states: by lemma \ref{lem4} all individuals are maximally complementary to two of the three known questions such that $O$ cannot acquire {\it any} individual information about the three rebits at the same time. The three rebits are pairwise maximally entangled, albeit not in a state of maximal information (see also section \ref{sec_rebtang}).

However, the three rebits can be similarly non-monogamous in a tripartite maximally entangled state of maximal information. As in the qubit case (\ref{ghz2}) in section \ref{sec_ghz}, one can write down a question graph corresponding to the rebit analogues of GHZ states, representing the eight possible answers to the tripartite questions $Q_{211},Q_{121},Q_{112}$

 \vspace*{.2cm}
  \begin{eqnarray}
 \psfrag{33}{$Q_{3_A3_B}$}
 \psfrag{331}{$Q_{3_A3_C}$}
 \psfrag{332}{$Q_{3_B3_C}$}
{\includegraphics[scale=.2]{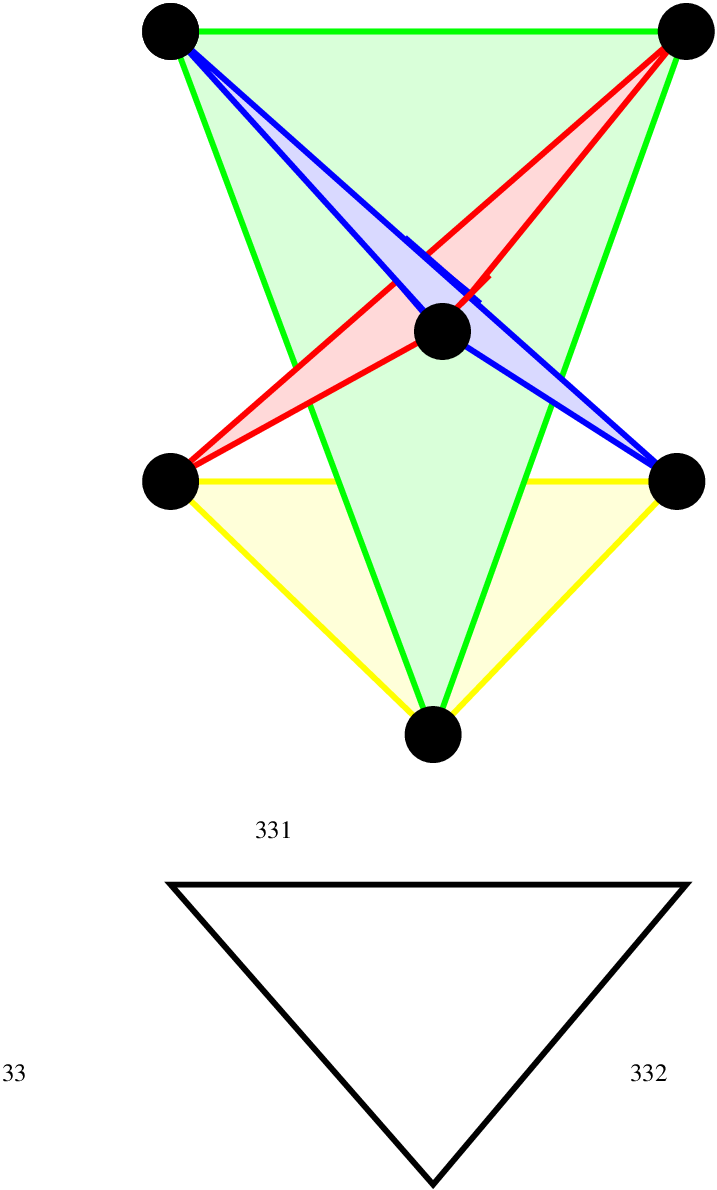}}\nn
\end{eqnarray}   
 \vspace*{.1cm}

\noindent (to justify this graph one has to employ theorem \ref{thm_rebit2}). Thanks to the information about the answers to $Q_{3_A3_B},Q_{3_A3_C},Q_{3_B3_C}$, such a state would contain a non-monogamous distribution of entanglement over $A,B,C$. In particular, if $O$ `marginalized' over rebit $C$, by discarding all information involving $C$, he would be left with the answer to $Q_{3_A3_B}$ which defines a maximally entangled two-rebit state of non-maximal information (see section \ref{sec_rebtang}) -- in contrast to the qubit case of section \ref{sec_ghz}.

\subsection{Correlation structure for $N=2$ gbits}\label{sec_corr}

While we were able to settle the relative negation $\neg$ between the correlations of bipartite correlation questions for both qubits and rebits in (\ref{qbit8}) (e.g., $Q_{11}\leftrightarrow Q_{22}=\neg(Q_{12}\leftrightarrow Q_{21})$), we still have to clarify the odd and even correlation structure for qubits in (\ref{qbit4}, \ref{qbit5}) and more generally in theorem \ref{thm_qubit}. For example, we have to clarify whether $Q_{33}=Q_{11}\leftrightarrow Q_{22}$ or $Q_{33}=\neg(Q_{11}\leftrightarrow Q_{22})$, etc. This will involve the notion of the `logical mirror' of an inference theory and require the tools from the previous section, covering the $N=3$ case. We recall that the odd and even correlation structure for rebits has already been settled in section \ref{sec_bell} through (\ref{qbit8}). This, as we shall see shortly, is a consequence of rebit theory being its own mirror image.

\subsubsection{The logical mirror image of an inference theory}\label{sec_mirror}

Consider a single qubit, described by $O$ via an informationally complete set $Q_{1},Q_{2},Q_{3}$ which can be viewed as a {\it question basis} on $\cq_1$. As such, it defines a `logical handedness' in terms of which outcomes to $Q_1,Q_2,Q_3$ $O$ calls `yes' and which `no'.\footnote{Ultimately, these three questions will define an orthonormal Bloch vector basis in the Bloch sphere (see section \ref{sec_n1}). The handedness or orientation of this basis will depend on the labeling of question outcomes.} Clearly, this is a convention made by $O$ and he can easily change the handedness by simply swapping the assignment `yes'$\leftrightarrow$'no' of one question, say, $Q_1$, which is tantamount to $Q_1\mapsto \neg Q_1$. This corresponds to taking the logical mirror image of a single qubit. For a single qubit this will not have any severe consequences and $O$ is free to choose whichever handedness he desires.

However, the handedness of a single qubit question basis does become important when considering composite systems. Consider, e.g., two observers $O,O'$ each having one qubit and describing it with a certain handed basis. They can also consider correlations $Q_{ij}$ of their qubit pair. If $O$ now decided to change the handedness of his local question basis by $Q_1\mapsto \neg Q_1$ this would result in
\ba
Q_{11},Q_{12},Q_{13}\q&\mapsto& \q\neg Q_{11},\neg Q_{12},\neg Q_{13},\q\nn\\
 Q_{ij}\q&\mapsto&\q Q_{ij},\,\,\, i\neq1\nn
\ea
 and therefore 
 \ba
&& Q_{11}\leftrightarrow Q_{22}\q\mapsto\q \neg(Q_{11}\leftrightarrow Q_{22}),\q\text{and}\q\nn\\&& Q_{12}\leftrightarrow Q_{21}\q\mapsto\q \neg(Q_{12}\leftrightarrow Q_{21}).\nn
 \ea
 Since $Q_{33}$ would remain unaffected by the change of local handedness ($Q_3,Q'_3$ remain invariant), $O$'s change of local handedness would have resulted in a switch between even and odd correlation structure for (\ref{qbit4}, \ref{qbit5}). Since the handedness of the local question basis is just a convention by $O$ or $O'$, we shall allow either. Consequently, whether three correlation questions are related by even or odd correlation (see theorem \ref{thm_qubit}) depends on the local conventions by $O,O'$ and both are, in fact, consistent. However, (\ref{qbit8}) will always hold. Usually, of course, one would favour the situation that both $O,O'$ make the same conventions such that their local question bases are equally handed (or that one observer $O$ describes all qubits with the same handedness). Nevertheless, physically all local conventions are fully equivalent, yet will lead to different representations of composite systems in the inference theory. 

But there are important consistency conditions on the distribution of odd and even correlations. This becomes obvious when examining three qubits $A,B,C$. To this end, consider the trivial conjunction
\ba
Q_{3_A3_B}\leftrightarrow Q_{3_A3_C}\leftrightarrow Q_{3_B3_C}\underset{\tiny\text{lemma \ref{lem5}}}{=}1.\label{333}
\ea
From (\ref{qbit4}) we know that the correlation is either even or odd, respectively, 
\ba
Q_{3_A3_B}&=&\q Q_{1_A1_B}\leftrightarrow Q_{2_A2_B},\q\text{or}\nn\\ Q_{3_A3_B}&=&\neg(Q_{1_A1_B}\leftrightarrow Q_{2_A2_B})\label{qbit4b}
\ea
and analogously for $A,C$ and $B,C$. Suppose now that all three bipartite correlations in the conjunction (\ref{333}) were even. Then we immediately get a contradiction because, using lemma \ref{lem5},
\begin{widetext}
\ba
Q_{3_A3_B}\leftrightarrow Q_{3_A3_C}\leftrightarrow Q_{3_B3_C}&=&\,\,\,\,(Q_{1_A1_B}\leftrightarrow Q_{2_A2_B})\leftrightarrow (Q_{1_A1_C}\leftrightarrow Q_{2_A2_C})\leftrightarrow (Q_{1_B1_C}\leftrightarrow Q_{2_B2_C})\nn\\
&\underset{\tiny(\ref{qbit8})}{=}&\neg \underset{=\,Q_{1_B1_C}}{\underbrace{(Q_{1_A1_B}\leftrightarrow Q_{1_A1_C})}}\leftrightarrow \underset{=\,Q_{2_B2_C}}{\underbrace{(Q_{2_A2_B}\leftrightarrow Q_{2_A2_C})}}\leftrightarrow (Q_{1_B1_C}\leftrightarrow Q_{2_B2_C})\nn\\
&=&0.\nn
\ea
\end{widetext}
But this violates the identity (\ref{333}) and results from the relative negation between the left and right hand side in (\ref{qbit8}). One would get the same contradiction if one of the correlations in (\ref{333}) was even and two were odd because then the two negations from the odd correlation would cancel each other and one would still be left with the negation coming from (\ref{qbit8}). 

On the other hand, everything is consistent if either all bipartite correlations in (\ref{qbit4b}) are odd or one is odd and two are even because then the odd number of negations from the odd correlations cancels the negation coming from (\ref{qbit8}):
\begin{widetext}
\ba
Q_{3_A3_B}\leftrightarrow Q_{3_A3_C}\leftrightarrow Q_{3_B3_C}&=&\neg(Q_{1_A1_B}\leftrightarrow Q_{2_A2_B})\leftrightarrow \neg(Q_{1_A1_C}\leftrightarrow Q_{2_A2_C})\leftrightarrow \neg(Q_{1_B1_C}\leftrightarrow Q_{2_B2_C})\nn\\
&\underset{\tiny(\ref{qbit8})}{=}&\neg \underset{=\,Q_{1_B1_C}}{\underbrace{(Q_{1_A1_B}\leftrightarrow Q_{1_A1_C})}}\leftrightarrow \,\,\,\,\underset{=\,Q_{2_B2_C}}{\underbrace{(Q_{2_A2_B}\leftrightarrow Q_{2_A2_C})}}\leftrightarrow \neg(Q_{1_B1_C}\leftrightarrow Q_{2_B2_C})\nn\\
&=&1.\nn
\ea
\end{widetext}
We can also quickly check that this is consistent with (\ref{qbit5}) 
\ba
Q_{3_A3_B}&=&\q Q_{1_A2_B}\leftrightarrow Q_{2_A1_B},\q\text{or}\nn\\ Q_{3_A3_B}&=&\neg(Q_{1_A2_B}\leftrightarrow Q_{2_A1_B})\label{qbit5b}
\ea
and (\ref{qbit7}) (and analogously for $A,C$ and $B,C$)
\ba
Q_{1_A2_B}\leftrightarrow Q_{2_A1_B}=\neg (Q_{1_A1_B}\leftrightarrow Q_{2_A2_B}).\nn
\ea
That is, if all correlations in (\ref{qbit4b}) are odd, then all correlations in (\ref{qbit5b}) must be even. Indeed,
\begin{widetext}
 \ba
Q_{3_A3_B}\leftrightarrow Q_{3_A3_C}\leftrightarrow Q_{3_B3_C}&=&\,\,\,\,(Q_{1_A2_B}\leftrightarrow Q_{2_A1_B})\leftrightarrow (Q_{1_A2_C}\leftrightarrow Q_{2_A1_C})\leftrightarrow (Q_{1_B2_C}\leftrightarrow Q_{2_B1_C})\nn\\
&\underset{\tiny(\ref{qbit8})}{=}&\neg \underset{=\,Q_{2_B2_C}}{\underbrace{(Q_{1_A2_B}\leftrightarrow Q_{1_A2_C})}}\leftrightarrow \underset{=\,Q_{1_B1_C}}{\underbrace{(Q_{2_A1_B}\leftrightarrow Q_{2_A1_C})}}\leftrightarrow (Q_{1_B2_C}\leftrightarrow Q_{2_B1_C})\nn\\
&=&Q_{3_B3_C}\leftrightarrow Q_{3_B3_C}=1\nn
\ea
\end{widetext}
is consistent. 

In conclusion, if $O$ wants to treat the bipartite relations among all three $A,B,C$ {\it identically}, then the following distribution of odd and even correlations
\ba
Q_{3_A3_B}&=&\q Q_{1_A2_B}\leftrightarrow Q_{2_A1_B}\nn\\&=&\neg(Q_{1_A1_B}\leftrightarrow Q_{2_A2_B}),\label{qbit9}
\ea
and analogously for $A,C$ and $B,C$, is the {\it only} consistent solution. {\it We shall henceforth make the convention that the bipartite correlation structure among any pair of qubits be the same such that (\ref{qbit9}) holds}. This turns out to be the case of qubit quantum theory. 

The tacit assumption, underlying the standard representation of qubit quantum theory, is that the handedness of each local qubit question basis for $A,B,C$ is the same, e.g., all `left' or `right' handed. But we emphasize, that it would be equally consistent to choose one basis as `left' (`right') and the other two as `right' (`left') handed. A qubit pair with equally handed bases will be described by the odd and even correlation distribution as in (\ref{qbit9}), while a qubit pair with oppositely handed bases will be described by the opposite distribution of odd and even correlations. In terms of whether $Q_{33}=Q_{11}\leftrightarrow Q_{22}$ is even or odd $Q_{33}=\neg(Q_{11}\leftrightarrow Q_{22})$, the three qubit relations yield only four consistent graphs for the distribution of `left' and `right' handedness\begin{widetext}
 \begin{eqnarray}
 \psfrag{o}{\footnotesize odd}
 \psfrag{e}{\footnotesize even}
 \psfrag{l}{\footnotesize`left'}
 \psfrag{r}{\footnotesize `right'}
\hspace*{-1cm}{\includegraphics[scale=.2]{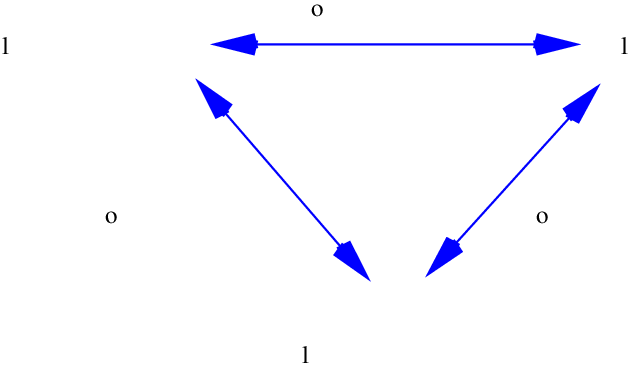}}\q\q\q\q\q{\includegraphics[scale=.2]{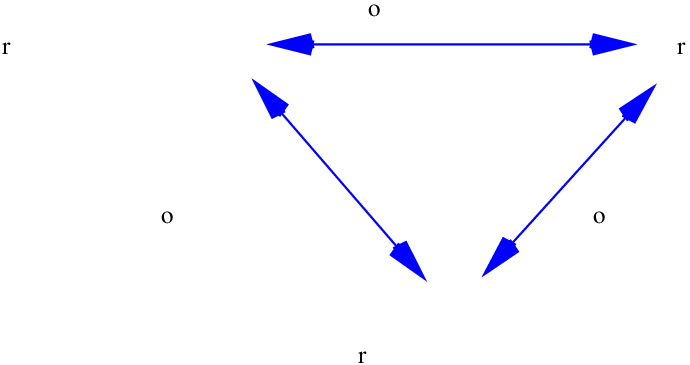}}\q\q\q\q\q{\includegraphics[scale=.2]{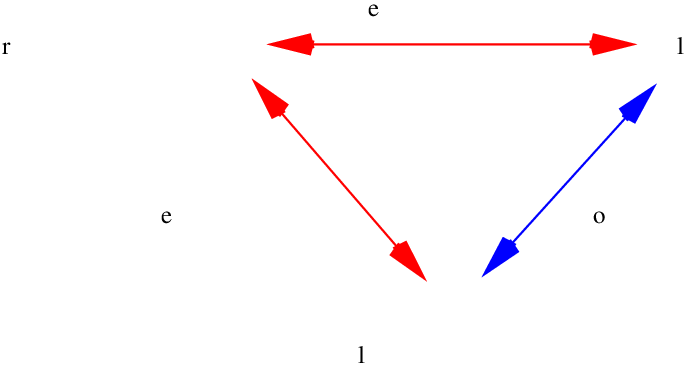}}\q\q\q\q\q{\includegraphics[scale=.2]{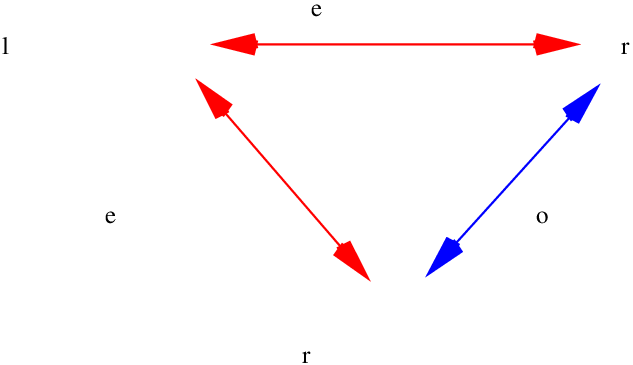}}\label{mirror}
\end{eqnarray}   
\end{widetext}
This gives a simple graphical explanation for the consistency observations above. The first two graphs correspond to quantum theory. The framework with the opposite correlation structure of quantum theory, corresponding to the last two graphs, is sometimes referred to as {\it mirror quantum theory} \cite{Dakic:2009bh}. While mirror quantum theory was considered inconsistent in \cite{Dakic:2009bh}, we see here that inconsistencies would only arise if the bipartite correlation structure of mirror quantum theory (and in particular the even correlation $Q_{33}=Q_{11}\leftrightarrow Q_{22}$) was used for all three pairs of qubits. However, the proper formulation of mirror quantum theory for three qubits corresponds to the last two graphs in (\ref{mirror}) and gives a perfectly consistent framework.\footnote{In fact, one could even produce the correlation structure of mirror quantum theory in the lab by using oppositely handed bases for two qubits in an entangled pair. The resulting state would be represented by the partial transpose of an entangled qubit state. This state would not be positive if represented in terms of the standard Pauli matrices and thus not correspond to a legal quantum state \cite{Peres:1996dw}. However, the point is that mirror quantum theory would not be represented in the standard Pauli matrix basis but in the partially transposed basis which corresponds to replacing
\ba
\sigma_y=\left(\begin{array}{cc}0 & -i \\i & 0\end{array}\right)\q\q\q\text{by}\q\q\q\sigma^T_y=\left(\begin{array}{cc}0 & i \\-i & 0\end{array}\right),\nn
\ea
for one of the two qubits (this is a switch of the $y$-axis orientation). In this basis, the state would be positive.}
 It could be easily generalized in an obvious manner to arbitrarily many qubits.\footnote{Notice, however, that for more than three qubits one would get more than two different consistent distributions of odd and even correlations. For example, for four qubits there will be three cases: (1) all four have equally handed bases, (2) three have equally handed bases, (3) two have equally handed bases.}

Obviously, the same argument can be carried out for any other triple of maximally compatible bipartite correlations appearing in theorem \ref{thm_qubit}. In conclusion, the two distinct consistent distributions of odd and even correlations, corresponding to quantum theory and mirror quantum theory, 
\begin{itemize}
\item[(a)] result from different conventions of local basis handedness, and
\item[(b)] are in one-to-one correspondence through a local relabeling `yes'$\leftrightarrow$`no' of one individual question.
\end{itemize}
As such, the two distinct correlation structures ultimately give rise to two distinct representations of the same physics and are thus fully equivalent. The transformation (b) between the two representations -- being a translation between two conventions/descriptions -- is a {\it passive} one and can be carried out on a piece of paper; it is {\it always} allowed. However, clearly, there cannot exist any actual {\it physical} transformation in the laboratory which maps states from one convention into the other. 

Within the formalism of quantum theory the transformation (b) $Q_1\mapsto\neg Q_1$ corresponds to the {\it partial transpose} (e.g., see \cite{Dakic:2009bh}). It is well known that the partial transpose defines a separability criterion for quantum states which is both necessary and sufficient for a pair of qubits \cite{Peres:1996dw}: a two-qubit density matrix $\rho$ is separable 
if and only if its partial transpose is positive relative to a basis of standard Pauli matrix products (i.e., represents a legal quantum state). This criterion holds analogously in our language here: as seen at the beginning of this section, the transformation $Q_1\mapsto \neg Q_1$ changes between odd and even correlations of bipartite correlations $Q_{ij}$. This would, in fact, be unproblematic if $O$ only had individual information about the two qubits; it would map a classically composed state even of maximal information, say, $Q_1=1$ and $Q'_1=1$ and thus $Q_{11}=1$, to another legal classically composed state $Q_1=0$, $Q'_1=1$ and thus $Q_{11}=0$. Both states exist within both conventions. However, applying this transformation to a maximally entangled state with odd correlation, say, $Q_{11}=Q_{22}=1$ and thus, by (\ref{qbit9}), $Q_{33}=0$ yields an even correlation $Q_{11}=0=Q_{33}$ and $Q_{22}=1$. The former state only exists in the quantum theory representation, while the latter exists only in the mirror image. For other entangled states one would similarly find that $Q_1\mapsto\neg Q_1$ necessarily maps from one representation into the other. The same conclusion also holds for a transformation $\{Q_1,Q_2,Q_3\}\mapsto\{\neg Q_1,\neg Q_2,\neg Q_3\}$, which one might call {\it total inversion}, because the odd number of negations involved in the transformation would likewise lead to a swap of odd and even correlations.

Lastly, we note that the situation is very different for rebit theory because it is its own mirror image, i.e.\ rebit theory and {\it mirror rebit theory} are {\it identical} representations. $O$ will describe a single rebit by a question basis $Q_1,Q_2$. Suppose $O$ decided to swap the `yes' and `no' assignments to the outcomes of $Q_1$, such that equivalently $Q_1\mapsto \neg Q_1$. For a pair of rebits, this would have the following ramification
\ba
Q_{11},Q_{12}\q&\mapsto&\q \neg Q_{11},\neg Q_{12},\nn\\ Q_{21},Q_{22}\q&\mapsto& \q\q Q_{21},\, Q_{22},\nn
\ea
 and therefore 
 \ba
Q_{12}\leftrightarrow Q_{21}\q&\mapsto&\q \neg(Q_{12}\leftrightarrow Q_{21})\nn\\\Rightarrow\q\q\q\q Q_{33}\q&\mapsto&\q\neg Q_{33}.\nn
 \ea
In contrast to the qubit case, $Q_{33}$ is {\it defined} as a correlation of correlations $Q_{33}:=Q_{12}\leftrightarrow Q_{21}$ (\ref{q33re}) and can {\it not} be written in terms of local questions $Q_{3},Q'_3$. Hence, $Q_{33}$ also changes under this transformation by construction. Accordingly, $Q_1\mapsto\neg Q_1$ does {\it not} lead to a swap of odd and even correlations for the rebit case. This `partial transpose' therefore always maps states to other states within the {\it same} representation. For example, even a maximally entangled state of maximal information and even correlation, say, $Q_{12}=Q_{21}=1$ and $Q_{33}=1$, would be mapped to another evenly correlated state $Q_{12}=0$, $Q_{21}=1$ and $Q_{33}=0$. As a consequence, the Peres separability criterion \cite{Peres:1996dw} which is valid for qubits does {\it not} hold in analogous fashion for rebits. 
For completely equivalent reasons, the {\it total inversion}, corresponding to $\{Q_1,Q_2\}\mapsto\{\neg Q_1,\neg Q_2\}$, is also a transformation which preserves the representation.

\subsubsection{Collecting the results: odd and even correlation structure for $N=2$}

After the many technical details it is useful to collect all the results concerning the compatibility, complementarity and correlation structure for two qubits, derived in lemmas \ref{lem1} and \ref{lem3}, theorem \ref{thm_qubit}, equation (\ref{qbit8}) and in the previous section \ref{sec_mirror} in a graph to facilitate a visualization. We shall henceforth abide by the convention that all bipartite relations for arbitrarily many qubits be treated equally such that (\ref{qbit9}) must hold. For the other relations of theorem \ref{thm_qubit} one finds the analogous results. As can be easily verified, the ensuing question structure has the lattice pattern of figure \ref{fig_corr}, where the triangles
\begin{widetext}
\begin{eqnarray}
\psfrag{-}{\hspace*{.05cm}\vspace*{.1cm}\footnotesize$-$}
\psfrag{+}{\footnotesize$+$}
\psfrag{A}{\footnotesize $Q$}
\psfrag{B}{\footnotesize$Q'$}
\psfrag{C}{\footnotesize$Q''$}
{\includegraphics[scale=.3]{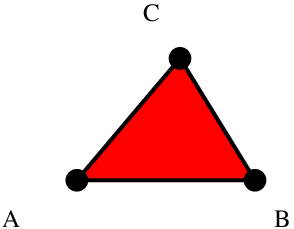}}\q\q\q\Leftrightarrow Q=\neg(Q'\leftrightarrow Q''),\q\q\q\q\q\q
{\includegraphics[scale=.3]{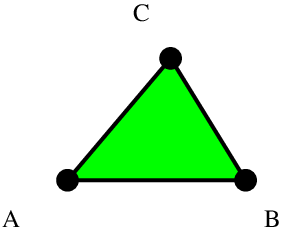}}\q\q\q\Leftrightarrow Q=Q'\leftrightarrow Q''\label{oddeven}
\end{eqnarray}
denote odd and even correlation, respectively.  
\begin{figure}[hbt!]
\begin{center}
\psfrag{+}{$+$}
\psfrag{-}{\hspace*{-.1cm}$-$}
\psfrag{1}{$Q_1$}
\psfrag{2}{$Q_2$}
\psfrag{3}{$Q_3$}
\psfrag{1p}{$Q'_1$}
\psfrag{2p}{$Q'_2$}
\psfrag{3p}{$Q'_3$}
\psfrag{11}{$Q_{11}$}
\psfrag{22}{$Q_{22}$}
\psfrag{12}{$Q_{12}$}
\psfrag{33}{$Q_{33}$}
\psfrag{13}{$Q_{13}$}
\psfrag{21}{$Q_{21}$}
\psfrag{23}{$Q_{23}$}
\psfrag{31}{$Q_{31}$}
\psfrag{32}{$Q_{32}$}
\psfrag{i}{\hspace*{-.45cm}\footnotesize identify}
{\includegraphics[scale=.3]{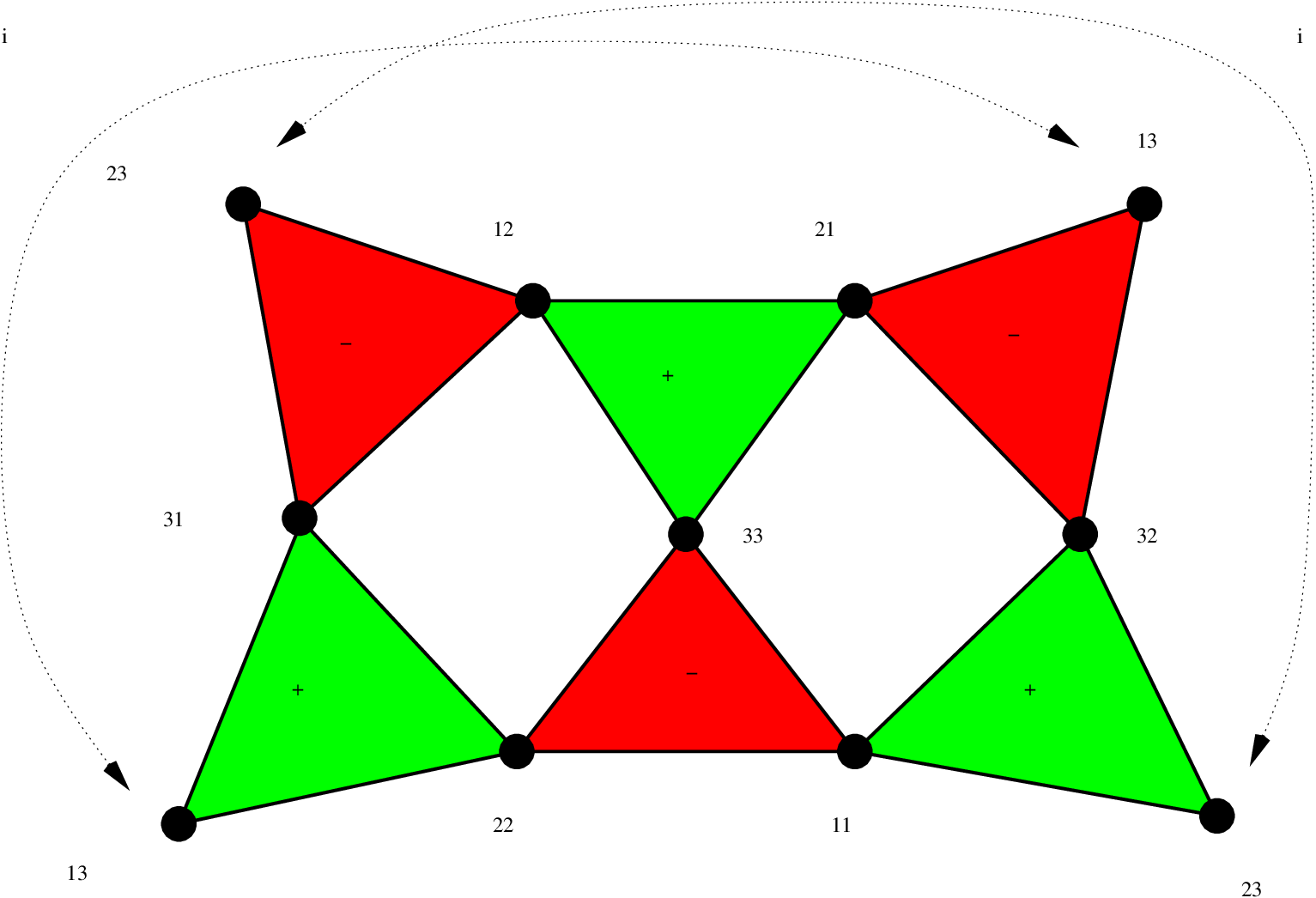}}\\\vspace*{1cm}
{\includegraphics[scale=.3]{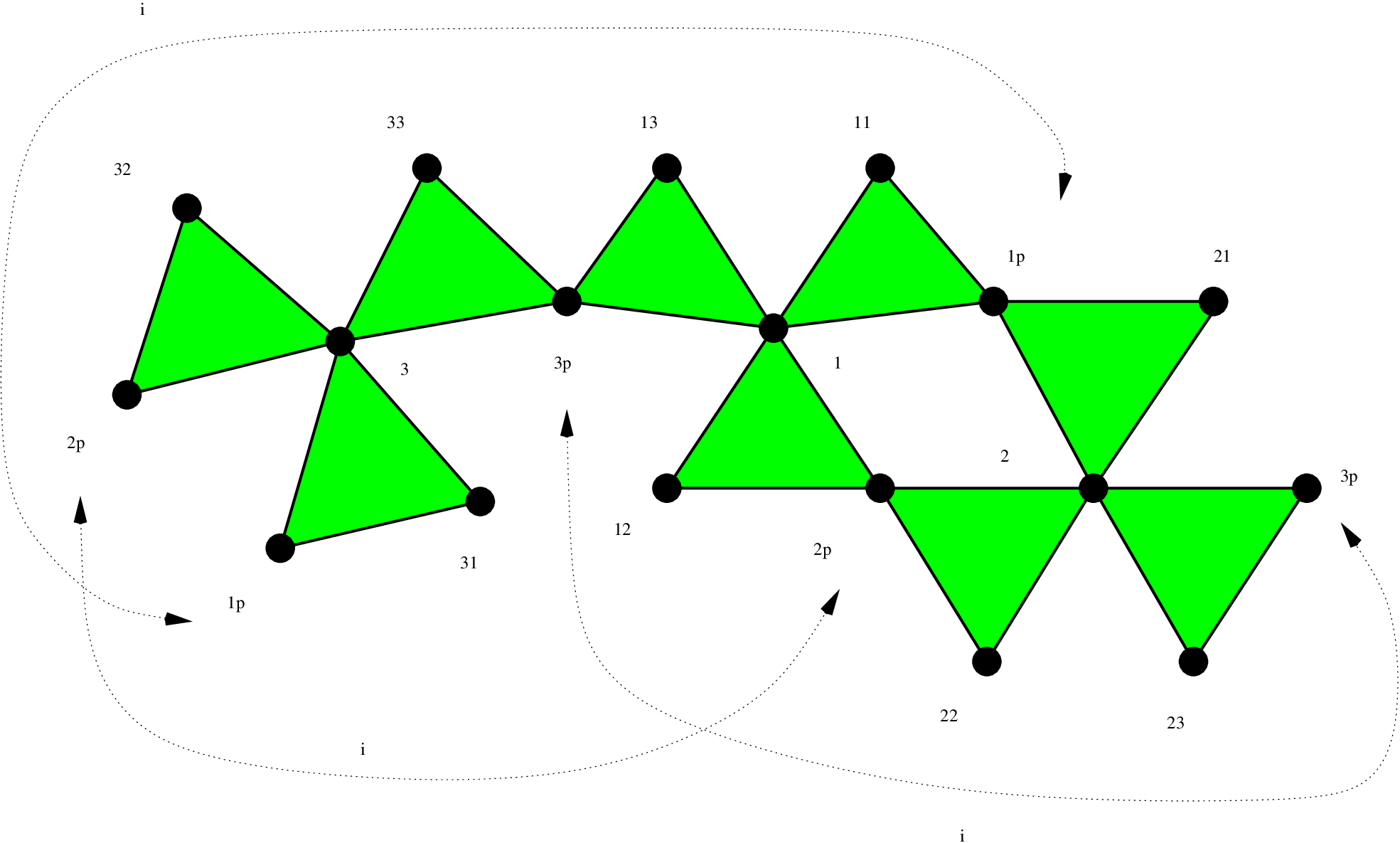}}
\caption{\small A lattice representation of the complete compatibility, complementarity and correlation structure of the informationally complete set $\cq_{M_2}$ for {\bf two qubits}. Every vertex corresponds to one of the 15 pairwise independent questions. If two questions are connected by an edge, they are maximally compatible. If two questions are {\it not} connected by an edge, they are maximally complementary. Thanks to the logical structure of the XNOR $\leftrightarrow$, defining a question as a correlation of two other questions, the compatibility structure results in a lattice of triangles. As clarified in (\ref{oddeven}), red triangles denote odd, while green triangles denote even correlation. Note that the two lattices represented here are connected through the nine correlation questions $Q_{ij}$ and form a single closed lattice (which, however, is easier to represent in this disconnected manner). Every question resides in exactly three triangles and is thereby maximally compatible with six and maximally complementary to eight other questions.}\label{fig_corr}
\end{center}
\end{figure}
\end{widetext}

We recall that (\ref{qbit7}, \ref{qbit8}) imply alternating odd and even correlation triangles for the bipartite correlation questions. However, we emphasize that the Bell scenario argument of section \ref{sec_bell} does {\it not} require the correlation triangles involving individual questions to also admit such an alternating odd and even pattern. For instance, the following relations of $Q_{33}$
\ba
Q_{33}=Q_3\leftrightarrow Q'_3=Q_{12}\leftrightarrow Q_{21}=\neg(Q_{11}\leftrightarrow Q_{22})\nn
\ea
are consistent with assumptions \ref{assump4} and \ref{assump4b}, despite the absence of a negation in $Q_3\leftrightarrow Q'_3=Q_{12}\leftrightarrow Q_{21}$ because the latter is {\it not} a classical logical identity. The graph corresponding to the last relation,
 \vspace*{.2cm} \begin{eqnarray}
\psfrag{d}{$\vdots$}
\psfrag{Q11}{$Q_{11}$}
\psfrag{Q22}{$Q_{22}$}
\psfrag{Q12}{$Q_{12}$}
\psfrag{Q21}{$Q_{21}$}
\psfrag{Q33}{$Q_{33}$}
\psfrag{q3}{$Q_3$}
\psfrag{p3}{$Q'_3$}
{\includegraphics[scale=.2]{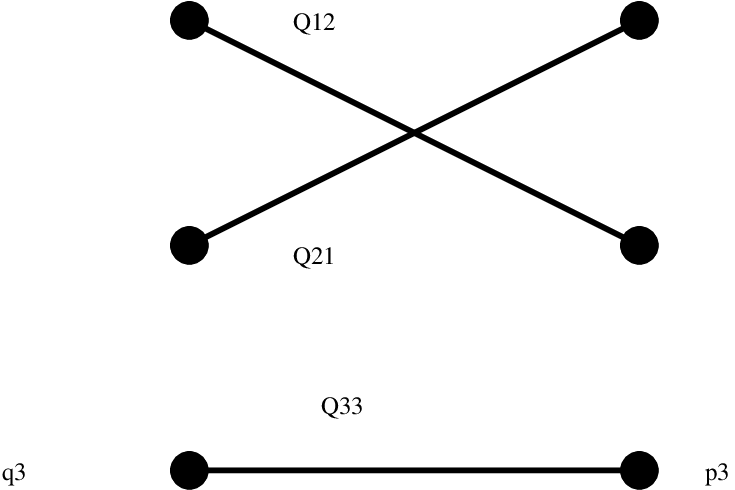}}\notag
\end{eqnarray}
 \vspace*{.2cm}
 
\noindent is also different from the graphs (\ref{join}) as it involves all six individual questions and the argument leading to (\ref{qbit8}) does not apply. Clearly, given the definition $Q_{ij}:=Q_i\leftrightarrow Q_j'$, all such triangles must be even. 

The lattice structure in figure \ref{fig_corr} contains 15 triangles and $15\times3=45$ distinct edges corresponding to compatibility relations. There are 60 `missing edges' corresponding to complementarity relations. Every question resides in three compatibility triangles and is therefore maximally compatible with the six other questions in these three triangles and maximally complementary to the eight remaining questions not contained in those three triangles. This means that once one question is fully known, all other information available to $O$ must be distributed over the three adjacent triangles. In particular, if $O$ knows the answers to two of the questions in the lattice, he will also know the answer to the third sharing the same triangle such that every triangle corresponds to a specific set of states of maximal information. (Although, as we shall see shortly in section \ref{sec_infomeasure} and, in more detail, in \cite{hw}, the time evolution rule \ref{time} implies that states of maximal information will likewise exist such that two independent \texttt{bits} can be distributed differently over the lattice.) 

Notice that every triangle is connected by an edge (adjacent to one of its three vertices) to every other vertex in the lattice. This embodies the statement `whenever $O$ asks $S$ a new question, he experience no net loss of information' of the complementarity rule \ref{unlim}. For instance, if $O$ has maximal knowledge about two questions in the lattice, corresponding to maximal information about one triangle, it is impossible for him to loose information by asking another question from the lattice because it will be connected to one of the questions from the previous triangle. As a concrete example, suppose $O$ knew the answers to $Q_{21},Q_{13},Q_{32}$. He could ask $Q_{22}$ next. Since $Q_{22}$ is connected by an edge to $Q_{13}$, he would know the answers to both upon asking $Q_{22}$ and thereby then also the answer to $Q_{31}$ -- no net loss of information occurs.

It is straightforward to check, e.g., using the ansatz
\ba
|\psi\rangle=\alpha\,|z_-z_-\rangle+\beta\,|z_+z_+\rangle+\gamma\,|z_-z_+\rangle+\delta|z_+z_-\rangle\nn
\ea
for a two qubit pure state and translating it into the various basis combinations $xx,xy,yx,yy,\ldots$, that the lattice structure of figure \ref{fig_corr} is precisely the compatibility and correlation structure of qubit quantum theory. For instance, $Q_{11},Q_{22},Q_{33}$ correspond to projectors onto the $+1$-eigenspaces of $\sigma_x\otimes\sigma_x$, $\sigma_y\otimes\sigma_y$ and $\sigma_z\otimes\sigma_z$. The three questions sharing a red triangle means, e.g., that $Q_{11}=Q_{22}=1$ and $Q_{33}=0$ is an allowed state while $Q_{11}=Q_{22}=Q_{33}=1$ is illegal. Indeed, ignoring normalization, in quantum theory one finds 
\ba
|x_+x_+\rangle-|x_-x_-\rangle&=&-i|y_+y_+\rangle+i|y_-y_-\rangle\nn\\&=&\q\,|z_+z_-\rangle\,+\,|z_-z_+\rangle\nn
\ea
for $\alpha=\beta=0$ and $\gamma=\delta=1$, corresponding to the propositions $Q_{11}=1$: ``the spins are correlated in $x$-direction"; $Q_{22}=1$: ``the spins are correlated in $y$-direction"; and $Q_{33}=0$: ``the spins are anti-correlated in $z$-direction".\footnote{Similarly, $|\psi\rangle=\alpha\,|z_-z_-\rangle+\beta\,|z_+z_+\rangle=\f{\alpha+\beta}{2}(|x_+x_+\rangle+|x_-x_-\rangle)+\f{\beta-\alpha}{2}(|x_+x_-\rangle+|x_-x_+\rangle)=\f{\alpha+\beta}{2}(|y_+y_-\rangle+|y_-y_+\rangle)+\f{\beta-\alpha}{2}(|y_+y_+\rangle+|y_-y_-\rangle)$. But $|x_+x_+\rangle+|x_-x_-\rangle=|y_+y_-\rangle+|y_-y_+\rangle$ and $|y_+y_+\rangle+|y_-y_-\rangle=|x_+x_-\rangle+|x_-x_+\rangle$. Hence, once $Q_{33}=1$, one indeed gets $Q_{11}\leftrightarrow Q_{22}=0$ even though $\alpha,\beta$ are unspecified such that $Q_{11},Q_{22}$ may be unknown.} But no quantum state exists such that the spins are also correlated in $z$-direction if they are correlated in $x$- and $y$-direction (this would be mirror quantum theory). Every other triangle in the lattice corresponds similarly to four pure quantum states (representing the answer configurations `yes-yes', `yes-no', `no-yes', `no-no' to the two independent questions per triangle).

Finally, we also collect the results on the compatibility, complementarity and correlation structure of two rebits, derived in lemmas \ref{lem1}, \ref{lem3} and \ref{lem4}, theorem \ref{thm_rebit} and equation (\ref{qbit7}), in a lattice structure in figure \ref{fig_corrd2}.
\begin{figure}[hbt!]
\begin{center}
\psfrag{+}{$+$}
\psfrag{-}{\hspace*{-.1cm}$-$}
\psfrag{1}{$Q_1$}
\psfrag{2}{$Q_2$}
\psfrag{3}{$Q_3$}
\psfrag{1p}{$Q'_1$}
\psfrag{2p}{$Q'_2$}
\psfrag{3p}{$Q'_3$}
\psfrag{11}{$Q_{22}$}
\psfrag{22}{$Q_{11}$}
\psfrag{12}{$Q_{12}$}
\psfrag{33}{$Q_{33}$}
\psfrag{13}{$Q_{1}'$}
\psfrag{21}{$Q_{21}$}
\psfrag{23}{$Q_{2}'$}
\psfrag{31}{$Q_{1}$}
\psfrag{32}{$Q_{2}$}
\psfrag{i}{\hspace*{-.45cm}\footnotesize identify}
{\includegraphics[scale=.28]{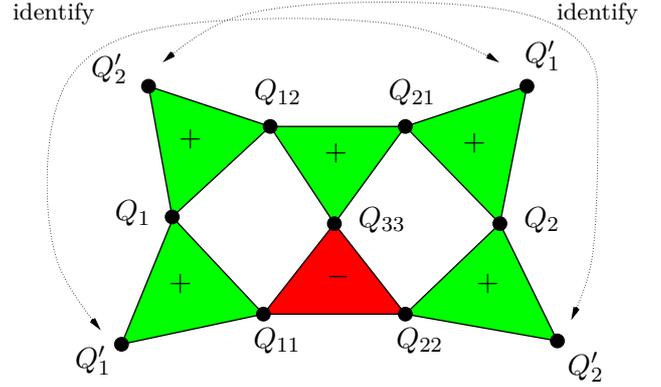}}
\caption{\small A lattice representation of the complete compatibility, complementarity and correlation structure of the informationally complete set $\cq_{M_2}$ for {\bf two rebits}. The explanations of the caption of figure \ref{fig_corr} also apply here. }\label{fig_corrd2}
\end{center}
\end{figure}
There are six triangles and $6\times3=18$ edges representing compatibility relations. Every question resides in exactly two triangles and is thereby maximally compatible with four and maximally complementary to four other questions. As in the qubit case in figure \ref{fig_corr}, every triangle is connected by an edge to every other vertex in the lattice, in conformity with rule \ref{unlim} asserting that $O$ shall not experience a net loss of information by asking further questions.

\subsection{The general case of $N>3$ gbits}\label{sec_ngbits}

We are now prepared to investigate the independence, compatibility and correlation structure ensuing from rules \ref{lim} and \ref{unlim} on a general $\cq_{N}$, i.e.\ of $O$'s possible questions to a system $S$ composed of $N>3$ gbits. In particular, we shall exhibit an informationally complete set $\cq_{M_N}$ for both qubits and rebits. $O$ can view the $N$ gbits as a composite system in many different ways: as $N$ individual gbits, as one individual gbit and a composite system of $N-1$ gbits, as a system composed of $2$ gbits and a system composed of $N-2$ gibts, and so on. All these different compositions yield, of course, the same question structure. It is simplest to interpret the $N$ gbits recursively as being composed of a composite system of $N-1$ gbits and a new individual gbit. Definition \ref{def_comp} of a composite system then implies that $\cq_{N}$ must contain (1) the questions of $\cq_{{N-1}}$ for $N-1$ gbits, (2) the set $\cq_1$ of the new gbit, and (3) all logical connectives of the maximally compatible questions of those two sets. This entails slightly different repercussions for qubits and rebits.

\subsubsection{An informationally complete set and entanglement for $N>3$ qubits}\label{sec_nqbits}

 A natural candidate for an informationally complete question set is given by the set of all possible XNOR conjunctions of the individual questions of the $N$ gbits 
  \ba
 Q_{\mu_1\mu_2\cdots \mu_N}:=Q_{\mu_1}\leftrightarrow Q_{\mu_2}\leftrightarrow\cdots\leftrightarrow Q_{\mu_N}\label{ngbit}
 \ea
(we recall from (\ref{truth}) that the logical connective yielding independent questions is either $\leftrightarrow$ or $\oplus$). Here we have introduced a new index notation: $\mu_a$ is the question index for gbit $a\in\{1,\ldots,N\}$ and can take the values $0,1,2,3$. As before the index values $i=1,2,3$ correspond to the individuals $Q_{1_a},Q_{2_a},Q_{3_a}$. On the other hand, the index value $\mu_a=0$ implies that none of the three individual questions of gbit $a$ appears in the conjunction (\ref{ngbit}), i.e.\ always $Q_{0_a}\equiv1$. For instance, 
\ba
Q_{100000\cdots000}:=Q_{1_1},\q\,\, Q_{0003020\cdots0}:=Q_{3_4}\leftrightarrow Q_{2_6},\nn
\ea
 etc. Note that $Q_{000\cdots000}$ corresponds to no question. The set (\ref{ngbit}) thus contains all individual, bipartite, tripartite, and up to $N$-partite correlation questions for $N$ gbits. We emphasize that, due to the special multipartite structure (and associativity) of the XNOR, a question such as $Q_{111\cdots111}$ does {\it not} incarnate the question `are the answers to $Q_{1_1},Q_{1_2},\cdots,Q_{1_N}$ all the same?'. For example, for $N=4$, $Q_{1_1}=Q_{1_2}=0$ and $Q_{1_3}=Q_{1_4}=1$ also yield $Q_{1_11_21_31_4}=1$. This is important for the entanglement structure.
 
 We begin with an important result.
 \begin{lem}\label{lem_Nq0}
 The $4^N-1$ questions\footnote{We deduct the trivial question $Q_{000\cdots000}$. Obviously, one arrives at the same number by counting the distribution of individuals, bipartite,...and $N$-partite correlations over $N$ qubits as a binomial series $\sum_{k=1}^{N}\binom{N}{k}\,3^k=(3+1)^N-1$.} $Q_{\mu_1\cdots\mu_N}$, $\mu=0,1,2,3$, are pairwise independent.
 \end{lem}
 
 \begin{proof}
 Consider $Q_{\mu_1\cdots\mu_N}$ and $Q_{\nu_1\cdots\nu_N}$. The two questions must disagree in at least one index otherwise they would coincide. Let the questions differ on the index of gbit $a$, i.e.\ $\mu_a\neq\nu_a$. Suppose $\mu_a,\nu_a\neq0$. Then $Q_{\mu_a}$ is maximally compatible with  $Q_{\mu_1\cdots\mu_a\cdots\mu_N}=Q_{\mu_1}\leftrightarrow\cdots\leftrightarrow Q_{\mu_a}\leftrightarrow\cdots\leftrightarrow Q_{\mu_N}$ and maximally complementary to $Q_{\nu_1\cdots\nu_a\cdots\nu_N}=Q_{\mu_1}\leftrightarrow\cdots\leftrightarrow Q_{\nu_a}\leftrightarrow\cdots\leftrightarrow Q_{\nu_N}$ since $Q_{\mu_a},Q_{\nu_a}$ are maximally complementary. Using the same argument as in the proof of lemma \ref{lem2} this implies independence of $Q_{\mu_1\cdots\mu_N}$ and $Q_{\nu_1\cdots\nu_N}$. Lastly, suppose now $\mu_a=0$ and $\nu_a\neq0$ (as we have $\mu_a\neq\nu_a$ not both can be $0$). Then $Q_{i_a\neq\nu_a}$ with $i\neq0$ is maximally compatible with $Q_{\mu_1\cdots\mu_N}$ and maximally complementary to $Q_{\nu_1\cdots\nu_N}$. By the same argument, this again implies independence of $Q_{\mu_1\cdots\mu_N}$ and $Q_{\nu_1\cdots\nu_N}$.
  \end{proof}
 
Consequently, the set (\ref{ngbit}) will be part of an informationally complete set. We note that a hermitian matrix of trace equal to $1$ on a $2^N$-dimensional complex Hilbert space (i.e.\ qubit density matrix) is described by $4^N-1$ parameters. 

Next, we must elucidate the compatibility and complementarity structure of this set.
 
\begin{lem}\label{lem_Nq}
 $Q_{\mu_1\cdots\mu_N}$ and $Q_{\nu_1\cdots\nu_N}$ are
 \begin{description}
 \item[maximally compatible] if the index sets $\{\mu_1,\ldots,\mu_N\}$ and $\{\nu_1,\ldots,\nu_N\}$ differ by an {\bf even} number (incl.\ $0$) of non-zero indices, \vspace*{-.1cm}
 \item[maximally complementary] if the index sets $\{\mu_1,\ldots,\mu_N\}$ and $\{\nu_1,\ldots,\nu_N\}$ differ by an {\bf odd} number of non-zero indices.
  \end{description}
 \end{lem}
 For example, for $N=2$, $Q_{11}$ and $Q_{22}$ differ by two non-zero indices and are thus maximally compatible. By contrast, $Q_{10}=Q_1$ and $Q_{22}$ differ by an odd number of non-zero indices and are thereby maximally complementary.
 
 \vspace*{-.2cm}
\begin{proof}
Let  $Q_{\mu_1\cdots\mu_N}$ and $Q_{\nu_1\cdots\nu_N}$ disagree in an {\bf odd} number, call it $2n+1$, of non-zero indices. We can always reshuffle the index labeling of the $N$ qubits such that now $a=1,\ldots,2n+1$ corresponds to the qubits on which $Q_{\mu_1\cdots\mu_N}$ and $Q_{\nu_1\cdots\nu_N}$ differ by non-zero indices, i.e.\ $\mu_a\neq\nu_a$ and $\mu_a,\nu_a\neq0$. The remaining qubits labeled by $b=2n+2,\ldots,N$ are then such that $Q_{\mu_1\cdots\mu_N}$ and $Q_{\nu_1\cdots\nu_N}$ either agree on the non-zero index, $\mu_b=\nu_b\neq0$ or at least one of $\mu_b,\nu_b$ is $0$. That is, after the reshuffling the index labeling, we can write the questions as
\begin{widetext}
\ba
Q_{\mu_1\cdots\mu_N}&=&\underset{\text{disagreement}}{\underbrace{(Q_{\mu_1}\leftrightarrow\cdots\leftrightarrow Q_{\mu_{2n+1}})}}\leftrightarrow \underset{\text{maximally compatible}}{\underbrace{(Q_{\mu_{2n+2}}\leftrightarrow \cdots\leftrightarrow Q_{\mu_N})}}\nn\\
Q_{\nu_1\cdots\nu_N}&=&\overbrace{(Q_{\nu_1}\leftrightarrow\cdots\leftrightarrow Q_{\nu_{2n+1}})}\,\,\leftrightarrow \overbrace{(Q_{\nu_{2n+2}}\leftrightarrow \cdots\leftrightarrow Q_{\nu_N})}.\label{ngbit2}
\ea
\end{widetext}
The parts of the questions where the index sets either agree or feature zeros coincide with $Q_{\mu_{2n+2}\cdots\mu_N}$ and $Q_{\nu_{2n+1}\cdots\nu_N}$ and are clearly maximally compatible (dropping here zero indices).

We can now proceed by induction. Lemmas \ref{lem1}, \ref{lem3}, \ref{lem5}--\ref{lem_tri3} imply that the statement of this lemma is correct for $n=0,1$. Let the statement therefore be true for $n$ and consider $n+1$. Then, (\ref{ngbit2}) reads
\begin{widetext}
\ba
Q_{\mu_1\cdots\mu_N}\!\!&=&\!\!\underset{\text{disagreement}}{\underbrace{(Q_{\mu_1\cdots\mu_{2n+3}})}}\leftrightarrow \underset{\text{maximally compatible}}{\underbrace{Q_{\mu_{2n+4}\cdots\mu_N}}}=\underset{\text{maximally complementary}}{\underbrace{(Q_{\mu_1}\leftrightarrow\cdots\leftrightarrow Q_{\mu_{2n+1}})}}\leftrightarrow \underset{\text{disagreement}}{\underbrace{(Q_{\mu_{2n+2}\mu_{2n+3}})}}\leftrightarrow \underset{\text{max.\ compat.}}{\underbrace{Q_{\mu_{2n+4}\cdots\mu_N}}}\nn\\
Q_{\nu_1\cdots\nu_N}\!\!&=&\!\!\overbrace{(Q_{\nu_1\cdots\nu_{2n+3}})}\,\leftrightarrow\,\,\q \overbrace{Q_{\nu_{2n+4}\cdots\nu_N}}\,\,\q={\overbrace{(Q_{\nu_1}\leftrightarrow\cdots\leftrightarrow Q_{\nu_{2n+1}})}}\,\leftrightarrow {\overbrace{(Q_{\nu_{2n+2}\nu_{2n+3}})}}\,\,\leftrightarrow {\overbrace{Q_{\nu_{2n+4}\cdots\nu_N}}}.\nn
\ea  
\end{widetext}
But, by lemma \ref{lem3}, $Q_{\mu_{2n+2}\mu_{2n+3}}$ and $Q_{\nu_{2n+2}\nu_{2n+3}}$ are maximally compatible with each other and therefore also with $Q_{\mu_1\cdots\mu_N}$ and $Q_{\nu_1\cdots\nu_N}$. This implies that both $Q_{\mu_1\cdots\mu_N}$ and $Q_{\nu_1\cdots\nu_N}$ are maximally compatible with and, thanks to lemma \ref{lem_Nq0}, independent of, e.g., $Q_{\mu_{2n+2}\cdots\mu_{N}}$. Complementarity of $Q_{\mu_1\cdots\mu_N}$ and $Q_{\nu_1\cdots\nu_N}$ now follows from lemma \ref{lem0} and noting that $Q_{\mu_1\cdots\mu_{2n+1}}$ and $Q_{\nu_1\cdots\nu_{N}}$ are maximally complementary because they disagree in $2n+1$ non-zero indices (for which, by assumption, the statement of the lemma holds). 
 
 Finally, let $Q_{\alpha_1\cdots\alpha_N}$ and $Q_{\beta_1\cdots\beta_N}$ disagree in an {\bf even} number $2n$ of non-zero indices. Using an analogous reshuffling of the index labeling as in the odd case above, one can rewrite the two questions as 
 \begin{widetext}
 \ba
Q_{\alpha_1\cdots\alpha_N}&=&\underset{\text{differ}}{\underbrace{Q_{\alpha_1\alpha_2}}}\leftrightarrow\underset{\text{differ}}{\underbrace{Q_{\alpha_3\alpha_4}}}\leftrightarrow\cdots\leftrightarrow \underset{\text{differ}}{\underbrace{Q_{\alpha_{2n-1}\alpha_{2n}}}}\leftrightarrow \underset{\text{maximally compatible}}{\underbrace{Q_{\alpha_{2n+1}\cdots\alpha_N}}}\nn\\
Q_{\beta_1\cdots\beta_N}&=&{\overbrace{Q_{\beta_1\beta_2}}}\,\leftrightarrow{\overbrace{Q_{\beta_3\beta_4}}}\,\leftrightarrow\cdots\leftrightarrow {\overbrace{Q_{\beta_{2n-1}\beta_{2n}}}}\,\leftrightarrow\q\,\,{\overbrace{Q_{\beta_{2n+1}\cdots\beta_N}}}.\label{ngbit3}
\ea 
\end{widetext}
That is, one can decompose the disagreeing parts of the questions into bipartite correlations. But thanks to lemma \ref{lem3} two bipartite correlations of the same qubit pair are maximally compatible if and only if they differ in both indices. Consequently, all the pairings of question components of the upper and lower line, as written in (\ref{ngbit3}), are maximally compatible and, hence, so must be $Q_{\alpha_1\cdots\alpha_N}$ and $Q_{\beta_1\cdots\beta_N}$. 
  \end{proof}

Fortunately, it turns out that the $4^N-1$ questions (\ref{ngbit}) are logically closed under the XNOR.

\begin{Theorem}\label{thm_qubit3}{\bf(Qubits)}
The $4^N-1$ questions $Q_{\mu_1\cdots\mu_N}$, $\mu=0,1,2,3$, are logically closed under $\leftrightarrow$ and thus form an informationally complete set $\cq_{M_N}$ with $D_N=4^N-1$ for the case $D_1=3$.
\end{Theorem}

\begin{proof}
We shall prove the statement by induction. The statement is trivially true for $N=1$ and, by theorems \ref{thm_qubit} and \ref{thm_qubit2}, holds also for $N=2,3$. Let the statement therefore be true for $N-1$ and consider a composite system of $N$ qubits. Since $O$ can treat the $N$ qubit system in many different ways as a composite system and the statement holds for $N-1$, all XNOR conjunctions of questions involving at least one zero index will be contained in (\ref{ngbit}). We thus only need to show that all $\leftrightarrow$ conjunctions involving at least one $N$-partite correlation $Q_{i_1\cdots i_N}$, $i_a\neq0$ $\forall\,a=1,\ldots,N$, produce questions already included in (\ref{ngbit}). 

Consider $Q_{i_1\cdots i_N}$ and $Q_{\nu_1\cdots\nu_N}$. Lemma \ref{lem_Nq} implies that these two questions are maximally compatible -- and can thus be connected by $\leftrightarrow$ -- if and only if $\{i_1,\ldots,i_N\}$ and $\{\nu_1,\ldots,\nu_N\}$ disagree in an even number of non-zero indices. There are now two cases that we must consider:

(a) Suppose $Q_{i_1\cdots i_N}$ and $Q_{\nu_1\cdots\nu_N}$ disagree in an even number of non-zero indices and, furthermore, agree on at least one index $i_a=\nu_a$. Thanks to $Q_{i_a}\leftrightarrow Q_{i_a}=1$, the conjunction then yields
\begin{widetext}
\ba
Q_{i_1\cdots i_a\cdots i_N}\leftrightarrow Q_{\nu_1\cdots i_a\cdots\nu_N}=Q_{i_1\cdots\underset{\text{position }a}{0}\cdots i_N}\leftrightarrow Q_{\nu_1\cdots\underset{\text{position }a}{0}\cdots\nu_N}.\nn
\ea
\end{widetext}
Hence, the two questions on the right hand side both contain less than $N$ non-zero indices such that the result must lie in the set (\ref{ngbit}) because the statement is true up to $N-1$ by assumption.

(b) Suppose $Q_{i_1\cdots i_N}$ and $Q_{\nu_1\cdots\nu_N}$ disagree in an even number $2n$ of non-zero indices and do not agree on {\it any} non-zero index. Reshuffling the index labelings as in the proof of lemma \ref{lem_Nq}, one can then write
\begin{widetext}
 \ba
Q_{i_1\cdots i_N}&=&\,\underset{\text{differ}}{\underbrace{Q_{i_1i_2}}}\,\leftrightarrow\,\underset{\text{differ}}{\underbrace{Q_{i_3i_4}}}\,\,\leftrightarrow\cdots\leftrightarrow \,\underset{\text{differ}}{\underbrace{Q_{i_{2n-1}i_{2n}}}}\,\,\leftrightarrow Q_{i_{2n+1}\cdots i_N}\nn\\
Q_{\nu_1\cdots\nu_N}&=&{\overbrace{Q_{\nu_1\nu_2}}}\,\leftrightarrow{\overbrace{Q_{\nu_3\nu_4}}}\,\leftrightarrow\cdots\leftrightarrow {\overbrace{Q_{\nu_{2n-1}\nu_{2n}}}}\,\,\leftrightarrow{{Q_{0_{2n+1}\cdots0_N}}}.\nn
\ea
\end{widetext}
(We obtain here $Q_{0_{2n+1}\cdots0_N}$ because $Q_{i_1\cdots i_N}$ and $Q_{\nu_1\cdots\nu_N}$ do not agree on any common index and $Q_{i_1\cdots i_N}$ does not feature zero-indices.) By lemma \ref{lem3}, the pairs of bipartite correlations differing in two indices, e.g., $Q_{i_1i_2}$ and $Q_{\nu_1\nu_2}$, are maximally compatible and by theorem \ref{thm_qubit} their XNOR conjunction will yield another bipartite correlation of the same qubit pair. For example, $Q_{i_1i_2}\leftrightarrow Q_{\nu_1\nu_2}$ equals either $Q_{j_1j_2}$ or $\neg Q_{j_1j_2}$ for $j_1\neq i_1,\nu_1$ and $j_2\neq i_2,\nu_2$. Accordingly, up to negation, one finds ($j_a\neq i_a,\nu_a$, $a=1,\ldots,2n$)
\begin{widetext}
\ba
Q_{i_1\cdots i_N}\leftrightarrow Q_{\nu_1\cdots\nu_N}=Q_{j_1j_2}\,\leftrightarrow\,{Q_{j_3j_4}}\,\,\leftrightarrow\cdots\leftrightarrow \,{{Q_{j_{2n-1}j_{2n}}}}\,\,\leftrightarrow Q_{i_{2n+1}\cdots i_N}=Q_{j_1\cdots j_{2n}i_{2n+1}\cdots i_N}\nn
\ea
\end{widetext}
which is another $N$-partite correlation contained in the set (\ref{ngbit}).
\end{proof}

As in the cases $N\leq3$, we could represent the compatibility and complementarity relations for $N>3$ qubits geometrically by a simplicial question graph where a $D$-partite question corresponds to a $(D-1)$-dimensional simplex within the graph ($D\leq N$). The compatibility and complementarity relations of lemma \ref{lem_Nq} then translate into abstract geometric relations according to whether a $D$-simplex and a $D'$-simplex share or disagree on subsimplices. In particular, the criterion that two distinct questions in $\cq_{M_N}$ are maximally compatible if and only if they disagree on an {\it even} number of non-zero indices means geometrically that the two questions are maximally compatible if and only if the two subsimplices in the two question simplices which correspond to the qubits they share in common either
\begin{itemize}
\item[(a)] do not overlap and are {\it odd}-dimensional, or 
\item[(b)] coincide.
\end{itemize} 
For example, for $N=2$, $Q_{11},Q_{22}$ correspond to two non-intersecting edges, i.e.\ $1$-simplices; they are maximally compatible because they disagree on their $1$-simplices, both involving qubit $1$ and $2$. On the other hand, for $N=3$, $Q_{1_A1_B},Q_{2_B2_C}$ also correspond to two non-intersecting $1$-simplices but are maximally complementary because the two edges involve different qubits -- $A,B$ for the first and $B,C$ for the second edge. 
 
This also helps us to understand entanglement for arbitrarily many qubits. Specifically, maximal entanglement will correspond to $O$ spending the $N$ independent \texttt{bits} he is allowed to acquire about the system $S$ of $N$ qubits on $N$-partite correlation questions. Lemma \ref{lem_Nq} guarantees that for every $N$ there will exist $N$ maximally compatible $N$-partite questions. For instance, there are $\binom{N}{2}$ ways of having $N-2$ of the indices take value $1$ and $2$ indices take the value $2$. Any two of the $\binom{N}{2}$ corresponding questions will disagree in two or four non-zero indices (and agree on the rest) and will thus be maximally compatible (even mutually according to theorem \ref{assump5}). These will correspond to a set of $\binom{N}{2}$ maximally compatible $(N-1)$-simplices in the question graph such that any two of them either disagree on an edge or a tetrahedron. (There will exist even more maximally compatible $N$-partite questions.) For $N\geq3$ it also holds that $\binom{N}{2}\geq N$. 

It is definitely possible to choose $N$ such maximally compatible $N$-partite correlation questions out of the $\binom{N}{2}$ many such that these $N$ questions do not all agree on a single index. For similar reasons to the $N=2,3$ cases, this choice will constitute a mutually independent set such that every individual question $Q_{i_1},\ldots,Q_{i_N}$ will be maximally complementary to at least one of these $N$ $N$-partite questions. (As a consequence of rule \ref{lim}, once the answers to these $N$ $N$-partite questions are known, they will also imply the answers to the $\binom{N}{2}-N$ remaining ones by the same reasoning as in section \ref{sec_3D}.) Accordingly, $O$ can exhaust the information limit with these $N$-partite correlation questions, while not being able to have any information whatsoever about the individuals -- a necessary condition for maximal entanglement.

There will exist many different ways of having such multipartite entanglement for arbitrary $N$. 
One could describe such different ways of entanglement by generalizing the correlation measures (\ref{monomeasure}) and informational monogamy inequalities resulting therefrom. These monogamy inequalities could also be considered as simplicial relations: they restrict the way in which the available (independent and dependent) information can be distributed over the various simplices and subsimplices in the question graph. However, we abstain from analyzing such relations here further.

\subsubsection{An informationally complete set and entanglement for $N>3$ rebits}

We briefly repeat the same procedure for $N$ rebits. In analogy to (\ref{ngbit}), the natural candidate set for an informationally complete $\cq_{M_N}$ will contain
\ba
 Q_{\mu_1\mu_2\cdots \mu_N}&:=&Q_{\mu_1}\leftrightarrow Q_{\mu_2}\leftrightarrow\cdots\leftrightarrow Q_{\mu_N},\nn\\\mu_a&=&0,1,2,\q\q a=1,\ldots,N,\label{nrebit}
\ea 
where the notation should be clear from section \ref{sec_nqbits}. However, being a composite system, by definition \ref{def_comp} we must permit the correlation of correlations $Q_{3_a3_b}$ (\ref{q33re}) for all $a,b\in\{1,\ldots,N\}$ because clearly $O$ is allowed to ask $Q_{1_a},Q_{2_a},Q_{1_b},Q_{2_b}$. Furthermore, thanks to (\ref{re-close}) we have, e.g.,
\ba
Q_{3_13_23_33_4}&=&Q_{3_13_2}\leftrightarrow Q_{3_33_4}=Q_{3_13_3}\leftrightarrow Q_{3_23_4}\nn\\&=&Q_{3_13_4}\leftrightarrow Q_{3_23_3}\nn
\ea
such that no confusion can arise about the meaning of $Q_{3_13_23_33_4}$ although there are no individuals $Q_{3_a}$ into which the question could be decomposed. The same holds similarly for any other even number of indices taking the value $3$. Consequently, the candidate set for $\cq_{M_N}$ can be written as
\begin{widetext}
\ba
\tilde{\cq}_{M_N}:=\big\{Q_{\mu_1\cdots\mu_N},\q\mu_a=0,1,2,3, \q a=1,\ldots,N\q\big|\q\text{only even number of indices taking value $3$}\big\}\nn
\ea
\end{widetext}
with an evident meaning of each such question.

Let us count the number of elements within $\tilde{\cq}_{M_N}$.

\begin{lem}
$\tilde{\cq}_{M_N}$ contains $2^{N-1}(2^N+1)-1$ non-trivial questions.
\end{lem}
\begin{proof}
If arbitrary distributions of the values $\mu_a=0,1,2,3$ over the $a=1,\ldots,N$ were permitted, we would obtain $4^N-1$ non-trivial questions upon subtracting $Q_{0_10_2\cdots0_N}$ as in the qubit case. In order to obtain the number of questions within $\tilde{\cq}_{M_N}$ we thus still have to subtract all the possible ways of distributing an odd number of $3$'s over the $N$ indices. There are precisely
\ba
O_N:=&&\binom{N}{1}\,3^{N-1}+\binom{N}{3}\,3^{N-3}\nn\\&&\q\q\q\q+\cdots+\binom{N}{2n+1}\,3^{N-(2n+1)}\nn
\ea
such ways, where $2n+1$ is the largest odd number smaller or equal to $N$. Similarly, the number of ways an even number of $3$'s can be distributed over $N$ indices is given by
\ba
E_N:=3^N+\binom{N}{2}\,3^{N-2}+\cdots+\binom{N}{2m}\,3^{N-2m},\nn
\ea
where $2m$ is the largest even number smaller or equal to $N$. We then have
\ba
E_N+O_N&=&(3+1)^N=4^N,\nn\\ E_N-O_N&=&(3-1)^N=2^N\nn
\ea
and thus
\ba
O_N=\f{1}{2}(4^N-2^N)\nn
\ea
which yields
\ba
4^N-1-O_N=2^{N-1}(2^N+1)-1\nn
\ea
non-trivial questions in $\tilde{\cq}_{M_N}$.
\end{proof}

We note that the number of parameters in a symmetric matrix with trace equal to $1$ on a $2^N$-dimensional real Hilbert space (i.e.\ rebit density matrix) is precisely $\f{1}{2}\,2^N(2^N+1)-1$.

Next, we assert pairwise independence as required for an informationally complete set.
\begin{lem}
The $2^{N-1}(2^N+1)-1$ non-trivial questions in $\tilde{\cq}_{M_N}$ are pairwise independent.
\end{lem}

\begin{proof}
The proof is entirely analogous to the proofs of lemmas \ref{lem2}, \ref{lem8}, \ref{lem9} and \ref{lem_Nq0}.
\end{proof}

Likewise, the complementarity and compatibility structure of $\tilde{\cq}_{M_N}$ is analogous to the qubit case.

\begin{lem}\label{lem_re20}
 $Q_{\mu_1\cdots\mu_N},Q_{\nu_1\cdots\nu_N}\in\tilde{\cq}_{M_N}$ are
 \begin{description}
 \item[maximally compatible] if the index sets $\{\mu_1,\ldots,\mu_N\}$ and $\{\nu_1,\ldots,\nu_N\}$ differ by an {\bf even} number (incl.\ $0$) of non-zero indices, and
 \item[maximally complementary] if the index sets $\{\mu_1,\ldots,\mu_N\}$ and $\{\nu_1,\ldots,\nu_N\}$ differ by an {\bf odd} number of non-zero indices.
  \end{description}
 \end{lem}
 
\begin{proof}
Thanks to lemmas \ref{lem4}, \ref{lem_re1}--\ref{lem_re3} and \ref{lem_re5}, the proof of lemma \ref{lem_Nq} also applies to the rebit case with the sole difference that only an even number of indices in the questions can take the value $3$ and that correlations of correlations $Q_{3_a3_b}$ cannot be decomposed into individuals $Q_{3_a},Q_{3_b}$.
\end{proof}

Finally, $\tilde{\cq}_{M_N}$ is indeed logically closed.

\begin{Theorem}\label{thm_rebit3}{\bf(Rebits)}
$\tilde{\cq}_{M_N}$ is logically closed under $\leftrightarrow$ and is thus an informationally complete set $\tilde{\cq}_{M_N}=\cq_{M_N}$ with $D_N=2^{N-1}(2^N+1)-1$ for the case $D_1=2$.
\end{Theorem}

\begin{proof}
Thanks to theorems \ref{thm_rebit}, \ref{thm_rebit2} and lemma \ref{lem_re20}, the proof of theorem \ref{thm_qubit3} also applies here, except that only an even number of indices can take the value $3$ and $Q_{3_a3_b}$ cannot be decomposed into individuals $Q_{3_a},Q_{3_b}$.
\end{proof}

We close with the observation that similarly to the $N$ qubit case in section \ref{sec_nqbits}, one could represent the compatibility and complementarity structure via a simplicial question graph. Maximally entangled states (of maximal information) will correspond to $O$ spending all $N$ available \texttt{bits} over a mutually independent set of $N$ $N$-partite questions which is maximally complementary to every individual question. In fact, the prescription for constructing a maximally $N$-partite entangled qubit state provided at the end of section \ref{sec_nqbits} also applies to $N$ rebits since only indices with values $1,2$ were employed. Furthermore, for rebits one can similarly generate maximally entangled states of non-maximal information; e.g., $O$ could ask only the $N-1$ questions $Q_{3_13_2000\cdots},Q_{03_23_300\cdots},Q_{003_33_40\cdots},\ldots,Q_{0\cdots03_{N-1}3_N}$. If $O$ has maximal information about each such question, every rebit pair will be maximally entangled, while $O$ has still not reached the information limit (see also section \ref{sec_rebtang}).

\section{Information measure, time evolution and state space}\label{sec_time}\label{sec_infomeasure}

Thus far we have only applied rules \ref{lim} and \ref{unlim} to derive the basic question structure on $\cq_N$. We shall now slightly switch topic and return to the problems of time evolution, begun in subsection \ref{subsec_time}, and of explicitly {\it quantifying} the information which $O$ has acquired about $S$ by means of interrogations with questions. The information measure and time evolution of $S$'s state in between interrogations are intertwined through rule \ref{pres} of information preservation and rule \ref{time} of maximality of time evolution such that we have to discuss these topics together.
%
%
%

In (\ref{totalinfo}) we have implicitly defined the total amount of $O$'s information about $S$ -- once the latter is in a state $\vec{y}_{O\rightarrow S}$ -- as
\ba
I_{O\rightarrow S}(\vec{y}_{O\rightarrow S})=\sum_{i=1}^{D_N}\,\alpha_i(\vec{y}_{O\rightarrow S}),\nn
\ea
where $\alpha_i$ quantifies $O$'s information about the outcome of question $Q_i\in\cq_{M_N}$ and satisfies the bounds (\ref{infoineq}). We shall begin in subsection \ref{sec_shannon} by clarifying that the information measure we seek to derive here is conceptually distinct from the standard Shannon entropy, before imposing elementary consistency conditions on the relation $\alpha_i(\vec{y}_{O\rightarrow S})$ in subsection \ref{sec_elcond}. Subsequently, we use these conditions and implement rules \ref{pres} and \ref{time} to establish that the set of possible time evolutions defines a group and, finally, to derive the explicit functional form of $I_{O\rightarrow S}$ in subsection \ref{sec_infomeasure}. This discussion will also unravel properties of state spaces and lead to the notion of pure states.

\subsection{Why the Shannon entropy does not apply here}\label{sec_shannon}

Consider a random variable experiment with distinct outcomes each of which is described by a particular probability. Outcomes with a smaller likelihood are more informative (relative to the prior information) than those with a larger likelihood. The Shannon information quantifies the information an observer gains, {\it on average}, via a specific outcome if she repeated the experiment many times. (Equivalently, it quantifies the uncertainty of the observer before an outcome of the experiment occurs.) The probability distribution over the different outcomes is assumed to be known. In particular, none of the outcomes in any run of the experiment will lead to an update of the probability distribution. 

Here, by contrast, we are {\it not} interested in how informative one specific question outcome is relative to another and how much information $O$ would gain, on average, with a specific answer if he repeated the interrogation with a fixed question many times on identically prepared systems. Instead, we wish to quantify $O$'s {\it prior information} (e.g., gained from previous interrogations) about the possible question outcomes on {\it only} the next system to be interrogated. The role of the probability distribution over the various outcomes (to all questions) is here assumed by the state which $O$ ascribes to $S$. This state is not assumed to be `known' absolutely, instead, as discussed in section \ref{sec_infland}, it is defined only relative to the observer and represents $O$'s `catalogue of knowledge' or `degree of belief' about $S$. In particular, this state is updated after any interrogation if the outcome has led to an information gain. This update takes the form of a `collapse' of the prior into a posterior state of s single system in a single shot interrogation, while it yields an update of the ensemble state in a multiple shot interrogation (see subsection \ref{sec_bayes}). The probability distribution represented by the state is thus a prior distribution for the next interrogation and only valid until the next answer -- unless the latter yields no information gain. We thus seek to quantify {\it the information content in the state of $S$ relative to $O$}. This is {\it not} equivalent to the sum of the average information gains, relative to that same state used as a fixed prior probability distribution, which are provided by specific answers to the questions in an informationally complete set over {\it repeated} interrogations. 

The reader should thus not be surprised to find that the end result of the below derivation will not yield the Shannon entropy. We emphasize, however, that this clearly does not invalidate the use of the Shannon entropy (in the more general form of the von Neumann entropy) in quantum theory as long as one employs it what it is designed for: to describe average information gains in repeated experiments on identically prepared systems -- relative to a prior state which is assigned {\it before} the repeated experiments are carried out and which represents the `known' probability distribution.

\subsection{Elementary conditions on the measure}\label{sec_elcond}

There are a few natural requirements on the relation between $\alpha_i$ and $\vec{y}_{O\rightarrow S}$: 
\begin{itemize}
\item[(i)] The $Q_i\in\cq_{M_N}$ are pairwise independent. Accordingly, $\alpha_i$ should {\it not} depend on the `yes'-probabilities $y_{j\neq i}$ of other questions $Q_{j\neq i}$ such that $\alpha_i=\alpha_i(y_i)$.
\item[(ii)] All $Q_i\in\cq_{M_N}$ are informationally of equivalent status. The functional relation between $y_i$ and $\alpha_i$ should be the same for all $i$: $\alpha_i=\alpha(y_i)$, $i=1,\ldots,D_N$.
\item[(iii)] If $O$ has no information about the outcome of $Q_i$, i.e.\ $y_i=1/2$, then $\alpha_i=0$ \texttt{bit}.
\item[(iv)] If $O$ has maximal information about the outcome of $Q_i$, i.e.\ $y_i=1$ or $y_i=0$, then $\alpha_i=1$ \texttt{bit}; both possible answers give $1\,\texttt{bit}$ of information.
\item[(v)] The assignment of which answer to $Q_i$ is `yes' and which is `no' is arbitrary and the functional relation between $\alpha_i$ and $y_i$ (or $n_i$) should {\it not} depend on this choice. Hence, $\alpha(y_i)=\alpha(n_i)$ must be symmetric around $y_i=1/2$ (see also (iv)); $\alpha_i$ quantifies the amount of information about $Q_i$, but does not encode {\it what} the answer to $Q_i$ is.
\item[(vi)] On the interval $y_i\in(1/2,1]$, the relation between $\alpha$ and $y_i$ should be monotonically {\it increasing} such that $O$'s information about the answer to $Q_i$ is quantified as higher, the higher the assigned probability for a `yes' outcome. Likewise, on $[0,1/2)$, $\alpha(y_i)$ must be monotonically {\it decreasing}. In particular, $\alpha$ shall be continuous and strictly convex; i.e., for all $y^1_i,y^2_i\in[0,1]$ and every $\lambda\in[0,1]$
\ba
\!\!\!\!\!\alpha(\lambda\,y^1_i+(1-\lambda)\,y_i^2)\leq\lambda\,\alpha(y_i^1)+(1-\lambda)\,\alpha(y_i^2)\nn
\ea
and the inequality shall be strict whenever $y^1_i\neq\,y^2_i$ and $0<\lambda<1$. Hence, $y_i=\f{1}{2}$ is the global minimum of $\alpha$ on $[0,1]$.
\end{itemize}

Since now
\ba
I_{O\rightarrow S}(\vec{y}_{O\rightarrow S})=\sum_{i=1}^{D_N}\,\alpha(y_i)\label{eqn_newI}
\ea
and each summand is strictly convex, we also have that $I_{O\rightarrow S}$ is a strictly convex function on $\Sigma_N$. Accordingly, in the coin flip scenario of subsection \ref{subsec_time}, $O$'s information about the mixed state $\vec{y}_{O\rightarrow S_{12}}=\lambda\,\vec{y}_{O\rightarrow S_1}+(1-\lambda)\,\vec{y}_{O\rightarrow S_2}$ is \emph{smaller} than his maximal information about any of its constituents $\vec{y}_{O\rightarrow S_1},\vec{y}_{O\rightarrow S_2}$ -- unless the outcome of the coin flip was certain, or $S_1,S_2$ are in the same state. This is consistent with the fact that
$O$'s information about the outcome of any question (asked to either $S_1$ or $S_2$, depending on the outcome of the coin flip) is clouded by the random coin flip outcome.


\subsection{The state of no information is an interior state}\label{ssec_interior}

In order to discuss the action of time evolution on the state space $\Sigma_N$, we have to derive a few further properties of the latter, some of which depend on the properties of $I_{O\rightarrow S}$.

Firstly, since $D_N$ is finite for finite $N$, $\Sigma_N$ is finite dimensional in this case. In fact, $\Sigma_N$ contains a basis of $\mathbb{R}^{D_N}$\footnote{Each of the following $D_N$-dimensional vectors $(1,\f{1}{2},\f{1}{2},\cdots,\f{1}{2}), (\f{1}{2},1,\f{1}{2},\cdots,\f{1}{2}),\cdots,(\f{1}{2},\cdots,\f{1}{2},1)$ represents a legal state $\vec{y}_{O\rightarrow S}$, corresponding to $O$ only knowing the answer to precisely one question from $\cq_{M_N}$ with certainty and `nothing else'.} so that it is a $D_N$-dimensional closed convex subset of $\mathbb{R}^{D_N}$. Furthermore, since $y_i\in[0,1]$ for all $y_i$ in $\vec{y}_{O\rightarrow S}$, $\Sigma_N$ is clearly bounded and thus, by the Heine-Borel-theorem, compact. As a compact convex set with non-empty interior, $\Sigma_N$ is thus a convex body and, in particular, has volume \cite{webster}.

Next, we focus on the state of no information $\vec{y}_{O\rightarrow S}=\f{1}{2}\,\vec{1}$. Given strict convexity of $I_{O\rightarrow S}$ on $\Sigma_N$, the conditions of the previous subsection entail that the state of no information is the {\it global minimum of} $I_{O\rightarrow S}$ with $I_{O\rightarrow S}(\f{1}{2}\,\vec{1})=0$ \texttt{bits} -- in agreement with its uniqueness. It is also an interior state.

\begin{lem}\label{lem_interior}
The state of no information lies in the interior of $\Sigma_N$.
\end{lem}
\begin{proof}
The $Q_i\in\cq_{M_N}$ are pairwise independent so for each such $Q_i$ there exist two states $\vec{y}\,^{y_i}_{O\rightarrow S}=(\f{1}{2},\cdots,\f{1}{2},1,\f{1}{2},\cdots,\f{1}{2})$ and $\vec{y}\,^{n_i}_{O\rightarrow S}=(\f{1}{2},\cdots,\f{1}{2},0,\f{1}{2},\cdots,\f{1}{2})$ corresponding to $O$ having {\it only} asked $Q_i$ to $S$ in the state of no information and having received the answers $Q_i=$ `yes' and $Q_i=$ `no', respectively. Since $\Sigma_N$ is convex, it also contains all convex mixtures of these special states. The $D_N$ $\vec{y}\,^{y_i}_{O\rightarrow S}$, as well as the $D_N$ $\vec{y}\,^{n_i}_{O\rightarrow S}$ each define a basis of $R^{D_N}$ such that the set of all their convex mixtures define a $D_N$-dimensional convex polytope contained in $\Sigma_N\subset\mathbb{R}^{D_N}$. In particular, given that $\f{1}{2}\,\vec{y}\,^{y_i}_{O\rightarrow S}+\f{1}{2}\,\vec{y}\,^{n_i}_{O\rightarrow S}=\f{1}{2}\,\vec{1}$ yields the state of no information for each $i$, it is clear that it lies in the interior of the convex polytope.
\end{proof}

Given that time evolution preserves the total information by rule \ref{pres}, we shall also be interested in the level sets of $I_{O\rightarrow S}$.

\begin{lem}\label{lem_level}
Let $I_{O\rightarrow S}$ fulfill conditions (i)--(vi) of section \ref{sec_elcond}. For sufficiently small $\varepsilon>0$, the level set 
\ba
L_\varepsilon:=\{\vec{y}_{O\rightarrow S}\in\Sigma_N|I_{O\rightarrow S}(\vec{y}_{O\rightarrow S})=\varepsilon\}\label{levelset}
\ea
lies in the interior of $\Sigma_N$ and is homeomorphic to $S^{D_N-1}$. In this case, $L_\varepsilon$ constitutes the full boundary of the fat level set $L_{I\leq\varepsilon}:=\{\vec{y}_{O\rightarrow S}\in\Sigma_N|I_{O\rightarrow S}(\vec{y}_{O\rightarrow S})\leq\varepsilon\}$ which contains the state of no information and is homeomorphic to the closed unit ball in $\mathbb{R}^{D_N}$.
\end{lem}

\begin{proof}
The state of no information is the global minimum of $I_{O\rightarrow S}$ with $I_{O\rightarrow S}(\f{1}{2}\,\vec{1})=0$ and, thanks to lemma \ref{lem_interior}, lies in the interior of $\Sigma_N$. Furthermore, condition (vi) of subsection \ref{sec_elcond} requires that $\alpha(y_i)$ increases monotonically away from $y_i=\f{1}{2}$ for each $i$ so that $I_{O\rightarrow S}$ in (\ref{eqn_newI}) must likewise increase monotonically in each direction away from $\vec{y}_{O\rightarrow S}=\f{1}{2}\,\vec{1}$ in $\mathbb{R}^{D_N}$. Owing to the continuity of $I_{O\rightarrow S}$, we can therefore find a sufficiently small $\varepsilon>0$ such that any state attained by moving away from the state of no information in {\it any} direction until reaching $I_{O\rightarrow S}=\varepsilon$ still lies in the interior of $\Sigma_N$. Accordingly, for such $\varepsilon>0$, $L_\varepsilon$ must lie in the interior of $\Sigma_N$ and so must $L_{I\leq\varepsilon}$ with the state of no information in its own interior. Given that $I_{O\rightarrow S}:\Sigma_N\rightarrow \mathbb{R}$ is continuous and $\Sigma_N$ is closed, $I_{O\rightarrow S}$ is a closed convex function. Closed convex functions have closed convex fat level sets \cite{webster} so that $L_{I\leq\varepsilon}$ is a closed convex subset in the interior of $\Sigma_N$. Clearly, $L_{I\leq\varepsilon}$ is $D_N$-dimensional since we moved away in all directions from $\vec{y}_{O\rightarrow S}=\f{1}{2}\,\vec{1}$ to reach $L_\varepsilon$ 
and the states with $I_{O\rightarrow S}=\varepsilon$ define the limit points of $L_{I\leq\varepsilon}$ in each direction so that $L_\varepsilon$ constitutes its full boundary. Any $D_N$-dimensional closed convex subset of $\mathbb{R}^{D_N}$ with non-empty interior is homeomorphic to the closed unit ball and its boundary is homeomorphic to $S^{D_N-1}$. 
\end{proof}

These results will become important for establishing that the set of time evolutions is a group.

\vspace*{.2cm}
\subsection{Time evolution of the `Bloch vector'}


In subsection \ref{subsec_time}, we have only considered the time evolution of the redundantly parametrized $2D_N$-dimensional state $\vec{P}_{O\rightarrow S}$. According to (\ref{linear}), this time evolution is linear, described by a matrix $A(\Delta t)$, and only depends on the interval $\Delta t=t_2-t_1$, where $t_1,t_2$ are arbitrary instants of time in between $O$'s interrogations on a given system.
However, because of (\ref{norm}), it clearly suffices to consider the yes-vector $\vec{y}_{O\rightarrow S}$ (or no-vector $\vec{n}_{O\rightarrow S}$) alone to describe the state of $S$ relative to $O$. Let us therefore determine, how $\vec{y}_{O\rightarrow S}$ and $\vec{n}_{O\rightarrow S}$ evolve under time evolution. To this end, we decompose $A(\Delta t)$ and $\vec{P}_{O\rightarrow S}$ in (\ref{linear}),
\ba
\!\!\!\!\left(\begin{array}{c}\vec{y}_{O\rightarrow S}(\Delta t) \\\vec{n}_{O\rightarrow S}(\Delta t)\end{array}\right)\!=\!\left(\begin{array}{cc}a(\Delta t) & b(\Delta t) \\c(\Delta t) & d(\Delta t)\end{array}\right)\!\!\left(\begin{array}{c}\vec{y}_{O\rightarrow S}(0) \\\vec{n}_{O\rightarrow S}(0)\end{array}\right),\nn
\ea
\vspace*{.15cm}

\noindent where $a(\Delta t), b(\Delta t),c(\Delta t),d(\Delta t)$ are nonnegative (real) $D_N\times D_N$ matrices. Using the normalization (\ref{norm}) (with $p=1$), one finds that both $\vec{y}_{O\rightarrow S},\vec{n}_{O\rightarrow S}$ evolve {\it affinely},
\begin{widetext}
\ba
\!\!\!\!\!\!\vec{y}_{O\rightarrow S}(\Delta t)&=&\left[a(\Delta t)-b(\Delta t)\right]\vec{y}_{O\rightarrow S}(0)
+b(\Delta t)\,\vec{1},\nn\\
\!\!\!\!\!\!\vec{n}_{O\rightarrow S}(\Delta t)&=&\left[d(\Delta t)-c(\Delta t)\right]\vec{n}_{O\rightarrow S}(0)
+c(\Delta t)\,\vec{1}.\label{afft}
\ea
\end{widetext}
We shall now determine relations among the four matrices $a,b,c,d$. Firstly, the normalization (\ref{norm}) must hold for all initial states and all times. Hence, for all $\Delta t\in\mathbb{R}$,
\begin{widetext}
\ba
\vec{1}&=&\vec{y}_{O\rightarrow S}(\Delta t)+\vec{n}_{O\rightarrow S}(\Delta t)\nn\\
&=&\left[a(\Delta t)-b(\Delta t)\right]\vec{y}_{O\rightarrow S}(0)+\left[d(\Delta t)-c(\Delta t)\right]\vec{n}_{O\rightarrow S}(0)+\left[b(\Delta t)+c(\Delta t)\right]\,\vec{1}\nn\\
&\underset{\tiny(\ref{norm})}{=}&\left[a(\Delta t)-b(\Delta t)+c(\Delta t)-d(\Delta t)\right]\vec{y}_{O\rightarrow S}(0)+\left[b(\Delta t)+d(\Delta t)\right]\,\vec{1}.\nn
\ea
\end{widetext}
Since the set of all possible initial states $\vec{y}_{O\rightarrow S}(0)$ contains a basis of $\mathbb{R}^{D_N}$, we must conclude
\ba
a(\Delta t)-b(\Delta t)&=&d(\Delta t)-c(\Delta t),\nn\\ \left[b(\Delta t)+d(\Delta t)\right]\,\vec{1}&=&\vec{1}.\label{wkh}
\ea
Next, we note that, by assumption \ref{assump3}, the state of no information is {\it unique} and, by rule \ref{pres}, clearly preserved under time evolution. Inserting the state of no information, e.g., as $\vec{n}_{O\rightarrow S}(\Delta t)=\f{1}{2}\,\vec{1}$ into (\ref{afft}), this entails
\ba
\left[c(\Delta t)+d(\Delta t)\right]\,\vec{1}=\vec{1}.\label{wkh2}
\ea
In conjunction, (\ref{wkh}, \ref{wkh2}) therefore imply\footnote{As an aside, note that $A(\Delta t)$ can thus be replaced by
\ba
A(\Delta t)=\left(\begin{array}{cc}a(\Delta t) & b(\Delta t) \\b(\Delta t) & a(\Delta t)\end{array}\right),\nn
\ea
so that time evolution is symmetric under a swap of the `yes/no'-labeling for all $Q\in\cq_{M_N}$ at once, i.e.\ under $\vec{y}_{O\rightarrow S}\leftrightarrow\vec{n}_{O\rightarrow S}$. }
\ba
\vec{y}_{O\rightarrow S}(\Delta t)&=&\left[a(\Delta t)-b(\Delta t)\right]\vec{y}_{O\rightarrow S}(0)+b(\Delta t)\,\vec{1},\nn\\
\vec{n}_{O\rightarrow S}(\Delta t)&=&\left[a(\Delta t)-b(\Delta t)\right]\vec{n}_{O\rightarrow S}(0)+b(\Delta t)\,\vec{1}.\nn
\ea
Consequently, defining the $D_N\times D_N$ time evolution matrix
\ba
T(\Delta t):=a(\Delta t)-b(\Delta t),\label{eqntdef}
\ea
yields a {\it linear} time evolution of what we shall henceforth call the {\it generalized Bloch vector} $2\,\vec{y}_{O\rightarrow S}-\vec{1}$:
\ba
\!\!\!\!\!\!\!\!\!\!\!2\,\vec{y}_{O\rightarrow S}(\Delta t)-\vec{1}&=&\vec{y}_{O\rightarrow S}(\Delta t)-\vec{n}_{O\rightarrow S}(\Delta t)\nn\\&=&T(\Delta t)\left(2\,\vec{y}_{O\rightarrow S}(0)-\vec{1}\right).\label{blochtime}
\ea
Notice that $T(\Delta t)$ need neither be nonnegative nor stochastic in any pair of its components (in contrast to $A(\Delta t)$). Henceforth, we shall only consider the $T(\Delta t)$ governing the time evolution of the Bloch vector 
\ba
\vec{r}_{O\rightarrow S}:=2\,\vec{y}_{O\rightarrow S}-\vec{1}
\ea
and often employ the latter to parametrize states.

\subsection{Time evolution is injective}

We continue by demonstrating that time evolution of $S$'s states must be injective on $\Sigma_N$. In subsection \ref{sec_surj}, we will use this result to prove that time evolution is also surjective on $\Sigma_N$ and thus, in fact, {\it reversible}.


\begin{lem}\label{lem_injective}
Let $T(\Delta t):\Sigma_N\rightarrow\Sigma_N$ be a time evolution as given in (\ref{eqntdef}) for any $\Delta t\in\mathbb{R}$. If $I_{O\rightarrow S}$ is strictly convex, rule \ref{pres} entails that $T(\Delta t)$ is injective.
\end{lem}
\begin{proof}
Assume $T(\Delta t)$ was not injective. Then there would exist states $\vec{y}\,'_{O\rightarrow S}(t_1)\neq\vec{y}\,''_{O\rightarrow S}(t_1)$ such that
\ba
\vec{y}_{O\rightarrow S}(t_2)&:=&T(\Delta t)\vec{y}\,'_{O\rightarrow S}(t_1)\nn\\&=&T(\Delta t)\vec{y}\,''_{O\rightarrow S}(t_1).\label{noninv}
\ea
Now consider the coin flip scenario of section \ref{subsec_time}. $O$ can prepare $S_1$ in the state $\vec{y}_{O\rightarrow S_1}(t_1):=\vec{y}\,'_{O\rightarrow S}(t_1)$ and $S_2$ in the state $\vec{y}_{O\rightarrow S_2}(t_1):=\vec{y}\,''_{O\rightarrow S}(t_1)$ at time $t_1$ before tossing the coin. By (\ref{convex}),
\ba
\vec{y}_{O\rightarrow S_{12}}(t_1)=\lambda\,\vec{y}_{O\rightarrow S_1}(t_1)+(1-\lambda)\,\vec{y}_{O\rightarrow S_2}(t_1),\nn
\ea
and, on account of (\ref{noninv}) and rule \ref{pres}, 
\ba
\!\!\!\!\!\!\!\!\!\!\!\!I_{O\rightarrow S_1}(\vec{y}_{O\rightarrow S_1}(t_1))&=&I_{O\rightarrow S_2}(\vec{y}_{O\rightarrow S_2}(t_1))\nn\\&=&I_{O\rightarrow S_{1}}(\vec{y}_{O\rightarrow S_1}(t_2))\nn\\&=&I_{O\rightarrow S_{2}}(\vec{y}_{O\rightarrow S_2}(t_2)).\label{noninv2}
\ea

Thus, using (\ref{noninv2}) and strict convexity of $I_{O\rightarrow S}$, this yields for a coin flip with $0<\lambda<1$ 
\begin{widetext}
\ba
I_{O\rightarrow S_{12}}(\vec{y}_{O\rightarrow S_{12}}(t_1))<
I_{O\rightarrow S_{1}}(\vec{y}_{O\rightarrow S_1}(t_2))=I_{O\rightarrow S_{2}}(\vec{y}_{O\rightarrow S_2}(t_2)).\label{incr1}
\ea
\end{widetext}
On the other hand, (\ref{noninv}) implies that, at time $t_2$, $O$ would find
\ba
\vec{y}_{O\rightarrow S_{12}}(t_2)=\vec{y}_{O\rightarrow S_1}(t_2)=\vec{y}_{O\rightarrow S_2}(t_2)\nn
\ea
such that 
\ba
\!\!\!\!\!\!\!\!\!\!I_{O\rightarrow S_{12}}(\vec{y}_{O\rightarrow S_{12}}(t_2))&=&I_{O\rightarrow S_{1}}(\vec{y}_{O\rightarrow S_1}(t_2))\nn\\&=&I_{O\rightarrow S_{2}}(\vec{y}_{O\rightarrow S_2}(t_2)).\label{incr2}
\ea
Hence, (\ref{incr1}, \ref{incr2}) entail that $O$'s information about $S_{12}$ has increased between $t_1$ and $t_2$, despite not having tossed the coin and asked any questions. This is in contradiction with rule \ref{pres}. We conclude that $T(\Delta t)$ must be injective. 
\end{proof} 

Since $\Sigma_N$ contains a basis of $\mathbb{R}^{D_N}$ (see subsection \ref{ssec_interior}) it is clear that $T(\Delta t)$ is also injective on $\mathbb{R}^{D_N}$. For finite dimensional square matrices, $T(\Delta t)$ being injective is equivalent to it being bijective on $\mathbb{R}^{D_N}$. Hence, to every $T(\Delta t)$ there must exist an inverse matrix $T^{-1}(\Delta t)$ such that $T(\Delta t)T^{-1}(\Delta t)=T^{-1}(\Delta t)T(\Delta t)=\mathbb{1}$.

However, is this $T^{-1}(\Delta t)$ also a legal time evolution?\footnote{I thank a referee for pointing out that addressing this question was overlooked in a previous draft version.} This would be the case if $T(\Delta t)$ was also surjective and thereby reversible {\it on} $\Sigma_N\subset\mathbb{R}^{D_N}$. Indeed, we shall demonstrate this property next.

\subsection{Time evolution is also reversible}\label{sec_surj}

To establish reversibility of time evolution we shall resort to tools from topology and convex sets.

\begin{lem}\label{lem_clopen}
Restricting any time evolution $T(\Delta t):\Sigma_N\rightarrow \Sigma_N$ to an interior level set $L_\varepsilon$ (\ref{levelset}) defines a continuous, open and closed map from $L_\varepsilon$ to itself.
\end{lem}
\begin{proof}
As an injective square matrix $T(\Delta t)$ defines a continuous invertible map from $\mathbb{R}^{D_N}$ to itself and so the pre-image of any open set in $\mathbb{R}^{D_N}$ is another open set. Rule \ref{pres} implies $[T(\Delta t)](L_\varepsilon)\subseteq L_\varepsilon$ and that no state from outside $L_\varepsilon$ can evolve into $L_\varepsilon$. The open sets of $L_\varepsilon$ in the induced topology are the intersections of $L_\varepsilon$ with open sets of $\mathbb{R}^{D_N}$. Consider any such open set $V\subset L_\varepsilon$ and any open set $U\subset\mathbb{R}^{D_N}$ such that $V=U\cap L_\varepsilon$. The pre-image $[T^{-1}(\Delta t)](U)$ is again open as is therefore $[T^{-1}(\Delta t)](U)\cap L_\varepsilon$
which, thanks to rule \ref{pres}, must be non-empty if $V$ is non-empty. Hence, $T(\Delta t)$ defines a continuous map from $L_\varepsilon$ to itself. For similar reasons, this map is also open. Finally, the closed map lemma states that a continuous map from a compact to a Hausdorff space is also closed. Thanks to lemma \ref{lem_level}, we have $L_\varepsilon\simeq S^{D_N-1}$ which is both compact and Hausdorff and so the map from $L_\varepsilon$ to itself defined by $T(\Delta t)$ is also closed.
\end{proof}

Next, we employ this result to show that time evolution is surjective on interior level sets.

\begin{lem}\label{lem_surjint}
Any allowed time evolution is surjective on the interior level sets, i.e.\ $[T(\Delta t)](L_\varepsilon)=L_\varepsilon$ and thus also $[T(\Delta t)](L_{I\leq\varepsilon})=L_{I\leq\varepsilon}$.
\end{lem}
\begin{proof}
Lemma \ref{lem_clopen} entails that $T(\Delta t)$ defines a continuous, open and closed map from $L_\varepsilon$ to itself. It therefore maps all open sets to open sets and all closed sets to closed sets in $L_\varepsilon$. In particular, it must map all sets that are both open and closed in $L_\varepsilon$ to other sets which are both open and closed. Since $L_\varepsilon\simeq S^{D_N-1}$ is connected, the empty set and the full $L_\varepsilon$ are the only sets in $L_\varepsilon$ which are both closed and open. Now, by lemma \ref{lem_injective}, $T(\Delta t)$ is injective
so that $[T(\Delta t)](L_\varepsilon)=L_\varepsilon$.
\end{proof}

The last result is sufficient to finally establish that time evolution is surjective also on the state space.
\begin{Theorem}\label{thm_surj}
Any allowed time evolution is surjective on $\Sigma_N$, i.e.\ $[T(\Delta t)](\Sigma_N)=\Sigma_N$, and has $|\det T(\Delta t)|=1$.
\end{Theorem}
\begin{proof}
Given that an interior fat level set $L_{I\leq\varepsilon}$ is, by lemma \ref{lem_level}, homeomorphic to the closed unit ball in $\mathbb{R}^{D_N}$ it is a convex body and has volume \cite{webster}. Lemma \ref{lem_surjint} yields $[T(\Delta t)](L_{I\leq\varepsilon})=L_{I\leq\varepsilon}$ and so its volume is left invariant by $T(\Delta t)$ which, according to (\ref{blochtime}), acts linearly on states $\vec{r}_{O\rightarrow S}$ in the Bloch vector parametrization. Theorem 6.2.14 in \cite{webster} then implies that $|\det T(\Delta t)|=1$ and therefore also that the volume of $[T(\Delta t)](\Sigma_N)$ is equal to the volume of $\Sigma_N$. Consider now the volume difference
\begin{widetext}
\ba
\delta\Big([T(\Delta t)](\Sigma_N),\Sigma_N\Big):=\rm{vol}\Big([T(\Delta t)](\Sigma_N)\cup\Sigma_N\Big)-\rm{vol}\Big([T(\Delta t)](\Sigma_N)\cap\Sigma_N\Big).\nn
\ea
\end{widetext}
Clearly, $[T(\Delta t)](\Sigma_N)\subseteq\Sigma_N$, and so $\delta\big([T(\Delta t)](\Sigma_N),\Sigma_N\big)=\rm{vol}\left(\Sigma_N\right)-\rm{vol}\big([T(\Delta t)](\Sigma_N)\big)=0$. Exercise 6.2.2 in \cite{webster} then shows that $\delta\big([T(\Delta t)](\Sigma_N),\Sigma_N\big)=0$ is only possible if $[T(\Delta t)](\Sigma_N)=\Sigma_N$.
\end{proof}

Hence, $T^{-1}(\Delta t)$ maps all legal states to legal states and must, by rule \ref{time}, be legal also.

\begin{Corollary}
Any time evolution permitted by the rules is \emph{reversible}.
\end{Corollary}

\begin{widetext}
\subsection{Time evolution defines a group}

Given that any time interval can be decomposed into two time intervals, $\Delta t=\Delta t_1+\Delta t_2$, and the evolution of $\vec{r}_{O\rightarrow S}=2\,\vec{y}_{O\rightarrow S}-\vec{1}$ is \emph{continuous} by rule \ref{time}, $O$ must find
\ba
T(\Delta t_2)\,T(\Delta t_1)\,\vec{r}_{O\rightarrow S}(0)=T(\Delta t_2)\,\vec{r}_{O\rightarrow S}(\Delta t_1)=\vec{r}_{O\rightarrow S}(\Delta t)=T(\Delta t)\,\vec{r}_{O\rightarrow S}(0),\nn
\ea
\end{widetext}
and thus that multiplication is {\it abelian},
\ba
T(\Delta t_1+\Delta t_2)=T(\Delta t_2)\,T(\Delta t_1)=T(\Delta t_1)\,T(\Delta t_2).\nn
\ea
(The last equality follows from time translation invariance.) From the last equation and $T(0)=\mathbb{1}$, using $\Delta t-\Delta t=0$, we can also infer that $T^{-1}(\Delta t)=T(-\Delta t)$. If we permit $O$ to consider the time evolution of $S$ for {\it any} duration $\Delta t\in\mathbb{R}$, it follows that the product of any two time evolution matrices is again a time evolution matrix. In summary, we therefore gather that a {\it given} time evolution, as perceived by $O$ and under the circumstances to which he has subjected $S$, is such that: 
\begin{itemize}
\item[(i)] $T(0)=\mathbb{1}$, 
\item[(ii)] to every $T(\Delta t)$ there exists an inverse $T^{-1}(\Delta t)=T(-\Delta t)$, 
\item[(iii)] the multiplication of any two time evolution matrices is again a time evolution matrix, and 
\item[(iv)] matrix multiplication is obviously associative. 
\end{itemize}
In conclusion, under the assumptions of section \ref{sec_infland} and rules \ref{pres} and \ref{time}, a {given} time evolution defines therefore an abelian, one-parameter matrix group. 
Hence, a given time evolution is described by a single evolution generator and single parameter $\Delta t$ parametrizing the duration. We note, however, that a multiplicity of time evolutions of $S$ is possible, depending on the physical circumstances (interactions) to which $O$ may subject $S$. Different time evolutions will be generated by different generators, but each time evolution will form a one-parameter group as discussed above. 

In fact, the full set of time evolutions $\ct_N$ which $O$ is able to implement must likewise be a group. Namely, if $T(\Delta t),T'(\Delta t')$ correspond to two distinct interactions, then rule \ref{time} implies that also $T(\Delta t)\cdot T'(\Delta t')$ is a legal time evolution for any state and since both $T,T'$ are invertible their full set must be a group. This implies, in particular, that $\ct_N$ is a group. We shall return to this further below and in \cite{hw}.

\subsection{The squared length of the Bloch vector as information measure}\label{sec_infomeasure}

We now have sufficient structure in our hand to determine the functional relation between $\alpha_i$ and $y_i$. Given that the generalized Bloch vector $2\,\vec{y}_{O\rightarrow S}-\vec{1}$ transforms nicely under time evolution (\ref{blochtime}), it is useful to parametrize $\alpha_i$ by $2y_i-1$, i.e.\ $\alpha_i=\alpha(2y_i-1)$. Rule \ref{pres} entails that $O$'s total information about (an otherwise non-interacting) $S$ is a `conserved charge' of time evolution
\ba
I_{O\rightarrow S}(\vec{y}_{O\rightarrow S}(\Delta t))=I_{O\rightarrow S}(\vec{y}_{O\rightarrow S}(0))\nn
\ea
which translates into the condition
\begin{widetext}
\ba
I_{O\rightarrow S}\left(T(\Delta t)\left(2\,\vec{y}_{O\rightarrow S}(0)-\vec{1}\right)\right)&=&\sum_{i=1}^{D_N}\alpha\left(\sum_{j=1}^{D_N}T_{ij}(\Delta t)\left(2y_j(0)-1\right)\right)\nn\\
&=&\sum_{i=1}^{D_N}\alpha\left(2y_i(0)-1\right)=I_{O\rightarrow S}(2\,\vec{y}_{O\rightarrow S}(0)-\vec{1}).\label{infocons}
\ea
\end{widetext}
If $T(\Delta t)$ was a permutation matrix, (\ref{infocons}) would hold for {\it any} function $\alpha(2y_i-1)$. For example, $N$ classical bits are governed by the evolution group $\mathbb{Z}_2\times\cdots\times \mathbb{Z}_2$. However, permutations form a discrete group, while in our present case $\{T(\Delta t),\,\Delta t\in \mathbb{R}\}$ constitutes a {\it continuous} one-parameter group. This is where {\it continuity} of time evolution, as asserted by rule \ref{time}, becomes crucial. Under a reasonable assumption on the information measure, we shall now show that {continuity} of time evolution, together with conditions (i)--(vi) of subsection \ref{sec_elcond}, enforces the quadratic relation $\alpha_i=(2y_i-1)^2$. To this end, we once more invoke the coin flip scenario.\footnote{We suspect that this result may be derivable from purely group theoretic arguments {\it without} an operational setup by employing the mathematical fact that to every continuous matrix group acting linearly on some space there corresponds a conserved inner product which is quadratic in the components of the vectors.}


Given our parametrization in terms of the Bloch vector $2\,\vec{y}_{O\rightarrow S}-\vec{1}$, $O$'s information about the outcomes of his questions in the coin flip scenario can be written as follows
\begin{widetext}
\ba
I_{O\rightarrow S_{12}}\left(2\left(\lambda\,\vec{y}_{O\rightarrow S_1}+(1-\lambda)\,\vec{y}_{O\rightarrow S_2}\right)-\vec{1}\right)&=&I_{O\rightarrow S_{12}}\left(\lambda\,(2\,\vec{y}_{O\rightarrow S_1}-\vec{1})+(1-\lambda)(2\,\vec{y}_{O\rightarrow S_2}-\vec{1})\right)\nn\\
&=&\sum_{i=1}^{D_N}\,\alpha\left(\lambda\,(2y_i^1-1)+(1-\lambda)(2y_i^2-1)\right).\nn
\ea
\end{widetext}
It is instructive to consider the case in which $O$ is entirely oblivious about $S_2$ such that the latter is in the state of no information $\vec{y}_{O\rightarrow S_2}=\f{1}{2}\,\vec{1}$ relative to him, but that $O$ has some information about $S_1$. In this case, $2\,\vec{y}_{O\rightarrow S_{12}}-\vec{1}=\lambda(2\,\vec{y}_{O\rightarrow S_1}-\vec{1})$ and (assuming the outcome of the coin flip is not certain) strict convexity of $I_{O\rightarrow S}$ (see subsection \ref{sec_elcond}) implies
\ba
I_{O\rightarrow S_{12}}(\lambda\, (2\,\vec{y}_{O\rightarrow S_1}-\vec{1}))<\lambda\,I_{O\rightarrow S_1}(2\,\vec{y}_{O\rightarrow S_1}-\vec{1})\nn
\ea
or, equivalently,
\ba
I_{O\rightarrow S_{12}}(\lambda\, (2\,\vec{y}_{O\rightarrow S_1}&-&\vec{1}))\\\label{rescale}
=f&\cdot& I_{O\rightarrow S_1}(2\,\vec{y}_{O\rightarrow S_1}-\vec{1}),\nn
\ea
where $f<1$ is a factor parametrizing $O$'s information loss relative to the case in which he does {\it not} toss a coin and, instead, directly asks $S_1$.\footnote{In fact, one can equivalently interpret the situation as follows: $O$ {\it only} considers a system $S_1$ which would be in the state $\vec{y}_{O\rightarrow S_1}$ if it was present with certainty. However, $\lambda$ in the state $\lambda(2\,\vec{y}_{O\rightarrow S_1}-\vec{1})$ represents the probability that $S_1$ `is there' at all. Indeed, in this case, (\ref{norm}) can be written as $\lambda(\vec{y}_{O\rightarrow S_1}+\vec{n}_{O\rightarrow S_1})=\lambda\cdot\vec{1}$ such that $\lambda(\vec{y}_{O\rightarrow S_1}-\vec{n}_{O\rightarrow S_1})=\lambda(2\,\vec{y}_{O\rightarrow S_1}-\vec{1})$.} The reason $O$ experiences such a relative information loss about the outcome of his interrogation is, of course, entirely due to the randomness of the coin flip. But the coin flip is independent of the systems $S_{1,2}$ and, in particular, of the states in which these are relative to $O$; the factor $\lambda$ by which the probabilities $\vec{y}_{O\rightarrow S_1}$ become rescaled is state independent. For that reason, the relative information loss should likewise depend {\it only} on the coin flip, quantified by $\lambda$, and {\it not} on the state $\vec{y}_{O\rightarrow S_1}$. For instance, if we also considered the case that the coin flip was certain, i.e.\ $\lambda=0,1$, then clearly for $\lambda=1$ we must have $f=1$ $\forall\,\vec{y}_{O\rightarrow S_1}\in\Sigma_N$ and for $\lambda=0$ it must hold $f=0$ $\forall\,\vec{y}_{O\rightarrow S_1}\in\Sigma_N$. We shall make this into a requirement on the information measure for all values of $\lambda$:

\begin{Req}
The relative information loss factor $f$ in (\ref{rescale}) is a \emph{state independent} (continuous) function of the coin flip probability $\lambda$ with $f(\lambda)<1$ for $\lambda\in(0,1)$.
\end{Req}

The binary Shannon entropy $H(y)=-y\log y-(1-y)\log(1-y)$ in the role of $\alpha$ would fail this requirement. Namely, the measure $\alpha$ thus factorizes, $\alpha(\lambda\,(2y_i-1))=f(\lambda)\,\alpha(2y_i-1)$, while $H(y)$ would not. Setting $\lambda=\lambda_1\cdot\lambda_2$ yields\footnote{Such a factorization of coin flip probabilities could be achieved, e.g., if $O$ decided to use one coin, with `heads' probability $\lambda_1$, to firstly decide which of two possible convex mixtures to prepare where both possible mixtures are generated with a second coin with `heads' probability $\lambda_2$. If three of the four states within the two mixtures are chosen as the state of no information, one would obtain precisely such an equation.}
\ba
f(\lambda_1\cdot\lambda_2)\,\alpha(2\,y_i-1)&=&\alpha(\lambda_1\cdot\lambda_2\,(2\,y_i-1))\nn\\&=&f(\lambda_1)\,\alpha(\lambda_2(2\,y_i-1))\nn\\&=&f(\lambda_1)f(\lambda_2)\,\alpha(2\,y_i-1)\nn
\ea
and therefore $f(\lambda_1)\cdot f(\lambda_2)=f(\lambda_1\cdot\lambda_2)$, which implies $f(\lambda)=\lambda^p$, for some power $p\in\mathbb{R}$. But then, $\alpha$ must be a homogeneous function $\alpha(2\,y_i-1)=k\,(2\,y_i-1)^p$ with some constant $k\in\mathbb{R}$. In consequence, the information $I_{O\rightarrow S}(2\,\vec{y}_{O\rightarrow S}-\vec{1})$ is (up to $k$) the $p$-norm of the Bloch vector $2\,\vec{y}_{O\rightarrow S}-\vec{1}$.

We can rule out that $p\in(-\infty,0]$ because in this case, as one can easily check, it is impossible to satisfy all the consistency conditions (i)--(vi) of subsection \ref{sec_elcond}. Hence, $p>0$. At this stage we can make use of (\ref{infocons}) and a result by Aaronson \cite{Aaronson:fk} which implies that the only vector $p$-norm with $p>0$ which is preserved by a {\it continuous }matrix group is the $2$-norm. Since any given time evolution of the Bloch vector $2\,\vec{y}_{O\rightarrow S}-\vec{1}$ is governed by a continuous, one-parameter matrix group, we conclude that $\alpha(2\,y_i-1)=k\,(2\,y_i-1)^2$. Imposing condition (iv) yields $k=1$ and therefore ultimately
\ba
I_{O\rightarrow S}(\vec{y}_{O\rightarrow S})=\sum_{i=1}^{D_N}\,(2\,y_i-1)^2.\label{info2}
\ea
It is straightforward to convince oneself that all of (i)--(vi) of subsection \ref{sec_elcond} are satisfied by this quadratic information measure. $O$'s total amount of information about $S$ is thus the squared length of the generalized Bloch vector, thereby assuming a geometric flavour. It is important to emphasize that, had we not imposed continuity of time evolution in rule \ref{time}, we would not have been able to arrive at (\ref{info2}); if time evolution was {\it not} continuous, many solutions to $\alpha$ in terms of $y_i$ would be possible.


The quadratic information measure (\ref{info2}) has been proposed earlier by Brukner and Zeilinger in \cite{Brukner:1999qf,Brukner:ly,Brukner:2001ve,Brukner:ys} from a different perspective, emphasizing that this is the most natural measure taking into account an observer's uncertainty -- due to statistical fluctuations -- about the outcome of the next trial of measurements on a system in a multiple shot experiment. 
Furthermore, taking the formalism of quantum theory as given, Brukner and Zeilinger \cite{brukner2009information} later singled out the quadratic measure from the set of Tsallis entropies by imposing an `information invariance principle', according to which a continuous transformation among any two complete sets of mutually complementary measurements in quantum theory should leave an observer's information about the system invariant. While \cite{brukner2009information} is certainly compatible with the present framework, here we come from farther away to the same result: we do not pre-suppose quantum theory and derive the quadratic measure more generally by starting from the landscape of information inference theories and imposing rule \ref{pres} of information preservation and rule \ref{time} of maximality of time evolution thereon. In this regard, the present derivation may similarly be taken as a strong justification for the original Brukner-Zeilinger proposal.

\subsection{The set of all time evolutions is a subgroup of $\rm{SO}(D_N)$}\label{ssec_SO}

Thus far, we only gathered that a given time evolution is described by a one-parameter group and, by theorem \ref{thm_surj}, must have $|\det T(\Delta t)|=1$. Now, given (\ref{info2}), we are in the position to say quite a bit more. The full matrix group leaving (\ref{info2}) invariant is $\mathrm{O}(D_N)$. However, a given time evolution is a continuous one-parameter group and must therefore be connected to the identity. Hence, $T(\Delta t)\in\SO(D_N)$, $\forall\,T(\Delta t)$. Consequently, the group corresponding to any fixed time evolution is a one-parameter subgroup of $\SO(D_N)$. We also noted before that the set of all possible time evolutions $\ct_N$ is a group thanks to rule \ref{time}. It must therefore likewise satisfy $\ct_N\subset\SO(D_N)$. This topic is thoroughly discussed in the companion articles, where it is shown that the rules imply $\ct_N=\rm{PSU}(2^N)$ for the $D_1=3$ \cite{hw} and $\ct_N=\rm{PSO}(2^N)$ for the $D_1=2$ case \cite{hw2}. Note that $\rm{PSU}(2^N)$ is a proper subgroup of $\rm{SO}(D_N=4^N-1)$ for $N>1$. The generators of $\ct_N$ are the set of possible time evolution generators of $S$'s states.

\subsection{Pure and mixed states}\label{sec_puremixed}

The explicit quantification of $O$'s information now permits us to render the distinction between three informational classes of $S$'s states -- which we already loosely referred to as `states of maximal knowledge' or `states of non-maximal information' in previous sections -- precise. 

Firstly, we determine the maximally attainable (independent and dependent) information content within a state $\vec{y}_{O\rightarrow S}$ of a system of $N$ gbits. This can be easily counted: once $O$ knows the answers to $N$ mutually independent questions (these do not need to be individuals), he will also know the answers to all their bipartite, tripartite,... and $N$-partite correlation questions -- all of which are contained in $\cq_{M_N}$ too by theorems \ref{thm_qubit3} and \ref{thm_rebit3}. But these are then
\ba
\binom{N}{1}+\binom{N}{2}+\cdots\binom{N}{N}=\sum_{i=1}^{N}\,\binom{N}{i}=2^N-1\nn
\ea
answered questions from $\cq_{M_N}$, while all remaining questions in $\cq_{M_N}$ will be maximally complementary to at least one of the known ones. $O$'s total information, as quantified by $I_{O\rightarrow S}$ (\ref{eqn_newI}), thus contains plenty of dependent \texttt{bits} of information -- a result of the fact that the questions in $\cq_{M_N}$ are pairwise but not necessarily mutually independent.

Using this observation, 
we shall characterize $S$'s states according to their information content, i.e.\ squared length of the Bloch vector. By rules \ref{pres} and \ref{time}, this distinction applies to all states which are connected via some time evolution to the states above, including those for which the information may be distributed partially over many elements of a fixed $\cq_{M_N}$. Specifically, we shall refer to a state $\vec{y}_{O\rightarrow S}$ as a
\begin{widetext}

\begin{description}
\item {\bf pure state:} if it is a state of maximal information content, i.e.\ maximal length
\ba
I_{O\rightarrow S}(\vec{y}_{O\rightarrow S})=\sum_{i=1}^{D_N}\,(2\,y_i-1)^2=(2^N-1)\, \texttt{bits},\nn
\ea
\item {\bf mixed state:} if it is a state of non-extremal information content, i.e.\  non-extremal length
\ba
0\,\texttt{bit}<I_{O\rightarrow S}(\vec{y}_{O\rightarrow S})=\sum_{i=1}^{D_N}\,(2\,y_i-1)^2<(2^N-1)\, \texttt{bits}, \nn
\ea
\item {\bf totally mixed state:} if it is the state of no information $\vec{y}_{O\rightarrow S}=\f{1}{2}\vec{1}$ with zero length 
\ba
I_{O\rightarrow S}\left(\vec{y}_{O\rightarrow S}=\f{1}{2}\vec{1}\right)=0\, \texttt{bit}.\nn
\ea
\end{description}
\end{widetext}
We note that these characterizations of states in terms of their length are indeed true in quantum theory. In particular, $N$ qubit pure states actually have a Bloch vector squared length equal to $2^N-1$. Our reconstruction gives this peculiar fact a clear informational interpretation.

One may wonder whether the above definition of a pure state is directly equivalent to being extremal in $\Sigma_N$ and thereby to the usual definition of pure states in GPTs or quantum theory. This is not obvious from what we have established thus far. However, since clearly we assume $I_{O\rightarrow S}=(2^N-1)$ \texttt{bits} to be the maximally attainable information, we can already conclude that pure states so defined must lie on the boundary of $\Sigma_N$ because a non-constant convex function cannot have its maximum in the interior of its convex domain \cite{webster}. Showing that the pure states, as defined above, of the theory surviving the imposition of all rules -- quantum theory -- are, indeed, the extremal states, requires more work. For $N=1$ we shall demonstrate this shortly, while the discussion for $N>1$ is deferred to \cite{hw}.

\section{The $N=1$ case and the Bloch ball}\label{sec_n1}

Before closing this toolkit for now, we quickly give a flavor of the capabilities of the newly developed concepts and tools by applying them to show that in the simplest case of a single gbit ($N=1$) rules \ref{lim}--\ref{time} indeed only have two solutions within $\cl_{\rm gbit}$, namely the qubit and the rebit state space including their respective time evolution groups. This proves the claim of section \ref{sec_principles} for $N=1$.

To this end, recall theorem \ref{thm_3d} which asserts that the dimension of the $N=1$ state space $\Sigma_1$ is either $D_1=2$ or $D_1=3$ which thus far we suggestively referred to as the `rebit' and `qubit case', respectively.

\newpage
\begin{widetext}
\subsection{A single qubit and the Bloch ball}

We begin with the $D_1=3$ case. $\Sigma_1$ will be parametrized by a three-dimensional vector 
\ba
\vec{y}_{O\rightarrow S}=\left(\begin{array}{c}y_1 \\y_2 \\y_2\end{array}\right),\nn
\ea 
where $y_1,y_2,y_3$ are the `yes'-probabilities of a mutually maximally complementary question set $Q_1,Q_2,Q_3$ constituting an informationally complete $\cq_1$ (\ref{qm1}). 

The informational distinction of states introduced in section \ref{sec_puremixed} reads in this case


\begin{description}
\item {\bf pure states:} $\vec{y}_{O\rightarrow S}$ such that 
\ba
I_{O\rightarrow S}=(2\,y_1-1)^2+(2\,y_2-1)^2+(2\,y_3-1)^2=1\, \texttt{bit},\nn
\ea
\item {\bf mixed states:} $\vec{y}_{O\rightarrow S}$ such that 
\ba
0\, \texttt{bit}<I_{O\rightarrow S}=(2\,y_1-1)^2+(2\,y_2-1)^2+(2\,y_3-1)^2<1\, \texttt{bit},\nn
\ea
\item {\bf totally mixed state:} $\vec{y}_{O\rightarrow S}=\f{1}{2}\,\vec{1}$ such that 
\ba
I_{O\rightarrow S}=(2\,y_1-1)^2+(2\,y_2-1)^2+(2\,y_3-1)^2=0\, \texttt{bit}.\nn
\ea
\end{description}
\end{widetext}

Recall from section \ref{ssec_SO} that the set of possible time evolutions $\ct_1$ must be contained in $\SO(3)$ and that the time evolution rule \ref{time} requires the set of states into which {\it any} state $\vec{y}_{O\rightarrow S}$ can evolve to be maximal, while being compatible with the rules. Thus, in particular, the image of the certainly legal pure state $(1,0,0)$ (rules \ref{lim} and \ref{unlim} imply its existence in $\Sigma_1$) under $\ct_1$ must be maximal. Applying all of $\SO(3)$ to this state generates {\it all} states with $|2\vec{y}-\vec{1}|^2=1$ \texttt{bit} and these certainly abide by the rules. 
Consequently, the set of all possible time evolutions, compatible with rules \ref{lim}--\ref{time}, is the rotation group
\ba
\ct_1\simeq\rm{SO}(3)\simeq\rm{PSU}(2).\nn
\ea
This is the component of the isometry group of the Bloch ball which is connected to the identity and it is required in full in order to maximize the number of states into which $(1,0,0)$ can evolve. However, since time evolution is state independent (c.f.\ assumption \ref{assump7}), this is the time evolution group for all states. 
$\rm{PSU}(2)$ is precisely the adjoint action of $\rm{SU}(2)$ on density matrices $\rho_{2\times2}$ over $\mathbb{C}^2$, $\rho_{2\times2}\mapsto U\,\rho_{2\times2}\,U^\dag$, $U\in\rm{SU}(2)$ and thus coincides with the set of all possible unitary time evolutions of a single qubit in standard quantum theory.

The set of allowed states populates the entire unit ball in the three-dimensional Bloch ball, i.e.\ $\Sigma_1\simeq B^3$. This follows from the fact that applying the full group $\ct_1=\SO(3)$ to $(1,0,0)$ generates the entire Bloch sphere as a closed set of extremal states and the fact that $\Sigma_1$ must be closed convex according to assumption \ref{assump2}. Hence, we recover the well-known three-dimensional Bloch ball state space of a single qubit of standard quantum theory with the set of all pure states defining the boundary sphere $S^2$, the totally mixed state (as the state of no information) constituting the center and the set of mixed states filling the interior in between, as illustrated in figure \ref{fig_bloch}. This is precisely the geometry of the set of all normalized density matrices on $\mathbb{C}^2$. Notice that the pure state space $S^2\simeq\mathbb{C}P^1$ indeed coincides with the set of all unit vectors in $\mathbb{C}^2$ (modulo phase).

\begin{figure}[hbt!]
\begin{center}
\psfrag{cm}{totally mixed state}
\psfrag{m}{mixed states}
\psfrag{p}{pure states}
\psfrag{r}{$\vec{r}$}
\psfrag{y}{$\vec{r}=2\,\vec{y}_{O\rightarrow S}-\vec{1}$}
\vspace*{.4cm}\hspace*{-4cm}\begin{subfigure}[b]{.22\textwidth}
\centering
\includegraphics[scale=.3]{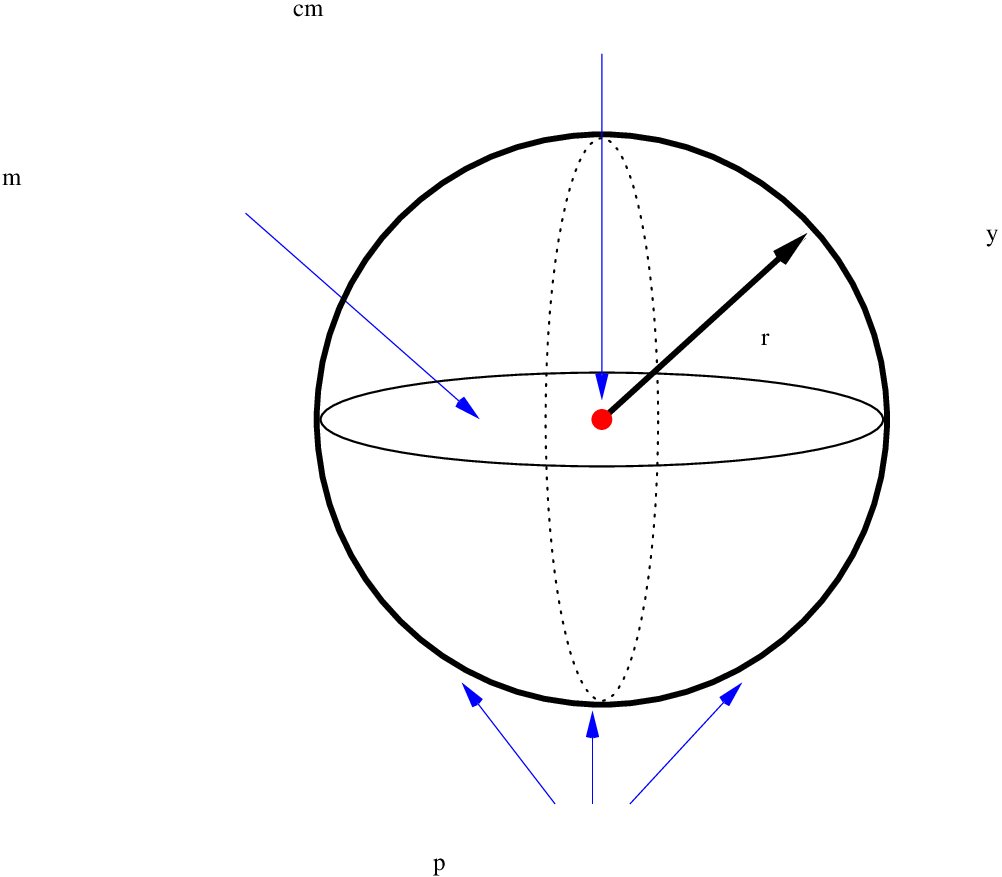}
\centering
\vspace*{.3cm}
\caption{\small 3D qubit Bloch ball}\label{fig_bloch}\end{subfigure}\\\vspace*{1cm}
\hspace*{-4cm}
\begin{subfigure}[b]{.22\textwidth}
\centering
\includegraphics[scale=.3]{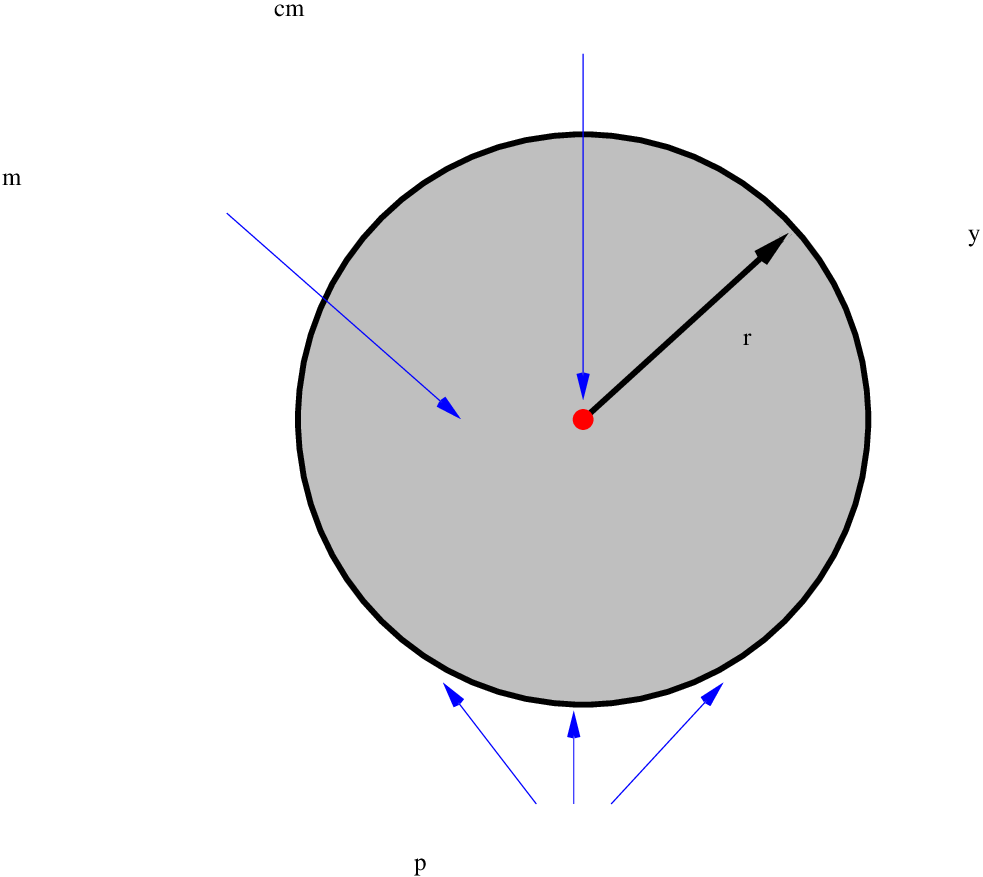}
\vspace*{.3cm}\caption{\small 2D rebit Bloch disc }\label{fig_blochrebit}
\end{subfigure}
\caption{\small The three-dimensional Bloch ball (a) and the two-dimensional Bloch disc (b) are the correct state spaces $\Sigma_1$ of a single qubit in standard quantum theory and a single rebit in real quantum theory, respectively. The vector $\vec{r}$ parametrizing the states is the Bloch vector $2\,\vec{y}_{O\rightarrow S}-\vec{1}$.}
\end{center}
\end{figure}

Given the complete symmetry of the Bloch ball as the state space for $N=1$, there should not exist a distinguished informationally complete question set $\cq_{M_1}$, corresponding to a distinguished orthonormal Bloch vector basis, by means of which $O$ can interrogate $S$. While it is an additional assumption, it is thus natural to stipulate that to {\it every} $Q\in\cq_1$ there exists a unique pure state in $\Sigma_1$ which represents the truth value $Q=$`yes' and, conversely, that every pure state of this system corresponds to the definite answer to one question in $\cq_1$. But then $\cq_1\simeq S^2$ which also coincides with the set of all possible projective measurements onto the $+1$ (or, equivalently, the $-1$) eigenspaces of the Pauli operators $\vec{n}\cdot\vec{\sigma}$ over $\mathbb{C}^2$ which is parametrized by $\vec{n}\in\mathbb{R}^3$, $|\vec{n}|^2=1$, and where $\vec{\sigma}=(\sigma_x,\sigma_y,\sigma_z)$ are the usual Pauli matrices. This set of permissible questions $\cq_N$ for all $N$ will be discussed more thoroughly in \cite{hw} together with a derivation of the Born rule for projective measurements.

The `ballness' and three-dimensionality of the state space of a single qubit can also be derived from various operational axioms within GPTs \cite{Dakic:2009bh,Masanes:2011kx,masanes2011derivation,Masanes:2012uq,Mueller:2012pc,Barnum:2014fk,Garner:2014uq}
and constitutes a crucial step in most GPT based reconstructions of quantum theory \cite{Hardy:2001jk,Dakic:2009bh,de2012deriving,masanes2011derivation,Mueller:2012pc}. The principle of {\it continuous reversibility}, according to which every pure state of the convex set can be mapped into any other by means of a continuous and reversible transformation, usually assumes a crucial role in such derivations. Here we offer a novel perspective on the origin of the Bloch ball by deriving it from elementary rules for the informational relation between $O$ and $S$; in particular, we recover continuous reversibility as a by-product.

\subsection{A single rebit and the Bloch disc}

The analogous result holds for the $D_1=2$ case: $\Sigma_1$ will be parametrized by a two-dimensional vector 
\ba
\vec{y}_{O\rightarrow S}=\left(\begin{array}{c}y_1 \\y_2\end{array}\right),\nn
\ea 
where $y_1,y_2$ are the `yes'-probabilities of two maximally complementary questions $Q_1,Q_2$ constituting an informationally complete $\cq_1$ (\ref{qm1}). We then have
\begin{description}
\item {\bf pure states:} $\vec{y}_{O\rightarrow S}$ such that
\ba
I_{O\rightarrow S}=(2\,y_1-1)^2+(2\,y_2-1)^2=1\, \texttt{bit},\nn
\ea
\item {\bf mixed states:} $\vec{y}_{O\rightarrow S}$ such that 
\ba
\!\!\!\!\!\!\!\!\!0\, \texttt{bit}<I_{O\rightarrow S}=(2y_1-1)^2+(2y_2-1)^2<1\, \texttt{bit},\nn
\ea 
and the
\item {\bf totally mixed state:} $\vec{y}_{O\rightarrow S}=\f{1}{2}\,\vec{1}$ such that 
\ba
I_{O\rightarrow S}=(2\,y_1-1)^2+(2\,y_2-1)^2=0\, \texttt{bit}.\nn
\ea
\end{description}
In analogy to the qubit case, the time evolution rule \ref{time} implies that
\begin{itemize}
\item[(a)] the set of all possible time evolutions is the (projective) rotation group
\ba
\ct_1\simeq\rm{SO}(2)\overset{\tiny{\#}}{\simeq}\rm{PSO}(2)\simeq\rm{SO}(2)/\mathbb{Z}_2\nn
\ea
because $\rm{SO}(2)$ is the full (connected component of the orientation preserving) isometry group of the Bloch disc and all time evolutions are permitted which preserve the Bloch vector length, representing $O$'s total information about $S$, and the orientation of the Bloch vector basis. Orientation preservation means that $O$'s convention about the `yes'-`no'-labelling of the question outcomes is preserved: 
The isomorphism denoted with a \# requires some explaining because $\rm{SO}(2)$ is actually a double cover of $\rm{PSO}(2)$ as indicated on the right. Indeed, the real density matrix of a single rebit, $\rho_{2\times2}=\f{1}{2}(\mathds{1}+r_x\sigma_x+r_z\sigma_z)$ on $\mathbb{R}^2$ where $r_i=(2y_i-1)$, evolves under the adjoint action of $\rm{SO}(2)$, $\rho_{2\times2}\mapsto O\,\rho_{2\times2}\,O^T$ for $O=e^{i\sigma_y t}\in\rm{SO}(2)$. This is equivalent to an action of $\rm{PSO}(2)$ because the non-trivial center element $-\mathbb{1}\in\rm{SO}(2)$ acts trivially on $\rho_{2\times2}$ and thus factors out. However, as a subcase of the single qubit case this adjoint action of $\rm{SO}(2)$ on $\rho_{2\times2}$ is also equivalent to $O'\,\vec{r}$ for $O'\in\rm{SO}(2)$ in the Bloch vector representation.\vspace*{-.1cm}

\item[(b)]  the state space $\Sigma_1$ coincides with the two-dimensional Bloch disc, as depicted in figure \ref{fig_blochrebit}, with the totally mixed state in the center, the pure states on the boundary circle $S^1$ and the mixed states in the interior in between. This is precisely the geometry of the set of unit trace, positive-semidefinite, symmetric matrices over $\mathbb{R}^2$ -- the space of density matrices of a single rebit. Similarly, the pure state space $S^1\simeq\mathbb{R}P^1$ coincides with the set of all unit vectors in the Hilbert space $\mathbb{R}^2$ modulo $\mathbb{Z}_2$.
\end{itemize}

Again, given the symmetry of the Bloch disc, there is no reason for a distinguished informationally complete question basis $\cq_{M_1}$ on $\cq_1$, incarnated as a distinguished Bloch vector basis, to exist. Accordingly, we assume that {\it any} pure state on $S^1$ corresponds to the definite answer of some question in $\cq_1$ which $O$ may ask the rebit. In this case, $\cq_1\simeq S^1$, which coincides with the set of projective measurements onto the $+1$ eigenspaces of the Pauli operators $\vec{n}\cdot\vec{\sigma}$ over $\mathbb{R}^2$ which are parametrized by $\vec{n}\in\mathbb{R}^2$, $|\vec{n}|^2=1$ and where $\vec{\sigma}=(\sigma_x,\sigma_z)$ are the two real and symmetric Pauli matrices.

\section{Discussion and outlook}\label{sec_conc}

Based on the premise to only speak about information which an observer (or more generally system) has access to by interaction with other systems, we have developed a novel framework for characterizing and (re-)constructing the quantum theory of qubit systems. 
We have laid down, without ontological statements, the mathematical and conceptual foundations for a landscape of theories describing how an observer acquires information about physical systems through interrogation with yes-no-questions.


 Within this landscape four elementary rules have been given which constrain the observer's acquisition of information for qubit (and rebit) systems. The rule of limited information and the complementarity rule imply an independence and compatibility structure of the binary questions which, in fact, reproduces that of projective measurements onto the $+1$ eigenspaces of Pauli operators in quantum theory. In particular, these rules entail in a constructive and simple way
\begin{enumerate}
\item a novel argument for the dimensionality of the Bloch ball, 
\item a new method for determining the correct number of independent questions/measurements necessary to describe a system of $N$ qubits (or rebits), 
\item a natural explanation for entanglement, monogamy of entanglement and quantum non-locality,
\item the explicit correlation structure of two qubits and rebits, and 
\item more generally the correlation structure for arbitrarily many qubits and rebits.
\end{enumerate}

\noindent Furthermore, the rules of information preservation and maximality of time evolution are shown to result
\begin{description}
\item  \hspace*{.2cm}6. in a reversible time evolution, and
\item  \hspace*{.2cm}7. under elementary consistency conditions, in a {\it quadratic} information measure, quantifying the observer's prior information about the answers to the various questions he may ask the system. 
\end{description}
This measure has been earlier proposed by Brukner and Zeilinger from a different perspective \cite{Brukner:1999qf,Brukner:ly,Brukner:2001ve,Brukner:ys,brukner2009information}, complementing our present derivation. 

Combining these results, we then show, as the simplest example, how the rules entail that
\begin{description}
\item\hspace*{.2cm}8. the Bloch ball and Bloch disc are recovered as state spaces for a single qubit and rebit, respectively, together with the correct time evolution groups and question sets. 
\end{description}

The full reconstruction of qubit quantum theory, following from the present four rules (and two additional ones), is completed in the companion paper \cite{hw} (the rebit case which violates one of the additional rules is considered separately \cite{hw2}). In conjunction, this derivation highlights the {\it partial} interpretation of quantum theory as a law book, describing and governing an observer's acquisition of information about physical systems. In particular, it highlights the quantum state as the observer's `catalogue of knowledge'.
 
Certainly, there are some limitations to the present approach: First of all, we employ a clear distinction between observer and system which cannot be considered as fundamental. Secondly, the construction is specifically engineered for the quantum theory of qubit systems. While arbitrary finite dimensional quantum systems could, in principle, be immediately encompassed by imposing the so-called {\it subspace axiom} of GPTs \cite{Hardy:2001jk,masanes2011derivation}, the latter does not naturally fit into our set of principles, mostly because rules \ref{lim} and \ref{unlim} quantify the information content of a system in terms of $N\in\mathbb{N}$ bits which is only suitable for composite gbit systems. More generally, the limitation is that the current approach only encompasses finite dimensional quantum theory, but not quantum {\it mechanics}. As it stands, the mechanical phase space language does not naturally fit into the present framework and more sophisticated tools are required in order to cover mechanical systems, let alone anything beyond that.

While this informational construction recovers the state spaces, the set of possible time evolutions and projective measurements of qubit quantum theory, in other words, the architecture of the theory, it clearly does not tell us much about the concrete physics (other than demanding it to fit into this general construction). This purely informational framework is rather universal and information carrier independent. But qubits as information carriers can be physically incarnated in many different ways: as electron or muon spins, photon polarization, quantum dots, etc. The framework cannot distinguish among the different physical incarnations, underlining the observation that not everything can be reduced to information and additional inputs are necessary in order to do proper physics. 

Nevertheless, despite its current limitations, this 
informational 
approach teaches us something non-trivial about the structure and physical content of quantum theory.

While this work is more generally motivated by the effort to understand physics from an informational and operational perspective and highlights a {\it partial} interpretation of quantum theory, it clearly does not single out `the right one' (see also the discussion in \cite{Hoehn:2017gst}). For example, by speaking exclusively about the information accessible to an observer, we are by default silent on the fate of hidden variables and on whether they could give rise to the assumptions and rules which we impose. The status of hidden variables is simply not relevant here (other than that {\it local} hidden variables are ruled out).

Let us, nonetheless, make a few possible (but not inescapable) interpretational statements. 
Given the absence of ontological commitments, this informational approach is generally compatible with (but does {\it not} rely on) the relational \cite{Rovelli:1995fv,Smerlak:2006gi}, Brukner-Zeilinger informational \cite{zeilinger1999foundational,Brukner:1999qf,Brukner:vn,Brukner:ys,Brukner:2002kx}, or QBist \cite{Fuchs:fk,Caves:1996nq,Caves:2002uq,caves2002unknown} interpretations of quantum mechanics. They depart from the traditional idea that systems necessarily have, absolute, i.e.\ observer independent properties (or, more generally, properties independent of their relations with other systems). Instead, many physical properties are interpreted as relational; the interaction between systems establishes a relation between them, permitting an information exchange which reveals certain physical properties relative to one another and in the absence of hidden variables this would be all there is.  Certainly, in order not to render such a view hopelessly solipsistic, systems ought to have certain intrinsic attributes, e.g.\ a corresponding state space or set of permissible interactions/measurements, such that different observers have a basis for agreeing or disagreeing on the description of physical objects. However, a state of a system or measurement/question outcome is taken relative to whoever performs the measurement or asks the question. Different observers may agree on states or measurement outcomes by communication (i.e., physical interaction) but if one rejects the idea of an absolute and omniscient observer it is natural to also abandon the idea of an absolute and external standard by means of which properties of systems could be defined. 


From such a perspective, one could say that the relational character of the qubits' properties is a consequence of a universal limit on the amount of information accessible to the observer and the mere existence of complementary information such that the observer can {\it not} know the answers to all his questions simultaneously. It is the observer who determines which questions he will ask and thereby what kind of system property the interrogation will reveal and thus, ultimately, which kind of information he will acquire (although clearly he does not determine what the question outcome is). But relative to the observer the system of qubits does not have properties other than those accessible to him.

\subsection{An operational alternative to the wave function of the universe}

Pushing an informational interpretation of quantum theory to the extreme, one may speculate whether the quantum state could also represent a state of information in a gravitational or cosmological context. For instance, is such an interpretation adequate for the `wave function of the universe' which is ubiquitous in standard approaches to quantum cosmology (e.g., see \cite{Hartle:1983ai,Bojowald:2011zzb,Ashtekar:2011ni})? Such an interpretation would require the existence of an absolute and omniscient observer, an idea which we just abandoned. 

Alternatively, one could adopt one of the central ideas of relational quantum mechanics \cite{Rovelli:1995fv,Smerlak:2006gi}, according to which all physical systems can assume the role of an `observer', recording information about other systems, thereby relieving the clear distinction used in this manuscript. Extending this idea to a space-time context, one could interpret the universe as an abstract network of subsystems/subregions, viewed as information registers, which can communicate and exchange information through interaction (see figure \ref{fig_global}). 
 \begin{figure}[hbt!]
\begin{center}
\psfrag{a}{absolute observer}
\psfrag{s1}{$S_1$}
\psfrag{s2}{$S_2$}
\psfrag{s3}{$S_3$}
\psfrag{s4}{$S_4$}
\psfrag{s5}{\tiny$S_5$}
\psfrag{s6}{\tiny$S_6$}
\psfrag{s7}{\hspace*{-.08cm}\tiny$S_7$}
\psfrag{c}{\hspace*{-.2cm}communication}
{\includegraphics[scale=.3]{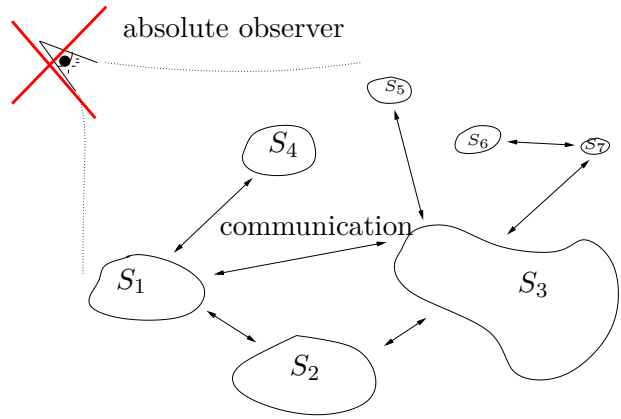}}
\caption{\small The universe as an information exchange network of subsytems without absolute observer.}\label{fig_global}
\end{center}
\end{figure}
In this background independent context, any information acquisition by any register is {\it internal}, i.e.\ occurs within the network; a global observer outside the network becomes meaningless. This may appear as a purely philosophical observation, but it implies concrete consequences for the description of the network: {\it there should be no global state} (aka `wave function of the universe') for the entire network at once. Indeed, admitting any register in the network to act as `observer', the {\it self-reference problem} \cite{dalla1977logical,breuer1995impossibility} impedes a given register to infer the global state of the entire network -- including itself -- from its interactions with the rest. Accordingly, relative to any subsystem, one could assign a state to the rest of the network, but a global state and thus a global Hilbert space would not arise. The absence of a global state in quantum gravity has been proposed before \cite{Crane:1995qj,Markopoulou:1999cz,Markopoulou:2002ru,Markopoulou:2007qg} -- albeit from a different, less informational perspective. 

This offers an operational alternative to the problematic concept of the `wave function of the universe' which is ubiquitous in quantum cosmology. However, while a `wave function of the universe' might not make sense as a fundamental concept, it could still be given meaning as an effective notion, namely as the state of the large scale structure of our universe on whose description `late time' observers better agree. 

Clearly, if this was to offer a coherent picture of physics, there would need to exist non-trivial consistency relations among the different registers' descriptions such as, e.g., the requirement that `late time observers' agree on the state of the large scale structure. This is not a practically unrealistic expectation as 
a concrete playground for this idea has recently been constructed from a different motivation \cite{Hackl:2014txa}: a scalar field on the background geometry of elliptic de Sitter space can only be quantized in an observer dependent manner. A global Hilbert space for the quantum field does not exist, but consistency conditions between different observers' descriptions can be derived. (Observer consistency has also been discussed in general terms within Relational Quantum Mechanics \cite{Rovelli:1995fv,Smerlak:2006gi} and QBism \cite{Cavescompat,mermin2002compatibility}.) Although not being a quantum gravity model, it may serve as a platform for concretizing this proposal further.


\section*{Acknowledgements}
This work has benefited from interactions with many colleagues. First and foremost, it is a pleasure to thank Markus M\"uller for uncountably many insightful discussions, patient explanations of his own work and suggestions. I am grateful to Chris Wever for collaboration on \cite{hw,hw2} and important discussions, to Sylvain Carrozza for drawing my attention to Relational Quantum Mechanics in the first place and many related conversations, to Tobias Fritz  and Lucien Hardy for asking numerous clarifying questions, and to Rob Spekkens for many careful comments and suggestions. I would also like to thank Caslav Brukner, Borivoje Dakic, Martin Kliesch, Matt Leifer, Miguel Navascu\'es, Matteo Smerlak and Rafael Sorkin for useful discussions. Research at Perimeter Institute is supported by the Government of Canada through Industry Canada and by the Province of Ontario through the Ministry of Research and Innovation. The project leading to this publication has also received funding from the European Union's Horizon 2020 research and innovation programme under the Marie Sklodowska-Curie grant agreement No 657661.

\providecommand{\href}[2]{#2}\begingroup\raggedright\endgroup

\end{document}